\documentclass[%
 reprint,
 superscriptaddress,
 amsmath,amssymb,
 aps,
 pra,
 longbibliography,
 preprintnumbers
]{revtex4-2}

\usepackage{graphicx}
\usepackage{amsmath}
\usepackage{dcolumn}
\usepackage{bm}

\usepackage[utf8]{inputenc}
\usepackage[english]{babel}
\usepackage[T1]{fontenc}

\usepackage{tikz}
\usepackage{lipsum}

\usepackage{amsmath,amsfonts,amssymb,amsthm,ascmac}
\allowdisplaybreaks 
\usepackage{mathtools}
\usepackage{mathrsfs}
\usepackage{siunitx}
\usepackage{bm}
\usepackage{cancel}

\usepackage{enumerate}
\usepackage{siunitx}
\usepackage{physics}
\usepackage{pxrubrica}
\usepackage[version=3]{mhchem}
\usepackage{array}
\usepackage{float}

\usepackage{booktabs}
\usepackage{longtable} 
\usepackage{multirow}
\usepackage{hhline}
\usepackage{subcaption}
\usepackage{graphicx}
\usepackage{tabularx}
\usepackage{tikz}
\usetikzlibrary{intersections,calc,patterns,through,positioning,arrows,backgrounds,cd}
\usepackage{pgfplots}
\usepgfplotslibrary{patchplots}
\pgfplotsset{compat=1.15}
\usepackage{titlesec}
\usepackage{picture}
\usepackage{fancybox}
\usepackage{boites}
\usepackage{tcolorbox}
\usepackage{dsfont}
\usepackage{quantikz}
\usepackage[normalem]{ulem}
\usepackage{titletoc}
\usepackage{enumitem}

\usepackage{algorithm}
\usepackage{algpseudocode}
\usepackage{listings}
\usepackage{mathtools}
\usepackage{amsmath,amssymb,booktabs}
\usepackage{xcolor}
\usepackage{fancybox}
\usepackage{stmaryrd}
\usepackage{ytableau}
\usepackage{dsfont}
\usepackage{physics}
\usepackage{ragged2e}
\usepackage[offset=1.2em,sep=4pt]{simpler-wick}

\usepackage[colorlinks,linkcolor=blue,citecolor=blue]{hyperref}

\makeatletter
\@ifundefined{longtable*}{}{%
  \expandafter\let\csname longtable*\endcsname\relax
  \expandafter\let\csname endlongtable*\endcsname\relax
}
\makeatother
\usepackage[justification=raggedright, singlelinecheck=true]{caption}

\definecolor{revisionpurple}{RGB}{128,0,128}
\newenvironment{revision}
  {\begingroup\color{black}}
  {\endgroup}

\providecommand{\red}[1]{\textcolor{black}{#1}}

\providecommand{\purple}[1]{\textcolor{black}{#1}}
\providecommand{\mac}[1]{\mathcal{#1}}
\providecommand{\mtr}[1]{\mathrm{#1}}

\providecommand{\sgn}{\mathrm{sgn} }
\providecommand{\Gr}{\operatorname{Gr}}

\providecommand{\coloneq}{:=}

\newcolumntype{L}[1]{>{\begin{minipage}[t]{#1}\raggedright\arraybackslash}l<{\end{minipage}}}
\newtheorem{theorem}{Theorem}
\newtheorem{definition}[theorem]{Definition}
\newtheorem{lemma}[theorem]{Lemma}
\newtheorem{corollary}[theorem]{Corollary}

\theoremstyle{remark}

\newcommand{\1}{\mathds{1}}

\newcommand{\Wg}{{\mathrm{Wg}}}

\newcommand{\fall}[2]{(#1)_{#2}}
\usepackage{mathtools}
\newcommand{\given}{\nonscript\,\vert\nonscript\,\mathopen{}}

\DeclareRobustCommand{\erase}{\bgroup\markoverwith{\textcolor{red}{\rule[.5ex]{2pt}{0.4pt}}}\ULon}

\makeatletter
\newcommand{\extp}{\@ifnextchar^\@extp{\@extp^{\,}}}
\def\@extp^#1{\mathop{\bigwedge\nolimits^{\!#1}}}
\makeatother

\begin{document}

\preprint{YITP-26-104}
\title{\mbox{Provably Efficient Learning of Fermionic Correlations under Particle-Number Symmetry}}

\author{Yuki Koizumi}
\email{koizumiyuki903@gmail.com}
 \affiliation{Department of Physics, The University of \mbox{Tokyo, 7-3-1 Hongo, Bunkyo-ku, Tokyo 113-0033, Japan}}

\author{Kaito Wada}
 \affiliation{\mbox{International Center for Elementary Particle Physics, The University of Tokyo, 7-3-1 Hongo, Bunkyo-ku, Tokyo 113-0033, Japan}}
 \affiliation{\mbox{Graduate School of Science and Technology, Keio University, Hiyoshi 3-14-1, Kohoku, Yokohama 223-8522, Japan}}

\author{Toshinori P. Takama}
\affiliation{Center for Gravitational Physics and Quantum Information,
Yukawa Institute for Theoretical Physics, Kyoto University,
Kitashirakawa Oiwakecho, Sakyo-ku, Kyoto 606-8502, Japan}

\author{Nobuyuki Yoshioka}
\email{ny.nobuyoshioka@gmail.com}
\affiliation{\mbox{International Center for Elementary Particle Physics, The University of Tokyo, 7-3-1 Hongo, Bunkyo-ku, Tokyo 113-0033, Japan}}

\begin{abstract}
    Predicting local fermionic correlations is a central task in quantum many-body physics, as these correlations encode many physically relevant local observables. The ubiquitous particle-number symmetry imposes strong structural constraints on quantum states, suggesting that local correlations should be learned with fewer samples than by symmetry-agnostic approaches. However, it has remained unclear whether such a provable advantage exists in collective learning of local correlations. Here, we develop a framework of number-conserving fermionic-shadow tomography based on random orbital rotations. We prove that, for every given order $k$, we can simultaneously estimate {\it all} $k$-body fermionic correlations of an $N$-mode $\eta$-particle state with a given variance $\varepsilon^2$ using only $O_k(\eta^k/\varepsilon^2)$ samples, which are independent of the system size $N$. We further establish a matching information-theoretic lower bound $\Omega_k(\eta^k/\varepsilon^2)$ for any adaptive protocol based on single-copy measurements, showing that the $(\eta^k,\varepsilon)$-dependence is optimal up to constants depending only on $k$. Furthermore, numerical studies show a 20-fold query reduction for one-body correlation estimation at $N=200$, $\eta=20$, and $\varepsilon=10^{-2}$, compared with the best alternative including Heisenberg-limited estimation methods. Relative to fermionic Gaussian-unitary shadows, sample costs for the square-lattice Fermi-Hubbard model at $N=288$ are reduced by factors of $148$ for energy density and $161$ for spin structure factor estimation. For molecular energy derivatives with respect to atomic coordinates, our method reduces sample requirements by factors of $53$ for $\mathrm{Cr}_2$ and $18.5$ for $[\mathrm{Fe}_2\mathrm{S}_2]^{2-}$ at $N=128$. These results demonstrate the practical advantage of particle-number symmetry for fermionic observable estimation.
\end{abstract}

\maketitle

\addtocontents{toc}{\protect\setcounter{tocdepth}{-10}}

\section{Introduction} 
Learning fermionic many-body systems from measurement data is a basic
primitive in quantum many-body physics and quantum simulation.  A central
instance of this task is the estimation of low-order fermionic correlation
functions, corresponding to $k$-body reduced density matrices ($k$-RDMs)~\cite{lowdin1955quantum},  which determine many relevant observables.
\begin{revision}
For instance, 2-RDM provide access to energy derivatives required for molecular geometry
optimization~\cite{obrien2019calculating,yalouz2022analytical,obrien2022efficient} and information about entanglement structure~\cite{cheong2004many,gullans2019entanglement,grover2013entanglement}.
They also encode pairing correlations that characterize superconducting
order~\cite{yang1962concept} and spin correlations such as spin strcture factors~\cite{varney2009quantum,mazurenko2016experimental}.
$3$-RDM determines the expectation values of effective
three-body interactions in transcorrelated
Hamiltonians~\cite{tsuneyuki2008transcorrelated,cohen2019similarity,dobrautz2024toward},
furthermore $4$-RDMs also enter internally contracted multireference
configuration interaction~\cite{saitow2013multireference,saitow2015fully}.
Reducing the measurement cost of estimating these quantities is an
essential step toward using quantum computers to address grand challenges
in physics and chemistry, including understanding emergent phases of strongly correlated matters, as well as predicting molecular structure and chemical reaction from first principles.
\end{revision}

From this motivation, given many identical copies of an unknown fermionic state, one aims to estimate all \(k\)-body correlation functions. 
Measurement-scheduling methods and optimized fermion-to-qubit mappings provide important baselines for fermionic learning task~\cite{bonet2020nearly,jiang2020optimal}. 
Classical shadow algorithm tailored to fermionic systems further exploit the algebraic structure of fermionic operators~\cite{Zhao:2020vxp,wan2023matchgate,ogorman2022fermionic, king2025triply}.
From these lines of work, it is known that existing fermionic learning protocols can achieve \(\mathcal O_k( N^k)\) sample complexity for learning all \(k\)-body correlations of an \(N\)-mode system with a fixed error, where \(\mac{O}_k (\cdot)\) suppresses constants depending only on the degree \(k\).
Moreover, Zhao et al.~\cite{Zhao:2020vxp} gave a simple argument regarding its optimality: the $k$-RDM contains $\mac O(N^{2k})$ independent entries, whereas only $\mac O(N^k)$ commuting observables can be accessed within a single measurement setting. This implies an $\mac{O}(N^k)$ requirement on the number of measurement settings. 
Thus, without additional structure, one should not
expect an \(N\)-independent single-copy measurement protocol for the fermionic
learning task.

In this context, it is  tempting to ask whether particle-number symmetry can remove this $N$-dependence. Such a setup is ubiquitously considered in quantum information science including quantum simulation~\cite{yoshioka2025krylov, xu2025neutralatoma, hartnett2026fast} as well as quantum interferometry~\cite{gong2021quantum}, and thus
many practical quantum algorithms are designed to exploit the reduced effective Hilbert space~\cite{babbush2018exponentially, babbush2019quantum, Koizumi2026Faster, Koizumi2026Heiseberg}. 
While this naturally motivates fermionic learning protocols with particle-number symmetry, it is highly nontrivial to 
leverage symmetry into an intrinsic learning advantage. 
Indeed, most existing symmetry-aware protocols do not eliminate the dependence on the number of modes~\cite{Zhao:2020vxp, hearth2024efficient, zhao2024group}.  

A particularly appealing exception is Low's orbital-rotation protocol~\cite{low2022classical}, whose estimator may have variance depending only on the particle number \(\eta\).
However, due to the difficulty of analyzing the randomized measurement
channel, the estimation variance in Ref.~\cite{low2022classical} was evaluated only after averaging over the target observables.
Recent works have made
progress on related symmetry-aware learning problems, including more structured
fermionic observables and fermionic linear-optics or Slater-determinant
learning settings~\cite{christensen2026learning, west2026particle}.  Nevertheless,
these results address restricted learning tasks rather than a general
entrywise guarantee for arbitrary \(k\)-RDM entries of an unknown
fixed-particle-number state.  Thus, it remains open whether particle-number
symmetry yields an \(N\)-independent sample-complexity guarantee for estimating all fermionic correlations simultaneously. 
\\

\begin{figure*}
\hspace{-1cm}
\captionsetup{type=table}
\caption{
\justifying{
Sample-complexity comparison for estimating all \(k\)-RDM elements of an \(N\)-mode  fermionic state to a fixed additive error \(\varepsilon\), comparing existing fermionic learning protocols~\cite{bonet2020nearly, jiang2020optimal,
Zhao:2020vxp, wan2023matchgate,  low2022classical} with our proposal.  The protocols are grouped according to whether they incorporate particle-number symmetry {with $\eta$ particles}.
Here, \(\mac O_k(\cdot)\) and \(\Omega_k(\cdot)\) suppress constants depending only on \(k\). 
}
}
\includegraphics[width=0.95\linewidth]{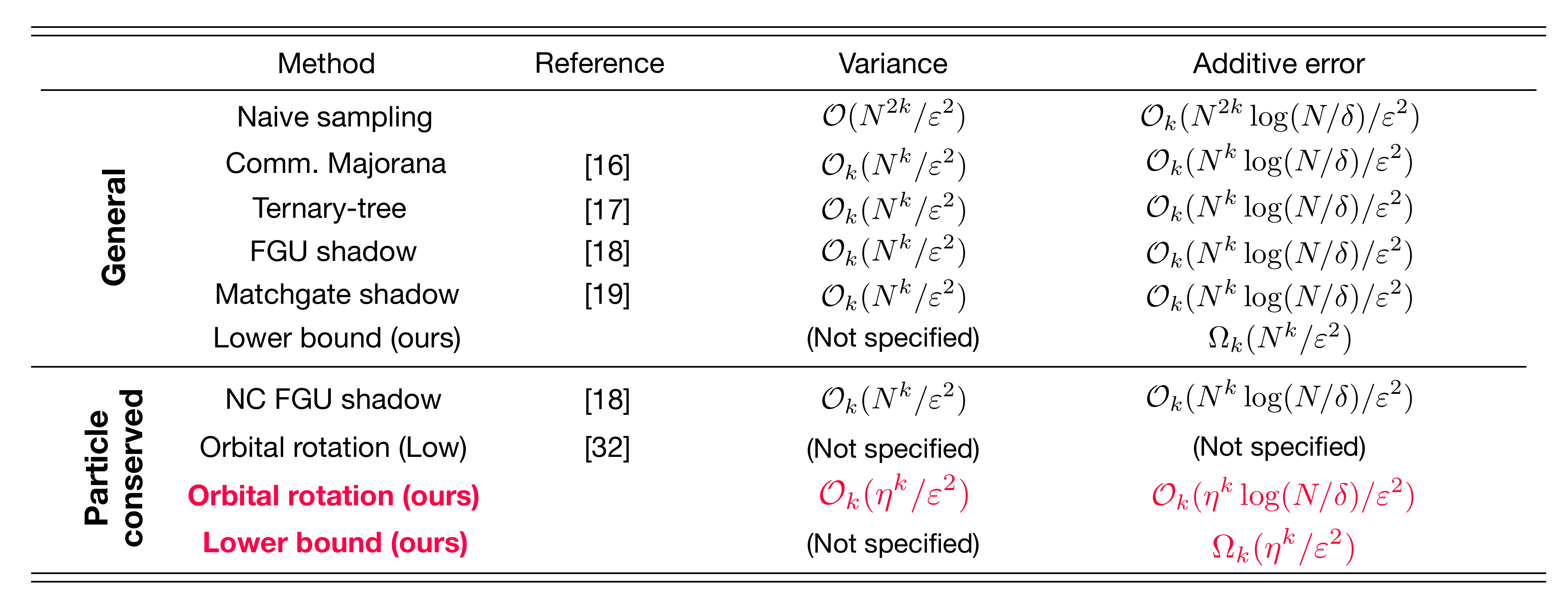}
\label{tab:sample_complexity}
\end{figure*}

An outstanding question {in this context} is as follows.
\begin{quote}
\begin{center}
    {\it Does fermionic number conservation allow a provable advantage in learning tasks?}
\end{center}    
\end{quote} 
Our contributions are summarized threefold below. 

First, we give a rigorous justification of the inverse measurement channel used in orbital rotation fermionic shadows. We identify a technical gap in the existing tomographic-completeness proof in Ref.~\cite{low2022classical} and resolve it by constructing the necessary operator basis directly from valid number-conserving measurement outcomes.

Second, we establish rigorous performance guarantees for orbital-rotation 
shadows in partial fermionic tomography. 
For every given order $k$, we can simultaneously estimate {\it all} $k$-body fermionic correlations of an $N$-mode $\eta$-particle state 
with a given variance $\varepsilon^2$ using only $\mathcal{O}_k(\eta^k/\varepsilon^2)$ samples, which are independent of the system size $N$.
We further prove matching information-theoretic lower bound
\(\Omega_k(\eta^k/\varepsilon^2)\) for any adaptive protocol based on single-copy measurements, where \(\varepsilon\) denotes the standard additive
error, which implies that the \((\eta^k,\varepsilon)\)-dependence is optimal up to constants depending only on \(k\).
For \(k=1\), we further derive an explicit closed-form variance formula.
\begin{revision}
We also provide a systematic procedure for deriving exact variance and covariance formulas for the $1$- and $2$-RDM estimators, including cross-covariances, enabling accurate evaluation of estimator variances and, consequently, the sample costs required to estimate one- and two-body observables.
\end{revision}
Table~\ref{tab:sample_complexity} highlights this particle-number-dependent 
improvement over existing methods.

\begin{revision}
Third, we demonstrate the practical impact of particle-number symmetry through numerical benchmarks in three settings: $1$-RDM estimation, estimation of the energy density and antiferromagnetic spin structure factor in Fermi-Hubbard model, and estimation of molecular energy-derivative expectation values for $\mathrm{Cr}_2$~\cite{larsson2022chromium} and $[\mathrm{Fe}_2\mathrm{S}_2]^{2-}$~\cite{georges2025quantum} molecules. These results demonstrate the utility of orbital-rotation shadows across a range of fermionic measurement tasks.  
\end{revision}

\section{$k$-RDM learning problem}

Let \(\rho\) be an \(N\)-mode \(\eta\)-particle fermionic state, and let $\mathcal{H}_{\eta}$ be the $\eta$-particle subspace.
For an integer
\(1\le k\le \eta\), let
\(\mathcal S_{N,k}\) denote the set of increasing \(k\)-tuples of mode
indices.  For \(\vec p=(p_1<\cdots<p_k)\) and
\(\vec q=(q_1<\cdots<q_k)\) in \(\mathcal S_{N,k}\), the
\((\vec p,\vec q)\)-entry of the \(k\)-RDM ${}^k{D}$ is defined by
\begin{align}
    \purple{{}^{k}D_{\mkern6mu\vec q}^{\vec p}}
    \coloneq
        \operatorname{tr}
        \left[
            \rho\,
            a_{p_1}^\dagger \cdots a_{p_k}^\dagger
            a_{q_k} \cdots a_{q_1}
        \right],
\end{align}
where \(a_i^\dagger\) and \(a_i\) are the fermionic creation and
annihilation operators for mode \(i\).  We assume access to independent
copies of \(\rho\).  Our goal is to estimate all entries
\(\purple{{}^{k}D_{\mkern6mu\vec q}^{\vec p}}\), with
\(\vec p,\vec q\in\mathcal S_{N,k}\), to additive error \(\varepsilon\in(0,1)\) with high probability $1-\delta \in (0,1)$, using as few copies of \(\rho\) as possible.

\section{Main Results} 

Before diving into our main results, it is useful to describe the orbital-rotation classical-shadow protocol originally introduced in Ref.~\cite{low2022classical}:
\begin{enumerate}[
    label=\textbf{Step~\arabic*:},
    labelsep=0.6em,
    leftmargin=*,
    itemsep=2pt
]
    \item Sample a single-particle unitary $u$ from
    \(\mathrm{Haar}(\mtr{U}(N))\) and apply the induced orbital
    rotation \(U_\eta(u)\) to the $N$-mode \(\eta\)-particle state \(\rho\).
    \item Measure the rotated state in the occupation-number basis and
    obtain an outcome \(\vec z\in\mac{S}_{N,\eta}\).
    \item From the classical measurement data \((u,\vec z)\), construct
    the single-shot classical shadow \(\widehat{\rho}\).
\end{enumerate}
Here, \(U_\eta(u)\) denotes the \(\eta\)-particle orbital rotation induced
by the single-particle unitary \(u\).  Its action on the occupation-number
basis is given by
\begin{align}
    U_\eta(u)\ket{\vec z}
    =
    \sum_{\vec y\in\mac S_{N,\eta}}
    \det(u_{\vec y,\vec z})\ket{\vec y},
\end{align}
where \(u_{\vec y,\vec z}\) is the submatrix of \(u\) with rows indexed by
\(\vec y\) and columns indexed by \(\vec z\).  A generic orbital rotation can
be decomposed into \(\mac{O}(N^2)\) two-mode Givens rotations~\cite{kivlichan2018quantum}.
Efficient classical post-processing to construct $\widehat{\rho}$ or the orbital-rotation $k$-RDM estimator \(\widehat{D}^{(k)}\) in Step~3 is provided in Sec.~\ref{sec:Grassmannian_reformulation} in Supplementary Material.

Our first result is the complete proof for the tomographic completeness for the orbital-rotation shadows.
Note that tomographic completeness technically refers to the algebraic property that guarantees that the measurement channel can be inverted. While this property is invoked in Ref.~\cite{low2022classical}, its proof contains a technical step that requires additional justification.  We provide
a self-contained argument below (see Sec.~\ref{sec:tomographic_completeness} in Supplementary Material for the details).
\begin{theorem}[Tomographic completeness (Informal)]
\label{thm:tomographic_completeness_main}
The orbital-rotation measurement ensemble is tomographically complete.  Namely, for \(1\le \eta\le N\), every operator on
\(\mac H_\eta\) can be expressed as a finite complex linear combination of
projectors expressed by
\begin{align}
    U_\eta(u) \Pi_{\eta } U_\eta(u)^\dagger, \quad 
    u\in\mtr U(N),
\end{align}
where \(\Pi_\eta \coloneq \ketbra{1,\ldots,\eta}\) is the projector onto the occupation-number basis state with the first \(\eta\) modes
occupied.
\end{theorem}

By tomographic completeness, the measurement channel is
invertible, so the inverse channel is well defined. 
This eventually ensures that \(\mathbb E[\widehat\rho]=\rho\), where $\mathbb{E}$ means the average over all possible classical data $(u,\vec{z})$.
Moreover, the same classical data
\((u,\vec z)\) can be used to construct a single-shot estimator of the
\(k\)-RDM, which is also unbiased for any $1\leq k\leq \eta$ and every $\eta$-particle state $\rho$.


Our second main result is an optimal performance guarantee, in terms of the particle number \(\eta\), for estimating \(k\)-RDM entries.  First, we derive an entrywise
variance of the orbital-rotation shadow.  The upper-bound analysis exploits the geometric structure underlying the randomness in the protocol.  Although the protocol is described using a Haar-random unitary
\(u \sim \mtr U(N)\) and an occupation-number outcome \(\vec z\), the resulting
estimator depends on these data only through the rank-\(\eta\) projector
$
    R \coloneq u^\dagger P_{\vec z} u ,
$
where \(P_{\vec z}\) is the one-particle projector onto the occupied modes in
\(\vec z\).  Thus, the relevant randomness is naturally described by the Grassmannian
$
    \mtr U(N)/(\mtr U(\eta)\times \mtr U(N-\eta))
$~\cite{milnor1974characteristic}.
The key step is to express the second moment, and hence the variance, of each
entry of the single-shot estimator \(\purple{{}^{k}\widehat D_{\mkern6mu\vec q}^{\vec p}}\) as an
expectation of a polynomial of bounded degree in the matrix entries of \(R\): its variance reduces to the evaluation of finitely many polynomial integrals over the Grassmannian, rather than to a direct analysis of the full Haar unitary \(u\).  The
basic objects are integrals of the form
\begin{align}
    \int_{\mtr U(N)/(\mtr U(\eta)\times \mtr U(N-\eta))}
    R_{i_1j_1}\cdots R_{i_tj_t}\,dR ,
    \label{eq:Weingarten_integral}
\end{align}
where \(dR\) denotes the invariant measure on the Grassmannian.
We bound such quantities and leads
to the following asymptotic variance bound (see Theorem~\ref{thm:entrywise-variance} in the Supplementary Material for the proof).

\begin{theorem}[Entrywise variance bound]
\label{thm:entrywise-variance_main}
For any fixed order \(k\), the single-shot estimator for each \(k\)-RDM entry
has variance of order \(\eta^k\), independent of the number of modes $N$.
Namely, for all \(1\le k\le \min(\eta, N-\eta)\) and $\eta < N$, all \(\eta\)-particle states
\(\rho\), and all \(\vec p,\vec q\in\mac S_{N,k}\),
\begin{align}
    \operatorname{Var}
    \left(
        \purple{{}^{k}\widehat D_{\mkern6mu\vec q}^{\vec p}}
    \right) := \mathbb{E} \qty[\abs{\purple{{}^{k}\widehat D_{\mkern6mu\vec q}^{\vec p}} - \purple{{}^{k}D_{\mkern6mu\vec q}^{\vec p}} }^2]
    \leq
    \frac{k+1}{k!} \eta^k,
\end{align}
where $\mathbb{E}$ is taken over all possible $(u,\vec{z})$. 
\end{theorem}

\begin{revision}
    For $k>\eta$, all $k$-RDM elements vanish identically and can be estimated with zero variance. In the regime $N-\eta<k\leq\eta$, particle--hole duality reduces the task to learning lower-order correlations of the $N-\eta$ holes. We expect a suitable adaptation of our analysis to yield analogous bounds governed by the hole number. We therefore focus on the nonredundant regime $1\leq k\leq\min(\eta,N-\eta)$.
\end{revision}

From this theorem, a typical median-of-means of orbital rotation {shadows provides estimators 
to complete the estimation task for all $k$-RDM elements within an additive error $\varepsilon$ with high probability.
A standard analysis clarifies} the sample complexity that scales $\frac{k+1}{k!} \eta^k\log \binom{N}{k}^2/\varepsilon^2$
\begin{revision}
    (see Corollary~\ref{cor:entrywise_additive_sample_complexity} in the Supplementary Material for the proof)
\end{revision}
.

We further complement the above upper bound with an information-theoretic lower bound for the same learning task.  Following the framework of Ref.~\cite{chen2022exponential}, the proof reduces \(k\)-RDM tomography to a two-hypothesis distinguishing problem.  More specifically, we distinguish the
maximally mixed state \(\tau\) on a restricted \(\eta\)-particle subspace from a
family of perturbed states whose \(k\)-RDM entries differ from those of
\(\tau\) by order \(\varepsilon\).
Hence, any protocol that estimates all
\(k\)-RDM entries to element-wise error \(\varepsilon\) can distinguish
these two cases.
For any adaptive single-copy measurement protocol, the probability of successfully distinguishing the two cases is bounded by the total variation distance between the classical measurement-record distributions generated in the two cases~\cite{chen2022exponential}.
We establish that this distance can be bounded using known Hilbert--Schmidt norm estimates for \(k\)-RDMs in the fixed-particle-number sector
\cite{christiansen2024hilbert,visconti2026hilbert}. The above argument yields the following lower-bound theorem (see Theorem~\ref{thm:lower_bound} in Supplementary Material for the details).

\begin{theorem}[Sample-complexity lower bound]
\label{thm:lower_bound_main}
Let \(k,\eta,N\) be integers satisfying \(1\le k\le \eta\) and
\(N\ge 2\eta\). 
Then, 
any quantum
algorithm based on single-copy adaptive measurements of a given unknown $N$-mode $\eta$-particle state \(\rho\) requires at least
\(c_k\eta^k/\varepsilon^2\) copies of \(\rho\) to estimate all elements of the
associated \(k\)-RDM ${D}^{(k)}$
within element-wise additive error \(\varepsilon\in (0,\varepsilon_k)\) with high probability.
Here, $c_k,\varepsilon_k$ are some constants depending only on $k$.
\end{theorem}

Combining Theorem~\ref{thm:entrywise-variance_main} with
Theorem~\ref{thm:lower_bound_main}, we conclude that the orbital-rotation shadow protocol achieves the optimal \((\eta,\varepsilon)\)-dependence over all protocols that learn entrywise
\(k\)-RDMs to additive error \(\varepsilon\) using arbitrary adaptive single-copy measurements,
up to constants depending only on \(k\).
This theorem also shows that by taking $\eta=N/2$, the lower bound $\Omega_k(N^k/\varepsilon^2)$ holds over all adaptive single-copy protocols for the entrywise $k$-RDM tomography. The same lower bound holds even for symmetry-agnostic protocols.
Hence, from the information-theoretic perspective, the previous general approaches in Table~\ref{tab:sample_complexity} are optimal in $N$ up to a logarithmic factor for a fixed $k$.

Beyond the asymptotic analysis, we also derive closed-form variance formulas in
low-order cases.  The entrywise variance bound above only requires uniform
bounds on the Grassmannian integrals in Eq.~\eqref{eq:Weingarten_integral}.
For explicit low-order formulas, however, we can evaluate the relevant
Weingarten integrals over the Grassmannian~\cite{coulter2025integration} exactly.  In particular, for
\(k=1\), this gives exact expressions for all matrix entries (see Sec.~\ref{sec:1-rdm_variance} in
Supplementary Material for details).

\begin{theorem}[Exact variance for 1-RDM]
\label{thm:exact_variance_1rdm_main}
For \(p,q\in[N]\), the single-shot orbital-rotation estimator satisfies
\begin{align}
\label{eq:exact_variance_1rdm}
\operatorname{Var}\!\left(\purple{{}^{1}\widehat D_{\mkern6muq}^{p}}\right)
&=
\frac{(N+1-\delta_{pq})(N-\eta+1)}{N(N+2)}
\Bigl(
\eta+\purple{{}^{1}D_{\mkern6mup}^{p}}+\purple{{}^{1}D_{\mkern6muq}^{q}}
\Bigr)
\notag\\
&-
(1-\delta_{pq})\frac{N+1}{N}
\purple{{}^{2}D_{\mkern6mup q}^{p q}}
-
\left|\purple{{}^{1}D_{\mkern6muq}^{p}}\right|^2,
\end{align}
where 
$
\purple{{}^{2}D_{\mkern6mup q}^{p q}}
=
\operatorname{tr}\!\left[
\rho\,
a_p^\dagger a_q^\dagger a_q a_p
\right].
$
\end{theorem}

We remark that the derivation of Eq.~\eqref{eq:exact_variance_1rdm} is not
specific to the \(1\)-RDM variance calculation.  For small order \(k\) up to $k =2$
, each matrix element of \(\purple{{}^k \widehat{D} }\) can be written as a
polynomial of bounded degree in the entries of the rank-\(\eta\) projector
\(R\).  Hence the variances and, more generally, the covariance blocks
\begin{align}
    &\operatorname{Cov}
    \left(
        \purple{{}^{\ell}\widehat D_{\mkern6mu\vec q}^{\vec p}},
        \purple{{}^{m}\widehat D_{\mkern6mu\vec s}^{\vec r}}
    \right)
    \coloneq \notag\\
    &\mathbb E
    \left[
        \left(
            \purple{{}^{\ell}\widehat D_{\mkern6mu\vec q}^{\vec p}}
            -
            \purple{{}^{\ell}D_{\mkern6mu\vec q}^{\vec p}}
        \right)
        \overline{
        \left(
            \purple{{}^{m}\widehat D_{\mkern6mu\vec s}^{\vec r}}
            -
            \purple{{}^{m}D_{\mkern6mu\vec s}^{\vec r}}
        \right)}
    \right],
\end{align}
for fixed small \(\ell\) and \(m\), reduce to finitely many Weingarten
integrals over
$
    \mtr{U}(N)/
    \left(
        \mtr{U}(\eta)\times \mtr{U}(N-\eta)
    \right).
$

\begin{revision}
Therefore, we provide a systematic procedure for deriving exact closed-form
expressions through the explicit evaluation of
Weingarten integrals (See the details in Sec.~\ref{sec:verify} of Supplementary Material).
Applying this procedure, we obtain the variances and covariances
of the $1$- and $2$-RDM estimators, including cross-covariances
between the two orders.
We confirm that the resulting expressions agree with direct Monte Carlo
simulations of the orbital-rotation shadow protocol within statistical
uncertainty.
These expressions enable the exact evaluation of the single-shot
variance of any linear combination of $1$- and $2$-RDM estimators for an
arbitrary state with fixed particle number.
They thus accurately determine the number
of sample shots required to estimate the observables considered
in practical applications to a fixed accuracy.
\end{revision}

\begin{figure}[t]
\begin{center}
\includegraphics[width=0.50\textwidth]{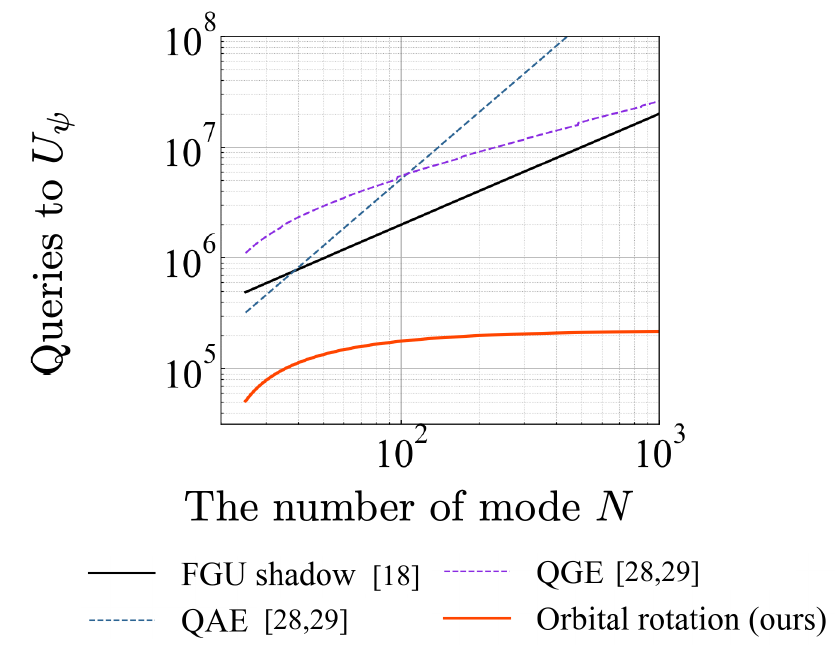}
\end{center}
    \vspace{-18pt}
\caption{
\justifying{
Cost for simultaneously learning all entries of the 1-RDM at fixed particle
number \(\eta=20\) and target accuracy \(\varepsilon=10^{-2}\), in terms of the number of calls to the state preparation unitary $U_{\psi}|0\rangle = |\psi\rangle$ for a target state $|\psi\rangle$.  We compare the
orbital-rotation shadow with the FGU shadow and Heisenberg-limited
estimation strategies like quantum amplitude estimation (QAE) and Quantum gradient estimation (QGE), whose 
costs are evaluated following
Refs.~\cite{Koizumi2026Faster, Koizumi2026Heiseberg}.
}
}
\label{fig:resource_estimator-main}
    \vspace{-10pt}
\end{figure}

\section{Applications}

In the following, we numerically demonstrate the advantage of orbital-rotation shadows over state-of-the-art methods for fermionic partial tomography.
\begin{revision}
    Here, the root mean-squared error is adopted as the accuracy criterion for these benchmarks, as it directly quantifies the statistical precision relevant to experimental measurements. We remark that additive-error guarantees can also be obtained using a standard median-of-means construction.
\end{revision}


First, we numerically assess the 
computational cost
required to estimate all entries of the 1-RDM for fermionic systems with fixed particle number \(\eta=20\) (details are given in 
Sec.~\ref{sec:Sample_cost} of Supplementary Material).  
\begin{revision}
The $1$-RDM determines one-body observables, including electronic
densities, momentum distributions, and dipole moments, and provides
the natural orbitals and their occupation numbers~\cite{Helgaker2000}.
For particle-number-conserving fermionic Gaussian states, it also
determines subsystem reduced density matrices, enabling the
characterization of entanglement
structure~\cite{cheong2004many, gullans2019entanglement}.
\end{revision}
Figure~\ref{fig:resource_estimator-main} compares the 
required number of queries to state preparation unitary $U_{\psi}$ for a target state $|\psi\rangle$
at target accuracy
\(\varepsilon=10^{-2}\). 
In contrast to the FGU shadow, the orbital-rotation protocol exhibits essentially $N$-independent scaling, with the dominant dependence appearing only through $\eta$.
We also include idealized
Heisenberg-limited strategies as benchmarks, which assume access to both a state-preparation circuit and its inverse; their cost should therefore be
interpreted as coherent queries to the state-preparation oracles~\cite{huggins2022nearly, Koizumi2026Heiseberg,
Koizumi2026Faster}. Among the learning protocols considered, the orbital-rotation protocol requires the fewest resources in the regime shown.

\begin{revision}

Beyond learning individual RDM entries, we also consider practical applications that often require estimating observables whose expectation values are linear combinations of $1$- and $2$-RDM elements.
Their measurement costs depend on the full estimator variances,
including covariances contributions.
We therefore focus on specific applications in condensed-matter
physics and quantum chemistry to assess the advantage of
orbital-rotation shadows, accounting for these covariance
contributions. 

Here, we compare the sample costs of orbital-rotation shadows with two other representatives methods, namely fermionic
Gaussian-unitary (FGU) shadows~\cite{Zhao:2020vxp} and uniform local Pauli shadows~\cite{Huang:2020tih} with Jordan-Wigner encoding.
All three protocols use measurement ensembles chosen independently
of the target observables, allowing the same snapshots to be reused
for multiple estimation tasks. In our comparison, 
the target state is taken as Slater determinant so that the measurement can be evaluated efficiently on a classical computer.
Although this is a simplification, we find that evaluations for Slater determinants in small systems show behavior similar to that for correlated ground states of Fermi–Hubbard models at the same particle number.
Concretely, the relative discrepancies remain below $1.5\%$ for the energy density and $4.3\%$ for the spin structure factor across all cases examined numerically (see Table~\ref{tab:hubbard-ground-slater} in the Supplementary Material).
These results support the use of Slater determinants as tractable
reference states for comparing the system-size scaling of measurement
costs across the different protocols.

\begin{figure}[t]
    \centering
    \includegraphics[width=0.48\textwidth]{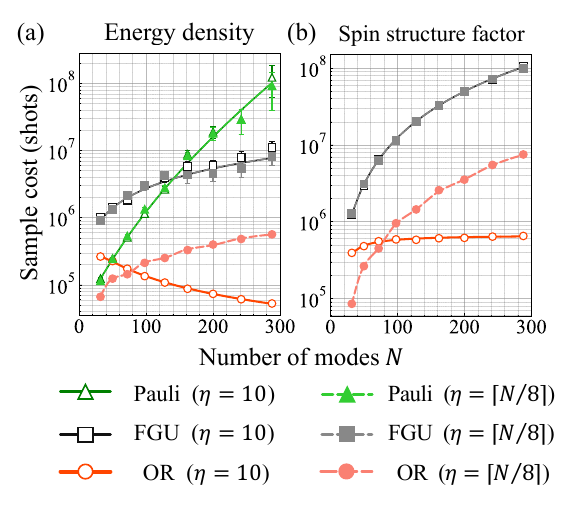}
        \vspace{-18pt}
    \caption{%
    \justifying
    Sample costs for estimating (a) the energy density and (b) the spin structure factor $S(\pi,\pi)$ of the square-lattice Fermi--Hubbard model using Pauli, FGU, and orbital-rotation (OR) shadows, for particle numbers $\eta = 10$ and $\eta = \lceil N/8 \rceil$. In Pauli CS, we consider the Jordan-Wigner encoding. 
    The target accuracies are $0.01t$ for energy density and
    $\varepsilon=0.01$ for spin structure factor, respectively.
    In both panels, markers show Monte Carlo estimates of the sample
    costs obtained from $10^4$ shots per method and system size,
    with error bars indicating one standard error.
    Solid lines show analytical predictions. In particular, those for orbital-rotation shadows are obtained using the closed-form expressions derived via Weingarten integrals over the Grassmannian.
    In panel (b), the Pauli shadow cost exceeds $10^{16}$ shots
    throughout the plotted range and is therefore omitted.
    }
    \label{fig:hubbard-sample-cost-main}
        \vspace{-10pt}
\end{figure}

\begin{figure*}[t]
\hspace{-1cm}
\includegraphics[width=0.95\linewidth]{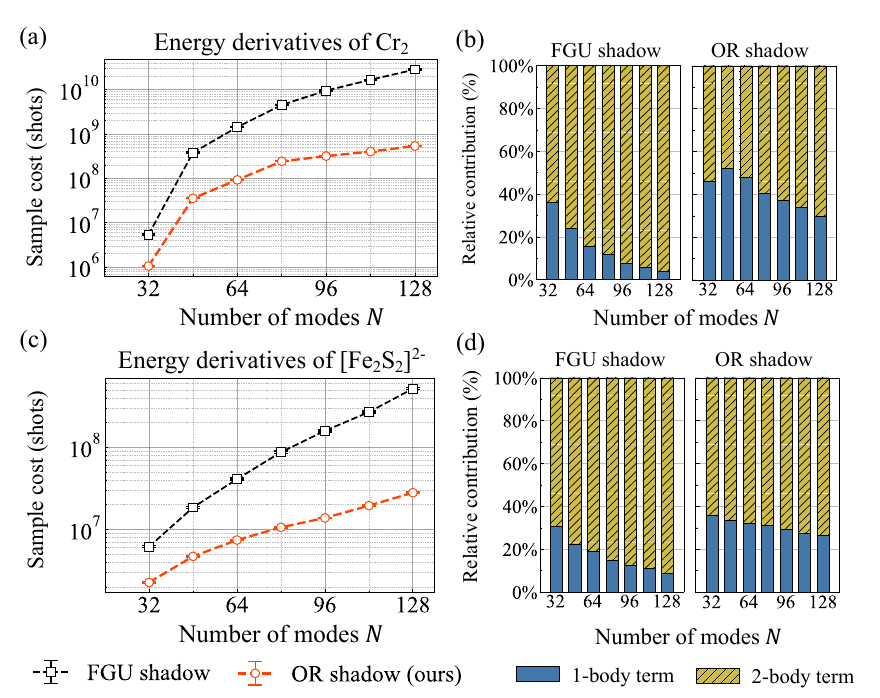}
\caption{
\justifying{
Sample costs and variance contributions for molecular derivative operators.
(a,c) Sample costs for FGU and orbital-rotation (OR) shadows
for $\mathrm{Cr}_2$ molecule and an idealized
$[\mathrm{Fe}_2\mathrm{S}_2]^{2-}$ cluster, respectively.
Slater reference states with $\eta=12$ and $14$ active electrons,
respectively, are used as the active space is enlarged.
The target root-mean-square error is
$\varepsilon=0.01 \ \mathrm{Hartree}/\text{\AA}$.
Markers show Monte Carlo estimates from $10^4$ shots per method
and system size, with error bars indicating one standard error.
(b,d) Relative variance contributions for $\mathrm{Cr}_2$ in (b)
and $[\mathrm{Fe}_2\mathrm{S}_2]^{2-}$ in (d), shown for
FGU (left) and OR (right) shadows.
The 1-body (blue) and 2-body (hatched yellow) contributions
arise from estimators of $\langle a_p^\dagger a_q\rangle$ and
$\langle a_p^\dagger a_q^\dagger a_s a_r\rangle$, respectively.
The bars show the 1- and 2-body variances normalized by their sum; the covariance contribution is not shown for simplicity.
Detailed molecular geometries and computational settings
are provided in the Supplementary Material.
}
}
\label{fig:chemsitry-app-main}
\end{figure*}

As a representative problem for condensed-matter physics, we evaluate the measurement cost for the energy density and
the antiferromagnetic spin structure factor of the square-lattice
Fermi--Hubbard model on an $L\times L$ lattice~\cite{arovas_hubbard_2022}.
The energy-density operator is
\begin{align}
\frac{H_{\mathrm{Hub}}}{L^2}
={}& -\frac{t}{L^2}
\sum_{\langle p,q\rangle,\sigma}
\left(a^\dagger_{p\sigma}a_{q\sigma}+\mathrm{h.c.}\right)
+\frac{U}{L^2}\sum_p n_{p\uparrow}n_{p\downarrow},
\end{align}
where $\langle p,q\rangle$ denotes nearest-neighbor sites and
$n_{p\sigma}=a^\dagger_{p\sigma}a_{p\sigma}$.
Its expectation value characterizes the energetics of correlated
fermions and provides a central benchmark for quantum simulation.
The spin structure factor at $(\pi,\pi)$ is
\begin{align}
S(\pi,\pi)
=
\frac{1}{L^2}\sum_{p,r}
(-1)^{x_p-x_r+y_p-y_r}
\mathbf S_p\cdot\mathbf S_r,
\end{align}
where $\mathbf S_p$ is the local spin operator at site $p$
with integer coordinates $(x_p,y_p)$.
Its expectation value probes staggered spin correlations and
serves as a diagnostic of antiferromagnetic order~\cite{varney2009quantum,mazurenko2016experimental}.  

Figure~\ref{fig:hubbard-sample-cost-main} shows the sample costs at fixed
target precision for $\eta=10$ and $\eta = \lceil N/8 \rceil $ particles and increasing lattice size,
with $U/t=4$. 
For fixed particle number, $\eta=10$, the energy-density estimation cost of orbital-rotation
shadows decreases with system size, whereas the FGU and local Pauli
costs increase. At $N=288$ modes, orbital-rotation shadows reduce the sample cost by factors of $148$ and $1943$ relative to FGU and local Pauli shadows, respectively.
For the spin structure factor, the orbital-rotation cost remains
approximately constant over the range examined, while the FGU cost
grows approximately quadratically with the number of modes.
The reduction relative to FGU shadows reaches a factor of
$161$ at the largest system size. 
Even at fixed filling represented by $\eta=\lceil N/8\rceil$, a natural setting in condensed-matter physics where $\eta\propto N$, orbital-rotation shadows retain a substantial advantage, reducing the sample costs for both observables by approximately a factor of $14$ relative to FGU shadows at $N=288$. 
These results demonstrate that exploiting particle-number symmetry
can dramatically reduce sample costs not only for $k$-RDM learning
but also for estimating observables in practical applications.

We remark that, for the fixed particle setting, Monte Carlo estimates based on $10^4$ measurement shots are plotted alongside the analytical predictions in Fig.~\ref{fig:hubbard-sample-cost-main}. Their agreement indicates that, for the system sizes considered here, this modest sampling budget is sufficient for comparing the sample costs of the different protocols. We therefore use Monte Carlo estimates for orbital-rotation and FGU shadows in the quantum-chemistry applications below, where analytical evaluation of the estimator variances is more computationally demanding.

For quantum chemistry, we consider energy gradients with respect
to nuclear coordinates, which are relevant to geometry optimization
and molecular dynamics~\cite{obrien2019calculating,obrien2022efficient}.
Under the appropriate stationarity conditions~\cite{mizukami2020orbital, obrien2022efficient}, the Hellmann--Feynman
relation expresses these gradients as expectation values of
derivatives of the molecular Hamiltonian.
The corresponding derivative operators can be written as
\begin{align}
\pdv{H_{\mathrm{act}}(\mathbf R)}{R_{A\alpha}}
&=
\sum_{p,q}
\pdv{h_{p,q}(\mathbf R)}{R_{A\alpha}}
a_p^\dagger a_q
\notag\\
&\quad+
\sum_{\vec p,\vec q}
\pdv{g_{p_1p_2,q_1q_2}(\mathbf R)}{R_{A\alpha}}
a_{p_1}^\dagger a_{p_2}^\dagger a_{q_2}a_{q_1},
\label{eq:energy-derivative-operator}
\end{align}
where $\mathbf R$ denotes the collection of nuclear positions, and $R_{A\alpha}$ is the $\alpha$ component of the position of nucleus $A$, with $\alpha\in \{x,y,z\}$. The quantities $h_{p,q}(\mathbf R)$ and $g_{p_1p_2,q_1q_2}(\mathbf R)$ denote the one- and two-body coefficients of the active-space Hamiltonian, respectively, with all orbital indices referring to active spin orbitals. Thus it retain a one- and two-body structure, allowing the same
shadow snapshots to be reused for different nuclear coordinates.

Using self-consistent-field Slater determinants as measurement
benchmarks, we study $\mathrm{Cr}_2$~\cite{larsson2022chromium} and an idealized
$[\mathrm{Fe}_2\mathrm{S}_2]^{2-}$ cluster~\cite{georges2025quantum}, retaining $12$ and $14$
active electrons, respectively, while enlarging the orbital spaces.
The molecular geometries and active-space construction are
specified in the Supplementary Material.
Panels (a) and (c) of Fig.~\ref{fig:chemsitry-app-main} show that
the orbital-rotation sample cost grows more slowly than that of
FGU shadows, yielding $53$-fold and $18.5$-fold
reductions at $N=128$ for $\mathrm{Cr}_2$ and
$[\mathrm{Fe}_2\mathrm{S}_2]^{2-}$, respectively.
These reductions are qualitatively consistent with the advantage
suggested by the comparison of the entrywise 2-RDM variance bounds,
which indicates a parametric improvement of order $(N/\eta)^2$
as the orbital space is enlarged at fixed particle number.

For molecular derivative estimation, however, coefficient weighting,
one-body contributions, and cancellations arising from covariances
between RDM estimators can modify this behavior.
Consequently, the entrywise comparison does not by itself determine
the total sample-cost ratio.
Nevertheless, in our fixed-Slater benchmarks, the two-body variance
becomes the largest individual contribution at larger $N$ for both
protocols as shown in Panels (b) and (d) of Fig.~\ref{fig:chemsitry-app-main}, although covariance cancellations remain appreciable.
An unweighted log--log fit over the four largest orbital spaces,
$N=80,96,112,128$, gives effective power-law behavior
$V_{\mathrm{FGU}}/V_{\mathrm{Orbital}}\propto N^{2.2}$
for $\mathrm{Cr}_2$ and $\propto N^{1.7}$ for
$[\mathrm{Fe}_2\mathrm{S}_2]^{2-}$, where $V{\mathrm{FGU}}$ and $V{\mathrm{Orbital}}$ denote the total variances of the FGU and orbital-rotation shadow estimators, respectively.
For unbiased sample means at a fixed target root-mean-square error,
these variance ratios also determine the relative shot requirements.
These finite-size trends show that the growing advantage remains
visible in the estimation of the full derivative observable,
supporting the expectation of further benefits in larger orbital
spaces at fixed particle number.
The fitted exponents characterize the explored range rather than
establishing asymptotic scaling laws.


\end{revision}

\section{Conclusion and Outlook}
In this work, we clarified how particle-number conservation leads to a
provable advantage in fermionic learning tasks through the orbital-rotation shadow protocol.  We first gave a self-contained justification of tomographic completeness of the protocol, resolving a technical gap in the existing argument (Theorem~\ref{thm:tomographic_completeness_main}).
We then proved as in Theorem~\ref{thm:entrywise-variance_main} that, for each fixed order \(k\), the orbital shadow protocol estimates any \(k\)-RDM element with variance \(\mac{O}_k(\eta^k)\). 
This yields an optimal $k$-RDM element-wise tomography protocol that matches the information-theoretic lower bound $\Omega_k(\eta^k/\varepsilon^2)$ for any single-copy adaptive measurement protocol, which we derived in Theorem~\ref{thm:lower_bound_main}.
\begin{revision}
    We also derived an exact closed-form variance formula for the $1$-RDM estimator and developed a systematic procedure for obtaining exact variance and covariance formulas for the $1$- and $2$-RDM estimators, including cross-covariances. These results enable accurate evaluation of estimator variances and the corresponding sample costs. 
    Additionally, our numerical analyses 
    for practical applications such as 
    spin structure evaluation in Fermi-Hubbard model and energy derivatives for strongly-correlated molecular electronic structure demonstrated substantial reductions in the number of measurement shots required to achieve a fixed accuracy.
\end{revision}

Several directions remain open.  
\begin{revision}
    First, the main remaining challenge is to sharpen the dependence of the lower bound on the order $k$. While our upper-bound analysis largely resolves this dependence, the current lower bound establishes optimality in the $\eta$-dependence for each fixed $k$, leaving a gap in the $k$-dependent factors. The structure observed in low-order cases and the combinatorics of occupied \(k\)-particle sectors suggest the sharper optimal scaling \(\Theta \left(k\binom{\eta}{k}/\varepsilon^2\right)\). Establishing a lower bound that matches this scaling remains an open problem.
\end{revision}

Second, from a practical viewpoint, it would be valuable to replace the random orbital rotations used here by a discrete ensemble of number-conserving FGU that retains the same particle-number-induced advantage. 
While the fermionic-shadow protocol of Ref.~\cite{Zhao:2020vxp,wan2023matchgate} already employs a discrete ensemble to achieve $\widetilde{\mac{O}}_k(N^k/\epsilon^2)$ sample complexity for $k$-RDM learning, it remains open whether a discrete, efficiently implementable ensemble can eliminate this polynomial dependence on $N$. 
Such a construction would reduce the compilation and calibration overhead associated with continuous orbital rotations and make the protocol more compatible with devices with finite-precision control.

Third,  it is highly nontrivial whether the advantage by particle-number symmetry persists under the task of state certification. 
It has been recently shown in qubit systems that adaptive single-qubit measurement suffices for efficient certification of any pure state~\cite{gupta2026few}, and thus it is intriguing to explore what is the minimum requirement for certification task in fermionic systems.

\section*{Acknowledgements}
The authors wish to thank Andrew Zhao, 
Luning Zhao, Pei Zeng, Qi Ye, Kento Tsubouchi, and Wataru Mizukami for fruitful discussions.
Y. K. is supported by the Program for Leading Graduate Schools (MERIT-WINGS), and JST BOOST
Grant Number JPMJBS2418.
K.W. is supported by JSPS KAKENHI Grant Number JP24KJ1963 and JST ASPIRE Grant Number JPMJAP2316.
T.P.T. is supported by JSPS KAKENHI Grant Number JP26KJ1547.
N.Y. is supported by JST Grant Number JPMJPF2221, JST CREST Grant Number JPMJCR23I4, IBM Quantum, Google Quantum AI, JST ASPIRE Grant Number JPMJAP2316, JST ERATO Grant Number JPMJER2302, JST [Moonshot R\&D] [Grant Number JPMJMS256J], and Institute of AI and Beyond of the University of Tokyo.
\ 
\\

{\it Note added.} During the preparation of this manuscript, we became aware of the independent and contemporaneous work of West, Cerezo, and Larocca~\cite{west2026particle}. Both works build on Low's orbital-rotation shadow estimator~\cite{low2022classical}, whose unbiasedness relies on tomographic completeness fixed by our work. 
The main focus of Ref.~\cite{west2026particle} is to establish mode-independent guarantees for
Slater-determinant overlaps and linear combination of
1-RDM elements,
while we study the simultaneous estimation of all \(k\)-RDM elements
to show that 
entrywise variance
at most \(\varepsilon^2\) is achieved with \(\mathcal{O}_k(\eta^k/\varepsilon^2)\) samples,
which provably matches the information-theoretic lower bound of \(\Omega_k(\eta^k/\varepsilon^2)\).
\begin{revision}
    Additionally, we develop a systematic procedure for deriving exact variance and covariance formulas for the 1- and 2-RDM estimators, including cross-covariances. These formulas provide a rigorous basis for evaluating the sample costs of physical observables and quantifying the significant advantage of orbital-rotation shadows.
\end{revision}

{\it Data availability.} The Python implementation code used to compute the closed-form expressions for covariances of $k$-RDMs is available in
Ref.~\cite{koizumi_or_shadow_covariances}.

\bibliography{ref} 

\clearpage
\onecolumngrid
\appendix

\setcounter{section}{0}
\renewcommand{\thesection}{S\arabic{section}}

\setcounter{subsection}{0}
\renewcommand{\thesubsection}{\arabic{subsection}}

\makeatletter
\renewcommand{\p@subsection}{\thesection.}
\makeatother

\counterwithout{equation}{section}
\setcounter{equation}{0}
\renewcommand{\theequation}{S\arabic{equation}}

\counterwithout{theorem}{section}
\setcounter{theorem}{0}
\renewcommand{\thetheorem}{S.\arabic{theorem}}

\renewcommand{\theHsection}{appendix.section.\arabic{section}}
\renewcommand{\theHsubsection}{appendix.subsection.\arabic{section}.\arabic{subsection}}
\renewcommand{\theHequation}{appendix.equation.\arabic{equation}}
\renewcommand{\theHtheorem}{appendix.theorem.\arabic{theorem}}

\startcontents[appendix]
\onecolumngrid

\begin{center}
    {\Large\bfseries Supplementary Material for ``Provably Efficient Learning }\\[0.5em]
    {\Large\bfseries  of Fermionic Correlations under Particle-Number Symmetry"}
\end{center}


\tableofcontents

\addtocontents{toc}{\protect\setcounter{tocdepth}{2}}

\section{Preliminary}

\subsection{Notation}
In this section, we introduce the notation used throughout papers.
Let $N, \eta, k$ be positive integers satisfying $1 \leq k \leq \eta < N$, where $N$, $\eta$, and $k$ denote the number of modes, the particle number, and the target order of the $k$-body reduced density matrix, respectively. Additionally, for a vector space $V$, we denote $\mtr{End}(V) \coloneq \{ f : V \to V \mid f \text{ is a linear operator.} \} $ as endomorphism algebra.
For a positive integer $m$, we denote $[m] \coloneq (1,2,\ldots,m)$ and the symmetry group of $m$ degrees as $\mathfrak{S}_m$. For a random variable \(X\) with probability distribution \(\mu\), we write
\(X\sim\mu\). For simplicity, when a compact group or a
compact homogeneous space is equipped with its normalized invariant probability measure, we use the same notation for sampling from this measure. For example,
$
    u\sim \mtr U(N)
$
means that \(u\) is sampled from the normalized Haar measure on the $N$-dimensional unitary group
\(\mtr U(N)\).

The target Hilbert space is the $\eta$-th exterior power of the
$N$-dimensional complex vector space $\mathbb{C}^N$, denoted by $\qty( \mathbb{C}^N)^{\wedge \eta}$. Throughout this section, we use the notation \(\mac H_\eta := (\mathbb C^N)^{{\wedge \eta }}\), so that \(\mac H_\eta=(\mathbb C^N)^{\wedge \eta }\) is the physical \(\eta\)-particle sector.
For simplicity, we denote by $\ket{z} \coloneqq a_z^\dagger \ket{\mathrm{vac}}$ the basis vector of
$\mathbb{C}^N$ for $1 \leq z \leq N$~\footnote{More precisely, the creation operator
$a_z^\dagger : \wedge^0 \mathbb{C}^N \simeq \mathbb{C} \to \mathbb{C}^N$
acts as $a_z^\dagger \ket{\mathrm{vac}} = \ket{z}$.}.
Then, the basis elements of the $\eta$-particle fermionic space $\qty( \mathbb{C}^N)^{\wedge \eta}$ are expressed as
$    \bigwedge_{j=1}^\eta \ket{z_j},$
with total dimension $\binom{N}{\eta}$.
These basis states are labeled by the occupied-mode configurations
$\vec{z} \in \mathcal{S}_{N,\eta}$, where
\begin{align}
    \mathcal{S}_{N,\eta} \coloneqq \{ (z_1, \ldots, z_\eta) \mid 1 \leq z_1 < z_2 < \cdots < z_\eta \leq N \}
\end{align}
is the set of all strictly increasing $\eta$-tuples of integers between $1$ and $N$. For notational simplicity, we write
\begin{align}
     \ket{\vec z} =   \ket{z_1,\ldots,z_\eta}
    :=
    \ket{z_1}\wedge\cdots\wedge\ket{z_\eta} \in \mac{H}_{\eta}.
\end{align}
For example, \(\ket{1,2,3,4}\) denotes
\(\ket{1}\wedge\ket{2}\wedge\ket{3}\wedge\ket{4}\). Additionally, following Ref.~\cite{low2022classical}, we often use a notation $\ket{[\eta] } $ to denote
\begin{align}
   \ket{ [\eta]} \coloneq \ket{1,2,\ldots, \eta} \in \mac{H}_\eta.
\end{align}

Throughout this paper, every
\(
\vec a\in\mac S_{N,m}
\)
is understood to be written in strictly increasing order.
When applying set operations, we identify \(\vec a\) with its
underlying set and list the resulting set again in increasing order.
Thus, for
\(
\vec p=(p_1,\ldots,p_k)
\),
\(
\vec q=(q_1,\ldots,q_\ell)
\),
and
\(
\vec r=(r_1,\ldots,r_m)
\),
we define
\begin{align}
    \vec p\cap\vec q
    &:=
    \operatorname{inc}
    \left(
        \{p_1,\ldots,p_k\}
        \cap
        \{q_1,\ldots,q_\ell\}
    \right),
    \\
    \vec r\setminus\vec p
    &:=
    \operatorname{inc}
    \left(
        \{r_1,\ldots,r_m\}
        \setminus
        \{p_1,\ldots,p_k\}
    \right),
\end{align}
where \(\operatorname{inc}(A)\) denotes the unique strictly
increasing tuple whose entries are the elements of a finite set
\(A\). In particular, \(\operatorname{inc}(\varnothing)\) is the
empty tuple.

\begin{revision}
In contrast, we retain the order of a tuple whenever it determines
a fermionic sign. Let
\(
\boldsymbol a=(a_1,\ldots,a_m)
\)
be a tuple of pairwise distinct mode indices, not necessarily written
in increasing order. We define its fermionic shuffle sign by
\begin{align}
    \varepsilon_{\wedge}(\boldsymbol a)
    &:=
    (-1)^{\operatorname{inv}(\boldsymbol a)},
    \label{eq:fermionic-shuffle-sign}
    \\
    \operatorname{inv}(\boldsymbol a)
    &:=
    \#\left\{
        (i,j):
        1\leq i<j\leq m,\;
        a_i>a_j
    \right\}.
\end{align}
Equivalently, if
\(
\pi_{\boldsymbol a}\in\mathfrak S_m
\)
is the unique permutation satisfying
\(
a_{\pi_{\boldsymbol a}(1)}
<
\cdots
<
a_{\pi_{\boldsymbol a}(m)}
\),
then
\(
\varepsilon_{\wedge}(\boldsymbol a)
=
\operatorname{sgn}(\pi_{\boldsymbol a})
\).
Consequently,
\begin{align}
    \ket{a_1}\wedge\cdots\wedge\ket{a_m}
    &=
    \varepsilon_{\wedge}(\boldsymbol a)\,
    \ket{\operatorname{sort}(\boldsymbol a)}.
    \label{eq:fermionic-sorting-identity}
\end{align}
For disjoint increasing tuples
\(
\vec a=(a_1<\cdots<a_p)\in\mac S_{N,p}
\)
and
\(
\vec b=(b_1<\cdots<b_q)\in\mac S_{N,q}
\),
we denote their increasing disjoint union by
\begin{align}
    \vec a\sqcup\vec b
    &:=
    \operatorname{inc}
    \left(
        \{a_1,\ldots,a_p\}
        \cup
        \{b_1,\ldots,b_q\}
    \right).
\end{align}
We define the corresponding fermionic shuffle sign by
\begin{align}
    \varepsilon_{\wedge}(\vec a,\vec b)
    &:=
    (-1)^{
        \#\left\{
            (i,j)\in[p]\times[q]:
            a_i>b_j
        \right\}
    }.
    \label{eq:two-block-fermionic-shuffle}
\end{align}
Equivalently,
\(
\varepsilon_{\wedge}(\vec a,\vec b)
\)
is the sign of the permutation that rearranges
\(
(a_1,\ldots,a_p,b_1,\ldots,b_q)
\)
into increasing order. Consequently,
\begin{align}
    \ket{\vec a}\wedge\ket{\vec b}
    &=
    \varepsilon_{\wedge}(\vec a,\vec b)\,
    \ket{\vec a\sqcup\vec b}.
    \label{eq:block-fermionic-shuffle}
\end{align}
If
\(
\vec a\cap\vec b\neq\varnothing
\),
then
\(
\ket{\vec a}\wedge\ket{\vec b}=0
\).
    
\end{revision}

For \(\vec z\in \mathcal S_{N,\eta}\), we denote by
\(P_{\vec z}\) the projector onto the subspace
spanned by the occupied modes \(\{\ket{z_1},\ldots,\ket{z_\eta}\}\).
More explicitly,
\begin{align}
    P_{\vec z}
    :=
    \sum_{j=1}^{\eta} \ketbra{z_j}{z_j}
    \in \mtr{End}(\mathcal H_1).
\end{align}
In this paper, we strictly distinguish this one-particle projector from the
corresponding projector on the \(\eta\)-particle sector.  The latter is
given by the exterior power
\begin{align}
    \Pi_{\vec z}^{(\eta) } = (P_{\vec z})^{\wedge\eta}
    \in \mtr{End}(\mathcal H_\eta).
\end{align}
For a decomposable vector
\(\ket{v_1}\wedge\cdots\wedge\ket{v_\eta}\in\mathcal H_\eta\), this
operator acts as
\begin{align}
    (P_{\vec z})^{\wedge\eta}
    \bigl(
        \ket{v_1}\wedge\cdots\wedge\ket{v_\eta}
    \bigr)
    =
    P_{\vec z}\ket{v_1}\wedge\cdots\wedge P_{\vec z}\ket{v_\eta}.
\end{align}
In particular, for an occupation-basis vector
\(\ket{\vec y}\in\mathcal H_\eta\), we have
$
    (P_{\vec z})^{\wedge\eta}\ket{\vec y}
    =
    \delta_{\vec y,\vec z}\ket{\vec z}.
$
Hence
$
    (P_{\vec z})^{\wedge\eta}
    =
    \ketbra{\vec z}{\vec z}.
$
Thus \(\Pi_{\vec z}^{(\eta)}\) is the Slater determinant projector
associated with the occupation pattern \(\vec z\), whereas
\(P_{\vec z}\) is the corresponding one-particle projector.

\subsection{Introduction of $k$-RDM, contraction map, and extension map}

In this section, we review \(k\)-body reduced density matrices (\(k\)-RDMs) through the contraction map and its trace-dual, the extension map. From the perspective of the contraction map, \(k\)-RDMs are obtained by reducing an original density operator to a lower-particle sector. 
This reduction viewpoint goes back to the early development of reduced-density-matrix theory, where the many-electron problem was reformulated in terms of lower-order density matrices \cite{Coleman1963,GarrodPercus1964}. The same viewpoint also underlies variational RDM methods. For Hamiltonians with at most two-body interactions, the energy is a linear functional of the 2-RDM, so the ground-state problem can be formulated as an optimization problem over 2-RDMs subject to \(N\)-representability constraints \cite{GarrodPercus1964,Mazziotti2004,Mazziotti2005}. 
Thus, in the present formalism, a \(k\)-RDM is naturally regarded as the lower-particle object obtained from a many-particle state by a contraction map. 

Conversely, the extension map introduced here is the trace-dual of the contraction map. It realizes an operator on the reduced \(k\)-particle space as a physical \(k\)-body observable acting on the \(\eta\)-particle sector. Thus, the contraction map describes the usual reduction from many-particle states to lower-order RDMs, whereas its trace-dual provides the corresponding operator-level lifting required to evaluate \(k\)-body observables on the physical sector. This contraction--extension duality will be used throughout the paper to formulate the orbital-rotation fermionic shadow estimator in a sector-consistent way.

In what follows, we consider an \(\eta\)-particle fermionic state
\begin{align}
    \rho \in \operatorname{End}(\mac H_\eta),
    \qquad
    \rho \ge 0,
    \qquad
    \operatorname{tr}_{\mac H_\eta}[\rho]=1,
\end{align}
where \(N\) is the number of fermionic modes and \(1\leq \eta < N\).

We now introduce \(k\)-RDMs as contractions of the physical \(\eta\)-particle density operator. This viewpoint is standard in reduced-density-matrix theory: reduced density matrices are obtained from the many-particle density matrix by tracing out, or equivalently contracting, the remaining particles \cite{Coleman1963,GarrodPercus1964}. Throughout this paper, we use the unnormalized contraction convention: no additional combinatorial prefactor is included in the contraction map, so that the contraction map can be defined as the trace-dual of the extension map introduced below.

For \(0\leq k\leq \eta\), we define the contraction map
\begin{align}
    \mathcal C_{\eta,k}:
    \operatorname{End}(\mac H_\eta)
    \longrightarrow
    \operatorname{End}(\mac H_k)
\end{align}
by requiring that its matrix elements are given by
\begin{align}
    \bra{\vec q}
        \mathcal C_{\eta,k}(Y)
    \ket{\vec p}
    &:=
    \tr_{\mathcal H_\eta}
    \left[
        Y\,
        a_{p_1}^{\dagger}\cdots a_{p_k}^{\dagger}
        a_{q_k}\cdots a_{q_1}
    \right].
\end{align}
for \(Y\in \operatorname{End}(\mac H_\eta)\) and \(\vec p=(p_1<\cdots<p_k),\vec q=(q_1<\cdots<q_k)\in \mac S_{N,k}\). 
Intuitively, the contraction map \(\mac C_{\eta,k}\) forgets
\(\eta-k\) particles while keeping the same \(N\)-mode system. Hence the
weight assigned to a \(k\)-particle configuration \(\vec a\) is obtained
by adding up the weights of all \(\eta\)-particle configurations
\(\vec b\) that contain \(\vec a\):
\begin{align}
    \bra{\vec a}
    \mac C_{\eta,k}(Y)
    \ket{\vec a}
    =
    \sum_{\substack{\vec b\in\mac S_{N,\eta}\\ \vec a\subseteq \vec b}}
    \bra{\vec b}Y\ket{\vec b}.
\end{align}
The \(k\)-body reduced density matrix of the state \(\rho\) is then defined as
\begin{align}
    D_\rho^{(k)}
    :=
    \mathcal C_{\eta,k}(\rho)
    \in \operatorname{End}(\mac H_k).
\end{align}
Its matrix elements are denoted by
\begin{align}
    \purple{{}^{k}D_{\rho;\,\vec q}^{\vec p}}
    &:=
    \left\langle \vec q \middle|
        D_\rho^{(k)}
    \middle| \vec p \right\rangle
    \nonumber\\
    &=
    \operatorname{tr}_{\mac H_\eta}
    \left[
        \rho\,
        a_{p_1}^{\dagger}\cdots a_{p_k}^{\dagger}
        a_{q_k}\cdots a_{q_1}
    \right].
\end{align}
Throughout this paper, we refer to ${}^{k}D_{\rho;,\vec q}^{\vec p}$ as the $(\vec p,\vec q)$-entry, listing the column index before the row index.
Thus, \(D_\rho^{(k)}\) is the reduced \(k\)-particle object obtained by contracting the \(\eta\)-particle density operator, and its matrix elements reproduce the number-conserving \(k\)-body correlation functions of the original fermionic state. With the above unnormalized convention, every normalized \(\eta\)-particle state \(\rho\) satisfies
\begin{align}
    \operatorname{tr}_{\mac H_k}[D_\rho^{(k)}]
    =
    \binom{\eta}{k}.
\end{align}
where \(\tr_{\mac H_\eta}\) denotes the trace on \(\mtr{End}(\mac H_\eta)\).

We next introduce the trace-dual operation of the contraction map. The extension of a lower-particle operator to a higher-particle sector is often expressed by the Grassmann wedge product with an identity operator \cite{Mazziotti2023, Mazziotti2007RDM}. Following this viewpoint, for \(X\in \operatorname{End}(\mac H_k)\), we write its extension to the physical \(\eta\)-particle sector as
\begin{align}
    X\wedge \1_{\mac H_{\eta-k}}
    \in
    \operatorname{End}(\mac H_\eta),
\end{align}
where \(\1_{\mac H_{\eta-k}}\) denotes the identity operator on \(\mac H_{\eta-k}=(\mathbb C^N )^{\wedge {(\eta-k)}}\), and \(\wedge\) denotes the Grassmann wedge product of operators.
In the normalization used in this paper, this Grassmann wedge extension is defined by its action on matrix units. For
\(\vec p=(p_1<\cdots<p_k)\) and \(\vec q=(q_1<\cdots<q_k)\) in \(\mac S_{N,k}\), we set
\begin{align}
    \ketbra{\vec p}{\vec q}
    \wedge \1_{\mac H_{\eta-k}}
    :=
    a_{p_1}^\dagger\cdots a_{p_k}^\dagger
    a_{q_k}\cdots a_{q_1}
    \big|_{\mac H_\eta}.
\end{align}
By linearity, if
\begin{align}
    X
    =
    \sum_{\vec p,\vec q\in \mac S_{N,k}}
    X_{\vec p,\vec q}
    \ketbra{\vec p}{\vec q}
    \in \operatorname{End}(\mac H_k),
\end{align}
then
\begin{align}
    X\wedge \1_{\mac H_{\eta-k}}
    =
    \sum_{\vec p,\vec q\in \mac S_{N,k}}
    X_{\vec p,\vec q}\,
    a_{p_1}^\dagger\cdots a_{p_k}^\dagger
    a_{q_k}\cdots a_{q_1}
    \big|_{\mac H_\eta}.
\end{align}
For simplicity, we define the extension map
\begin{align}
    \mathcal E_{k,\eta}:
    \operatorname{End}(\mac H_k)
    \longrightarrow
    \operatorname{End}(\mac H_\eta)
\end{align}
by
\begin{align}
    \mathcal E_{k,\eta}(X)
    :=
    X\wedge \1_{\mac H_{\eta-k}}.
\end{align}

\purple{This extension map is trace-dual to the contraction map \(\mathcal C_{\eta,k}\).} Indeed, for all \(Y\in \operatorname{End}(\mac H_\eta)\) and \(X\in \operatorname{End}(\mac H_k)\), the definition of \(\mathcal C_{\eta,k}\) gives
\begin{align}
    \operatorname{tr}_{\mac H_\eta}
    \left[
        Y\,
        \mathcal E_{k,\eta}(X)
    \right]
    &=
    \sum_{\vec p,\vec q\in \mac S_{N,k}}
    X_{\vec p,\vec q}
    \operatorname{tr}_{\mac H_\eta}
    \left[
        Y\,
        a_{p_1}^\dagger\cdots a_{p_k}^\dagger
        a_{q_k}\cdots a_{q_1}
    \right] \\
    &=
    \sum_{\vec p,\vec q\in \mac S_{N,k}}
    X_{\vec p,\vec q}
    \purple{\left\langle \vec q \middle|
        \mathcal C_{\eta,k}(Y)
    \middle| \vec p \right\rangle} .
\end{align}
Equivalently,
\begin{align}
    \operatorname{tr}_{\mac H_\eta}
    \left[
        Y\,
        \mathcal E_{k,\eta}(X)
    \right]
    =
    \operatorname{tr}_{\mac H_k}
    \left[
        \purple{X}
        \mathcal C_{\eta,k}(Y)
    \right],
    \label{eq:duality_contraction_extension}
\end{align}
\purple{Thus, no transpose is required with the matrix-element convention above.}

With this notation, a reduced \(k\)-particle operator \(X\in \operatorname{End}(\mac H_k)\) defines the physical observable
\begin{align}
    O_X^{(k)}
    :=
    \mathcal E_{k,\eta}(X)
    =
    X\wedge \1_{\mac H_{\eta-k}}
    \in \operatorname{End}(\mac H_\eta).
\end{align}
For an \(\eta\)-particle state \(\rho\), the \(k\)-RDM then satisfies
\begin{align}
    \operatorname{tr}_{\mac H_\eta}
    \left[
        \rho\,
        \left(
            X\wedge \1_{\mac H_{\eta-k}}
        \right)
    \right]
    =
    \sum_{\vec p,\vec q\in \mac S_{N,k}}
    X_{\vec p,\vec q}
    \purple{{}^{k}D_{\rho;\,\vec q}^{\vec p}}.
\end{align}
Therefore, contraction sends a physical \(\eta\)-particle operator to a reduced \(k\)-particle operator, whereas the Grassmann wedge product \(X\wedge \1_{\mac H_{\eta-k}}\) gives the corresponding trace-dual extension from the reduced sector back to the physical sector.

\subsection{Basic properties of the Grassmann wedge extension}
\label{subsec:extension-basic}

We collect the elementary properties of the extension map and the contraction map defined in the previous section that will be used later.
\begin{lemma}[Basic properties of the Grassmann wedge extension]
\label{lem:extension-basic}
Let $r,t$ be integers. Then, the extension map \(\mathcal E_{t,r}\) satisfies the following properties.
\begin{enumerate}
    \item If \(0\leq t\leq r\), then
    \begin{align}
        \mathcal E_{t,r}\!\left(\1_{\mac H_t}\right)
        =
        \binom{r}{t}\1_{\mac H_r}.
        \label{eq:extension_identity}
    \end{align}

    \item If \(X\in\operatorname{End}(\mac H_r)\), then
    \begin{align}
        \mathcal E_{r,r}(X)=X.
        \label{eq:extension_same_sector}
    \end{align}

    \item If \(0\leq s\leq r\leq m\), then for every \(X\in\operatorname{End}(\mac H_s)\),
    \begin{align}
        \mathcal E_{r,m}\!\left(\mathcal E_{s,r}(X)\right)
        =
        \binom{m-s}{r-s}\mathcal E_{s,m}(X).
        \label{eq:extension_composition}
    \end{align}
\end{enumerate}
\end{lemma}

\begin{proof}
The identity property follows by counting occupied \(t\)-mode subsets. Since
\begin{align}
    \1_{\mac H_t}
    =
    \sum_{\vec p\in\mac S_{N,t}}
    \ketbra{\vec p}{\vec p},
\end{align}
Equivalently, we obtain
\begin{align}
    \mathcal E_{t,r}\!\left(\1_{\mac H_t}\right)
    =
    \sum_{\vec p\in\mac S_{N,t}}
    a_{p_1}^\dagger\cdots a_{p_t}^\dagger
    a_{p_t}\cdots a_{p_1}
    \big|_{\mac H_r}.
\end{align}
Let \(\ket{\vec s}\in\mac H_r\) be a basis vector, with \(\vec s\in\mac S_{N,r}\). The annihilation string \(a_{p_t}\cdots a_{p_1}\) vanishes unless \(\vec p\subset\vec s\). If \(\vec p\subset\vec s\), it removes the occupied modes in \(\vec p\) up to a fermionic sign, and the corresponding creation string restores \(\ket{\vec s}\) with the same sign. Hence
\begin{align}
    \mathcal E_{t,r}\!\left(\1_{\mac H_t}\right)\ket{\vec s}
    =
    \sum_{\substack{\vec p\in\mac S_{N,t}\\ \vec p\subset\vec s}}
    \ket{\vec s}
    =
    \binom{r}{t}\ket{\vec s}.
\end{align}
Since the occupation basis spans \(\mac H_r\), this proves Eq.~\eqref{eq:extension_identity}.

The same-sector property follows directly from the matrix-unit action. For \(\vec p,\vec q\in\mac S_{N,r}\), the operator
\begin{align}
    a_{p_1}^\dagger\cdots a_{p_r}^\dagger
    a_{q_r}\cdots a_{q_1}
    \big|_{\mac H_r}
\end{align}
is equal to \(\ketbra{\vec p}{\vec q}\). Therefore \(\mathcal E_{r,r}(X)=X\) for all \(X\in\operatorname{End}(\mac H_r)\).

Finally, we prove the composition rule. The Grassmann wedge form gives
\begin{align}
    \mathcal E_{s,r}(X)
    =
    X\wedge \1_{\mac H_{r-s}}.
\end{align}
Applying \(\mathcal E_{r,m}\) once more yields
\begin{align}
    \mathcal E_{r,m}\!\left(\mathcal E_{s,r}(X)\right)
    =
    \left(
        X\wedge \1_{\mac H_{r-s}}
    \right)
    \wedge
    \1_{\mac H_{m-r}}.
\end{align}
Equivalently, on the \(m\)-particle sector, this operation first chooses \(r-s\) auxiliary particles to combine with the original \(s\)-particle operator, and then chooses the remaining \(m-r\) particles. For a fixed set of \(s\) particles on which \(X\) acts, the number of intermediate \(r\)-particle sectors containing it is \(\binom{m-s}{r-s}\). Hence
\begin{align}
    \mathcal E_{r,m}\!\left(\mathcal E_{s,r}(X)\right)
    =
    \binom{m-s}{r-s}
    \mathcal E_{s,m}(X),
\end{align}
which proves Eq.~\eqref{eq:extension_composition}.
\end{proof}

Next, we provide a basic property of the contraction map for a projection map.
\begin{lemma}[Contraction of exterior powers of projections]
\label{lem:contraction-exterior-projection}
Let \(P=P^2=P^\dagger\in\operatorname{End}(\mathbb C^N)\) be an orthogonal projection of rank \(\ell\). For \(0\leq t\leq r\leq N\), the contraction map satisfies
\begin{align}
    \mathcal C_{r,t}\!\left(P^{\wedge r}\right)
    =
    \binom{\ell-t}{r-t}
    P^{\wedge t}.
    \label{eq:contraction_exterior_projection}
\end{align}
In particular, if \(\operatorname{rank}P=r\), then
\begin{align}
    \mathcal C_{r,t}\!\left(P^{\wedge r}\right)
    =
    P^{\wedge t}.
    \label{eq:contraction_rank_r_projection}
\end{align}
\end{lemma}
\begin{proof}
Choose an orthonormal one-particle basis in which
$
    P
    =
    \sum_{i=1}^{\ell}
    \ketbra{i}{i}.
$
Then
\begin{align}
    P^{\wedge r}
    =
    \sum_{\substack{\vec z\in\mac S_{N,r}\\
    \vec z\subset [\ell]}}
    \ketbra{\vec z}{\vec z},
    \qquad
    P^{\wedge t}
    =
    \sum_{\substack{\vec y\in\mac S_{N,t}\\
    \vec y\subset [\ell]}}
    \ketbra{\vec y}{\vec y}.
\end{align}
We compute the matrix elements of
\(\mathcal C_{r,t}(P^{\wedge r})\). By the definition of the contraction map,
for \(\vec p,\vec q\in\mac S_{N,t}\),
\begin{align}
    \left\langle \vec p \middle|
        \mathcal C_{r,t}(P^{\wedge r})
    \middle| \vec q \right\rangle
    =
    \operatorname{tr}_{\mac H_r}
    \left[
        P^{\wedge r}
        \purple{
            a_{q_1}^\dagger\cdots a_{q_t}^\dagger
            a_{p_t}\cdots a_{p_1}
        }
    \right].
\end{align}
Using the above expansion of \(P^{\wedge r}\), this becomes
\begin{align}
    \sum_{\substack{\vec z\in\mac S_{N,r}\\
    \vec z\subset [\ell]}}
    \left\langle \vec z \middle|
        \purple{
            a_{q_1}^\dagger\cdots a_{q_t}^\dagger
            a_{p_t}\cdots a_{p_1}
        }
    \middle| \vec z \right\rangle .
\end{align}
The diagonal matrix element is nonzero only when \(\vec p=\vec q\) and
\(\vec p\subset \vec z\). In that case, the annihilation string removes the
occupied modes in \(\vec p\), and the creation string restores the same modes,
so the value is \(1\). Hence
\begin{align}
    \left\langle \vec p \middle|
        \mathcal C_{r,t}(P^{\wedge r})
    \middle| \vec q \right\rangle
    =
    \delta_{\vec p,\vec q}
    \#\left\{
        \vec z\in\mac S_{N,r}
        \mid
        \vec p\subset \vec z\subset [\ell]
    \right\}.
\end{align}
If \(\vec p\not\subset[\ell]\), this number is zero. If
\(\vec p\subset[\ell]\), then the remaining \(r-t\) occupied modes must be
chosen from the \(\ell-t\) elements of \([\ell]\setminus \vec p\). Therefore
\begin{align}
    \#\left\{
        \vec z\in\mac S_{N,r}
        \mid
        \vec p\subset \vec z\subset [\ell]
    \right\}
    =
    \binom{\ell-t}{r-t}.
\end{align}
Thus
\begin{align}
    \left\langle \vec p \middle|
        \mathcal C_{r,t}(P^{\wedge r})
    \middle| \vec q \right\rangle
    =
    \binom{\ell-t}{r-t}
    \left\langle \vec p \middle|
        P^{\wedge t}
    \middle| \vec q \right\rangle .
\end{align}
Since this holds for all \(\vec p,\vec q\in\mac S_{N,t}\), we obtain
\begin{align}
    \mathcal C_{r,t}\!\left(P^{\wedge r}\right)
    =
    \binom{\ell-t}{r-t}
    P^{\wedge t}.
\end{align}
If \(\ell=r\), then \(\binom{\ell-t}{r-t}=1\), and hence
\begin{align}
    \mathcal C_{r,t}\!\left(P^{\wedge r}\right)
    =
    P^{\wedge t}.
\end{align}
\end{proof}

\section{Orbital-rotation Classical Shadows}

\subsection{Previous work on fermionic shadow tomography}

In this section, we briefly overview the related works of fermionic shadow tomography. Classical shadow tomography provides a general framework for estimating many
observables from randomized measurements
\cite{Huang:2020tih, elben2023randomized}.  In a standard formulation, one
draws a random unitary \(U\) from an ensemble \(\mathcal{U}\), applies it to a
state \(\rho\), and measures the rotated state in a fixed computational basis.
For a measurement outcome \(z\), the corresponding averaged measurement channel
is given by
\begin{align}
\label{eq:shadow_channel}
    \mathcal{M}_{\mathcal{U}}(\rho)
    \coloneqq
    \mathbb{E}_{U\sim\mathcal{U}}
    \sum_{z}
    \bra{z}U\rho U^\dagger\ket{z}\,
    U^\dagger\ketbra{z}U .
\end{align}
When this channel is invertible, an unbiased estimator of the state is obtained
as
\begin{align}
\label{eq:shadow_def}
    \hat{\rho}_{U,z}
    \coloneqq
    \mathcal{M}_{\mathcal{U}}^{-1}
    \bigl(
        U^\dagger\ketbra{z}U
    \bigr).
\end{align}
This formalism is particularly effective for qubit systems with local Pauli
measurements, where the sample complexity can be controlled in terms of the
Pauli weight of the observables.

For fermionic systems, however, the relevant notion of locality is not
necessarily qubit locality.  A \(k\)-body fermionic observable is local in the
fermionic mode algebra, but after a fermion-to-qubit encoding
\cite{bravyi2002fermionic}, it is generally mapped to a linear combination of
Pauli strings whose weights depend on the encoding and may grow with the number
of modes.  Therefore, the usual local-Pauli shadow bound, which scales as
\(\mathcal{O}(3^w)\) for Pauli weight \(w\), does not directly provide an
encoding-independent guarantee for estimating \(k\)-particle reduced density
matrices.

This observation motivates shadow-tomography protocols that exploit the
fermionic structure directly
\cite{Zhao:2020vxp,wan2023matchgate,ogorman2022fermionic,low2022classical}.
In particular, when the particle number is fixed, it is natural to use
particle-number-preserving random single-particle basis rotations.  Low's
protocol \cite{low2022classical} follows this approach by
drawing a Haar-random unitary \(u\sim \mathrm{Haar}(\mtr{U}(N))\) on the
single-particle space and applying its second-quantized action to the
\(\eta\)-particle sector.  We denote this induced action by
\begin{align}
    U_\eta(u) \coloneqq  u^{\wedge \eta } .
\end{align}
After applying \(U_\eta(u)\), the state is measured in the occupation-number
basis, yielding an occupation pattern
\(\vec z\in\mac{S}_{N,\eta}\). 
By the Born rule, this outcome occurs with probability
$
    p_\rho(\vec z\mid u)
    =
    \tr_{\mac{H}_\eta}\!\qty[
        \rho\, U_\eta^\dagger(u) P_{\vec z}^{\wedge \eta} U_\eta(u)
    ] .
$
As a result, the target system \(\rho\) is projected onto \(U_\eta^\dagger(u) P_{\vec z}^{\wedge \eta} U_\eta(u)\).

Low's estimator for the \(k\)-RDM is then constructed from the classical data
\((u,\vec z)\).  In Low's notation, it can be written as
\begin{align}
    \widehat{D}^{(k)}(u,\vec z)
    &=
    U_k^\dagger\!\left(v_{\vec z}^\dagger u\right)
    E_{\eta,k}
    U_k\!\left(v_{\vec z}^\dagger u\right),
    \\
    E_{\eta,k}
    &=
    \sum_{\vec r\in\mac{S}_{N,k}}
    \ketbra{\vec r}\,
    \frac{
        \binom{\eta-s'}{k-s'}
        \binom{N-\eta+s'}{s'}
    }{
        (-1)^{k+s'}\binom{k}{s'}
    },
\end{align}
where \(U_k(w)\coloneqq  w^{\wedge k}\), \(v_{\vec z}\in\mtr{U}(N)\) is a mode
permutation that maps the reference occupation pattern
\([\eta]\coloneqq (1,\ldots,\eta)\) to \(\vec z\), and
\(
    s' \coloneqq |\vec r\cap[\eta]|.
\)
This protocol exploits the fixed-particle-number structure and gives improved
sample-complexity guarantees compared with approaches that do not use the
particle-number constraint explicitly.

The expression above is written with respect to the auxiliary reference pattern
\([\eta]\) and the choice of a permutation \(v_{\vec z}\).  In the following
sections, we reformulate the same estimator in a basis-free Grassmannian form,
where the measurement outcome is regarded as an \(\eta\)-dimensional subspace
\(R\in\mac{P}_\eta\).  This reformulation separates the intrinsic measurement
outcome from the auxiliary coordinate choices, gives a more transparent
description of the inverse measurement channel on the \(k\)-RDM space, and
provides a convenient starting point for deriving variance and covariance
bounds for general \(k\)-body observables.

\subsection{Tomographic completeness}
\label{sec:tomographic_completeness}

In this section, we revisit the tomographic-completeness
argument. The inverse-channel construction used in Low's estimator relies on an informational-completeness property~\cite{dAriano2004informationally, Huang:2020tih} of the measurement projectors.  In the orbital-rotation shadow protocol, this means that the orbit of a single reference Slater projector under number-conserving basis rotations must span the whole operator space.  
We state this property formally as follows.

\begin{theorem}[Tomographic completeness]
\label{thm:tomograhic_completeness_app}
Let \(N\) and \(\eta\) be integers satisfying \(1\le \eta < N\). Then
\begin{align}
    \label{eq:tomo-complete-span}
    \mtr{span}_{\mathbb C}
    \left\{
        U_\eta(u)\Pi^{(\eta)}_{[\eta]}U_\eta(u)^\dagger
        \mid
        u\in\mtr U(N)
    \right\}
    =
    \mtr{End}(\mac H_\eta),
\end{align}
where
$
    \Pi^{(\eta)}_{[\eta]} \coloneq \ketbra{[\eta]}
    \in \mtr{End}(\mac H_\eta)$, and $
    [\eta] \coloneq (1,2,\ldots,\eta)\in \mac S_{N,\eta}.
$
\end{theorem}
We remark that this tomographic-completeness statement is essential to justify the inversion of the measurement channel.  However, one step in the original proof requires an additional justification.  We therefore give a self-contained rigorous proof of Theorem~\ref{thm:tomograhic_completeness_app} 
after explaining why the
proof strategy in Ref.~\cite{low2022classical} is insufficient.

To explain this reason, we introduce the concept of Slater rank \cite{eckert2002quantum}.  Recall that a vector in
\( (\mathbb C^N)^{\wedge \eta} \) is called a Slater determinant, or a
decomposable \(\eta\)-vector, when it can be written as $
    \ket{\psi}
    =
    \ket{v_1}\wedge\cdots\wedge \ket{v_\eta}
$
for linearly independent one-particle orbitals
\(\ket{v_1},\ldots,\ket{v_\eta}\in\mathbb C^N\). Generally, the Slater rank
of \(\ket{\Psi} \in \qty( \mathbb{C}^N)^{\wedge \eta} \) is the smallest integer \(r\) for which
\begin{align}
    \ket{\psi}
    =
    \sum_{\alpha=1}^r
    c_\alpha\,
    \ket{v_{\alpha,1} }\wedge\cdots\wedge \ket{v_{\alpha,\eta} }.
\end{align}
Thus Slater rank one is equivalent to being a single Slater
determinant.  Moreover, single-particle rotations preserve Slater
rank, because 
\begin{align}
    U_\eta(u)
    \left(
        \ket{v_1}\wedge\cdots\wedge \ket{v_\eta}
    \right)
    =
    u\ket{v_1}\wedge\cdots\wedge u\ket{v_\eta} .
\end{align}

Using this concept, we point out a subtle gap in the proof of
Theorem~2 in Ref.~\cite{low2022classical}. To prove tomographic completeness, the original proof
considers the Hermitian combination of matrix units for distinct
\(
    \vec p,\vec q\in\mac S_{N,k}
\)
and \(\phi\in[0,2\pi)\):
\begin{align}
    D^{\vec p}_{\vec q;\phi}
    \coloneq
    e^{i\phi}\ketbra{\vec p}{\vec q}
    +
    e^{-i\phi}\ketbra{\vec q}{\vec p}
    \in \mtr{End}(\mac H_k).
\end{align}
This operator has normalized eigenvectors
$
    \ket{\phi_\pm}
    =
    \frac{1}{\sqrt2}
    \left(
        e^{i\phi/2}\ket{\vec p}
        \pm
        e^{-i\phi/2}\ket{\vec q}
    \right)
    \in \mac H_k,
$
which satisfy
$
    D^{\vec p}_{\vec q;\phi}\ket{\phi_\pm}
    =
    \pm \ket{\phi_\pm}.
$
By definition, $ D^{\vec p}_{\vec q;\phi}$ can be expressed as follows, 
\begin{align}
    D^{\vec p}_{\vec q;\phi}
    =
    \ketbra{\phi_+}{\phi_+}
    -
    \ketbra{\phi_-}{\phi_-}.
\end{align}
We remark that this is a valid spectral decomposition, but the resulting
eigenvectors $\ket{\phi_\pm}$ are not generally obtained by single-particle rotations.  More
precisely, one cannot find a single-particle rotation \(U_k(w)\) for $k\geq 2$ satisfying
\begin{align}
    U_k(w)\ket{\vec p}=\ket{\phi_+},
    \qquad
    U_k(w)\ket{\vec q}=\ket{\phi_-}.
    \label{eq:contradicted_unitary}
\end{align}
Note that, for \(k=1\), this causes no difficulty, since every nonzero one-particle
vector has Slater rank one and such a rotation can be chosen.  For
\(k\ge 2\), however, the eigenvectors \(\ket{\phi_\pm}\) need not have
Slater rank one, so Eq.~\eqref{eq:contradicted_unitary} cannot generally be
satisfied, as shown in the simple example below.  Indeed, this would contradict the preservation of Slater rank
under single-particle rotations.  Therefore, for \(k\ge 2\), the step in the
original proof that uses such a \(U_k(w)\) to express
\(D^{\vec p}_{\vec q}\) as a linear combination of projectors
\(U_\eta(u)\Pi^{(\eta)}_{[\eta]}U_\eta(u)^\dagger\) is not fully justified.

For instance, when we consider a case for $k=2$ and set
\(
    \vec p=(1,2) \in \mac{S}_{4,2}
\)
and
\(
    \vec q=(3,4) \in \mac{S}_{4,2}
\), the corresponding eigenvectors are expressed as follows,
\begin{align}
    \ket{\phi_+}
    =
    \frac{1}{\sqrt2}
    \left(
        e^{i\phi/2}\ket{1,2}
        +
        e^{-i\phi/2}\ket{3,4}
    \right)
    \in (\mathbb C^4)^{\wedge 2} .
\end{align}
This vector is the sum of two Slater determinants, and hence has
Slater rank at most two.  However, the Slater rank of this vector is not one. Indeed, if a
two-particle vector were a single Slater determinant, say
\(\ket{\chi}=\ket{v}\wedge \ket{w} \in \mac{H}_2 \), then its exterior square would vanish:
\begin{align}
    \ket{\chi}\wedge\ket{\chi}
    =
    \ket{v}\wedge \ket{w} \wedge \ket{v} \wedge \ket{w}
    =
    0.
\end{align}
The last equality follows from the alternating property of the exterior
product, namely \(\ket{v}\wedge\ket{v}=0\) for every
\(\ket{v}\in\mathbb{C}^N\).
Thus a nonzero exterior square shows that the vector is not of Slater
rank one.  For the present vector, we compute
\begin{align}
    \ket{\phi_+}\wedge\ket{\phi_+}
    &=
    \frac12
    \left(
        e^{i\phi/2}\ket{1,2}
        +
        e^{-i\phi/2}\ket{3,4}
    \right)
    \wedge
    \left(
        e^{i\phi/2}\ket{1,2}
        +
        e^{-i\phi/2}\ket{3,4}
    \right) \\
    &=
    \ket{1,2,3,4}
    \neq 0 .
\end{align}
Therefore \(\ket{\phi_+}\) is not a single Slater determinant, and it does not have Slater rank one.

We therefore do not prove tomographic completeness by diagonalizing
matrix-unit combinations.  Instead, we extract the desired matrix
units as discrete Fourier components of projectors onto genuine
rotated Slater determinants.
\\
\noindent \textit{Proof of Theorem~\ref{thm:tomograhic_completeness_app}.} 
We first outline the proof.  Since the matrix units
$\ketbra{\vec p}{\vec q}$ form a basis of
\(\mtr{End}(\mac H_\eta)\), it is sufficient to show that each such
matrix unit is contained in the span on the LHS of
Eq.~\eqref{eq:tomo-complete-span}.  The diagonal case like $\ketbra{\vec p}$ is obtained
directly by applying a permutation matrix to the reference Slater
determinant \(\ket{[\eta]}\).  For distinct configurations \(\vec p\neq \vec q\), write
\(\vec a=(a_1,\ldots,a_m)\) and
\(\vec b=(b_1,\ldots,b_m)\) for the non-common modes of
\(\vec p\) and \(\vec q\), respectively. For each \(j\in[m]\), we introduce an independent phase
\(\phi_j\in[0,2\pi)\) and consider a two-dimensional Givens rotation
acting only on the subspace
\(\operatorname{span}\{\ket{a_j},\ket{b_j}\}\).  With respect to the
ordered basis \((\ket{a_j},\ket{b_j})\), its action is given by
\begin{align}
    \frac{1}{\sqrt{2}}
    \mqty[
        1 & -e^{-i\phi_j} \\
        e^{i\phi_j} & 1
    ] .
\end{align}
Since these subspaces are mutually orthogonal, the product of these
rotations defines a well-defined block-diagonal one-particle unitary,
which we denote by \(u(\vec\phi;\vec a,\vec b)\) where $\vec \phi = (\phi_1 , \ldots, \phi_m ) \in [0,2\pi)^{m}$.  We then consider
the following rotated Slater projector
\begin{align}
        U_\eta(u(\vec\phi;\vec a,\vec b))\Pi_{\vec p}
    U_\eta(u(\vec\phi;\vec a,\vec b))^\dagger .
    \label{eq:target_projector}
\end{align}
The core idea is to extract the desired matrix unit from a linear combination of rotated Slater projectors.  Namely, we
choose an appropriate phase set \(A\subset[0,2\pi)^m\) and coefficients
\(c_{\vec\phi}\) such that
\begin{align}
    \sum_{\vec\phi\in A}
    c_{\vec\phi}\,
    U_\eta(u(\vec\phi;\vec a,\vec b))\Pi_{\vec p} \
    U_\eta(u(\vec\phi;\vec a,\vec b))^\dagger
    =
    \ketbra{\vec p}{\vec q}.
\end{align}
Such a choice is possible because the unitary
\(u(\vec\phi;\vec a,\vec b)\) is a product of independent
two-dimensional Givens rotations. By expanding the rotated projector in the occupation-number basis, we
obtain a linear combination of matrix units whose coefficients depend on
the phases $\vec \phi$:
\begin{align}
    U_\eta(u(\vec\phi;\vec a,\vec b))\Pi_{\vec p}\
    U_\eta(u(\vec\phi;\vec a,\vec b))^\dagger
    =
    \frac{1}{2^m}
    \sum_{\vec s,\vec t}
\,
    e^{i\Theta_{\vec s,\vec t}(\vec\phi)}
    \ketbra{\vec s}{\vec t}.
\end{align}
Here, the phase \(\Theta_{\vec s,\vec t}(\vec\phi)\) is determined by
the pair \((\vec s,\vec t)\).  Indeed, different replacement patterns
give different phase factors.  In particular, every unwanted matrix unit has a phase  different from that of the target term
\(\ketbra{\vec p}{\vec q}\).  This allows us to choose a finite phase set
and suitable coefficients so that the discrete phase sum cancels all
unwanted terms.

Following this core idea, we will prove the theorem.
For $   \vec p,\vec q\in\mac S_{N,\eta},$
the rank-one operators
$
    \ket{\vec p}\!\bra{\vec q}
$
form a basis for
\(
    \mtr{End}(\mac H_\eta)
\).
Therefore, it suffices to show that each operator
\(
    \ket{\vec p}\!\bra{\vec q}
\)
belongs to the span in Eq.~\eqref{eq:tomo-complete-span}.

For the remainder of the proof, we use the following notation. Let
\(
    \vec p,\vec q\in\mac S_{N,\eta}
\)
with \(\vec p\neq \vec q\). When forming intersections and set differences,
we identify ordered tuples with their underlying sets. To this end, we introduce
$
    \vec r
    \coloneq
    \vec p\cap \vec q,
    m
    \coloneq
    \eta-|\vec r|.
$
We also write the non-common modes as
$
    \vec a \coloneq  \vec p\setminus\vec r
    =
    (a_1 < \cdots < a_m), \vec b \coloneq  \vec q\setminus\vec r
    =
    (b_1 < \cdots < b_m) \in \mac{S}_{N, m}.
$
The modes in \(\vec r\) are kept fixed throughout the argument, and we rotate
only within the two-dimensional subspaces
\(
    \mtr{span}\{\ket{a_j},\ket{b_j}\}
\).
This keeps the proof entirely within the \(\eta\)-particle sector
\(\mac H_\eta\).

We first recall the one-particle calculation.  For distinct
$a,b\in[N]$ and $\phi\in[0,2\pi)$, define
$u(\phi;a,b)\in\mtr U(N)$ by
\begin{equation}
    u(\phi;a,b)\ket{a}
    \coloneq
    \frac{1}{\sqrt 2}
    \left(
        \ket{a}+e^{i\phi}\ket{b}
    \right),
    \qquad
    u(\phi;a,b)\ket{b}
    \coloneq
    \frac{1}{\sqrt 2}
    \left(
        -e^{-i\phi}\ket{a}+\ket{b}
    \right),
\end{equation}
and let it act as the identity on all other basis vectors. By using this Givens rotation, we will introduce
\begin{equation}
    \Pi^{(1)}(\phi;a,b)
    \coloneq
    U_1\!\left(u(\phi;a,b)\right)
    \ketbra{a}\,
    U_1\!\left(u(\phi;a,b)\right)^\dagger .
\end{equation}
This matrix is also expressed as 
\begin{equation}
    \Pi^{(1)}(\phi;a,b)
    =
    \frac12
    \left(
        \ketbra{a}{a}
        +
        \ketbra{b}{b}
        +
        e^{-i\phi}\ketbra{a}{b}
        +
        e^{i\phi}\ketbra{b}{a}
    \right).
\end{equation}
Let
$\Phi_4\coloneq\{0,\pi/2,\pi,3\pi/2\}$.  By the discrete
Fourier orthogonality relation
\begin{equation}
    \sum_{\phi\in\Phi_4} e^{i\ell\phi}
    =
    \begin{cases}
        4, & \ell\equiv 0 \pmod 4,\\
        0, & \ell\not\equiv 0 \pmod 4,
    \end{cases}
\end{equation}
we have
\begin{equation}
\begin{aligned}
    \frac12
    \sum_{\phi\in\Phi_4}
    e^{i\phi}\,
    \Pi^{(1)}(\phi;a,b)
    &=
    \frac14
    \sum_{\phi\in\Phi_4}
    e^{i\phi}
    \left(
        \ketbra{a}{a}
        +
        \ketbra{b}{b}
        +
        e^{-i\phi}\ketbra{a}{b}
        +
        e^{i\phi}\ketbra{b}{a}
    \right)
    \\
    &=
    \frac14
    \left[
        \left(\sum_{\phi\in\Phi_4} e^{i\phi}\right)
        \left(
            \ketbra{a}{a}
            +
            \ketbra{b}{b}
        \right)
        +
        \left(\sum_{\phi\in\Phi_4} 1\right)
        \ketbra{a}{b}
        +
        \left(\sum_{\phi\in\Phi_4} e^{2i\phi}\right)
        \ketbra{b}{a}
    \right]
    \\
    &=
    \frac14
    \left(
        0
        +
        4\ketbra{a}{b}
        +
        0
    \right)
    =
    \ketbra{a}{b}.
\end{aligned}
\label{eq:one-particle-fourier-extraction}
\end{equation}
Therefore, $
    \ketbra{a}{b}
    =
    \frac12
    \sum_{\phi\in\Phi_4}
    e^{i\phi}\,
    \Pi^{(1)}(\phi;a,b) 
$ holds, so a one-particle off-diagonal matrix unit is a linear
combination of rotated diagonal projectors.

We now apply the same idea to $\ketbra{\vec p}{\vec q}$ with $\vec a, \vec b \in \mac{S}_{N,m}$.
For $\vec\phi=(\phi_1,\ldots,\phi_m)\in\Phi_4^{\times m}$, let
$u(\vec\phi;\vec a,\vec b)\in\mtr U(N)$ be the unitary that
acts as $u(\phi_j;a_j,b_j)$ on
$\mtr{span}\{\ket{a_j},\ket{b_j}\}$ for each $j=1,\ldots,m$,
and as the identity on the orthogonal complement.  Since the
modes $a_1,\ldots,a_m,b_1,\ldots,b_m$ are all distinct,
these two-dimensional rotations act on mutually orthogonal
subspaces and hence define a well-defined block-diagonal unitary.

We will introduce the extended version of $\Pi^{(1)} (\phi;  a, b)$ as follows,
\begin{equation}
    \Pi^{(\eta)}(\vec\phi;\vec p,\vec q)
    \coloneq
    U_\eta \!\left(u(\vec\phi;\vec a,\vec b)\right)
    \Pi_{\vec p}
    U^\dagger_\eta\!\left(u(\vec\phi;\vec a,\vec b)\right) .
\end{equation}
We now expand this projector explicitly.  For $I\subseteq[m]$,
we will define the occupation configuration of length $m$ as follows,
\begin{equation}
    \vec \sigma_I
    \coloneq
    \mtr{sort}
    \left(
        \{b_j:j\in I\}
        \cup
        \{a_j:j\notin I\}
    \right)
\end{equation}
where $\mtr{sort}$ means that the occupied modes are written in
increasing order.  We also write
$\phi_I\coloneq\sum_{j\in I}\phi_j$, with
$\phi_\emptyset=0$. For example, suppose that
$
    \vec p=(1,3,5),
    \vec q=(1,4,6) \in \mac{S}_{6,3}.
$ Then the common part is $\vec r=(1)$, and the non-common modes
are $ \vec a = (3,5), \vec b=(4,6)$, hence $m=|\vec p| - |\vec r| = 2$. The four subsets of $[m]$
give the following occupation configurations:
\begin{equation}
\begin{array}{c|c|c}
    I & \vec \sigma_I & \vec r \cup \vec \sigma_I\\ \hline
    \emptyset & (3,5) = \vec a & (1,3,5)=\vec p \\
    \{1\}    & (4,5) & (1,4,5) \\
    \{2\}  & (3,6) & (1,3,6) \\
    \{1,2\} & (4,6) = \vec b & (1,4,6)=\vec q
\end{array}
\end{equation}
Thus, $I$ records which of the modes $a_j$ are replaced by
$b_j$.  For instance, $I=\{1\}$ means that $a_1=3$ is
replaced by $b_1=4$, while $a_2=5$ is left unchanged. In the same example, the corresponding phase sums are $
    \phi_\emptyset=0,
    \phi_{\{1\}}=\phi_1,
    \phi_{\{2\}}=\phi_2,
    \phi_{\{1,2\}}=\phi_1+\phi_2.
$

Using this notation, we will introduce appropriate Givens rotations 
\(u(\vec\phi;\vec a,\vec b)\). We will define the action of
\(u(\vec\phi;\vec a,\vec b)\) on the non-common element basis vector as follows,
\begin{align}
    \forall j \in [m], \ u(\vec\phi;\vec a,\vec b)\ket{a_j}
    \coloneq 
    \frac{1}{\sqrt{2}}
    \left(
        \ket{a_j}+e^{i\phi_j}\ket{b_j}
    \right),
\end{align}
Additionally, the action of \(u(\vec\phi;\vec a,\vec b)\) on the common vector $\ket{\vec r} = \ket{\vec p \cap \vec q}$ is defined by
\begin{align}
    \forall j \in [ \abs{\vec r}], \ u(\vec\phi;\vec a,\vec b)\ket{ r_j } \coloneq \ket{r_j}. 
    \label{eq:def_action_Givens_rot_common}
\end{align}
By multilinearity of the exterior product, we obtain
\begin{align}
    u(\vec\phi;\vec a,\vec b)^{\wedge m}\ket{\vec a}
    &=
    \bigwedge_{j=1}^m
    u(\vec\phi;\vec a,\vec b)\ket{a_j} \\
    &=
    \frac{1}{2^{m/2}}
    \bigwedge_{j=1}^m
    \left(
        \ket{a_j}+e^{i\phi_j}\ket{b_j}
    \right) \\
    &=
    \frac{1}{2^{m/2}}
    \sum_{I\subseteq[m]}
    e^{i\phi_I}
    \purple{
        \varepsilon_{\wedge}
        \left(
            \vec a_{I^c},
            \vec a_I
        \right)
    }
    \qty(\bigwedge_{j\notin I}\ket{a_j} )
    \wedge
    \qty( \bigwedge_{j\in I}\ket{b_j} ),
\end{align}
where \(\phi_I:=\sum_{j\in I}\phi_j\).  After reordering each exterior
product into increasing order, a fermionic sign appears.  Denoting this
sign by \(\purple{\varepsilon_{\wedge}(\vec a_{I^c},\vec b_I)}\in\{\pm1\}\), we obtain
\begin{align}
    u(\vec\phi;\vec a,\vec b)^{\wedge m}\ket{\vec a}
    =
    \frac{1}{2^{m/2}}
    \sum_{I\subseteq[m]}
    \purple{
        \varepsilon_{\wedge}
        \left(
            \vec a_{I^c},
            \vec a_I
        \right)
        \varepsilon_{\wedge}
        \left(
            \vec a_{I^c},
            \vec b_I
        \right)
    }
    e^{i\phi_I}\ket{\vec \sigma_I}.
    \label{eq:expanded-rotated-slater-vector}
\end{align}
Here $\purple{\varepsilon_{\wedge}(\vec a,\varnothing)}=1$, because $I=\emptyset$ corresponds to
the original ordered vector $\ket{\vec a}$.  

Taking the outer product of both sides of
Eq.~\eqref{eq:expanded-rotated-slater-vector} and using the fact that
\(u(\vec\phi;\vec a,\vec b)\) acts trivially on the common mode
\(\ket{\vec r}\), as shown in
Eq.~\eqref{eq:def_action_Givens_rot_common}, we obtain
\begin{equation}
    \Pi(\vec\phi;\vec p,\vec q)
    =
    \frac{1}{2^m}
    \sum_{I,J\subseteq[m]}
    \purple{
        \varepsilon_{\wedge}
        \left(
            \vec a_{I^c},
            \vec a_I
        \right)
        \varepsilon_{\wedge}
        \left(
            \vec a_{I^c},
            \vec b_I
        \right)
        \varepsilon_{\wedge}
        \left(
            \vec a_{J^c},
            \vec a_J
        \right)
        \varepsilon_{\wedge}
        \left(
            \vec a_{J^c},
            \vec b_J
        \right)
    }
    e^{i(\phi_I-\phi_J)}
    (\ket{\vec \sigma_I} \wedge \ket {\vec r} ) (\bra {\vec \sigma_J}  \wedge \bra{\vec r}) .
    \label{eq:expanded-rotated-slater-projector}
\end{equation}
The desired matrix unit corresponds to $I=\emptyset$ and
$J=[m]$. \purple{ 
In this case, the four shuffle signs preceding the phase
factor in Eq.~\eqref{eq:expanded-rotated-slater-projector} are all equal
to one.  Moreover,
\begin{align}
    \ket{\vec a \ }\wedge\ket{\vec r}
    =
    \purple{\varepsilon_{\wedge}(\vec a,\vec r)}
    \ket{\vec p},\
    \bra{\vec b}\wedge\bra{\vec r}
    =
    \purple{\varepsilon_{\wedge}(\vec b,\vec r)}
    \bra{\vec q }.
\end{align}
Therefore, the corresponding term is
\begin{align}
    \frac{1}{2^m}
    \purple{
        \varepsilon_{\wedge}(\vec a,\vec r)
        \varepsilon_{\wedge}(\vec b,\vec r)
    }
    e^{-i\sum_{j=1}^m\phi_j}
    \ketbra{\vec p}{\vec q}.
\end{align}
Similarly, the reverse matrix unit $\ketbra{\vec q}{\vec p}$
corresponds to $I=[m]$ and $J=\emptyset$, carries the same ordering
sign, and has phase
$e^{i\sum_j\phi_j}$.
}

We now isolate $\ketbra{\vec p}{\vec q}$ by a discrete
Fourier projection. We will multiply
Eq.~\eqref{eq:expanded-rotated-slater-projector} by
$e^{i\sum_{\ell=1}^m\phi_\ell}$ and sum over
$\vec\phi\in\Phi_4^{\times m}$.  For fixed $I,J\subseteq[m]$, the phase
factor becomes
$\prod_{\ell=1}^m e^{i(1+\purple{\mathbf 1_{\{\ell\in I\}}}
-\purple{\mathbf 1_{\{\ell\in J\}}})\phi_\ell}$, where $\purple{\mathbf 1}$ denotes an indicator function defined by
\begin{align}
    \purple{\mathbf 1_{\{l \in I\}}} =\begin{cases}
        1 \quad & \text{when } l \in I, \\
        0 \quad & \text{otherwise.}
    \end{cases}
\end{align}
Using this function, we obtain the following equality,
\begin{equation}
    \sum_{\vec\phi\in\Phi_4^{\times m}}
    e^{i\sum_{\ell=1}^m\phi_\ell}
    e^{i(\phi_I-\phi_J)}
    =
    \prod_{\ell=1}^m
    \left(
        \sum_{\phi_\ell\in\Phi_4}
        e^{i\left(
            1+\purple{\mathbf 1_{\{\ell\in I\}}}
            -
            \purple{\mathbf 1_{\{\ell\in J\}}}
        \right)\phi_\ell}
    \right).
    \label{eq:product-fourier-factor}
\end{equation}
For each $\ell$, the integer
$1+\purple{\mathbf 1_{\{\ell\in I\}}}-\purple{\mathbf 1_{\{\ell\in J\}}}$ is one of
$0,1,2$.  By the discrete Fourier orthogonality relation on
$\Phi_4$, the inner sum in
Eq.~\eqref{eq:product-fourier-factor} is nonzero only when this
integer is $0$.  This happens if and only if
$\ell\notin I$ and $\ell\in J$.  For this to hold for every
$\ell=1,\ldots,m$, we must have $I=\emptyset$ and $J=[m]$.

Thus every term in the expansion vanishes under the Fourier sum
except the desired matrix unit $\ketbra{\vec a}{\vec b}$.
Consequently, recalling that $\ket{\vec a} = \ket{\vec p \setminus \vec r} $ and $\ket{\vec b} = \ket{\vec q \setminus \vec r}$ provides
\begin{equation}
    \sum_{\vec\phi\in\Phi_4^{\times m}}
    e^{i\sum_{j=1}^m\phi_j}
    \Pi(\vec\phi;\vec p,\vec q)
    = 2^m
    \purple{
        \varepsilon_{\wedge}(\vec a,\vec r)
        \varepsilon_{\wedge}(\vec b,\vec r)
    }
    \ketbra{\vec p}{\vec q}.
\end{equation}
Since $\purple{
    \left(
        \varepsilon_{\wedge}(\vec a,\vec r)
        \varepsilon_{\wedge}(\vec b,\vec r)
    \right)^{-1}
}
:=\purple{
    \varepsilon_{\wedge}(\vec a,\vec r)
    \varepsilon_{\wedge}(\vec b,\vec r)
}$, we obtain
\begin{equation}
    \ketbra{\vec p}{\vec q}
    =
    \frac{
        \purple{
            \varepsilon_{\wedge}(\vec a,\vec r)
            \varepsilon_{\wedge}(\vec b,\vec r)
        }
    }{2^m}
    \sum_{\vec\phi\in\Phi_4^{\times m}}
    e^{i\sum_{j=1}^m\phi_j}
    \Pi(\vec\phi;\vec p,\vec q).
    \label{eq:matrix-unit-fourier-extraction-detailed}
\end{equation}
This proves that every off-diagonal matrix unit
$\ketbra{\vec a}{\vec b}$ is a linear combination of rotated
Slater projectors with initial projector $\Pi_{\vec p}$. 

It remains only to replace the initial projector $\Pi_{\vec p}$
by the fixed reference projector $\Pi_{[\eta]}$.  Choose a
permutation matrix $v_{\vec p}\in\mtr U(N)$ such that
$U_\eta(v_{\vec p})\ket{[\eta]}=\ket{\vec p}$.  Then
$\Pi_{\vec p}
=U_\eta(v_{\vec p})\Pi_{[\eta]}U_\eta(v_{\vec p})^\dagger$.
Using the homomorphism property
$U_\eta(u)U_\eta(v)=U_\eta(uv)$, we get
\begin{equation}
    \Pi(\vec\phi;\vec p,\vec q)
    =
    U_\eta\!\left(
        u(\vec\phi;\vec a,\vec b)v_{\vec p}
    \right)
    \Pi_{[\eta]}
    U^\dagger_\eta\!\left(
        u(\vec\phi;\vec a,\vec b)v_{\vec p}
    \right) .
\end{equation}
Hence every projector appearing in
Eq.~\eqref{eq:matrix-unit-fourier-extraction-detailed} is of
the required form
$U_\eta(u)\Pi_{[\eta]}U_\eta(u)^\dagger$, with
$u\in\mtr U(N)$.  Therefore,
\begin{equation}
    \ketbra{\vec p}{\vec q}
    \in
    \mtr{span}_{\mathbb C}
    \left\{
        U_\eta(u)\Pi_{[\eta]}U_\eta(u)^\dagger
        \mid
        u\in\mtr U(N)
    \right\}.
\end{equation}
Since the matrix units $\ketbra{\vec p}{\vec q}$, where
$\vec p,\vec q\in\mac S_{N,\eta}$, form a basis of
$\mtr{End}(\mac H_\eta)$, the theorem follows.
\qed
\\
\\

This theorem establishes the informational completeness of the
orbital-rotation shadow protocol on the fixed-particle-number
sector: the corresponding measurement channel is injective on
this sector and admits an algebraic inverse on its image.  The proof is constructive, since each matrix unit
\(\ketbra{\vec p}{\vec q}\) is extracted by a finite Fourier
combination of rotated Slater projectors.  Thus, the inversion step is
justified at the algebraic level.

\section{Grassmannian reformulation}
\label{sec:Grassmannian_reformulation}

\subsection{Measurement channel as a Grassmannian integral}
\label{sec:Grassmannian_moment}

In this section, we introduce a Grassmannian formulation of the orbital-rotation shadow. 
This reformulation is useful for two reasons:
First, it rewrites the average over Haar-random orbital rotations and measurement outcomes in the computational basis as an invariant Grassmannian integral, which leads to a transparent expression for the measurement channel and its inverse. 
Second, it makes explicit what randomness is used in the protocol: the relevant measurement frame is not the full $\eta$-particle Hilbert space, but the Grassmannian of $\eta$-dimensional one-particle subspaces. 
In this sense, the protocol is randomized over the manifold of single Slater determinants compatible with particle-number symmetry.

To introduce the Grassmannian formulation, we focus on a sample
\((u,\vec z)\) of the orbital-rotation shadow protocol, where
\(u\sim \mathrm{Haar}(\mathrm U(N))\) and
\(\vec z\in\mac S_{N,\eta}\).  The measurement effect associated with
this outcome, pulled back to the original basis, is
$
    U_\eta(u)^\dagger P_{\vec z}^{\wedge\eta} U_\eta(u)
    =
    (u^\dagger P_{\vec z}u)^{\wedge\eta}.
$
Hence \((u,\vec z)\) naturally determines the rank-\(\eta\)
one-particle projector
$
R(u,\vec z)=u^\dagger P_{\vec z}u
$
and hence a point of the Grassmannian
\begin{align}
\mathcal P_\eta
=
\left\{
R\in \mathrm{End}(\mathbb C^N)
\mid
R^2=R,\ R^\dagger=R,\ \tr(R)=\eta
\right\}
\simeq \mathrm{Gr}_\eta(\mathbb{C}^N) \simeq \mtr{U}(N) / \qty( \mtr{U}(\eta) \times \mtr{U}(N-\eta) ).    
\end{align}
From this observation, the measurement outcome can be regarded intrinsically as a single Slater-determinant projector on $\mac{H}_\eta= (\mathbb C^N)^{\wedge \eta}$. 

In what follows, we use this Grassmannian picture only to state the resulting $\mtr{U}(N)$-invariant integral form of the measurement channel. The geometric construction of the well-defined invariant measure, together with the detailed justification of its equivalence to the original orbital-rotation protocol, is deferred
to Sec.\ref{sec:Geometric_costruction}.

Under Low's orbital-rotation shadow protocol, a measurement
outcome $(u,\vec z) \in \mtr{U}(N) \times \mac{S}_{N,\eta}$ determines the
rank-$\eta$ projector $u^\dagger P_{\vec z}u$, which represents the occupied
$\eta$-dimensional single-particle subspace pulled back to the original basis.
Let $F$ be a test function on an $\eta$-particle subspace $\mac{H}_\eta$. For this function, we define
the corresponding moment by
\begin{align}
    \purple{\mathbb{E}_{\rho}} [F]
    \coloneqq
    \int_{\mtr{U}(N)} \dd{u}
    \sum_{\vec z \in \mac{S}_{N,\eta}}
    \tr\!\left(
        \rho\,
        U_\eta^\dagger(u)
         P^{\wedge \eta }_{\vec z}
        U_\eta(u)
    \right)
    F\!\left(u^\dagger P_{\vec z}u\right).
\end{align}

Using the \(\mtr{U}(N)\)-invariant measure on
\(\mtr{Gr}_{\eta}(\mathbb{C}^N)\) induced by the Haar measure on
\(\mtr{U}(N)\), we obtain
\begin{align}
    \purple{\mathbb{E}_{\rho}} [F]
    &=
    \int_{\mtr{U}(N)} \dd{u}
    \sum_{\vec z \in \mac{S}_{N,\eta}}
    \tr\!\left(
        \rho\,
        U_\eta^\dagger(u)
        P^{\wedge \eta}_{\vec z}
        U_\eta(u)
    \right)
    F\!\left(u^\dagger P_{\vec z}u\right)                                      \\
    &=
    \int_{\mtr{U}(N)} \dd{u}
    \sum_{\vec z \in \mac{S}_{N,\eta}}
    \tr\!\left(
        \rho\,
        U_\eta^\dagger(w_{\vec z})
        P_{[\eta]}^{\wedge \eta}
        U_\eta(w_{\vec z})
    \right)
    F\!\left(w_{\vec z}^\dagger P_{[\eta]} w_{\vec z}\right)                    \\
    &=
    \sum_{\vec z \in \mac{S}_{N,\eta}}
    \int_{\mtr{U}(N)} \dd{w}\,
    \tr\!\left(
        \rho\,
        U_\eta^\dagger(w)
        P_{[\eta]}^{\wedge \eta}
        U_\eta(w)
    \right)
    F\!\left(w^\dagger P_{[\eta]} w\right)                                      \\
    &=
    \sum_{\vec z \in \mac{S}_{N,\eta}}
    \int_{\mtr{Gr}_\eta(\mathbb{C}^N)} \dd{R}\,
    \tr\!\left(
        \rho\,R^{\wedge \eta}
    \right)
    F(R)                                                                        \\
    &=
    \binom{N}{\eta}
    \int_{\mtr{Gr}_\eta(\mathbb{C}^N)} \dd{R}\,
    \tr\!\left(
        \rho\,R^{\wedge \eta}
    \right)
    F(R).
    \label{eq:Grassmannian_moment}
\end{align}
Here, for each \(\vec z\in \mac{S}_{N,\eta}\), we choose a permutation
unitary \(v_{\vec z}\in \mtr{U}(N)\) such that
$U_\eta(v_{\vec z})\ket{[\eta]}=\ket{\vec z}$. Equivalently,
$P_{\vec z}=v_{\vec z}P_{[\eta]}v_{\vec z}^\dagger$, and hence
$P_{\vec z}^{\wedge\eta}
=
U_\eta(v_{\vec z})
P_{[\eta]}^{\wedge\eta}
U_\eta^\dagger(v_{\vec z})$.
The second equality follows by setting
$w_{\vec z}=v_{\vec z}^\dagger u$. The third equality uses the left
invariance of the Haar measure, which allows us to replace
\(w_{\vec z}\) by a new integration variable \(w\).  The passage from the
unitary integral to the integral over  Grassmannian is obtained by pushing forward
the Haar measure under the map
$w\mapsto w^\dagger P_{[\eta]}w$. The corresponding geometric identification
is discussed in Sec.~\ref{sec:Geometric_costruction}.

\subsection{Derivation of the inverse measurement channel from Grassmannian moments}
\label{sec:inverse-measurement-channel-grassmannian}

In this section, we derive the inverse of the orbital-rotation measurement channel in terms of the contraction and extension maps. By applying the contraction map to the channel output, we obtain its \(t\)-particle exterior moment. We show that this moment depends only on the \(s\)-particle contractions of the input operator with \(s\leq t\). Hence, the measurement moments can be solved successively in increasing order of \(t\), and the resulting relations admit a simple binomial inversion. This gives both the inverse measurement channel and an unbiased estimator for the \(k\)-RDM. In this way, we recover Low's orbital-rotation fermionic shadow estimator \cite{low2022classical} without using the explicit hypergeometric-sum inversion appearing in the original derivation. This agreement confirms that the Grassmannian formulation gives a sector-consistent and equivalent description of the orbital-rotation shadow protocol.

To end this, we will use two technical facts. We first record a reference overlap formula that will be used to determine the coefficients in the Grassmannian moment expansion. This formula evaluates the overlap between the reference \(\eta\)-particle Slater projector and a \(t\)-particle coordinate minor.

\begin{lemma}[Reference overlap moment]
\label{lem:reference-overlap-moment}
Let $0 \leq t \leq \eta < N $. Additionally, we assume that \(\vec a\in\mac S_{N,t}\), and set
$
    q:=|\vec a\cap[\eta]|.
$
Then
\begin{align}
    \binom{N}{\eta}
    \int_{\mtr{Gr}_{\eta} (\mathbb{C}^N ) }
    \det R_{[\eta],[\eta]}\,
    \det R_{\vec a,\vec a}\,\dd R 
    =
    \frac{1}{\binom{N+1}{t}}
    \sum_{s=0}^{q}
    \frac{\binom{\eta-s}{t-s}}{\binom{t}{s}}
    \binom{q}{s}.
    \label{eq:reference-overlap-moment}
\end{align}
\end{lemma}

The proof of Lemma~\ref{lem:reference-overlap-moment} follows from the two-fold Grassmannian moment formula together with a combinatorial summation. Since the latter calculation is somewhat lengthy and does not play a conceptual role in the main argument, we defer it to Lemma~\ref{lem:technical-binomial-sum} in Sec.~\ref{sec:technical_lemmas}. 

Additionally, to write the measurement channel and inverse channel with contraction and extension map, we provide the following lemma,
\begin{lemma}
    \label{lem:expression_U(N)-invariant}
    Let $0 \leq t \leq \eta < N $. We assume a \purple{linear} map $\Psi: \mtr{End}(\mac{H}_\eta) \to \mtr{End}(\mac{H}_t)  $ satisfies
    \begin{align}
        \Psi( U_\eta(u) X U_\eta^\dagger (u)  ) = U_t (u)\Psi(X) U_t^\dagger (u), 
        \label{eq:U(N)-equivariant}
    \end{align}
    for every $u \in \mtr{U}(N)$ and $X \in \mtr{End} (\mac{H}_\eta)$. 
    Then, $\Psi$ is expressed as
    \begin{align}
        \Psi(\cdot ) = \sum^t_{s=0} c_{s} \mac{E}_{s,t} ( \mac{C}_{\eta, s} (\cdot)), 
    \end{align}
    where $c_s$ are coefficients.
\end{lemma}
The proof relies on representation theory and is deferred
to Sec.~\ref{sec:Irrep_endH}.

Using Lemmas~\ref{lem:reference-overlap-moment} and \ref{lem:expression_U(N)-invariant}, we can determine the explicit form of the measurement channel and the inverse channel. To this end, we will prove the following lemma.

\begin{lemma}
\label{lem:exterior-grassmannian-moment-hierarchy}
Let
\(
    \mac M_\eta:
    \operatorname{End}(\mac H_\eta)
    \longrightarrow
    \operatorname{End}(\mac H_\eta)
\)
be the orbital-rotation shadow measurement channel expressed by
\begin{align}
    \mac M_\eta(\rho)
    =
    \binom{N}{\eta}
    \int_{\mtr{Gr}_{\eta} (\mathbb{C}^N ) }
    \operatorname{tr}_{\mac H_\eta}
    \left[
        \rho R^{\wedge \eta}
    \right]
    R^{\wedge \eta}\,\dd R  .
    \label{eq:measurement-channel-wedge-dual}
\end{align}
For \(0\leq t\leq \eta\), define
\begin{align}
    \mac A_t(\rho)
    :=
    \mac C_{\eta,t}
    \left(
        \mac M_\eta(\rho)
    \right)
    \in \operatorname{End}(\mac H_t).
\end{align}
Then, for every \(\rho\in\operatorname{End}(\mac H_\eta)\),
\begin{align}
    \mac A_t(\rho )
    &=
    \binom{N}{\eta}
    \int_{\mtr{Gr}_{\eta} (\mathbb{C}^N ) }
    \operatorname{tr}_{\mac H_\eta}
    \left[
        \rho R^{\wedge \eta}
    \right]
    R^{\wedge t}\,\dd R 
    \label{eq:At-definition-wedge-dual}
    \\
    &=
    \frac{1}{\binom{N+1}{t}}
    \sum_{s=0}^{t}
    \frac{\binom{\eta-s}{t-s}}{\binom{t}{s}}
    \mac E_{s,t}
    \left(
        \mac C_{\eta,s}(\rho)
    \right).
    \label{eq:grassmannian-moment-identity-wedge-dual}
\end{align}
\end{lemma}

\begin{proof}
The first identity follows from the contraction identity
\begin{align}
    \mac C_{\eta,t}
    \left(
        R^{\wedge \eta}
    \right)
    =
    R^{\wedge t},
\end{align}
valid for \(R\in\mac P_\eta\). Applying \(\mac C_{\eta,t}\) to
\(\mac M_\eta(X)\) therefore gives Eq.~\eqref{eq:At-definition-wedge-dual}.

It remains to prove the explicit expression
\eqref{eq:grassmannian-moment-identity-wedge-dual}. 
To this end, we remark that $ \mac{A}_t (\rho) $ has the following property
\begin{align}
    \mac A_t\!\left(U_\eta(u)\rho U_\eta(u)^\dagger\right)
    &=
    \binom{N}{\eta}
    \int
    \tr_{\mac{H}_\eta} \left[
        \rho (u^\dagger R u)^{\wedge \eta}
    \right]
    R^{\wedge t}\,\dd R                                                    \\
    &=
    \binom{N}{\eta}
    \int
    \tr_{\mac{H}_\eta} \left[
        \rho (R')^{\wedge \eta}
    \right]
    (uR'u^\dagger)^{\wedge t}\,\dd R'                                      \\
    &=
    U_t(u)\mac A_t(\rho)U_t(u)^\dagger .
\end{align}
Here, we use the invariance of \(\dd R\).
Thus, by Lemma \ref{lem:expression_U(N)-invariant}, $\mac{A}_t$ can be expressed by a linear combination of $\mac{E}_{s,t}(\mac C_{\eta, s}(\cdot))$. 

We introduce the linear map \(\mac B_t\) by
\begin{align}
    \mac B_t(\cdot)
    \coloneqq
    \frac{1}{\binom{N+1}{t}}
    \sum_{s=0}^{t}
    \frac{\binom{\eta-s}{t-s}}{\binom{t}{s}}
    \mac E_{s,t}
    \left(
        \mac C_{\eta,s}(\cdot)
    \right).
\end{align}
It remains to prove that \(\mac A_t=\mac B_t\).

By Theorem~\ref{thm:tomograhic_completeness_app}, it is enough to verify this
identity on the \(\mtr U(N)\)-orbit of \(P_{[\eta]}^{\wedge \eta}\), namely
to show that
\begin{align}
    \mac A_t
    \left(
        U_\eta(u) P_{[\eta]}^{\wedge \eta} U_\eta(u)^\dagger
    \right)
    =
    \mac B_t
    \left(
        U_\eta(u) P_{[\eta]}^{\wedge \eta} U_\eta(u)^\dagger
    \right)
\end{align}
for every \(u\in\mtr U(N)\).
Since both \(\mac A_t\) and \(\mac B_t\) satisfy the condition
in Eq.~\eqref{eq:U(N)-equivariant}, the above identity follows once we prove
it at the reference point.  Therefore, it suffices to show that
\begin{align}
    \mac A_t
    \left(
        P_{[\eta]}^{\wedge \eta}
    \right)
    =
    \mac B_t
    \left(
        P_{[\eta]}^{\wedge \eta}
    \right).
\end{align}

Additionally, 
We remark that both \(\mac A_t(P_{[\eta]}^{\wedge \eta} )\) and \(\mac B_t(P_{[\eta]}^{\wedge \eta})\) are
diagonal in the occupation-number basis.
For \(\mac A_t(P_{[\eta]}^{\wedge \eta})\), let
$
    u_d=\operatorname{diag}(e^{i\theta_1},\ldots,e^{i\theta_N})
    \in \mtr U(N)
$
be an arbitrary diagonal unitary matrix.  Since \(P_{[\eta]}^{\wedge \eta}\) is invariant under \(U_\eta(u_d)\), and since \(\mac A_t\) satisfy the condition
in Eq.~\eqref{eq:U(N)-equivariant}, we have
\begin{align}
    U_t(u_d)\mac A_t(P_{[\eta]}^{\wedge \eta})U_t(u_d)^\dagger
    =
    \mac A_t(P_{[\eta]}^{\wedge \eta}).
\end{align}
Thus \(\mac A_t(P_{[\eta]}^{\wedge \eta})\) commutes with all diagonal phase rotations.  Hence
its off-diagonal matrix elements in the occupation-number basis vanish, so
\(\mac A_t(P_{[\eta]}^{\wedge \eta})\) is diagonal.
Similarly, for \(\mac B_t(P_{[\eta]}^{\wedge \eta})\), we use
\(
    \mac C_{\eta,s}(P_{[\eta]}^{\wedge \eta})
    =
    P_{[\eta]}^{\wedge s}.
\)
The operator \(P_{[\eta]}^{\wedge s}\) is also invariant under every
diagonal unitary \(U_s(d)\).  Since \(\mac E_{s,t}\) is
\(\mtr U(N)\)-equivariant, each operator
\(
    \mac E_{s,t}\!\left(P_{[\eta]}^{\wedge s}\right)
\)
commutes with all diagonal phase rotations on \(\mac H_t\).  Hence each of
them is diagonal in the occupation-number basis.  Therefore their linear
combination \(\mac B_t(P_{[\eta]}^{\wedge \eta})\) is also diagonal. Therefore, it suffices to compare the diagonal matrix elements.

We will choose arbitrary $\vec a \in \mac{S}_{N,t} $ and set $q\coloneq \abs{\vec a \cap [\eta]}$. Such $q$ satisfies $\max(0, t-(N-\eta )) \leq q \leq t$.
By substituting $\rho = P_{[\eta]}^{\wedge \eta}$ into Eq.~\eqref{eq:At-definition-wedge-dual}, we have 
\begin{align}
    \left\langle \vec a \middle|
        \mac A_t(P_{[\eta]}^{\wedge \eta})
    \middle| \vec a \right\rangle
    &= \binom{N}{\eta}
    \int_{\mtr{Gr}_{\eta} (\mathbb{C}^N ) }
    \operatorname{tr}_{\mac H_\eta}
    \left[
        P_{[\eta]}^{\wedge \eta} R^{\wedge \eta}
    \right]
    \bra{\vec a} R^{\wedge t} \ket{\vec a}\,\dd R 
    \\
    &= \binom{N}{\eta} 
    \int_{\mtr{Gr}_{\eta} (\mathbb{C}^N ) }
    \det R_{[\eta], [\eta]} \det R_{\vec a, \vec a}\,\dd R   
    \\
    &=\frac{1}{\binom{N+1}{t}}
    \sum_{s=0}^{q}
    \frac{\binom{\eta-s}{t-s}}{\binom{t}{s}}
    \binom{q}{s} 
    \label{eq:expression_At_1}
\end{align}
Here, the third equality follows from
Lemma~\ref{lem:reference-overlap-moment}.

On the other hand, the contraction identity gives  
$
    \mac C_{\eta,s}
    \left(
        P_{[\eta]}^{\wedge \eta}
    \right)
    =
    P_{[\eta]}^{\wedge s}
$, so we obtain,
\begin{align}
    \left\langle \vec a \middle|
        \mac E_{s,t}
        \left(
            P_{[\eta]}^{\wedge s}
        \right)
    \middle| \vec a \right\rangle
    =
    \binom{q}{s}.
\end{align}
Therefore, 
\begin{align}
    \left\langle \vec a \middle|
        \mac B_t(P_{[\eta]}^{\wedge \eta})
    \middle| \vec a \right\rangle
    &=
    \frac{1}{\binom{N+1}{t}}
    \sum_{s=0}^{t}
    \frac{\binom{\eta-s}{t-s}}{\binom{t}{s}} \binom{q}{s} \\
    &=\frac{1}{\binom{N+1}{t}}
    \sum_{s=0}^{q}
    \frac{\binom{\eta-s}{t-s}}{\binom{t}{s}}
    \binom{q}{s},
\end{align}
because we have $\binom{q}{s} = 0$ for $s > q$.
Now we have 
\begin{align}
    \bra{\vec a}\mac A_t
    \left(
        P_{[\eta]}^{\wedge \eta}
    \right) \ket{\vec a}
    =
    \bra{\vec a}
    \mac B_t
    \left(
        P_{[\eta]}^{\wedge \eta}
    \right) \ket{\vec a},
\end{align}
for arbitrary $\ket{\vec a} \in \mac{H}_t $.

\end{proof}

Intuitively, since \(\mac A_t(X)\) is obtained by contracting the orbital-rotation shadow \(\mac M_\eta(X)\), it extracts the \(t\)-body
reduced moment data of the measurement-channel output. Although this contraction lowers the order to \(t\), the binomial inversion formula shows that the resulting hierarchy of reduced moments can be reassembled, via the
extension maps \(\mac E_{t,k}\), into the \(k\)-body component for any
\(0\le k\le \eta\). In particular, by taking \(k=\eta\), we obtain the
explicit form of the inverse measurement channel.

\begin{lemma}[Binomial inversion of the extension map]
\label{lem:binomial-inversion-moment-hierarchy}
Let \(\rho \in\operatorname{End}(\mac H_\eta)\), and set
\begin{align}
    {}^s D_\rho
    :=
    \mac C_{\eta,s}(\rho)
    \in \operatorname{End}(\mac H_s),
    \qquad
    0\leq s\leq \eta.
\end{align}
Suppose that the reduced Grassmannian moments $\mac A_t$ satisfy
\begin{align}
    \mac A_t(\rho)
    =
    \frac{1}{\binom{N+1}{t}}
    \sum_{s=0}^{t}
    \frac{\binom{\eta-s}{t-s}}{\binom{t}{s}}
    \mac E_{s,t}
    \left(
        {}^s D_\rho
    \right),
    \qquad
    0\leq t\leq \eta.
    \label{eq:moment-hierarchy-before-inversion}
\end{align}
Then, for every \(0\leq k\leq \eta\),
\begin{align}
    {}^k D_\rho
    =
    \sum_{t=0}^{k}
    (-1)^{k+t}
    \binom{N+1}{t}
    \frac{\binom{\eta-t}{k-t}}{\binom{k}{t}}
    \mac E_{t,k}
    \left(
        \mac A_t(\rho)
    \right).
    \label{eq:binomial-inversion-moment-hierarchy}
\end{align}
\end{lemma}

\begin{proof}
Substituting Eq.~\eqref{eq:moment-hierarchy-before-inversion} into the
right-hand side of Eq.~\eqref{eq:binomial-inversion-moment-hierarchy}, the
factors \(\binom{N+1}{t}\) cancel and we obtain
\begin{align}
    &
    \sum_{t=0}^{k}
    (-1)^{k+t}
    \frac{\binom{\eta-t}{k-t}}{\binom{k}{t}}
    \mac E_{t,k}
    \left[
        \sum_{s=0}^{t}
        \frac{\binom{\eta-s}{t-s}}{\binom{t}{s}}
        \mac E_{s,t}
        \left(
            D_\rho^{(s)}
        \right)
    \right]
    \\
    &=
    \sum_{s=0}^{k}
    \alpha_{k,s}
    \mac E_{s,k}
    \left(
        D_\rho^{(s)}
    \right),
\end{align}
where we used the composition rule
\begin{align}
    \mac E_{t,k}
    \left(
        \mac E_{s,t}(Z)
    \right)
    =
    \binom{k-s}{t-s}
    \mac E_{s,k}(Z),
\end{align}
and
\begin{align}
    \alpha_{k,s}
    =
    \sum_{t=s}^{k}
    (-1)^{k+t}
    \frac{\binom{\eta-t}{k-t}}{\binom{k}{t}}
    \frac{\binom{\eta-s}{t-s}}{\binom{t}{s}}
    \binom{k-s}{t-s}.
\end{align}
Using
\begin{align}
    \binom{\eta-s}{t-s}
    \binom{\eta-t}{k-t}
    =
    \binom{\eta-s}{k-s}
    \binom{k-s}{t-s},
    \qquad
    \binom{k}{t}\binom{t}{s}
    =
    \binom{k}{s}\binom{k-s}{t-s},
\end{align}
we simplify the coefficient as
\begin{align}
    \alpha_{k,s}
    =
    \frac{\binom{\eta-s}{k-s}}{\binom{k}{s}}
    \sum_{t=s}^{k}
    (-1)^{k+t}
    \binom{k-s}{t-s}.
\end{align}
The last sum is the elementary binomial identity
\begin{align}
    \sum_{t=s}^{k}
    (-1)^{k+t}
    \binom{k-s}{t-s}
    =
    \delta_{s,k}.
\end{align}
Hence \(\alpha_{k,s}=0\) for \(s<k\), while \(\alpha_{k,k}=1\). Therefore only the \(s=k\) term remains, and
\begin{align}
    \sum_{t=0}^{k}
    (-1)^{k+t}
    \binom{N+1}{t}
    \frac{\binom{\eta-t}{k-t}}{\binom{k}{t}}
    \mac E_{t,k}
    \left(
        \mac A_t(\rho)
    \right)
    =
    \mac E_{k,k}
    \left(
        D_\rho^{(k)}
    \right)
    =
    D_\rho^{(k)}.
\end{align}
This proves the claim.
\end{proof}

By the tomographic completeness established in Theorem~\ref{thm:tomograhic_completeness_app}, the measurement channel
\(\mac M_\eta\) is injective on \(\operatorname{End}(\mac H_\eta)\).
Since \(\operatorname{End}(\mac H_\eta)\) is finite-dimensional and
\(\mac M_\eta\) is a linear endomorphism of this space, \(\mac M_\eta\)
is invertible. The following lemma identifies its inverse explicitly by Lemma~\ref{lem:binomial-inversion-moment-hierarchy}.

\begin{lemma}[Inverse measurement channel]
\label{lem:inverse-measurement-channel}
The measurement channel \(\mac M_\eta\) is invertible on
\(\operatorname{End}(\mac H_\eta)\), and its inverse is
\begin{align}
    \mac M_\eta^{-1}(Y)
    =
    \sum_{t=0}^{\eta}
    (-1)^{\eta+t}
    \frac{\binom{N+1}{t}}{\binom{\eta}{t}}
    \mac E_{t,\eta}
    \left(
        \mac C_{\eta,t}(Y)
    \right).
    \label{eq:inverse-measurement-channel-binomial}
\end{align}
\end{lemma}

\begin{proof}
By the tomographic completeness established in Theorem \ref{thm:tomograhic_completeness_app},
it remains to identify the inverse explicitly. Let \(Y=\mac M_\eta(X)\).
Since $\mac A_t(\cdot)$ is defined by $\mac C_{\eta,t}
    \left(
        \mac M_\eta(\cdot)
    \right)$, we obtain
\begin{align}
    \mac A_t(X)
    =
    \mac C_{\eta,t}
    \left(
        \mac M_\eta(X)
    \right)
    =
    \mac C_{\eta,t}(Y).
\end{align}
Applying Lemma~\ref{lem:binomial-inversion-moment-hierarchy} with \(k=\eta\),
and using
\begin{align}
    {}^\eta D_X
    =
    \mac C_{\eta,\eta}(X)
    =
    X,
\end{align}
we obtain
\begin{align}
    X
    &= {}^\eta D_{X} \\
    &=  \sum_{t=0}^{\eta}
    (-1)^{\eta+t}
    \binom{N+1}{t}
    \frac{\binom{\eta-t}{\eta-t}}{\binom{\eta}{t}}
    \purple{\mac E_{t,\eta}}
    \left(
        \mac A_t(X)
    \right). \\
    &=
    \sum_{t=0}^{\eta}
    (-1)^{\eta+t}
    \frac{\binom{N+1}{t}}{\binom{\eta}{t}}
    \mac E_{t,\eta}
    \left(
        \mac C_{\eta,t}(Y)
    \right).
\end{align}
Hence, the map expressed by
$
    \sum_{t=0}^{\eta}
    (-1)^{\eta+t}
    \frac{\binom{N+1}{t}}{\binom{\eta}{t}}
    \mac E_{t,\eta}
    \left(
        \mac C_{\eta,t}(\cdot)
    \right)
$
sends \(Y=\mac M_\eta(X)\) back to \(X\).
Since \(\mac M_\eta\) is invertible, this map is precisely
\(\mac M_\eta^{-1}\).
\end{proof}

With this expression, we obtain the single-shot estimator of orbital-rotation shadow. We remark that this estimator recovers the result of Ref.~\cite{low2022classical}. 

\begin{corollary}[Single-shot state shadow]
\label{cor:single-shot-state-shadow-binomial}
For a single measurement outcome \(R\in\mac P_\eta\), the corresponding state shadow is
\begin{align}
    \widehat\rho(R)
    :=
    \mac M_\eta^{-1}
    \left(
        R^{\wedge \eta}
    \right)
    =
    \sum_{t=0}^{\eta}
    (-1)^{\eta+t}
    \frac{\binom{N+1}{t}}{\binom{\eta}{t}}
    \mac E_{t,\eta}
    \left(
        R^{\wedge t}
    \right).
    \label{eq:single-shot-state-shadow-binomial}
\end{align}
Moreover,
\begin{align}
    \purple{\mathbb E_{\rho}}[\widehat\rho(R)]
    =
    \rho.
\end{align}
\end{corollary}

\begin{proof}
Apply Lemma~\ref{lem:inverse-measurement-channel} to \(Y=R^{\wedge \eta}\). Since \(R\) has rank \(\eta\), the contraction identity gives
\begin{align}
    \mac C_{\eta,t}
    \left(
        R^{\wedge \eta}
    \right)
    =
    R^{\wedge t}.
\end{align}
This proves Eq.~\eqref{eq:single-shot-state-shadow-binomial}. The unbiasedness follows from
\begin{align}
    \purple{\mathbb E_{\rho}}
    \left[
        R^{\wedge \eta}
    \right]
    =
    \mac M_\eta(\rho),
\end{align}
and hence
\begin{align}
    \purple{\mathbb E_{\rho}}[\widehat\rho(R)]
    =
    \mac M_\eta^{-1}
    \left(
        \mac M_\eta(\rho)
    \right)
    =
    \rho.
\end{align}
\end{proof}

\begin{corollary}[Reduced \(k\)-RDM estimator]
\label{cor:k-rdm-estimator-binomial}
Let \(0\leq k\leq \eta\). For a single Grassmannian measurement outcome \(R\in\mac P_\eta\), define
\begin{align}
    \purple{{}^k \widehat{D} }(R)
    :=
    \sum_{t=0}^{k}
    (-1)^{k+t}
    \binom{N+1}{t}
    \frac{\binom{\eta-t}{k-t}}{\binom{k}{t}}
    \mac E_{t,k}
    \left(
        R^{\wedge t}
    \right).
    \label{eq:k-rdm-estimator-binomial}
\end{align}
Then \(\purple{{}^k \widehat{D} }(R)\) is unbiased:
\begin{align}
    \purple{\mathbb E_{\rho}}
    \left[
        \purple{{}^k \widehat{D} }(R)
    \right]
    =
    \mac C_{\eta,k}(\rho)
    =
    D_\rho^{(k)}.
\end{align}
\end{corollary}

\begin{proof}
For the measurement distribution induced by \(\rho\), we have
\begin{align}
    \purple{\mathbb E_{\rho}}
    \left[
        R^{\wedge t}
    \right]
    =
    \mac A_t(\rho).
\end{align}
Taking the expectation of Eq.~\eqref{eq:k-rdm-estimator-binomial} and applying Lemma~\ref{lem:binomial-inversion-moment-hierarchy} with \(X=\rho\) gives
\begin{align}
    \purple{\mathbb E_{\rho}}
    \left[
        \purple{{}^k \widehat{D} }(R)
    \right]
    =
    D_\rho^{(k)}.
\end{align}
This proves the claim.
\end{proof}

\subsection{Comparison with Low's estimator}
\label{subsec:comparison-low-estimator}

In this section, we verify that the estimator obtained from the Grassmannian moment formula is equivalent to Low's orbital-rotation estimator~\cite{low2022classical}.  By unitary
covariance of exterior powers and of the extension maps, it suffices to
compare the two estimators at the reference outcome
$
    R=P_{[\eta]}
$.
We use the convention
\begin{align}
    R^{\wedge 0}
    =
    \1_{\mac H_0},
    \qquad
    \mac E_{0,k}(\1_{\mac H_0})
    =
    \1_{\mac H_k}.
\end{align}

Let \(\ket{\vec r}\) be a standard \(k\)-particle occupation-basis vector,
with \(\vec r\in\mac S_{N,k}\), and set
$
    s
    :=
    |\vec r\cap[\eta]|.
$
Since \(P_{[\eta]}\) projects onto the first \(\eta\) modes, the operator
\(\mac E_{\ell,k}(P_{[\eta]}^{\wedge \ell})\) counts the number of occupied
\(\ell\)-subsets of \(\vec r\) contained in \([\eta]\).  Hence
\begin{align}
    \mac E_{\ell,k}
    \left(
        P_{[\eta]}^{\wedge \ell}
    \right)
    \ket{\vec r}
    =
    \binom{s}{\ell}
    \ket{\vec r}.
    \label{eq:reference-extension-eigenvalue-general}
\end{align}
Therefore, from Eq.~\eqref{eq:k-rdm-estimator-binomial}, the diagonal coefficient of the Grassmannian-moment estimator is
\begin{align}
    \omega_{\eta,k}(s)
    &=
    \sum_{\ell=0}^{k}
    (-1)^{k+\ell}
    \binom{N+1}{\ell}
    \frac{\binom{\eta-\ell}{k-\ell}}{\binom{k}{\ell}}
    \binom{s}{\ell}
    \label{eq:grassmannian-coefficient-before-simplification}
    \\
    &=
    \sum_{\ell=0}^{s}
    (-1)^{k+\ell}
    \binom{N+1}{\ell}
    \frac{\binom{\eta-\ell}{k-\ell}}{\binom{k}{\ell}}
    \binom{s}{\ell},
    \label{eq:grassmannian-coefficient-truncated}
\end{align}
where we used \(\binom{s}{\ell}=0\) for \(\ell>s\).

We now simplify Eq.~\eqref{eq:grassmannian-coefficient-truncated}.  Using
$
    {\binom{s}{\ell}}{\binom{k}{\ell}}^{-1}
    =
    {\binom{k-\ell}{s-\ell}}{\binom{k}{s}}^{-1}
$
and
$
    \binom{\eta-\ell}{k-\ell}
    \binom{k-\ell}{s-\ell}
    =
    \binom{\eta-\ell}{s-\ell}
    \binom{\eta-s}{k-s},
$
we obtain
\begin{align}
    \omega_{\eta,k}(s)
    &=
    (-1)^k
    \frac{\binom{\eta-s}{k-s}}{\binom{k}{s}}
    \sum_{\ell=0}^{s}
    (-1)^\ell
    \binom{N+1}{\ell}
    \binom{\eta-\ell}{s-\ell}.
    \label{eq:grassmannian-coefficient-reduced-to-sum}
\end{align}

It remains to evaluate the finite sum in
Eq.~\eqref{eq:grassmannian-coefficient-reduced-to-sum}.
To this end, we will use the following relationship 
\begin{align}
    \qty[(1+z)^\eta \qty(\frac{-z}{1+z})^\ell ]_{z^{m}} 
    &= \qty[(-z)^\ell (1+z)^{\eta - \ell} ]_{z^m}  \\
    &= (-1)^{\ell} \binom{\eta - \ell}{m-\ell}
\end{align}
Here \([f(z)]_{z^{k}}\) denotes the coefficient of \(z^{k}\) in the power-series expansion of \(f(z)\).

By choosing
$
    a:=k-s,
    s=k-a,
$ we obtain
\begin{align}
    \sum_{\ell=0}^{s}
    (-1)^\ell
    \binom{N+1}{\ell}
    \binom{\eta-\ell}{s-\ell}
    &=
    \sum_{\ell=0}^{k-a}
    (-1)^\ell
    \binom{N+1}{\ell}
    \binom{\eta-\ell}{k-a-\ell}
    \label{eq:finite-sum-rewritten}
    \\
    &=
    \left[
        (1+z)^\eta
        \sum_{\ell\geq 0}
        \binom{N+1}{\ell}
        \left(
            \frac{-z}{1+z}
        \right)^\ell
    \right]_{z^{k-a}}
    \label{eq:finite-sum-generating-function-1}
    \\
    &=
    \left[
        (1+z)^\eta
        \left(
            1-\frac{z}{1+z}
        \right)^{N+1}
    \right]_{z^{k-a}}
    \label{eq:finite-sum-generating-function-2}
    \\
    &=
    \left[
        (1+z)^{\eta-N-1}
    \right]_{z^{k-a}}
    \label{eq:finite-sum-generating-function-3}
    \\
    &=
    (-1)^{k-a}
    \binom{N-\eta+k-a}{k-a}.
    \label{eq:finite-sum-evaluated}
\end{align}
Here, the third equality arises from $\sum_{\ell \geq 0} \binom{N+1}{\ell} \alpha^\ell = (1+\alpha)^{(N+1)}$ and the last equality comes from the definition of binomial coefficients.

Substituting Eq.~\eqref{eq:finite-sum-evaluated} into
Eq.~\eqref{eq:grassmannian-coefficient-reduced-to-sum}, and using
\(s=k-a\), we get
\begin{align}
    \omega_{\eta,k}(s)
    &=
    (-1)^k
    \frac{\binom{\eta-s}{k-s}}{\binom{k}{s}}
    (-1)^s
    \binom{N-\eta+s}{s}
    \\
    &=
    (-1)^{k+s}
    \frac{
        \binom{\eta-s}{k-s}
        \binom{N-\eta+s}{s}
    }{
        \binom{k}{s}
    }.
    \label{eq:grassmannian-coefficient-low-form}
\end{align}
This is exactly Low's diagonal coefficient,
\begin{align}
    \omega_{\eta,k}^{\rm Low}(s)
    =
    (-1)^{k+s}
    \frac{
        \binom{\eta-s}{k-s}
        \binom{N-\eta+s}{s}
    }{
        \binom{k}{s}
    }.
    \label{eq:low-diagonal-coefficient-general}
\end{align}
Hence the two estimators agree at the reference outcome \(R=P_{[\eta]}\).

Finally, for a general rotated outcome
$
    R
    =
    u^\dagger P_{[\eta]}u,
$
unitary covariance gives
\begin{align}
    \purple{{}^k \widehat{D} }(R)
    =
    (u^{\wedge k})^\dagger
    \purple{{}^k \widehat{D} }(P_{[\eta]})
    u^{\wedge k}.
\end{align}
Therefore, the Grassmannian-moment estimator agrees with Low's
orbital-rotation estimator.

\section{Variance of $k$-RDM Estimators}

In this section, we will derive provable variance guarantees for the orbital-rotation shadow estimator.  We first obtain closed-form variance formulae for the \(1\)-RDM entries, showing explicitly how the behavior of variance are determined by the true one- and two-body reduced density matrices.  We then prove a general entrywise bound for \(k\)-RDM single-shot estimators and translate it into a sample complexity for all $k$-RDM estimation task.

\subsection{Case: $k=1$}
\label{sec:1-rdm_variance}

Let \(\rho\) be an \(N\)-mode \(\eta\)-particle fermionic state, and let
\(D^{(1)}\) be its \(1\)-RDM. Let
\(R\in \mtr{Gr}_\eta(\mathbb C^N)\) denote the rank-\(\eta\)
one-particle projector induced by one sample of the number-conserving
fermionic shadow protocol. Throughout this section, we use the
single-shot \(1\)-RDM estimator
\begin{align}
    \widehat{D}^{(1)}(R)
    :=
    (N+1)R-\eta \1_{\mac H_1}
    \in \mtr{End}(\mac H_1).
\end{align}

We first treat off-diagonal entries, where the identity term in the
estimator does not contribute. We then treat diagonal entries, where the
same estimator contains an additional scalar shift.

\begin{theorem}[Off-diagonal entries]
\label{thm:explicit_variance_k=1}
For every \(p,q\in[N]\) with \(p\neq q\), the single-shot 1-RDM orbital-rotation shadow estimator satisfies
\begin{align}
\operatorname{Var}(\purple{{}^{1}\widehat D_{\mkern6muq}^{p}})
&:=
\purple{\mathbb E_{\rho}}
\left[
\left|
\purple{{}^{1}\widehat D_{\mkern6muq}^{p}}-\purple{{}^{1}D_{\mkern6muq}^{p}}
\right|^2
\right]
\\
&=
\frac{(N+1)(N+1-\eta)}{N(N+2)}
\left(
\eta+\purple{{}^{1}D_{\mkern6mup}^{p}}+\purple{{}^{1}D_{\mkern6muq}^{q}}
\right)
-
\frac{N+1}{N}\purple{{}^{2}D_{\mkern6mu(p,q)}^{(p,q)}}
-
|\purple{{}^{1}D_{\mkern6muq}^{p}}|^2,
\label{eq:explicit_formula_1RDM}
\end{align}
where
\begin{align}
\purple{{}^{2}D_{\mkern6mu(p,q)}^{(p,q)}}
=
\operatorname{tr}_{\mac H_\eta}
\left[
\rho\,a_p^\dagger a_q^\dagger a_q a_p
\right]
\end{align}
is the \(((p,q),(p,q))\) element of the \(2\)-RDM. 
\end{theorem}

\begin{proof}
First, we outline the proof.
We reduce the argument to the representative off-diagonal entry \((1,2)\) and express the second moment of the estimator $\purple{\mathbb{E}_{\rho}} \qty[ \abs{\purple{{}^{1}\widehat D_{\mkern6muq}^{p}} }^2 ]$ as an operator-valued Grassmannian moment \( {}^1B^1_{\ 2}\).  The symmetries
\(\mathrm U(1)\times \mathrm U(1)\times \mathrm U(N-2)\) imply that
\( {}^1B^1_{\ 2}\) is block-diagonal with respect to the occupations of modes
\(1\) and \(2\), and Schur's lemma reduces its evaluation to four scalar
coefficients.  These coefficients are explicitly computed from Weingarten integral over Grassmannian, and substituting them into the second
moment formula yields the exact variance and the stated upper bound.

It suffices to consider the representative case \((p,q)=(1,2)\). From
the estimator formula,
\begin{align}
    \purple{{}^{1}\widehat D_{\mkern6mu2}^{1}}(R)
    =
    (N+1)R_{2,1},
\end{align}
where \(R_{2,1}:=\langle 2|R|1\rangle\) denotes the \((2,1)\)-entry of the projector \(R\).
From this expression and the unbiasedness, the variance of this estimator can be bounded as follows,
\begin{align}
    \purple{\mathbb E_{\rho}}
\left[
\left|
\purple{{}^{1}\widehat D_{\mkern6muq}^{p}}-\purple{{}^{1}D_{\mkern6muq}^{p}}
\right|^2
\right] 
= \purple{\mathbb{E}_{\rho}} \qty[ \abs{\purple{{}^{1}\widehat D_{\mkern6muq}^{p}} }^2 ] - \abs{\purple{{}^{1}D_{\mkern6muq}^{p}}}^2.
\end{align}
Using the Grassmannian formulation of the orbital-rotation protocol
developed in Sec.~\ref{sec:Grassmannian_moment}, we can rewrite
Eq.~\eqref{eq:Grassmannian_moment} as
\begin{align}
    \purple{\mathbb{E}_{\rho}} \qty[ \abs{\purple{{}^{1}\widehat D_{\mkern6mu2}^{1}} }^2 ] =(N+1)^2  \binom{N}{\eta} \int_{\mtr{Gr}_\eta(\mathbb{C}^N)} \abs{R_{2,1}}^2 \tr_{\mac{H}_\eta }(\rho R^{\wedge \eta}) \dd{R}.
    \label{eq:1rdm_estimator_formula}
\end{align}
Therefore, it is sufficient to evaluate the corresponding operator-valued moment
\begin{align}
    {}^1B^1_{\ 2}
    \coloneq
    \binom{N}{\eta}
    \int_{\mtr{Gr}_\eta(\mathbb C^N)}
    \abs{R_{2,1} }^2
    R^{\wedge \eta}
    \dd R
    \in \mtr{End}(\mac H_\eta).
\end{align}

Regarding $ {}^1B^1_{\ 2}$, the relevant symmetry is the invariance under local phase rotations of the first two modes, together with unitary rotations of the remaining
\(N-2\) modes. More precisely, let
\begin{align}
    G
    =
    \mtr U(1)_{\{1\}}
    \times
    \mtr U(1)_{\{2\}}
    \times
    \mtr U(N-2),
\end{align}
where \(\mtr U(N-2)\) acts on $
    \operatorname{span}\{\ket{3},\ldots,\ket{N}\}
    \subset
    \mathbb C^N$.
For \(g\in G\), the Grassmannian measure is invariant under
\(R\mapsto gRg^\dagger\), and the scalar factor \(\abs{R_{2,1}}^2\) is
also invariant under this transformation.  Indeed, if
\(g=\operatorname{diag}(e^{i\theta_1},e^{i\theta_2},h)\) where $h \in \mtr{U}(N-2)$, then
$
    (gRg^\dagger)_{1,2}
    =
    e^{i(\theta_1-\theta_2)}R_{1,2},
$
and hence \(\abs{(gRg^\dagger)_{1,2}}^2=\abs{R_{2,1}}^2\).  Therefore,
by a change of variables in the Grassmannian integral, we obtain
we obtain
\begin{align}
    U_\eta(g) {}^1B^1_{\ 2} U_\eta(g)^\dagger
    =
    {}^1B^1_{\ 2}.
\end{align}
These two \(\mtr U(1)\) symmetries imply
\begin{align}
    [ {}^1B^1_{\ 2},n_1]=[ {}^1B^1_{\ 2},n_2]=0,
\end{align}
where $n_j \coloneq a_j^\dagger a_j \eval_{\mac{H}_\eta} \in \mtr{End}(\mac{H}_\eta ) $ for $j \in \{1,\ldots, N \}$.
Thus \( {}^1B^1_{\ 2}\) preserves each joint occupation sector of the modes
\(1\) and \(2\). On each such sector, the remaining degrees of freedom form
an exterior-power representation of \(\mtr U(N-2)\).

For each \(T\subseteq \{1,2\}\), consider the sector in which precisely
the modes in \(T\) are occupied among the modes \(1\) and \(2\).  After
fixing such an occupation pattern, the remaining part of the
\(\eta\)-particle space is identified with
\begin{align}
    \bigwedge^{\eta-\abs{T}}
    \operatorname{span}\{\ket{3},\ldots,\ket{N}\}.
\end{align}
This is an irreducible representation of \(\mtr U(N-2)\). Since \( {}^1B^1_{\ 2}\)
commutes with the induced \(\mtr U(N-2)\)-action, Schur's lemma implies that
the restriction of \( {}^1B^1_{\ 2}\) to each fixed occupation sector
\(T\subseteq \{1,2\}\) is a scalar multiple of the identity. Hence
\begin{align}
    {}^1B^1_{\ 2}
    =
    A_\emptyset \Pi^{(\eta)}_{\emptyset}
    +
    A_{\{1\}} \Pi^{(\eta)}_{\{1\}}
    +
    A_{\{2\}} \Pi^{(\eta)}_{\{2\}}
    +
    A_{\{1,2\}} \Pi^{(\eta)}_{\{1,2\}}.
\end{align}
Here
\begin{align}
    \Pi^{(\eta)}_T
    \coloneq
    \prod_{a\in T} n_a
    \prod_{b\in L\setminus T}(\1-n_b),
    \qquad
    T\subseteq L=\{1,2\},
\end{align}
where the operators are restricted to \(\mac H_\eta\).

To determine the coefficient \(A_T\), we use
\begin{align}
    A_T
    =
    \frac{
        \tr_{\mac H_\eta}
        \left[
            \Pi^{(\eta)}_T {}^1B^1_{\ 2}
        \right]
    }{
        \tr_{\mac H_\eta}
        \left[
            \Pi^{(\eta)}_T
        \right]
    }.
\end{align}
Together with
\begin{align*}
    \tr_{\mac H_\eta}
    \left[
        \Pi^{(\eta)}_T
    \right]
    =
    \binom{N-2}{\eta-\abs{T}},
\end{align*}
this gives
\begin{align}
    A_T
    =
    \binom{N}{\eta}
    \binom{N-2}{\eta-\abs{T}}^{-1}
    \int_{\mtr{Gr}_\eta(\mathbb C^N)}
    \abs{R_{2,1}}^2
    \tr_{\mac H_\eta}
    \left[
        \Pi^{(\eta)}_T
        R^{\wedge \eta}
    \right]
    \dd R .
\end{align}
The trace factors in the integrand are
\begin{align}
    \tr_{\mac H_\eta}
    \left[
        \Pi^{(\eta)}_\emptyset
        R^{\wedge \eta}
    \right]
    &=
    1-R_{1,1}-R_{2,2}
    +
    (R_{1,1}R_{2,2} - R_{2,1}R_{1,2}), \\
    \tr_{\mac H_\eta}
    \left[
        \Pi^{(\eta)}_{\{1\}}
        R^{\wedge \eta}
    \right]
    &=
    R_{1,1}
    -
    (R_{1,1}R_{2,2} - R_{2,1}R_{1,2}) , \\
    \tr_{\mac H_\eta}
    \left[
        \Pi^{(\eta)}_{\{2\}}
        R^{\wedge \eta}
    \right]
    &=
    R_{2,2}
    -
    (R_{1,1}R_{2,2} - R_{2,1}R_{1,2}), \\
    \tr_{\mac H_\eta}
    \left[
        \Pi^{(\eta)}_{\{1,2\}}
        R^{\wedge \eta}
    \right]
    &=
    R_{1,1}R_{2,2}-R_{2,1}R_{1,2}.
\end{align}
Here, we use the identities
$
    \tr_{\mac H_\eta}\!\left[n_1 R^{\wedge \eta}\right]
    =
    R_{1,1},
    \tr_{\mac H_\eta}\!\left[n_1 n_2 R^{\wedge \eta}\right]
    =
    \det R_{\{1,2\},\{1,2\}} = R_{1,1}R_{2,2} - R_{2,1}R_{1,2},
$
which are proved in Lemma~\ref{lem:trace_occupied_operator} in the Sec.\ref{app:technical_lemma}.
Consequently, the computation of all coefficients \(A_T\) reduces to the
following Weingarten integral over Grassmannian:
\begin{gather}
    \int_{\mtr{Gr}_\eta(\mathbb{C}^N ) } R_{2,1}R_{1,2}\dd R,
    \quad
    \int_{\mtr{Gr}_\eta(\mathbb{C}^N ) } R_{2,1}R_{1,2}R_{1,1}\dd R, 
    \quad
    \int_{\mtr{Gr}_\eta(\mathbb{C}^N ) } R_{2,1}R_{1,2}R_{22}\dd R, \\
    \int_{\mtr{Gr}_\eta(\mathbb{C}^N ) } R_{2,1}R_{1,2}R_{1,1}R_{2,2}\dd R,
    \quad
    \int_{\mtr{Gr}_\eta(\mathbb{C}^N ) } (R_{2,1}R_{1,2})^2\dd R.
\end{gather}
We left the explicit formula of these Grassmannian moments in Sec.~\ref{appendix:explicit_weingarten}.
Using the Weingarten formula for these Grassmannian moments, we
obtain
\begin{align}
    A_\emptyset
    &=
    \frac{\eta(N+1-\eta)}{N(N+1)(N+2)}, \\
    A_{\{1\}}
    =
    A_{\{2\}}
    &=
    \frac{(\eta+1)(N+1-\eta)}{N(N+1)(N+2)}, \\
    A_{\{1,2\}}
    &=
    \frac{(\eta+1)(N-\eta)}{N(N+1)(N+2)}.
\end{align}
From Eq.~\eqref{eq:1rdm_estimator_formula} and definition of $ {}^1B^1_{\ 2}$, we obtain
\begin{align}
    \purple{\mathbb E_{\rho}}
    \left[
        \abs{\purple{{}^{1}\widehat D_{\mkern6mu2}^{1}}(R)}^2
    \right]
    =
    (N+1)^2
    \tr_{\mac H_\eta}
    \left[
        \rho {}^1B^1_{\ 2}
    \right].
\end{align}
Using
$    {}^1B^1_{\ 2}
    =
    \sum_{T\subseteq\{1,2\}}A_T\Pi^{(\eta)}_T
$
and the identities
\begin{align}
    \tr[\rho\Pi^{(\eta)}_\emptyset]
    =
    1-\purple{{}^{1}D_{\mkern6mu1}^{1}}-\purple{{}^{1}D_{\mkern6mu2}^{2}}+\purple{{}^{2}D_{\mkern6mu(1,2)}^{(1,2)}}, &\quad
    \tr[\rho\Pi^{(\eta)}_{\{1\}}] =
    \purple{{}^{1}D_{\mkern6mu1}^{1}}-\purple{{}^{2}D_{\mkern6mu(1,2)}^{(1,2)}}, \\
    \tr[\rho\Pi^{(\eta)}_{\{2\}}] =
    \purple{{}^{1}D_{\mkern6mu2}^{2}}-\purple{{}^{2}D_{\mkern6mu(1,2)}^{(1,2)}}, &\quad
    \tr[\rho\Pi^{(\eta)}_{\{1,2\}}]
    =
    \purple{{}^{2}D_{\mkern6mu(1,2)}^{(1,2)}},
\end{align}
we obtain
\begin{align}
    \tr_{\mac H_\eta}
    \left[
        \rho {}^1B^1_{\ 2}
    \right]
    &=
    A_\emptyset
    +
    (A_{\{1\}}-A_\emptyset)(\purple{{}^{1}D_{\mkern6mu1}^{1}}+\purple{{}^{1}D_{\mkern6mu2}^{2}})
    +
    (A_\emptyset-2A_{\{1\}}+A_{\{1,2\}})
    \purple{{}^{2}D_{\mkern6mu(1,2)}^{(1,2)}}
    \\
    &=
    \frac{(N+1-\eta)(\eta+\purple{{}^{1}D_{\mkern6mu1}^{1}}+\purple{{}^{1}D_{\mkern6mu2}^{2}})}{N(N+1)(N+2)}
    -
    \frac{\purple{{}^{2}D_{\mkern6mu(1,2)}^{(1,2)}}}{N(N+1)}.
\end{align}
Since $\purple{\mathbb E_{\rho}}
    \left[
        \abs{\purple{{}^{1}\widehat D_{\mkern6mu2}^{1}}(R)}^2
    \right] = (N+1)^2\tr_{\mac H_\eta}
    \left[
        \rho {}^1B^1_{\ 2}
    \right]  $, the following identity holds,
\begin{align}
    \purple{\mathbb E_{\rho}}
    \left[
        \abs{\purple{{}^{1}\widehat D_{\mkern6mu2}^{1}}(R)}^2
    \right]
    =
    \frac{(N+1)(N+1-\eta)}{N(N+2)}
    \left(
        \eta+\purple{{}^{1}D_{\mkern6mu1}^{1}}+\purple{{}^{1}D_{\mkern6mu2}^{2}}
    \right)
    -
    \frac{N+1}{N}\purple{{}^{2}D_{\mkern6mu(1,2)}^{(1,2)}}.
\end{align}
Since the estimator is unbiased,
$
    \purple{\mathbb E_{\rho}}[\purple{{}^{1}\widehat D_{\mkern6mu2}^{1}}(R)]
    =
    \purple{{}^{1}D_{\mkern6mu2}^{1}},
$
the exact variance is evaluated as follows,
\begin{align}
    \operatorname{Var}
    \left(
        \purple{{}^{1}\widehat D_{\mkern6mu2}^{1}}
    \right)
    =
    \frac{(N+1)(N+1-\eta)}{N(N+2)}
    \left(
        \eta+\purple{{}^{1}D_{\mkern6mu1}^{1}}+\purple{{}^{1}D_{\mkern6mu2}^{2}}
    \right)
    -
    \frac{N+1}{N}\purple{{}^{2}D_{\mkern6mu(1,2)}^{(1,2)}}
    -
    \abs{\purple{{}^{1}D_{\mkern6mu2}^{1}}}^2.
\end{align}
In particular, since
$
    0\leq \purple{{}^{1}D_{\mkern6mu1}^{1}},\purple{{}^{1}D_{\mkern6mu2}^{2}}\leq 1,
    \purple{{}^{2}D_{\mkern6mu(1,2)}^{(1,2)}}\geq 0,
    \abs{\purple{{}^{1}D_{\mkern6mu2}^{1}}}^2\geq 0,
$
we have
\begin{align}
    \operatorname{Var}
    \left(
        \purple{{}^{1}\widehat D_{\mkern6mu2}^{1}}
    \right)
    \leq
    \eta+\purple{{}^{1}D_{\mkern6mu1}^{1}}+\purple{{}^{1}D_{\mkern6mu2}^{2}}
    \leq
    \eta+2.
\end{align}
The general off-diagonal case \(p\neq q\) follows by relabeling the modes.
\end{proof}

Next, we consider the variance of diagonal element estimator, 

\begin{lemma}
For every \(p\in[N]\), the single-shot 1-RDM orbital-rotation shadow estimator satisfies
\begin{align}
\operatorname{Var}(\purple{{}^{1}\widehat D_{\mkern6mup}^{p}})
&:=
\purple{\mathbb E_{\rho}}
\left[
\left|
\purple{{}^{1}\widehat D_{\mkern6mup}^{p}}-\purple{{}^{1}D_{\mkern6mup}^{p}}
\right|^2
\right]
\\
&=
\frac{
\eta(N+1-\eta)
+
2\purple{{}^{1}D_{\mkern6mup}^{p}}(N+1-\eta)
-
(N+2) (\purple{{}^{1}D_{\mkern6mup}^{p}})^2
}{N+2}.
\end{align}
In particular,
\begin{align}
    \operatorname{Var}(\purple{{}^{1}\widehat D_{\mkern6mup}^{p}})
    \leq
    \eta+2
    =
    \mac{O}(\eta).
\end{align}
\end{lemma}

\begin{proof}
It suffices to consider the representative case \(p=1\). In this case,
the identity term in the estimator contributes, and hence
\begin{align}
    \label{eq:expression_1rdm_diagonal_estimator}
    \purple{{}^{1}\widehat D_{\mkern6mu1}^{1}}(R)
    =
    (N+1)R_{1,1}-\eta .
\end{align}
For this estimator, we will introduce the corresponding operator-valued moment as follows,
\begin{align}
     {}^1B^1_{\ 1}
    \coloneq
    \binom{N}{\eta}
    \int_{\mtr{Gr}_\eta(\mathbb C^N)}
    R_{1,1}^2
    R^{\wedge \eta}
    \dd R
    \in \mtr{End}(\mac H_\eta).
\end{align}
By invariance under
$
    \mtr U(1)_{\{1\}}\times \mtr U(N-1),
$
the operator \( {}^1B^1_{\ 1}\) preserves the occupation sector of the first mode.
Schur's lemma therefore gives
\begin{align}
     {}^1B^1_{\ 1}
    =
    A_\emptyset(\1_{\mac{H}_\eta }-n_1)
    +
    A_{\{1\}}n_1 .
\end{align}
Following the similar procedure in the proof of Theorem~\ref{thm:explicit_variance_k=1}, the coefficients are obtained from
\begin{align}
    A_\emptyset
    &=
    \binom{N}{\eta}
    \binom{N-1}{\eta}^{-1}
    \int_{\mtr{Gr}_\eta(\mathbb C^N)}
    R_{1,1}^2(1-R_{1,1})
    \dd R,
    \\
    A_{\{1\}}
    &=
    \binom{N}{\eta}
    \binom{N-1}{\eta-1}^{-1}
    \int_{\mtr{Gr}_\eta(\mathbb C^N)}
    R_{1,1}^3
    \dd R.
\end{align}
By calculating Weingarten integral over Grassmannian, we obtain
\begin{align}
    A_\emptyset
    =
    \frac{\eta(\eta+1)}{(N+1)(N+2)},
    \qquad
    A_{\{1\}}
    =
    \frac{(\eta+1)(\eta+2)}{(N+1)(N+2)}.
\end{align}
The evaluation details is left to Sec.~\ref{appendix:explicit_weingarten}.
Consequently,
\begin{align}
    \tr_{\mac H_\eta}
    \left[
        \rho  {}^1B^1_{\ 1}
    \right]
    &=
    A_\emptyset(1-\purple{{}^{1}D_{\mkern6mu1}^{1}})
    +
    A_{\{1\}}\purple{{}^{1}D_{\mkern6mu1}^{1}}
    \\
    &=
    \frac{(\eta+1)(\eta+2\purple{{}^{1}D_{\mkern6mu1}^{1}})}{(N+1)(N+2)}.
\end{align}
Moreover, by recalling Eq.~\eqref{eq:expression_1rdm_diagonal_estimator}, the unbiasedness of the estimator gives
$
    \purple{\mathbb E_{\rho}}
    \left[
        R_{1,1}
    \right]
    =
    \frac{\purple{{}^{1}D_{\mkern6mu1}^{1}}+\eta}{N+1}.
$
From this fact, we obtain
\begin{align}
    \purple{\mathbb E_{\rho}}
    \left[
        \left(
            \purple{{}^{1}\widehat D_{\mkern6mu1}^{1}}(R)
        \right)^2
    \right]
    &=
    (N+1)^2
    \tr_{\mac H_\eta}
    \left[
        \rho  {}^1B^1_{\ 1}
    \right]
    -
    2\eta(N+1)
    \purple{\mathbb E_{\rho}}[R_{1,1}]
    +
    \eta^2
    \\
    &=
    \frac{(N+1-\eta)(\eta+2\purple{{}^{1}D_{\mkern6mu1}^{1}})}{N+2}.
\end{align}
Since
$
    \purple{\mathbb E_{\rho}}[\purple{{}^{1}\widehat D_{\mkern6mu1}^{1}}(R)]
    =
    \purple{{}^{1}D_{\mkern6mu1}^{1}},
$
the exact variance formula can be expressed as follows,
\begin{align}
    \operatorname{Var}
    \left(
        \purple{{}^{1}\widehat D_{\mkern6mu1}^{1}}
    \right)
    &=
    \frac{(N+1-\eta)(\eta+2\purple{{}^{1}D_{\mkern6mu1}^{1}})}{N+2}
    -
    (\purple{{}^{1}D_{\mkern6mu1}^{1}})^2
    \\
    &=
    \frac{
        \eta(N+1-\eta)
        +
        2\purple{{}^{1}D_{\mkern6mu1}^{1}}(N+1-\eta)
        -
        (N+2)(\purple{{}^{1}D_{\mkern6mu1}^{1}})^2
    }{N+2}.
\end{align}
In particular, since \(0\leq \purple{{}^{1}D_{\mkern6mu1}^{1}}\leq 1\),
\begin{align}
    \operatorname{Var}
    \left(
        \purple{{}^{1}\widehat D_{\mkern6mu1}^{1}}
    \right)
    \leq
    \eta+2\purple{{}^{1}D_{\mkern6mu1}^{1}}
    \leq
    \eta+2.
\end{align}
The general diagonal case follows by relabeling the modes.
\end{proof}

Combining the off-diagonal and diagonal cases, we obtain, for every
\(p,q\in[N]\),
\begin{align}
    \operatorname{Var}(\purple{{}^{1}\widehat D_{\mkern6muq}^{p}})
    \leq
    \eta+2
    =
    \mac{O}(\eta).
\end{align}

Taken together, we obtain the general variance formula for $1$-RDM as follows,
\begin{align}
\operatorname{Var}\!\left(\purple{{}^{1}\widehat D_{\mkern6muq}^{p}}\right)
={}&
\frac{(N+1-\delta_{pq})(N-\eta+1)}{N(N+2)}
\Bigl(
\eta+\purple{{}^{1}D_{\mkern6mup}^{p}}+\purple{{}^{1}D_{\mkern6muq}^{q}}
\Bigr)
\notag\\
&-
(1-\delta_{pq})\frac{N+1}{N}
\purple{{}^{2}D_{\mkern6mup q}^{p q}}
-
\left|\purple{{}^{1}D_{\mkern6muq}^{p}}\right|^2 .
\label{eq:explicit_variance_1rdm_app}
\end{align}

\begin{revision}

The same construction extends directly to covariances between two matrix
elements, where the corresponding estimators may have different RDM
orders.  We define
\begin{align}
    \operatorname{Cov}
    \left(
        \purple{{}^{\ell}\widehat D_{\mkern6mu\vec q}^{\vec p}},
        \purple{{}^{m}\widehat D_{\mkern6mu\vec s}^{\vec r}}
    \right)
    \coloneq
    \purple{\mathbb E_{\rho}}
    \left[
        \left(
            \purple{{}^{\ell}\widehat D_{\mkern6mu\vec q}^{\vec p}}
            -
            \purple{{}^{\ell}D_{\mkern6mu\vec q}^{\vec p}}
        \right)
        \overline{
            \left(
                \purple{{}^{m}\widehat D_{\mkern6mu\vec s}^{\vec r}}
                -
                \purple{{}^{m}D_{\mkern6mu\vec s}^{\vec r}}
            \right)
        }
    \right].
    \label{eq:mixed-covariance-definition}
\end{align}
For positive integers \(k\) and \(k'\), define the corresponding mixed
moment operator by
\begin{align}
    {}^{k,k'}B_{\vec q;\,\vec s}^{\vec p;\,\vec r}
    \coloneq
    \binom{N}{\eta}
    \int_{\mtr{Gr}_{\eta}(\mathbb C^N)}
    {}^{k}\widehat D_{\vec q}^{\vec p}(R)\,
    \overline{
        {}^{k'}\widehat D_{\vec s}^{\vec r}(R)
    }\,
    R^{\wedge\eta}\dd R .
    \label{eq:mixed_moment_operator}
\end{align}
The covariance is then expressed in terms of this moment operator as
\begin{align}
    \operatorname{Cov}
    \left(
        \purple{{}^{\ell}\widehat D_{\mkern6mu\vec q}^{\vec p}},
        \purple{{}^{m}\widehat D_{\mkern6mu\vec s}^{\vec r}}
    \right)
    &=
    \tr_{\mac H_\eta}\left[
        {}^{\ell,m}B_{\vec q;\,\vec s}^{\vec p;\,\vec r} \ \rho\,
    \right]
    -
    \purple{{}^{\ell}D_{\mkern6mu\vec q}^{\vec p}}\,
    \overline{
        \purple{{}^{m}D_{\mkern6mu\vec s}^{\vec r}}
    }.
    \label{eq:mixed-covariance-moment-relation}
\end{align}
The variance is recovered by setting \(k'=k\) and
\(
(\vec r,\vec s)=(\vec p,\vec q)
\).
As shown above, this moment-based method can also yield closed-form
variance and covariance formulas by evaluating finitely many
Weingarten integrals over the Grassmannian.

Since the evaluation requires a large number of distinct Weingarten
integrals over the Grassmannian, we defer the details to
Sec.~\ref{sec:verify}.  The implementation is
available in Ref.~\cite{koizumi_or_shadow_covariances}.
This procedure yields explicit formulas for
\(
\mtr{Var}
[
    {}^{2}\widehat D_{\vec q}^{\vec p}
]
\)
and, more generally, for the covariance blocks with
\(
\ell,m\leq 2
\).
The same procedure extends in principle to the variances of the
\(3\)-RDM estimator.

\end{revision}


While explicit evaluation of Weingarten integral over Grassmannian can yield closed-form
variance formulas for fixed \(k\), this approach does not scale well
to higher-order RDMs.  Indeed, the required Weingarten computations involve
high-degree coordinate moments and rapidly growing permutation sums, making
the exact symbolic evaluation computationally prohibitive.  In the next
section, we therefore replace exact moment evaluation with uniform moment
bounds, which are sufficient to prove an \(N\)-independent asymptotic
variance bound for general fixed \(k\).

\begin{revision}

\subsection{General case}

We first state the target result of this section. For \(\vec q,\vec p\in\mac S_{N,k}\), let
$
    \purple{{}^{k}\widehat D_{\mkern6mu\vec q}^{\vec p}}(R)
    :=
    \bra{\vec q}\purple{{}^k \widehat{D} }(R)\ket{\vec p}
$
denote the \((\vec q,\vec p)\)-entry of the single-shot estimator associated with a measurement outcome \(R\). Here \(R\) is the rank-\(\eta\) projector obtained from the orbital-rotation measurement, written as \(\purple{R=u^\dagger P_{\vec z}u}\) for \((u,\vec z)\in\mtr U(N)\times\mac S_{N,\eta}\). The entrywise variance of this estimator is defined by
\begin{align}
    \operatorname{Var}
    \left(
        \purple{{}^{k}\widehat D_{\mkern6mu\vec q}^{\vec p}}
    \right)
    :=
    \mathbb E_{\rho}
    \left[
        \left|
            \purple{{}^{k}\widehat D_{\mkern6mu\vec q}^{\vec p}}(R)
            -
            \purple{{}^{k}D_{\mkern6mu\vec q}^{\vec p}}
        \right|^2
    \right].
\end{align}

The main goal of this section is to prove the following entrywise variance bound.

\begin{theorem}[Upper bound for $k$-RDM estimator]
\label{thm:entrywise-variance}
For every fixed \(k\), there exists a constant \(C_k>0\), depending only on \(k\), such that for all \(N,\eta\) satisfying \(\purple{1\leq k\leq \min\{\eta,N-\eta\}}\), all \(\eta\)-particle states \(\rho\in\operatorname{End}(\mac H_\eta)\), and all \(\vec q,\vec p\in\mac S_{N,k}\),
\begin{align}
    \operatorname{Var}
    \left(
        \purple{{}^{k}\widehat D_{\mkern6mu\vec q}^{\vec p}}
    \right)
    \leq \frac{k+1}{k!}\eta^k  ,
\end{align}
\end{theorem}

\purple{
Before diving into the proof, let us briefly provide the outline in four steps. First, Lemma~\ref{lem:slater-reduction} uses the diagonal-torus invariance of the estimator to show that the positive operator \({}^{k}B_{\vec q}^{\vec p}\) is diagonal in the occupation basis. Consequently, the supremum of the second moment over all \(\eta\)-particle states reduces to the input state \(\ket{\vec s}\). Second, Lemma~\ref{lem:tilted-polar} realizes the corresponding Grassmannian measure \(\nu_{\vec s}\) as the angular marginal of the determinant-tilted Ginibre measure and separates the projector \(R\) from the Wishart variable \(T\). Third, Lemma~\ref{lem:rao-blackwell} expresses the orbital estimator as a conditional expectation of the Gaussian polynomial \(\mathfrak D_\eta^{(k)}(GG^\dagger)\). Conditional Jensen's inequality therefore reduces the required second moment to the determinant-weighted Gaussian moment in Eq.~\eqref{eq:conditional-jensen}. Fourth, Lemma~\ref{lem:gaussian-moment-bound} decomposes this Gaussian polynomial into mutually orthogonal complex Wiener--It\^o chaos components indexed by the contraction number \(r\). The norms of these components are evaluated using the exterior Pieri components and the corresponding projector diagonals, and summing the resulting bounds yields the desired entrywise variance bound.
}

\begin{lemma}
\label{lem:slater-reduction}
Let $k, \eta, N$ integers satisfying $\purple{1 \leq k \leq \min\{\eta,N-\eta\}}$.
For fixed $\vec q , \vec p \in \mac{S}_{N,k} $, define the positive operator
\begin{align}
    \purple{{}^{k}B_{\mkern6mu  \vec q}^{\vec p}}
    := 
    \binom{N}{\eta}
    \int_{\mtr{Gr}_\eta(\mathbb{C}^N)}
    \left|
        \purple{{}^{k}\widehat D_{\mkern6mu\vec q}^{\vec p}}(R)
    \right|^2
    R^{\wedge\eta}
    \dd R.
    \label{eq:def_Bpq}
\end{align}
Then $\purple{{}^{k}B_{\mkern6mu  \vec q}^{\vec p}}$ is diagonal in the occupation basis of
$\mac H_\eta$. Hence,
\begin{align}
    \sup_\rho
    \ \purple{\mathbb{E}_{\rho}} \!\left[
        \left|
            \purple{{}^{k}\widehat D_{\mkern6mu \vec q}^{\vec p}}
        \right|^2
    \right]
    &= \sup_\rho \purple{\binom{N}{\eta}} \int_{\mtr{Gr}_\eta(\mathbb{C}^N)} \left|
            \purple{{}^{k}\widehat D_{\mkern6mu \vec q}^{\vec p}} 
        \right|^2 \tr[\rho R^{\wedge \eta }] \dd R\\
    &=
    \max_{ \vec s \in \mac{S}_{N,\eta} }
    \tr_{\mac{H}_\eta }[ \purple{{}^{k}B_{\mkern6mu  \vec q}^{\vec p}}  \ketbra{\vec s} ]. 
    \label{eq:slater-reduction}
\end{align}
\end{lemma}

\begin{proof}
    First, we show that the following identity holds for diagonal unitary $v=\mtr{diag}(e^{i\theta_1},\dots,e^{i\theta_N})$,
    \begin{align}
       U_\eta(v)
    \purple{{}^{k}B_{\mkern6mu  \vec q}^{\vec p}}
    U_\eta^\dagger(v) =  {}^{k}B_{\mkern6mu  \vec q}^{\vec p}
    \end{align}
    Under the adjoint action of
$v=\mtr{diag}(e^{i\theta_1},\dots,e^{i\theta_N})$ where $\theta_i \in [0,2\pi)$, each entry in $R$ transforms as
$R_{ab}\mapsto e^{i(\theta_a-\theta_b)}R_{ab}$. Additionally, the extension map has the following property
\begin{align}
    \mac{E}_{a,k} ( U_a (v) R^{\wedge a}  U_a^\dagger  (v)) 
    &=[ U_a (v) R^{\wedge a}  U_a^\dagger  (v) ] \wedge \1_{\mac{H}_{k-a}} \\
    &= [ U_a (v) R^{\wedge a}  U_a^\dagger  (v) ] \wedge ( U_{k-a}(v) U_{k-a}^\dagger(v) ) \\
    &= U_k(v) [ R^{\wedge a} \wedge \1 _{ \mac{H}_{k-a}} ] U_k^\dagger (v) \\
    &= U_k(v) \mac{E}_{a,k} (R^{\wedge a} ) U_k^\dagger (v). 
\end{align}
Since the estimator $\purple{{}^k \widehat{D} }$ is a linear combination of the terms
$\mac E_{a,k}(R^{\wedge a})$, it follows that
\begin{align}
    \purple{{}^k \widehat{D} }(vRv^\dagger)
    =
    U_k(v)\purple{{}^k \widehat{D} }(R)U_k^\dagger(v).
\end{align}
When we denote 
\(
    \theta_{\vec q}
    :=
    \sum_{j=1}^k\purple{\theta_{q_j}},
    \theta_{\vec p}
    :=
    \sum_{j=1}^k\purple{\theta_{p_j}},
\)
we have
\(
    U_k(v)\ket{\vec q}
    =
    e^{i\theta_{\vec q}}\ket{\vec q},
    U_k(v)\ket{\vec p}
    =
    e^{i\theta_{\vec p}}\ket{\vec p}.
\)
Consequently,
\begin{align}
    \purple{{}^{k}\widehat D_{\mkern6mu\vec q}^{\vec p}}(vRv^\dagger)
    &=
    \bra{\vec q}
    U_k(v)\purple{{}^k \widehat{D} }(R)U_k^\dagger(v)
    \ket{\vec p}
    \\
    &=
    e^{i(\theta_{\vec q}-\theta_{\vec p})}
    \purple{{}^{k}\widehat D_{\mkern6mu\vec q}^{\vec p}}(R),
\end{align}
and hence
\begin{align}
    \left|
        \purple{{}^{k}\widehat D_{\mkern6mu\vec q}^{\vec p}}(vRv^\dagger)
    \right|^2
    =
    \left|
        \purple{{}^{k}\widehat D_{\mkern6mu\vec q}^{\vec p}}(R)
    \right|^2.
    \label{eq:estimator-phase-invariance}
\end{align}
Using the definition of $\purple{{}^{k}B_{\mkern6mu  \vec q}^{\vec p}}$, we therefore obtain
\begin{align}
    U_\eta(v)
    \purple{{}^{k}B_{\mkern6mu  \vec q}^{\vec p}}
    U_\eta^\dagger(v)
    &=
    \binom{N}{\eta}
    \int_{\Gr_\eta(\mathbb C^N)}
    \left|
        \purple{{}^{k}\widehat D_{\mkern6mu\vec q}^{\vec p}}(R)
    \right|^2
    U_\eta(v)R^{\wedge\eta}U_\eta^\dagger(v)
    \,\mathrm dR
    \\
    &\quad=
    \binom{N}{\eta}
    \int_{\Gr_\eta(\mathbb C^N)}
    \left|
        \purple{{}^{k}\widehat D_{\mkern6mu\vec q}^{\vec p}}(R)
    \right|^2
    (vRv^\dagger)^{\wedge\eta}
    \,\mathrm dR
    \\
    &\quad=
    \binom{N}{\eta}
    \int_{\Gr_\eta(\mathbb C^N)}
    \left|
        \purple{{}^{k}\widehat D_{\mkern6mu\vec q}^{\vec p}}
        (v^\dagger Rv)
    \right|^2
    R^{\wedge\eta}
    \,\mathrm dR
    \\
    &\quad=
    \binom{N}{\eta}
    \int_{\Gr_\eta(\mathbb C^N)}
    \left|
        \purple{{}^{k}\widehat D_{\mkern6mu\vec q}^{\vec p}}(R)
    \right|^2
    R^{\wedge\eta}
    \,\mathrm dR
    \\
    &\quad=
    \purple{{}^{k}B_{\mkern6mu  \vec q}^{\vec p}}.
\end{align}
Here, the third equality follows from the invariance of the Grassmannian measure $\dd{R} = \dd{(v R v^\dagger )}$,
whereas the fourth equality uses
Eq.~\eqref{eq:estimator-phase-invariance}.

Since the restriction of \(U_\eta\) to the diagonal torus
\(\mtr U(1)^N\) decomposes multiplicity-freely into the
one-dimensional weight spaces, 
\begin{align}
    \mac{H}_\eta = \bigoplus_{\vec s \in \mac{S}_{N,\eta } }\mathbb{C} \ket{\vec s}
\end{align}
Schur's lemma implies
that
\begin{align}
    \purple{{}^{k}B_{\mkern6mu  \vec q}^{\vec p}}
    =
    \sum_{\substack{\vec s\in \mac{S}_{N,\eta}}}
    c_{\vec s}^{(\vec q,\vec p)} \ketbra{\vec s},
\end{align}
where
$
    c_S^{(\vec q,\vec p)}
    =
    \bra{\vec s}\purple{{}^{k}B_{\mkern6mu  \vec q}^{\vec p}}\ket{\vec s}
    \geq 0.
$
From this fact, we finally obtain 
\begin{align}
    \sup_{\rho } \purple{\mathbb{E}_{\rho}} \left[
        |\purple{{}^{k}\widehat D_{\mkern6mu\vec q}^{\vec p}}|^2
    \right]
    &=
    \sup_{\rho } 
    \tr_{\mac{H}_\eta } (\rho \purple{{}^{k}B_{\mkern6mu  \vec q}^{\vec p}} )
    \\
    &\leq
    \max_{\vec{s} \in \mac{S}_{N,\eta }} \tr_{\mac{H}_\eta }
    \qty[ \ketbra{s} \purple{{}^{k}B_{\mkern6mu  \vec q}^{\vec p}} ]  = \max_{\vec{s} \in \mac{S}_{N,\eta }} \mathbb{E}_{\ket{\vec s} }\left[
        |\purple{{}^{k}\widehat D_{\mkern6mu\vec q}^{\vec p}}|^2
    \right] .
\end{align}

\end{proof}

For the remainder of the proof, we fix
$\vec s \in \mac{S}_{N,\eta}$, and use the input state
$\ket{\vec s}$. Since
$
    \langle \vec s|
    R^{\wedge\eta}
    |\vec s\rangle
    =
    \det(R_{\vec s, \vec s}),
$
the corresponding probability measure is expressed as 
\begin{align}
    \nu_{\vec s}(\dd R)
    =
    \binom{N}{\eta}
    \det(R_{\vec s,\vec s})
    \dd R.
    \label{eq:slater-measure}
\end{align}

Direct evaluation of the diagonal entries of ${}^{k}B_{\mkern6mu\vec q}^{\vec p}$ via Weingarten calculus involves cumbersome sums arising from determinant expansions. We therefore derive upper bounds on these entries in terms of determinant-weighted expectations over the complex Ginibre ensemble.

We first recall the definition of the standard complex Ginibre ensemble.
Let $G\in\mathbb C^{N\times\eta}$ be a standard complex Ginibre
matrix~\cite{ginibre1965statistical}, whose entries are independent
standard complex Gaussian random variables satisfying
\begin{align}
    \mathbb E_0[G_{ia}]
    &=
    0,
    \qquad
    \mathbb E_0\!\left[
        G_{ia}\overline{G_{jb}}
    \right]
    =
    \delta_{ij}\delta_{ab}.
    \label{eq:ginibre-covariance}
\end{align}
Here, $\mathbb E_0$ denotes expectation with respect to the
Gaussian probability measure $\gamma_0$, defined by
\begin{align}
    \gamma_0(\dd G)
    &:=
    \pi^{-N\eta}
    e^{-\tr(G^\dagger G)}
    \dd G,
    \label{eq:ginibre-measure}
\end{align}
where $\dd G$ denotes the Lebesgue measure on
$\mathbb C^{N\times\eta}\simeq\mathbb R^{2N\eta}$.

Using $\gamma_0$ as a reference measure, we next introduce a
determinant-tilted Ginibre measure associated with the
orbital-rotation shadow.
For $\vec s\in\mac S_{N,\eta}$, let $G_{\vec s}$ be the
$\eta\times\eta$ submatrix of $G$ formed by the rows indexed by
$\vec s$.
Its determinant satisfies the second-moment identity
\begin{align}
    \mathbb E_0\!\left[
        |\det G_{\vec s}|^2
    \right]
    &=
    \eta!.
    \label{eq:ginibre-minor-moment}
\end{align}
Based on that, we can define the following tilted Ginibre measure
\begin{align}
    \gamma_{\vec s}(\dd G)
    :=
    \frac{|\det G_{\vec s}|^2}{\eta!}\,
    \gamma_0(\dd G),
    \label{eq:tilted-ginibre-measure}
\end{align}
which is a probability measure.  We write
\begin{align}
    \mathbb E_{\vec s}[F(G)]
    :=
    \frac{1}{\eta!}
    \mathbb E_0\!\left[
        |\det G_{\vec s}|^2F(G)
    \right]
    \label{eq:tilted-ginibre-expectation}
\end{align}
for expectation with respect to $\gamma_{\vec s}$.

The relevance of this tilted measure is that the angular part of
$G$ reproduces precisely the orbital-shadow measurement
distribution. The proof of this lemma is presented separately in Lemma.~\ref{lem:gammaS-polar}
Sec.~\ref{app:technical_lemma}.

\begin{lemma}
\label{lem:tilted-polar}
Let $N,\eta$ be integers satisfying $1\leq\eta\leq N$, and fix
$\vec s\in\mac S_{N,\eta}$.
Let $G\in\mathbb C^{N\times\eta}$ be distributed according to $\gamma_{\vec s}$ and write its
polar decomposition as
\begin{align}
    G=LT^{1/2},
    \qquad
    L^\dagger L=\1_\eta,
    \qquad
    T=G^\dagger G.
\end{align}
Additionally, we denote the rank-$\eta$ projector as $R:=LL^\dagger$. 

Then $R$ and $T$
are independent.
In particular, the law of $R$ coincides with the orbital-shadow
measurement distribution associated with the input state
$\ketbra{\vec s}$:
\begin{align}
    R
    &\sim
    \nu_{\vec s}(\dd R)
    :=
    \binom{N}{\eta}
    \det(R_{\vec s,\vec s})\,\dd R,
\end{align}
Moreover, for every $0\leq t\leq\eta$,
\begin{align}
    \mathbb E_{\vec s}\!\left[
        T^{\wedge t}
    \right]
    =
    \fall{N+1}{t}\,
    \1_{\wedge^t\mathbb C^\eta}.
    \label{eq:wishart-wedge-moment}
\end{align}
\end{lemma}

We now use this fact to express the orbital-rotatoin shadow estimator
as a conditional expectation of a Gaussian polynomial.  For a
positive semidefinite operator $Y \in \mtr{End}(\mac{H}_1) $, define
\begin{align}
    \mathfrak D_\eta^{(k)}(Y)
    :=
    \sum_{t=0}^{k}
    (-1)^{k+t}
    \frac{(\eta-t)!}{k!(\eta-k)!}
    \mac E_{t,k}(Y^{\wedge t}).
    \label{eq:dilute-estimator}
\end{align}
Using this operator, we obtain the following upper bound on
the variance of the orbital-rotation shadow estimator.

\begin{lemma}[Rao--Blackwell identity]
\label{lem:rao-blackwell}
Let $N,\eta,k$ be integers satisfying $1\leq k\leq\eta<N$,
and fix $\vec s\in\mac S_{N,\eta}$.
Let $G\in\mathbb C^{N\times\eta}$ be distributed according to
the determinant-tilted measure $\gamma_{\vec s}$ defined in
Eq.~\eqref{eq:tilted-ginibre-measure}.
Write its polar decomposition as
\begin{align}
    G=LT^{1/2},
    \qquad
    T=G^\dagger G,
    \qquad
    L^\dagger L=\1_\eta,
\end{align}
and define $
    R:=LL^\dagger
$.

Then, the following identity holds,
\begin{align}
    \widehat D_N^{(k)}(R)
    =
    \mathbb E_{\vec s}\!\left[
        \mathfrak D_\eta^{(k)}(GG^\dagger)
        \given R
    \right].
    \label{eq:rao-blackwell}
\end{align}
Consequently, for every
$\vec p,\vec q\in\mac S_{N,k}$,
\begin{align}
    \mathbb E_{{\vec s}}\!\left[
        \left|
            \purple{{}^{k}\widehat D_{\mkern6mu\vec q}^{\vec p}}(R)
        \right|^2
    \right]
    \leq
    \frac{1}{\eta!}
    \mathbb E_0\!\left[
        |\det G_{\vec s}|^2
        \left|
            \purple{{}^{k}\mathfrak D_{\eta;\vec q}^{\vec p}}
            (GG^\dagger)
        \right|^2
    \right].
    \label{eq:conditional-jensen}
\end{align}
Here, the matrix elements of
$\mathfrak D_\eta^{(k)}(GG^\dagger)$ is defined by
$
    \bra{\vec q}
    \mathfrak D_\eta^{(k)}(GG^\dagger)
    \ket{\vec p}.
$
\end{lemma}

\begin{proof}
For simplicity, we denote $Y = LTL^\dagger$ which satifies $Y = GG^\dagger$. The functoriality of the exterior power gives
$
    Y^{\wedge t}
    =
    (L^{\wedge t})
    T^{\wedge t}
    (L^{\wedge t})^\dagger .
$
Under the determinant-tilted measure, \(L\) and \(T\) are independent.
Therefore,
\begin{align}
    \mathbb E_{\vec s}\!\left[
        Y^{\wedge t}
        \given L
    \right]
    &=
    (L^{\wedge t})
    \mathbb E_{\vec s}\!\left[
        T^{\wedge t}
    \right]
    (L^{\wedge t})^\dagger
\\
    &=
    \fall{N+1}{t}\,
    (L^{\wedge t})(L^{\wedge t})^\dagger
\\
    &=
    \fall{N+1}{t}\,
    R^{\wedge t}.
    \label{eq:conditional-on-L}
\end{align}
The second equality follows from
Eq.~\eqref{eq:wishart-wedge-moment}.
Since $R=LL^\dagger$ is a function of $L$ and the RHS
depends only on $R$, we condition both sides on $R$ and apply
the tower property of conditional expectation to obtain
\begin{align}
    \mathbb E_{\vec s}\!\left[
        Y^{\wedge t}
        \given R
    \right]
    &=
    \mathbb E_{\vec s}\!\left[
        \mathbb E_{\vec s}\!\left[
            Y^{\wedge t}
            \given L
        \right]
        \given R
    \right]
    \nonumber\\
    &=
    \fall{N+1}{t}\,
    R^{\wedge t}.
    \label{eq:conditional-wedge-moment}
\end{align}
Applying this identity term by term to
Eq.~\eqref{eq:dilute-estimator} yields
\begin{align}
    \mathbb E_{\vec s}\!\left[
        \mathfrak D_\eta^{(k)}(Y)
        \given R
    \right]
    &=
    \sum_{t=0}^{k}
    (-1)^{k+t}
    \frac{(\eta-t)!}{k!(\eta-k)!}
    \fall{N+1}{t}\,
    \mac E_{t,k}(R^{\wedge t})
    \nonumber\\
    &=
    \widehat D_N^{(k)}(R),
\end{align}
where the last equality follows from
\begin{align}
    \frac{(\eta-t)!}{k!(\eta-k)!}
    \fall{N+1}{t}
    =
    \binom{N+1}{t}
    \frac{\binom{\eta-t}{k-t}}{\binom{k}{t}}.
\end{align}
Therefore, we prove Eq.~\eqref{eq:rao-blackwell}.

For fixed \(\vec p,\vec q\in\mac S_{N,k}\),
taking the \((\vec p,\vec q)\) matrix element of
Eq.~\eqref{eq:rao-blackwell} yields
$
    \purple{{}^{k}\widehat D_{\mkern6mu\vec q}^{\vec p}}(R)
    =
    \mathbb E_{\vec s}\!\left[
        \purple{{}^{k}\mathfrak D_{\eta;\vec q}^{\vec p}} \given R
    \right].
$
We apply conditional Jensen's inequality for the convex function
$z\mapsto |z|^2$.
For a square-integrable complex-valued random variable $X_1$
and any random variable $X_2$ on the same probability space,
which may be vector- or matrix-valued, this gives
\begin{align}
    \abs{\mathbb E[X_1\mid X_2]}^2
    \leq
    \mathbb E\!\left[\abs{X_1}^2\mid X_2\right].
\end{align}
This inequality, followed by
the tower property, therefore gives
\begin{align}
    \mathbb E_{\vec s}\!\left[
        \left|
            \purple{{}^{k}\widehat D_{\mkern6mu\vec q}^{\vec p}}(R)
        \right|^2
    \right]
    &=
    \mathbb{E}_{\vec s}  \qty[ \abs{ \mathbb{E}_{\vec s}[\purple{{}^{k}\mathfrak D_{\eta;\vec q}^{\vec p}} (GG^\dagger)  \mid R] }^2 ]
    \\
    &\leq
    \mathbb E_{\vec s}\!\left[
        \mathbb E_{\vec s}\!\left[
            |\purple{{}^{k}\mathfrak D_{\eta;\vec q}^{\vec p}} (GG^\dagger)|^2
            \given R
        \right]
    \right]
    \\
    &=
    \mathbb E_{\vec s}\!\left[
        |\purple{{}^{k}\mathfrak D_{\eta;\vec q}^{\vec p}} (GG^\dagger)|^2
    \right].
\end{align}
The marginal law of \(R\) under \(\gamma_{\vec s}\) is the
orbital-shadow distribution associated with \(\ketbra{\vec s}\).
Moreover, by Eq.~\eqref{eq:tilted-ginibre-expectation},
\begin{align}
    \mathbb E_{\vec s}\!\left[
        |\purple{{}^{k}\mathfrak D_{\eta;\vec q}^{\vec p}} (GG^\dagger)|^2
    \right]
    =
    \frac{1}{\eta!}
    \mathbb E_0\!\left[
        |\det G_{\vec s}|^2
        \left|
            \purple{{}^{k}\mathfrak D_{\eta;\vec q}^{\vec p}}
            (GG^\dagger)
        \right|^2
    \right].
\end{align}
Combining these identities proves
Eq.~\eqref{eq:rao-blackwell}.
\end{proof}

Therefore, it remains to bound the right-hand side of
Eq.~\eqref{eq:rao-blackwell}, which is a purely Gaussian moment
estimate.  The main tool is the complex Wiener--It\^o chaos
decomposition~\cite{chen2019complex,chen2017product}:  It decomposes a square-integrable polynomial $F$ of
complex Gaussian variables into components $F_{p,q}$, where
$F_{p,q}$ has $p$ Gaussian and $q$ conjugate-Gaussian degrees and is
called the $(p,q)$-th complex Wiener--It\^o chaos component of $F$.
More formally, the following identity holds
\begin{align}
    F=\sum_{p,q\geq0}F_{p,q},
    \qquad
    \mathbb E_0\!\left[
        F_{p,q}\overline{F_{p',q'}}
    \right]=0
    \quad\text{if}\quad
    (p,q)\neq(p',q').
\end{align}
Applying this decomposition to the determinant-weighted estimator is
useful because its second moment reduces to a sum of
the squared norms of the individual chaos components, with all cross
terms vanishing.
This avoids termwise bounds on the alternating determinant
expansions, which would discard cancellations and introduce spurious
combinatorial factors.  We state the resulting tighter bound below
and sketch its proof.  The details are deferred to
Sec.~\ref{sec:gaussian-moment-bound}.

\begin{lemma}
\label{lem:gaussian-moment-bound}
Let $k,\eta,N$ be integers satisfying
$\purple{2\leq k\leq \min\{\eta,N-\eta\}}$, and let
$G\in\mathbb C^{N\times\eta}$ be a standard complex Ginibre
matrix.  For every $\vec s\in\mac S_{N,\eta}$ and
$\vec q,\vec p\in\mac S_{N,k}$,
\begin{align}
    \frac{1}{\eta!}
    \mathbb E_0\!\left[
        |\det G_S|^2
        \left|
            \purple{{}^{k}\mathfrak D_{\eta;\vec q}^{\vec p}}(GG^\dagger)
        \right|^2
    \right]
    \leq
    \frac{k+1}{k!} \eta^k.
\end{align}
\end{lemma}

\begin{proof}[Proof sketch]


For each $\vec a\in\mac S_{\eta,k}$, let
$G_{\vec q,\vec a}$ denote the $k\times k$ submatrix of $G$
whose rows and columns are indexed by $\vec q$ and $\vec a$,
respectively.  
As shown in Sec.~\ref{sec:gaussian-moment-bound}, the
determinant-weighted Gaussian polynomial admits the decomposition
\begin{align}
    \det G_{\vec s,[\eta]}\,
    \purple{{}^{k}\mathfrak D_{\eta;\vec q}^{\vec p}}(GG^\dagger)
    =
    \sum_{r=0}^{k}
    \purple{\Xi_{\vec s;\vec q,\vec p}^{(r)}},
\end{align}
where $\purple{\Xi_{\vec s;\vec q,\vec p}^{(r)}}$ belongs to the
$(\eta+k-r,k-r)$-th complex Wiener--It\^o chaos.
The construction of these components from \purple{the Wick-ordered exterior-power matrix element and the contraction
decomposition} is given in
Sec.~\ref{sec:gaussian-moment-bound}.

A fundamental property of the complex Wiener--It\^o chaos
decomposition is that chaoses of distinct bidegrees are orthogonal
in $L^2(\gamma_0)$; more precisely, if $X$ belongs to the
$(p,q)$-th chaos and $Y$ belongs to the $(p',q')$-th chaos, then
\begin{align}
    \mathbb E_0[X\overline{Y}]=0
    \qquad
    \text{whenever }
    (p,q)\neq(p',q').
\end{align}
Since different contraction numbers $r\neq r'$ give different
bidegrees, the corresponding components are orthogonal:
\begin{align}
    \mathbb E_0\!\left[
        \purple{\Xi_{\vec s;\vec q,\vec p}^{(r)}}
        \overline{
            \purple{\Xi_{\vec s;\vec q,\vec p}^{(r')}}
        }
    \right]
    =0
    \qquad
    (r\neq r').
\end{align}
Consequently, orthogonality separates the squared norm into the sum
of the squared norms of the individual chaos components:
\begin{align}
    \frac{1}{\eta!}
    \mathbb E_0\!\left[
        \left|\det G_{\vec s,[\eta]}\right|^2
        \left|
            \purple{{}^{k}\mathfrak D_{\eta;\vec q}^{\vec p}}
            (GG^\dagger)
        \right|^2
    \right]
    &=
    \frac{1}{\eta!}
    \mathbb E_0\!\left[
        \left|
            \sum_{r=0}^{k}
            \purple{\Xi_{\vec s;\vec q,\vec p}^{(r)}}
        \right|^2
    \right]
    \nonumber\\
    &=
    \frac{1}{\eta!}
!    \sum_{r=0}^{k}
    \sum_{r'=0}^{k}
    \mathbb E_0\!\left[
        \purple{\Xi_{\vec s;\vec q,\vec p}^{(r)}}
        \overline{
            \purple{\Xi_{\vec s;\vec q,\vec p}^{(r')}}
        }
    \right]
    \nonumber\\
    &=
    \frac{1}{\eta!}
    \sum_{r=0}^{k}
    \mathbb E_0\!\left[
        \left|
            \purple{\Xi_{\vec s;\vec q,\vec p}^{(r)}}
        \right|^2
    \right].
\end{align}

Using the representation theory, the norms
of these components can be evaluated explicitly; as shown in
Sec.~\ref{sec:gaussian-moment-bound}, the $r$-th component is bounded by
\begin{align}
    \frac{1}{\eta!}
    \mathbb E_0\!\left[
        \left|
            \purple{\Xi_{\vec s;\vec q,\vec p}^{(r)}}
        \right|^2
    \right]
    \leq
    \binom{\eta}{k}
    \frac{\eta+1}{\eta-k+r+1},
\end{align}
for $0\leq r\leq k$.
This inequlaity provides 
\begin{align}
    \frac{1}{\eta!}
    \mathbb E_0\!\left[
        \left|\det G_{\vec s,[\eta]}\right|^2
        \left|
            \purple{{}^{k}\mathfrak D_{\eta;\vec q}^{\vec p}}
            (GG^\dagger)
        \right|^2
    \right]
    \leq
    \binom{\eta}{k}
    (\eta+1)
    \sum_{r=0}^{k}
    \frac{1}{\eta-k+r+1}.
\end{align}
The RHS of the last inequality is bounded by $(k+1)\eta^k/k!$ for $k \geq 2$, which proves the target claim.
\end{proof}

\begin{proof}[Proof of Theorem~\ref{thm:entrywise-variance}]
For $k\geq2$, the claim follows from
Lemma~\ref{lem:slater-reduction},
Eq.~\eqref{eq:conditional-jensen}, and
Lemma~\ref{lem:gaussian-moment-bound}.

For $k=1$, we use the exact variance formulas derived in
Sec.~\ref{sec:1-rdm_variance}.
Write
$n_p:=a_p^\dagger a_p$,
$d_p:=\operatorname{tr}(\rho n_p)$, and
$d_{pq}:=\operatorname{tr}(\rho n_p n_q)$ for $p\neq q$.
For diagonal entries, these formulas give
\begin{align}
    \operatorname{Var}
    \left({}^{1}\widehat D_{\mkern6mu p}^{p}\right)
    &=
    \frac{N+1-\eta}{N+2}(\eta+2d_p)-d_p^2
    \nonumber\\
    &\leq
    \eta+2d_p-d_p^2
    \leq
    \eta+1,
\end{align}
where $0\leq d_p\leq1$.
For off-diagonal entries, the inequalities
\begin{align}
    0<
    \frac{(N+1)(N+1-\eta)}{N(N+2)}
    \leq1,
    \qquad
    \frac{N+1}{N}\geq1
\end{align}
and the exact variance formula imply
\begin{align}
    \operatorname{Var}
    \left({}^{1}\widehat D_{\mkern6mu q}^{p}\right)
    &\leq
    \eta+d_p+d_q-d_{pq}
    \nonumber\\
    &\leq
    \eta+1.
\end{align}
The last inequality follows from
$0\leq n_p+n_q-n_pn_q\leq I$.
Thus, for every $p,q\in[N]$,
\begin{align}
    \operatorname{Var}
    \left({}^{1}\widehat D_{\mkern6mu q}^{p}\right)
    \leq
    \eta+1
    \leq
    2\eta,
\end{align}
which is precisely the claimed bound for $k=1$.
\end{proof}

\end{revision}


Building on the variance bound established for the single-shot estimator, we analyze the performance of the corresponding empirical-mean.  Let \(M_{\rm obs}\) denote the number of target
observables, which is the number of entries of the \(k\)-RDM to be estimated.
For $ \omega_j:=\mathrm{tr}(O_j\rho) $, for $
    j=1,\ldots,M_{\rm obs},
$
suppose that \(\widehat\omega_j\) is a single-shot unbiased estimator satisfying
\begin{align}
    \purple{\mathbb{E}_{\rho}}[\widehat\omega_j]=\omega_j,
    \qquad
    \mathrm{Var}[\widehat\omega_j]
    =
    \purple{\mathbb{E}_{\rho}}
    \left[
        \left|
        \widehat\omega_j-\omega_j
        \right|^2
    \right]
    \le V_j .
    \label{eq:general_variance_bound}
\end{align}
Given \(L\) independent trials, we introduce the typical mean estimator expressed as follows,
\begin{align}
    \widehat\omega_j(L)
    :=
    \frac{1}{L}
    \sum_{\ell=1}^{L}
    \widehat\omega_j^{(\ell)} .
\end{align}
Since the estimator is unbiased, its MSE coincides with its variance.  Hence, the independence of single-shot estimators yields
\begin{align}
    \mathrm{MSE}\!\left[\widehat\omega_j(L)\right]
    &:=
    \mathbb{E}_{\rho}
    \left[
        \left|
        \widehat\omega_j(L)-\omega_j
        \right|^2
    \right] \\
    &=
    \mathrm{Var}_{\rho}\!\left[\widehat\omega_j(L)\right]
    =
    \frac{1}{L}\,
    \mathrm{Var}_{\rho}[\widehat\omega_j]
    \le
    \frac{V_j}{L}.
    \label{eq:mse_from_variance_bound}
\end{align}
Therefore, choosing
$
    L
    \ge
    {\max_{1\le j\le M_{\rm obs}} V_j}/{\varepsilon^2}
    \label{eq:sample_complexity_from_variance}
$
guarantees
\begin{align}
    \forall j \in \{1,\ldots,M_{\mtr{obs}} \}, 
    \mtr{Var}[\widehat{\omega}_j (L) ] \leq \varepsilon^2.
\end{align}
From this discussion, we directly derive the following corollary regarding the sample complexity regarding variance.
\begin{corollary}
    Let \(N,\eta,k\) be integers satisfying $ 1\le k\le \min(\eta ,N-\eta)$, $\eta < N,$ and $\varepsilon \in (0,1)$.
    For every \(N\)-mode \(\eta\)-particle state
    \(\rho\in\operatorname{End}(\mac H_\eta)\), there exists a quantum algorithm that outputs
    estimators \(\purple{{}^{k}\widehat D_{\mkern6mu\vec q}^{\vec p}}\) of all entries of the
    \(k\)-RDM \(D^{(k)}_\rho\) satisfying
    \begin{align}
        \forall \vec p, \vec q \in \mac{S}_{N,k}, \quad
        \operatorname{Var}
        \left(
            \purple{{}^{k}\widehat D_{\mkern6mu\vec q}^{\vec p}}
        \right)
        &\le
        \varepsilon^2
    \end{align}
    for all \(\vec p,\vec q\in\mac S_{N,k}\), using
    \purple{
    $ C_k \eta^k / \varepsilon^2  $
    copies of \(\rho\), where $C_k = (k+1)/k! $. 
    }
    \label{cor:entrywise_mse_sample_complexity}
\end{corollary}
This result also yields performance guarantees for the standard
median-of-means estimator introduced in Ref.~\cite{Huang:2020tih}.
Let \(n_s\) denote the number of samples in each batch, i.e., the batch size,
and let \(n_B\) denote the number of batches.  For each
\(r=1,\ldots,n_B\), define the average over the \(r\)-th batch by
\begin{align}
    \overline{\omega}^{(r)}_j(n_s)
    :=
    \frac{1}{n_s}
    \sum_{\ell=1}^{n_s}
    \widehat{\omega}^{(r,\ell)}_j,
    \qquad
    j=1,\ldots,M_{\rm obs},
\end{align}
where \(\widehat{\omega}^{(r,\ell)}_j\) denotes the single-shot estimator
obtained from the \(\ell\)-th sample in the \(r\)-th batch.  The
median-of-means estimator is then defined by
\begin{align}
    \widehat{\omega}^{\rm MoM}_j(n_s,n_B)
    :=
    \operatorname{med}
    \left\{
        \overline{\omega}^{(r)}_j(n_s)
    \right\}_{r=1}^{n_B}.
\end{align}
This estimator satisfies the following property, 

\begin{corollary}[Median-of-means estimator from Ref.~\cite{Huang:2020tih}]
    Let \(\varepsilon,\delta\in(0,1)\).  Suppose that, for each
    \(j=1,\ldots,M_{\rm obs}\), the single-shot estimator
    \(\widehat{\omega}_j\) is unbiased and satisfies
    $
        \operatorname{Var}\!\left[\widehat{\omega}_j\right]
        \le V_j .
    $
    Set
    \begin{align}
        n_B
        =
        \left\lceil
        2\log\frac{2M_{\rm obs}}{\delta}
        \right\rceil ,
        \qquad
        n_s
        =
        \left\lceil
        \frac{34}{\varepsilon^2}
        \max_{1\le j\le M_{\rm obs}} V_j
        \right\rceil
        .
    \end{align}
    Then, using \(n_s n_B\) independent single-shot classical shadow, the median-of-means
    estimators satisfy
    \begin{align}
        \forall j\in\{1,\ldots,M_{\rm obs}\}, \quad 
        \Pr
        \left[
            \left|
                \widehat{\omega}^{\rm MoM}_j(n_s,n_B)
                -
                \omega_j
            \right|
            \le
            \varepsilon
        \right]
        \ge
        1-\delta .
    \end{align}
    \label{cor:mom_from_variance_bound}
\end{corollary}

\begin{revision}

To estimate all $k$-RDM entries, we extend this construction
to complex-valued estimators, since off-diagonal entries
may be complex.
Suppose that each $\widehat{\omega}_j$ is unbiased and satisfies
$\mathbb{E}[|\widehat{\omega}_j-\omega_j|^2]\le V_j$.
We define the median-of-means estimate by taking the medians
of the real and imaginary parts of the batch means separately:
\begin{align}
    \widehat{\omega}^{\mathrm{MoM}}_j(n_s,n_B)
    &:=
    \operatorname{med}_{1\le r\le n_B}
    \operatorname{Re}\omega_j^{(r)}(n_s)
    +i\operatorname{med}_{1\le r\le n_B}
    \operatorname{Im}\omega_j^{(r)}(n_s).
\end{align}
Choose
\begin{align}
    n_B
    &=
    \left\lceil
        2\log\frac{4M_{\mathrm{obs}}}{\delta}
    \right\rceil,
    \qquad
    n_s
    =
    \left\lceil
        \frac{68}{\varepsilon^2}
        \max_{1\le j\le M_{\mathrm{obs}}}V_j
    \right\rceil.
\end{align}
Then, using $n_sn_B$ independent classical-shadow snapshots,
the resulting estimates satisfy
\begin{align}
    \Pr\!\left[
        \max_{1\le j\le M_{\mathrm{obs}}}
        \left|
            \widehat{\omega}^{\mathrm{MoM}}_j(n_s,n_B)
            -\omega_j
        \right|
        \le\varepsilon
    \right]
    \ge1-\delta.
\end{align}
Indeed, the real and imaginary parts of each
$\widehat{\omega}_j$ both have variance at most $V_j$.
Applying the real-valued median-of-means bound to each
component with accuracy $\varepsilon/\sqrt{2}$ and failure
probability at most $\delta/(2M_{\mathrm{obs}})$,
and taking a union bound over all $2M_{\mathrm{obs}}$ components,
gives the stated simultaneous guarantee. 

Applying this construction with
$M_{\mathrm{obs}}=\binom{N}{k}^{2}$ and the uniform
single-shot variance bound $V_j\le C_k\eta^k$,
where $C_k=(k+1)/k!$, yields the following corollary.

\begin{corollary}[Simultaneous additive-error estimation of the $k$-RDM]
    \label{cor:entrywise_additive_sample_complexity}
    Let $N,\eta,k$ be integers satisfying
    $1\le k\le\min\{\eta,N-\eta\}$,
    and let $\varepsilon,\delta\in(0,1)$.
    For every $N$-mode $\eta$-particle state
    $\rho\in\operatorname{End}(\mac H_\eta)$,
    there exists a quantum algorithm
    that outputs estimates
    ${}^{k}\widehat D_{\mkern6mu\vec q}^{\vec p}$
    of all entries of the $k$-RDM $D^{(k)}_\rho$
    satisfying
    \begin{align}
        \Pr\!\left[
            \max_{\vec p,\vec q\in\mac S_{N,k}}
            \left|
                {}^{k}\widehat D_{\mkern6mu\vec q}^{\vec p}
                -
                {}^{k}D_{\rho;\vec q}^{\vec p}
            \right|
            \le\varepsilon
        \right]
        \ge1-\delta,
    \end{align}
    using 
    \begin{align}
        L
        =
        \left\lceil
            \frac{68C_k\eta^k}{\varepsilon^2}
        \right\rceil
        \left\lceil
            2\log\frac{4\binom{N}{k}^{2}}{\delta}
        \right\rceil
        &=
        \mac{O}_k\!\left(
            \frac{\eta^k}{\varepsilon^2}
            \log\frac{N}{\delta}
        \right),
    \end{align}
    copies of $\rho$, where $C_k=(k+1)/k!$.
\end{corollary}

\end{revision}


\section{Lower bound}

In this section, we evaluate the lower bound of learning all $k$-RDM elements under particle number symmetry. 
As we discussed in the previous section, when we consider the particle number symmetry, the dependence on $N$ vanishes and the intrinsic cost is determined by the particle number $\eta$. In this context, the lower bound of the sample complexity under particle number symmetry arises as a natural question. We answer this question by showing that the lower bound matches the scaling of our proposed upper bound, which implies that our proposal is optimal in terms of $\eta$.
Our lower bound result is expressed as follows.
\begin{theorem}
    \label{thm:lower_bound}
    Let $k, \eta, N$ be integers satisfying $k \leq \eta \leq N/2$, and let $\varepsilon \in (0, \purple{1/(8\cdot 4^k)]}$.
    For a given $\eta$-particle $N$-mode state $\rho$, any quantum algorithm that performs adaptive single-copy measurements of $\rho$ requires $\Omega_k(\eta^k)/\varepsilon^2$ copies of $\rho$ to estimate all elements of the associated $k$-RDM $\{ \purple{{}^{k}D_{\mkern6mu\vec q}^{\vec p}} \}_{\vec p, \vec q \in \mac{S}_{N,k}}$ within an element-wise additive error $\varepsilon$ with high probability. Here, the associated $k$-RDM elements are defined by
    \begin{align}
        \purple{{}^{k}D_{\mkern6mu\vec q}^{\vec p}} \coloneq \purple{\tr \qty[ \rho a^\dagger_{p_1} \ldots a^\dagger_{p_k} a_{q_k} \ldots a_{q_1} ]}.
    \end{align}
\end{theorem}

\subsection{The core techniques}
Before showing the proof of this theorem, we briefly overview the core techniques related to an information-theoretic framework developed in Ref.~\cite{chen2022exponential}. This framework provides a general way to prove lower bounds on learning a target property of a given quantum state $\rho$ via any algorithm without external quantum memory: the only quantum memory we can access is the storage of $\rho$. Such a learning algorithm consists of several steps and a single step proceeds as follows: the algorithm measures a single copy of the given unknown state and records the measurement outcome in classical memory. Here, the state stored in classical memory is the collection of all previous measurement outcomes. This classical state is updated by a POVM measurement outcome on a single copy of the unknown quantum state $\rho$. By repeating this process $T$ times, the final classical state is used to make a prediction for a target property of $\rho$. To include any adaptive algorithm, the POVM measurement in each step can depend on the previous classical state.

The entire transition of the classical state can be represented by a rooted tree with depth $T$~\cite{chen2022exponential}. At the root, the classical state is initialized to $r$, which is independent of the unknown state $\rho$. The transition from a node $u$ at depth $t$ to a node $v$ at depth $t+1$ is determined by the outcome $s$ of the POVM $\{M_s^u\}$ on $\rho$. Here, the POVM $\{M_s^u\}$ depends on the previous classical state $u$ and can always be written as $M_s^u = w_s^u \ketbra{\psi_s^u}$ for a positive weight $w_s^u$ and a pure state $\ket{\psi_s^u}$. After $T$ measurements, the classical state lies in a leaf $l$ of the rooted tree. For instance, a single path from the root to a leaf is represented as
\begin{equation}
    r \equiv u_1 \xrightarrow{s_1} u_2 \xrightarrow{s_2} \cdots \xrightarrow{s_{T-1}} u_T \xrightarrow{s_T} l.
\end{equation}
Since the classical state is $r$ with probability 1 at the beginning, the probability $\mathbb{P}(l \given \rho)$ of arriving at $l$ along the above path is given by
\begin{equation}
    \label{def:probability_tree_prob}
    \mathbb{P}(l \given \rho) = \prod_{t=1}^{T} \left( w_{s_t}^{u_t} \langle \psi_{s_t}^{u_t} | \rho | \psi_{s_t}^{u_t} \rangle \right).
\end{equation}
The formal definition of this framework is given by

\begin{definition}[Tree representation for learning states~\cite{chen2022exponential}]
\label{def:tree_representation}
Fix an unknown $n$-qubit quantum state $\rho$. A learning algorithm without quantum memory can be expressed as a rooted tree $\mathcal{T}$ of depth $T$, where each node on the tree encodes all the measurement outcomes the algorithm has seen so far. The tree satisfies the following properties.
\begin{itemize}
    \item Each node $u$ is associated with a probability $\mathbb{P}(u \given \rho)$.
    \item For the root $r$ of the tree, $\mathbb{P}(r \given \rho) = 1$.
    \item At each non-leaf node $u$, we measure a POVM $\{\mathcal{M}_s^u\}_s$ on $\rho$ to obtain a classical outcome $s$. Each child node $v$ of the node $u$ is connected through the edge $e_{u,s}$.
    \item If $v$ is the child node of $u$ connected through the edge $e_{u,s}$, then
    \begin{equation}
        \mathbb{P}(v \given \rho) = \mathbb{P}(u \given \rho)\, w_s^u \langle \psi_s^u | \rho | \psi_s^u \rangle,
    \end{equation}
    where $\{w_s^u | \psi_s^u \rangle \langle \psi_s^u |\}_s$ is a rank-1 POVM that depends on the node $u$.
    \item Every root-to-leaf path is of length $T$. Note that for a leaf node $l$, $\mathbb{P}(l \given \rho)$ is the probability that the classical memory is in state $l$ after the learning procedure.
\end{itemize}
\end{definition}
The proof strategy for the lower bound on $T$ is to reduce the learning task of a target property to a two-hypothesis distinguishing task. In the distinguishing task, one of the following events is ensured to occur with equal probability:
\begin{itemize}
    \item The given unknown state $\rho$ is a reference state $\tau$,
    \item The given unknown state $\rho$ is selected uniformly at random from a set of states $\{\rho_x\}$,
\end{itemize}
where we later carefully design $\tau$ and $\{\rho_x\}$ to derive a strong lower bound. Then, the goal of the distinguishing task is to decide which event has occurred. If one can successfully perform the distinguishing task using a no-quantum-memory learning algorithm described by a rooted tree, the leaf probability distributions should be sufficiently different for the two events. According to Ref.~\cite{chen2022exponential}, this is formally stated as follows.
\begin{lemma}[Lemma 5.3 from Ref.~\cite{chen2022exponential}]
    \label{lem:The_condition_distinguishing}
    Consider any learning algorithm without quantum memory that is described by a rooted tree. The probability that the learning algorithm solves the two-hypothesis distinguishing task is upper bounded by
    \begin{align}
        \frac12 + \frac14 \sum_{l \in \mtr{leaf}} \abs{ \mathbb{E}_x [\mathbb{P}(l \given \rho_x)] - \mathbb{P}(l \given \tau) }.
    \end{align}
\end{lemma}
Therefore, any such learning algorithm that solves the distinguishing task with high probability, say $2/3$, must satisfy
\begin{align}
    \frac23 &\leq \frac12 + \frac14 \sum_{l \in \mtr{leaf}} \abs{ \mathbb{E}_x [\mathbb{P}(l \given \rho_x)] - \mathbb{P}(l \given \tau) } \\
    \Rightarrow \frac13 &\leq \frac12 \sum_{l \in \mtr{leaf}} \abs{ \mathbb{E}_x [\mathbb{P}(l \given \rho_x)] - \mathbb{P}(l \given \tau) }.
\end{align}
In the proof of Theorem~\ref{thm:lower_bound}, we first construct an appropriate state set $\{\rho_x\}$ and reference state $\tau$ that reduce learning the $k$-RDM to the two-hypothesis distinguishing task. After that, we derive the lower bound on $T$ by upper bounding the trace distance between the leaf probabilities $\mathbb{E}_x[\mathbb{P}(\cdot \given \rho_x)]$ and $\mathbb{P}(\cdot \given \tau)$.

\subsection{Proof of lower bound}

As a preparation for the proof of Theorem~\ref{thm:lower_bound}, we will prove the following lemma. 
\begin{lemma}
\label{lem:target_op_lower_bound}
Let $\eta, N, k$ be integers, and let $A\subseteq[N]$ be a subset of size
$m:=|A|$ with $k\le\eta\le m$ and $2k\le m$. Write
$\mathcal{H}_{\eta}:=\bigwedge^{\eta}\mathbb{C}^{N}$, and let $\Pi_{A}$
be the orthogonal projector onto
\begin{align}
    \mathcal{H}_{\eta,A}
    :=
    \operatorname{span}\bigl\{\ket{\vec x}\in\mathcal{H}_{\eta}
        : \operatorname{supp}(\vec x)\subseteq A\bigr\} \subset \mac{H}_\eta ,
\end{align}
the subspace of $\eta$-particle states occupying only the modes in $A$. For
disjoint $\vec p,\vec q\in\mac{S}_{N,k}$ satisfying $\ket{\vec p}, \ket{\vec q} \in \mac{H}_{k, A}$ , define the Hermitian operators on
$\mathcal{H}_{\eta}$
\begin{align}
    X_{\vec p,\vec q}
    &:=
    \sum_{\vec r}
    \purple{\varepsilon_{\wedge}(\vec p,\vec r)
        \varepsilon_{\wedge}(\vec q,\vec r)}
    \bigl(
        \ketbra{\vec p\cup\vec r}{\vec q\cup\vec r}
        +\ketbra{\vec q\cup\vec r}{\vec p\cup\vec r}\bigr) \in \mtr{End}(\mac{H}_\eta ),\\
    Y_{\vec p,\vec q}
    &:=
    -i\sum_{\vec r}
    \purple{\varepsilon_{\wedge}(\vec p,\vec r)
        \varepsilon_{\wedge}(\vec q,\vec r)
        }
    \bigl(
        \ketbra{\vec p\cup\vec r}{\vec q\cup\vec r}
        -\ketbra{\vec q\cup\vec r}{\vec p\cup\vec r}\bigr) \in \mtr{End}(\mac{H}_\eta ),
\end{align}
where the sum runs over all $\vec r\subseteq A\setminus(\vec p\cup\vec q)$
with $|\vec r|=\eta-k$ and a sign factor has the following relationship,
\begin{align}
    \purple{
        \varepsilon_{\wedge}(\vec p,\vec r)
        \varepsilon_{\wedge}(\vec q,\vec r)
    }
    &=
    \left\langle
        \vec p\cup\vec r
    \middle|
        a^\dagger_{p_1}\cdots a^\dagger_{p_k}
        a_{q_k}\cdots a_{q_1}
    \middle|
        \vec q\cup\vec r
    \right\rangle
    \in\{\pm1\}.
\end{align}

Then $X_{\vec p,\vec q}$ and $Y_{\vec p,\vec q}$ are
supported on $\mathcal{H}_{\eta,A}$ and traceless on $\mathcal{H}_{\eta}$, and their
eigenvalues belong to $\{\pm 1,0\}$. Moreover,
\begin{align}
    \operatorname{tr}_{\mathcal{H}_\eta}[X_{\vec p,\vec q}^2]
    =
    \operatorname{tr}_{\mathcal{H}_\eta}[Y_{\vec p,\vec q}^2]
    =
    2\binom{m-2k}{\eta-k}.
\end{align}
\end{lemma}
\begin{proof}
By construction, every term maps a basis state supported in $A$ to another
such state, so $X_{\vec p,\vec q}$ and $Y_{\vec p,\vec q}$ vanish on the
orthogonal complement of $\mathcal{H}_{\eta,A}$. Hence, denoting by
$\Pi_A\in\mtr{End}(\mac{H}_\eta)$ the projector onto $\mac{H}_{\eta,A}$, we have
$
    X_{\vec p,\vec q}=\Pi_{A}X_{\vec p,\vec q}\Pi_{A},
$
and likewise for $Y_{\vec p,\vec q}$.

For each spectator set $\vec r\subseteq A\setminus(\vec p\cup\vec q)$ with
$|\vec r|=\eta-k$, consider the two-dimensional subspace
\begin{align}
    \mathcal{V}_{\vec r}
    :=
    \operatorname{span}\{\ket{\vec p\cup\vec r},\ \ket{\vec q\cup\vec r}\}.
\end{align}
Since $\vec p$ and $\vec q$ are disjoint and both disjoint from $\vec r$,
distinct $\vec r$ yield mutually orthogonal subspaces, and
$X_{\vec p,\vec q}, Y_{\vec p,\vec q}$ leave each $\mathcal{V}_{\vec r}$
invariant, acting on it exactly as the Pauli matrices $\purple{\varepsilon_{\wedge}(\vec p,\vec r)\varepsilon_{\wedge}(\vec q,\vec r)}X$ and $    \purple{\varepsilon_{\wedge}(\vec p,\vec r)\varepsilon_{\wedge}(\vec q,\vec r)}Y$ in the
ordered basis $(\ket{\vec p\cup\vec r},\ \ket{\vec q\cup\vec r})$. On the
orthogonal complement of $\bigoplus_{\vec r}\mathcal{V}_{\vec r}$ both
operators vanish. Therefore their eigenvalues belong to $\{\pm 1,0\}$, and
since each two-dimensional block is traceless,
\begin{align}
    \operatorname{tr}_{\mathcal{H}_\eta}[X_{\vec p,\vec q}]
    =
    \operatorname{tr}_{\mathcal{H}_\eta}[Y_{\vec p,\vec q}]
    =0 .
\end{align}
Finally, on each block $X^{2}=Y^{2}=\1_{\mathcal{V}_{\vec r}}$, so
$X_{\vec p,\vec q}^{2}=Y_{\vec p,\vec q}^{2}$ equals the orthogonal projector
onto $\bigoplus_{\vec r}\mathcal{V}_{\vec r}$. The number of spectator sets is
$\binom{m-2k}{\eta-k}$ because this number is equivalent to choosing $\eta-k$ modes from the $m-2k$ elements of
$A\setminus(\vec p\cup\vec q)$), and each contributes a two-dimensional
block, so
\begin{align}
    \operatorname{tr}_{\mathcal{H}_\eta}[X_{\vec p,\vec q}^{2}]
    =
    \operatorname{tr}_{\mathcal{H}_\eta}[Y_{\vec p,\vec q}^{2}]
    =
    2\binom{m-2k}{\eta-k}.
\end{align}
\end{proof}

Additionally, we remark the expression of $X_{\vec p, \vec q}, Y_{\vec p, \vec q }$ via second-quantized operators. The operator
$\sum_{\vec r}\purple{\varepsilon_{\wedge}(\vec p,\vec r)\varepsilon_{\wedge}(\vec q,\vec r)}\bigl(\ketbra{\vec p\cup\vec r}{\vec q\cup\vec r}
+\ketbra{\vec q\cup\vec r}{\vec p\cup\vec r}\bigr)\in\mtr{End}(\mac{H}_\eta)$
is exactly the $A$-restriction of the second-quantized operator. Indeed,
$a^\dagger_{p_1}\cdots a^\dagger_{p_k}a_{q_k}\cdots a_{q_1}$ acts on an
occupation-basis state by first annihilating the modes $q_1,\ldots,q_k$ and then creating the modes
$p_1,\ldots,p_k$, so it sends
$\ket{\vec q\cup\vec r}\mapsto\purple{\varepsilon_{\wedge}(\vec p,\vec r)\varepsilon_{\wedge}(\vec q,\vec r)}\ket{\vec p\cup\vec r}$ with the spectator set
$\vec r$ left untouched, and annihilates every other basis state. Conjugating
by $\Pi_A$ keeps only the terms with $\vec p\cup\vec r,\vec q\cup\vec r\subseteq A$,
i.e. $\vec r\subseteq A\setminus(\vec p\cup\vec q)$ with $|\vec r|=\eta-k$;
hence,
\begin{align}
    \Pi_A\bigl(a^\dagger_{p_1}\cdots a^\dagger_{p_k}a_{q_k}\cdots a_{q_1}\bigr)\Pi_A
    =
    \sum_{\vec r} s_{\vec p,\vec q} (\vec r)\ketbra{\vec p\cup\vec r}{\vec q\cup\vec r}.
\end{align}
Therefore, the following identities hold
\begin{align}
    X_{\vec p, \vec q} &= \Pi_A \qty( a^\dagger_{p_1}\cdots a^\dagger_{p_k}a_{q_k} {\cdots a_{q_1}}+ a^\dagger_{q_1}\cdots a^\dagger_{q_k}a_{p_k} \cdots a_{p_1} ) \Pi_A, \\
    Y_{\vec p, \vec q} &= -i\Pi_A \qty( a^\dagger_{p_1}\cdots a^\dagger_{p_k}a_{q_k}{\cdots a_{q_1}} - a^\dagger_{q_1}\cdots a^\dagger_{q_k}a_{p_k} \cdots a_{p_1} ) \Pi_A, 
\end{align}

Using these operators, we construct one part of the hypothesis ensemble
$\{\rho_x\}$. \purple{
For the ensemble construction, fix $A\subseteq[N]$ with
$m:=|A|=2\eta$, which is possible since $2\eta\le N$. Under this assumption,
} 
we introduce the following sets of observables:
\begin{align}
    \label{eq:Lambda_+}
    \Lambda_+ &\coloneq \bigl\{\, X_{\vec p,\vec q},\ Y_{\vec p,\vec q}
        \in \mtr{End}(\mac{H}_\eta)
        \;\big|\; \vec p,\vec q \in \mac{S}_{A,k},\    \vec p\cap\vec q=\emptyset, 
        \
        \purple{\vec p \prec 
        \vec q} 
        \bigr\},\\
        \label{eq:Lambda_-}
    \Lambda_- &\coloneq \bigl\{\, -O \;\big|\; O\in\Lambda_+ \,\bigr\},
    \qquad
    \Lambda \coloneq \Lambda_+\cup\Lambda_- .
\end{align}
Here
$\mac{S}_{A,k}:=\{\vec r\subseteq A\mid|\vec r|=k\} \subset \mac{S}_{N,k} $ denotes the set of
size-$k$ subsets of $A$, which labels the occupation basis of the
$k$-particle sector $\mac{H}_{k,A}$, \purple{and $\prec$ denotes the lexicographic order on $\mac{S}_{A,k}$:
for distinct tuples $\vec p$ and $\vec q$, we write
$\vec p \prec \vec q$ if $p_j<q_j$ at the first index $j$
for which $p_j\neq q_j$}
. 
The number of disjoint unordered pairs $\{\vec p,\vec q\}$ is
$\tfrac12\binom{m}{k}\binom{m-k}{k}$, and each pair contributes the two
observables $X_{\vec p,\vec q}$ and $Y_{\vec p,\vec q}$, so the size of $\Lambda$ is expressed by
$    |\Lambda| = 2\binom{m}{k}\binom{m-k}{k}
$
~\footnote{Equivalently, one may work with ordered pairs $(\vec p,\vec q)$; the resulting identifications $X_{\vec p,\vec q}=X_{\vec q,\vec p}$ and $Y_{\vec p,\vec q}=-Y_{\vec q,\vec p}$ change the count only by a constant factor and do not affect the asymptotic scaling.}.
By Lemma~\ref{lem:target_op_lower_bound}, every $O\in\Lambda$ is traceless
with eigenvalues in $\{0,\pm1\}$, so that $\norm{O}_{\mtr{op}}=1$. Based on
this set, we define the ensemble $\{\rho_x\}:=\{\rho_O\}_{O\in\Lambda}$
together with the reference state $\tau$ as follows:
\begin{align}
    \tau &\coloneq \frac{\Pi_A}{d_{\eta,A}},
    \quad\text{where } d_{\eta,A}\coloneq \dim\mac{H}_{\eta,A}=\binom{m}{\eta},\\
    \label{eq:expression_of_rhoO}
    \rho_O &\coloneq \tau + \frac{\varepsilon'}{d_{\eta,A}}\,O
    \quad(\forall O\in\Lambda),
    \quad\text{where } \varepsilon'\coloneq \purple{4} 
        \varepsilon\,\frac{d_{\eta,A}}{\binom{m-2k}{\eta-k}}\in(0,1/2].
\end{align}
Here $\Pi_A \in \mtr{End} ( \mac{H}_\eta)$ denotes the projection operator on $\mac{H}_{\eta, A}$. From the definitions of $\Pi_A$ and $O\in\Lambda$, we can directly show that
$\Pi_A\tau\Pi_A=\tau$ and $\Pi_A\rho_O\Pi_A=\rho_O$ for every $O\in\Lambda$;
that is, $\tau$ and $\rho_O$ are $\eta$-particle states supported on
$\mac{H}_{\eta,A}$. 
\purple{Additionaly, we assume $0<\varepsilon\le(8\cdot4^k)^{-1}$.
Since $m=2\eta$ and
\begin{align}
        \frac{\binom{m-2k}{\eta-k}}{d_{\eta,A}}
    =
    \prod_{j=0}^{k-1}
    \frac{\eta-j}{2(2\eta-2j-1)}
    \ge 4^{-k},
\end{align}
we have $0<\varepsilon'\le1/2$.}
We will prove Theorem~\ref{thm:lower_bound} using these states. 
\\

\noindent
\textit{ Proof of Theorem~\ref{thm:lower_bound}.}
We first outline the proof. The argument proceeds in four steps.
First, we show that a lower bound for element-wise $k$-RDM tomography
reduces to a lower bound for the two-hypothesis distinguishing task between the reference state $\tau$ and the ensemble $\{\rho_O\}_{O}$. Second, for a learning algorithm without quantum memory described by a rooted tree, we upper bound the probability of solving this distinguishing task. Following Ref.~\cite{chen2022exponential}, this probability is controlled by the total variation distance between the leaf distributions of $\tau$ and $\{\rho_O\}_O$, which we in turn bound in terms of the number of copies $T$.
Third, we show that the resulting bound is governed by a single quantity $\Delta$, which is determined by the Hilbert--Schmidt norm of the $k$-RDM of an arbitrary $\eta$-particle state.
Finally, we bound this Hilbert--Schmidt norm using Refs.~\cite{christiansen2024hilbert,visconti2026hilbert}. Combining these ingredients and simplifying the resulting binomial coefficients, we obtain the lower bound $T=\Omega_k(\eta^k)/\varepsilon^2$.

First, we see that the element-wise $k$-RDM tomography can be reduced to the two-hypothesis
distinguishing task. Recall from the construction that the target operators are the Hermitian
combinations
\begin{align}
    X_{\vec p,\vec q}
    &=
    \Pi_A\bigl(
        a^\dagger_{p_1}\cdots a^\dagger_{p_k}a_{q_k}\cdots a_{q_1}
        +a^\dagger_{q_1}\cdots a^\dagger_{q_k}a_{p_k}\cdots a_{p_1}
    \bigr)\Pi_A,\\
    Y_{\vec p,\vec q}
    &=
    -i\,\Pi_A\bigl(
        a^\dagger_{p_1}\cdots a^\dagger_{p_k}a_{q_k}\cdots a_{q_1}
        -a^\dagger_{q_1}\cdots a^\dagger_{q_k}a_{p_k}\cdots a_{p_1}
    \bigr)\Pi_A.
\end{align}
Since $\sigma\in\{\tau,\rho_O\}$ is supported on $\mac{H}_{\eta,A}$, we have
$\Pi_A\sigma\Pi_A=\sigma$, so by linearity of the trace,
\begin{align}
    \tr_{\mac{H}_\eta }[\sigma X_{\vec p,\vec q}]
    &=
    \tr_{\mac{H}_\eta }\!\bigl[\sigma\,
        a^\dagger_{p_1}\cdots a^\dagger_{p_k}a_{q_k}\cdots a_{q_1}\bigr]
    +\tr_{\mac{H}_\eta }\!\bigl[\sigma\,
        a^\dagger_{q_1}\cdots a^\dagger_{q_k}a_{p_k}\cdots a_{p_1}\bigr],
        \label{eq:X_pq_decomposition}
        \\
    &= \purple{{}^{k}D_{\sigma;\,\vec q}^{\vec p}} + \purple{{}^{k}D_{\sigma;\,\vec p}^{\vec q}} \\
    \tr_{\mac{H}_\eta }[\sigma Y_{\vec p,\vec q}]
    &=
    \purple{i ({}^{k}D_{\sigma;\,\vec p}^{\vec q} - {}^{k}D_{\sigma;\,\vec q}^{\vec p} )}
    \label{eq:Y_pq_decomposition}.
\end{align}
From these identities, element-wise
$\varepsilon$-error estimates of all $k$-RDM elements yield element-wise
$2\varepsilon$-error estimates of $\tr_{\mac{H}_\eta }[\sigma O]$ for every
$O\in\Lambda$. If the target state is $\tau$,
\begin{align}
    \tr_{\mac{H}_\eta }[\tau O]
    =\frac{1}{d_{\eta,A}}\tr_{\mac{H}_\eta }[O]=0
    \qquad\forall O\in\Lambda.
\end{align}
On the other hand, if the target state is $\rho_O$,
\begin{align}
    \tr_{\mac{H}_\eta }[\rho_O O]
    =\tr_{\mac{H}_\eta }[\tau O]
    +\frac{\varepsilon'}{d_{\eta,A}}\tr_{\mac{H}_\eta }[O^2]
    =\frac{\varepsilon'}{d_{\eta,A}}\cdot
       2\binom{m-2k}{\eta-k}
    = \purple{8} \varepsilon,
\end{align}
where $\tr_{\mac{H}_\eta }[O^2]=2\binom{m-2k}{\eta-k}$ for every $O\in\Lambda$
follows from Lemma~\ref{lem:target_op_lower_bound}. 
\purple{
Let $\widehat{o}_O$ denote the estimate of
$\operatorname{tr}(\sigma O)$ obtained from the estimated
$k$-RDM elements.
Whenever all $k$-RDM estimates have additive error at most
$\varepsilon$, we have
$|\widehat{o}_O-\operatorname{tr}(\sigma O)|\le2\varepsilon$
for every $O\in\Lambda$ because of Eqs.~\eqref{eq:X_pq_decomposition} and \eqref{eq:Y_pq_decomposition}. Therefore
\begin{align}
    \sigma=\tau
    &\quad\Longrightarrow\quad
    \max_{O\in\Lambda}|\widehat{o}_O|
    \le2\varepsilon,\\
    \sigma=\rho_{O_*}
    &\quad\Longrightarrow\quad
    \max_{O\in\Lambda}|\widehat{o}_O|
    \ge|\widehat{o}_{O_*}|
    \ge8\varepsilon-2\varepsilon=6\varepsilon.
\end{align}
Thus, declaring the alternative hypothesis when
$\max_{O\in\Lambda}|\widehat{o}_O|>4\varepsilon$
and the reference hypothesis otherwise correctly distinguishes
the two hypotheses whenever the simultaneous RDM estimation
succeeds.
Consequently, a tomography algorithm with simultaneous success
probability at least $2/3$ yields a distinguishing algorithm
with success probability at least $2/3$, using the same
number of copies.
}

Second, we upper bound the success probability of any no-quantum-memory
learning algorithm that uses $T$ copies of the state. The total variation distance of the leaf distributions is evaluated as follows. 
\begin{align}
    \frac12\sum_{l\in\mtr{leaf}}
        \bigl|\mathbb{E}_O[\mathbb{P}(l\given\rho_O)]
            -\mathbb{P}(l\given\tau)\bigr|
    &=
    \sum_{\substack{l\in\mtr{leaf}:\\
        \mathbb{P}(l\given\tau)\ge\mathbb{E}_O[\mathbb{P}(l\given\rho_O)]}}
        \bigl(\mathbb{P}(l\given\tau)
            -\mathbb{E}_O[\mathbb{P}(l\given\rho_O)]\bigr)\\
    &=
    \sum_{\substack{l\in\mtr{leaf}:\\
        \mathbb{P}(l\given\tau)\ge\mathbb{E}_O[\mathbb{P}(l\given\rho_O)]}}
        \mathbb{P}(l\given\tau)
        \left(1-\frac{\mathbb{E}_O[\mathbb{P}(l\given\rho_O)]}
            {\mathbb{P}(l\given\tau)}\right).
\end{align}
where the first equality is the standard total variation identity
$\frac12\sum_{l}|\mathbb{P}_\alpha (l)-\mathbb{P}_\beta (l)|=\sum_{l:\,\mathbb{P}_\beta (l)\ge \mathbb{P}_\alpha (l)}(\mathbb{P}_\beta (l)-\mathbb{P}_\alpha (l))$, valid for any
two probability distributions $\mathbb{P}_\alpha , \mathbb{P}_\beta$, applied with
$\mathbb{P}_\alpha (l)=\mathbb{E}_O[\mathbb{P}(l\given\rho_O)]$ and
$\mathbb{P}_\beta (l)=\mathbb{P}(l\given\tau)$.

Below, we evaluate the ratio of the probabilities at the leaf $l$. By following Eq.~\eqref{def:probability_tree_prob} and  fixing an
arbitrary path from the root $r$ to the leaf $l$, we can write
\begin{align}
    \frac{\mathbb{E}_O[\mathbb{P}(l\given\rho_O)]}{\mathbb{P}(l\given\tau)}
    &=
    \mathbb{E}_O
    \frac{\prod_{t=1}^{T}\bigl(w_{s_t}^{u_t}
        \langle\psi_{s_t}^{u_t}|\rho_O|\psi_{s_t}^{u_t}\rangle\bigr)}
         {\prod_{t=1}^{T}\bigl(w_{s_t}^{u_t}
        \langle\psi_{s_t}^{u_t}|\tau|\psi_{s_t}^{u_t}\rangle\bigr)}
        \\
    &=
    \mathbb{E}_O
    \frac{\prod_{t=1}^{T}\bigl(
        \langle\psi_{s_t}^{u_t}|\tau|\psi_{s_t}^{u_t}\rangle
        +\varepsilon' d_{\eta,A}^{-1}
        \langle\psi_{s_t}^{u_t}|O|\psi_{s_t}^{u_t}\rangle\bigr)}
         {\prod_{t=1}^{T}\bigl(
        \langle\psi_{s_t}^{u_t}|\tau|\psi_{s_t}^{u_t}\rangle\bigr)}
            && (\because \text{Eq.}~\eqref{eq:expression_of_rhoO} ) \\
    &=
    \mathbb{E}_O\prod_{t=1}^{T}
    \left(1+\varepsilon'
        \frac{\langle\psi_{s_t}^{u_t}|\Pi_A O\Pi_A|\psi_{s_t}^{u_t}\rangle}
             {\langle\psi_{s_t}^{u_t}|\Pi_A|\psi_{s_t}^{u_t}\rangle}\right)\\
    &=
    \mathbb{E}_O\exp\left[\sum_{t=1}^{T}
        \ln\left(1+\varepsilon'
        \frac{\langle\psi_{s_t}^{u_t}|\Pi_A O\Pi_A|\psi_{s_t}^{u_t}\rangle}
             {\langle\psi_{s_t}^{u_t}|\Pi_A|\psi_{s_t}^{u_t}\rangle}\right)\right]\\
    &\ge
    \exp\left[\sum_{t=1}^{T}\mathbb{E}_O
        \ln\left(1+\varepsilon'
        \frac{\langle\psi_{s_t}^{u_t}|\Pi_A O\Pi_A|\psi_{s_t}^{u_t}\rangle}
             {\langle\psi_{s_t}^{u_t}|\Pi_A|\psi_{s_t}^{u_t}\rangle}\right)\right]
             && (\because \text{Jensen's inequality})
             \\
    &=
    \exp\left[\sum_{t=1}^{T}\frac{1}{|\Lambda|}\sum_{O\in\Lambda}
        \ln\left(1+\varepsilon'
        \frac{\langle\psi_{s_t}^{u_t}|\Pi_A O\Pi_A|\psi_{s_t}^{u_t}\rangle}
             {\langle\psi_{s_t}^{u_t}|\Pi_A|\psi_{s_t}^{u_t}\rangle}\right)\right]\\
    &=
    \exp\Biggl[\sum_{t=1}^{T}\frac{1}{2|\Lambda_+|}\sum_{O\in\Lambda_+}
        \biggl\{
        \ln\!\left(1+\varepsilon'
        \frac{\langle\psi_{s_t}^{u_t}|\Pi_A O\Pi_A|\psi_{s_t}^{u_t}\rangle}
             {\langle\psi_{s_t}^{u_t}|\Pi_A|\psi_{s_t}^{u_t}\rangle}\right)
    \notag\\
    &\qquad\qquad\qquad
        +\ln\!\left(1-\varepsilon'
        \frac{\langle\psi_{s_t}^{u_t}|\Pi_A O\Pi_A|\psi_{s_t}^{u_t}\rangle}
             {\langle\psi_{s_t}^{u_t}|\Pi_A|\psi_{s_t}^{u_t}\rangle}\right)
        \biggr\}\Biggr]
        && (\because \text{Eq}.~\eqref{eq:Lambda_-})
        \\
    &=
    \exp\left[\sum_{t=1}^{T}\frac{1}{2|\Lambda_+|}\sum_{O\in\Lambda_+}
        \ln\left(1-(\varepsilon')^{2}
        \left(\frac{\langle\psi_{s_t}^{u_t}|\Pi_A O\Pi_A|\psi_{s_t}^{u_t}\rangle}
             {\langle\psi_{s_t}^{u_t}|\Pi_A|\psi_{s_t}^{u_t}\rangle}\right)^{2}
        \right)\right]
        && (\because \ \ln(1+a) + \ln(1-a) = \ln(1-a^2) )
        \\
    &\ge
    \label{eq:long_inquality}
    \exp\left[-\sum_{t=1}^{T}\frac{(\varepsilon')^{2}}{|\Lambda_+|}
        \sum_{O\in\Lambda_+}
        \left(\frac{\langle\psi_{s_t}^{u_t}|\Pi_A O\Pi_A|\psi_{s_t}^{u_t}\rangle}
             {\langle\psi_{s_t}^{u_t}|\Pi_A|\psi_{s_t}^{u_t}\rangle}\right)^{2}
        \right]
        && (\because\ \ln(1-u)\ge-2u\ \text{for } u\in(0,1/2])
        .
\end{align}
Here, the second inequality follows from Jensen's inequality
$\mathbb{E}[e^{Z}]\ge e^{\mathbb{E}[Z]}$ for the convex function $\exp(\cdot)$ and a random variable $Z$. Regarding the last inequality, we remark that
\begin{align}
    (\varepsilon')^{2}
        \left(\frac{\langle\psi_{s_t}^{u_t}|\Pi_A O\Pi_A|\psi_{s_t}^{u_t}\rangle}
             {\langle\psi_{s_t}^{u_t}|\Pi_A|\psi_{s_t}^{u_t}\rangle}\right)^{2} \leq 1/2
\end{align}
because we set $\varepsilon' \in (0,1/2]$ in Eq.~\eqref{eq:expression_of_rhoO} and $\norm{O}_{\mtr{op}} = 1 $ leads to $\abs{ \bra{\psi} \Pi_A O \Pi_A \ket{\psi} } \leq \bra{\psi} \Pi_A^2 \ket{\psi}  =  \bra{\psi}  \Pi_A \ket{\psi}$.

\purple{
Third, we bound the right-hand side of
Eq.~\eqref{eq:long_inquality} using the quantity
\begin{align}
    \label{eq:expression_Delta}
    \Delta
    \coloneq
    \sup_{\substack{
        \psi\in\mtr{End}(\mac{H}_{\eta, A})\\
        \psi\ge0,\ \tr[\psi]=1
    }}
    \frac{1}{|\Lambda_+|}
    \sum_{O\in\Lambda_+}
    \bigl(\tr[\psi O]\bigr)^2.
\end{align}
}
Along any path with $\mathbb{P}(l\mid\tau)>0$, the denominator below
is positive at every step $t$, so we may define
\begin{align}
    \psi_{t,A}
    \coloneq
    \frac{
        \Pi_A
        \ket{\psi^{u_t}_{s_t}}\bra{\psi^{u_t}_{s_t}}
        \Pi_A
    }{
        \bra{\psi^{u_t}_{s_t}}
        \Pi_A
        \ket{\psi^{u_t}_{s_t}}
    }.
\end{align}
This is a density operator supported on
$\mac{H}_{\eta,A}$.
Consequently, the definition of $\Delta$ gives
\begin{align}
    &\frac{1}{|\Lambda_+|}
    \sum_{O\in\Lambda_+}
    \left(
        \frac{
            \bra{\psi^{u_t}_{s_t}}
            \Pi_A O\Pi_A
            \ket{\psi^{u_t}_{s_t}}
        }{
            \bra{\psi^{u_t}_{s_t}}
            \Pi_A
            \ket{\psi^{u_t}_{s_t}}
        }
    \right)^2=
    \frac{1}{|\Lambda_+|}
    \sum_{O\in\Lambda_+}
    \bigl(\tr[\psi_{t,A}O]\bigr)^2
    \le\Delta.
\end{align}

Substituting this into the exponent and using $e^{-x}\ge 1-x$, we obtain
\begin{align}
    \exp\left[-\sum_{t=1}^{T}\frac{(\varepsilon')^{2}}{|\Lambda_+|}
        \sum_{O\in\Lambda_+}
        \left(\frac{\langle\psi_{s_t}^{u_t}|\Pi_A O\Pi_A|\psi_{s_t}^{u_t}\rangle}
             {\langle\psi_{s_t}^{u_t}|\Pi_A|\psi_{s_t}^{u_t}\rangle}\right)^{2}\right]
    \ge
    \exp\bigl(-T(\varepsilon')^{2}\Delta\bigr)
    \ge
    1-T(\varepsilon')^{2}\Delta .
\end{align}
We remark that the restriction to the fixed-$\eta$ sector in the supremum of Eq.~\eqref{eq:expression_Delta} is essential for the analysis that follows: it lets us bound $\Delta$ by the Hilbert--Schmidt norm of the $k$-RDM of an
$\eta$-particle state, to which the fixed-$\eta$ estimates of
Refs.~\cite{christiansen2024hilbert,visconti2026hilbert} apply.

From these discussions, we can derive
\begin{align}
     \frac12 \sum_{l \in \mtr{leaf }} \abs{ \mathbb{E}_O [\mathbb{P} (l \given \rho_O)  ] - \mathbb{P} (l \given \tau ) }  
     &\leq \sum_{ \substack{l \in \mtr{leaf}, \\ \mathbb{P} (l \mid \tau ) \geq \mathbb{E}_O [ \mathbb{P} ( l \mid O)  ] } } \mathbb{P}(l \mid \tau )  \qty[ 1 - ( 1 - T (\varepsilon')^2 \Delta ) ] \\
     &\le T (\varepsilon')^2 \Delta.
\end{align}
From Lemma~\ref{lem:The_condition_distinguishing}, we conclude that any learning algorithm without quantum memory that solves the distinguish task with success probability at least $2/3$ requires
\begin{align}
    \label{eq:lower_bound_of_T}
    T \geq \frac1{3 (\varepsilon')^2 \Delta} = \purple{\frac1{48}} \qty( \frac{ \binom{m-2k}{\eta- k}}{ \binom{m}{\eta}} )^2 \frac1{\varepsilon^2 \Delta}. 
\end{align}
Here we use $\varepsilon' = \purple{4} \varepsilon d_{\eta , A} / \binom{m-2k}{\eta - k}$ from Eq.~\eqref{eq:expression_of_rhoO}. 

Next, we reduce the bound on the remaining factor $\Delta$ to evaluating
the Hilbert--Schmidt norm of the $k$-RDM. \purple{
For any density operator
$\psi\in\mtr{End}(\mac{H}_\eta)$ satisfying
$\Pi_A\psi\Pi_A=\psi$,
}
\begin{align}
    \sum_{O\in\Lambda_+}\bigl(\tr[\psi O]\bigr)^2
    &= \sum_{\substack{\vec p,\vec q\in\mac{S}_{A,k}\\ \vec p\cap\vec q=\emptyset, \purple{\vec p \prec \vec q} }}
        \Bigl[\bigl(\tr[\psi X_{\vec p,\vec q}]\bigr)^2
        +\bigl(\tr[\psi Y_{\vec p,\vec q}]\bigr)^2\Bigr]\\
    &= \sum_{\substack{\vec p,\vec q\in\mac{S}_{A,k}\\ \vec p\cap\vec q=\emptyset, \purple{\vec p \prec \vec q} }}
        \Bigl[\bigl(2\,\mtr{Re}\,\purple{{}^{k}D_{\psi;\,\vec q}^{\vec p}}\bigr)^2
        +\bigl(2\,\mtr{Im}\,\purple{{}^{k}D_{\psi;\,\vec q}^{\vec p}}\bigr)^2\Bigr]\\
    &= \sum_{\substack{\vec p,\vec q\in\mac{S}_{A,k}\\ \vec p\cap\vec q=\emptyset, \purple{\vec p \prec \vec q} }}
        4\,\bigl|\purple{{}^{k}D_{\psi;\,\vec q}^{\vec p}}\bigr|^2\\
    &\le \sum_{\vec p,\vec q\in\mac{S}_{A,k}}
        4\,\bigl|\purple{{}^{k}D_{\psi;\,\vec q}^{\vec p}}\bigr|^2\\
    &\le 4\sum_{\vec p,\vec q\in\mac{S}_{N,k}}
        \bigl|\purple{{}^{k}D_{\psi;\,\vec q}^{\vec p}}\bigr|^2
    = 4\,\bigl\|D^{(k)}_{\psi}\bigr\|_{\mtr{HS}}^2 .
\end{align}
Here $\|\cdot\|_{\mtr{HS}}$ denotes the Hilbert--Schmidt norm.

We evaluate $\|D^{(k)}_{\psi}\|_{\mtr{HS}}$ using the results of
Refs.~\cite{christiansen2024hilbert,visconti2026hilbert}. First, for $k=1$,
since $\|D^{(1)}_{\psi}\|_{\mtr{op}}\leq1$ and
$\|{}^1 D_{\psi}\|_{\mtr{tr}}=\eta$, the Hilbert--Schmidt norm
satisfies~\cite{christiansen2024hilbert}
\begin{align}
    \|{}^1 D_{\psi}\|_{\mtr{HS}}^2
    \le \| {}^1 D_{\psi}\|_{\mtr{op}}\,\|{}^1 D_{\psi}\|_{\mtr{tr}}
    = \eta .
\end{align}
For $k=2$, Theorem~1 of Ref.~\cite{christiansen2024hilbert} gives
\begin{align}
    \|{}^2 D_{\psi}\|_{\mtr{HS}}^2 \le 5\eta^2 .
\end{align}
More generally, Theorem~1 of Ref.~\cite{visconti2026hilbert} provides the bound
\begin{align}
    \|{}^k D_{\psi}\|_{\mtr{HS}}^2 \le C_k\,\eta^{k},
\end{align}
where $C_k$ is a constant depending only on $k$. Combining this with the
convexity of the Hilbert--Schmidt norm, Eq.~\eqref{eq:expression_Delta} is
bounded as
\begin{align}
    \Delta
    \le
    \frac{4}{|\Lambda_+|}\, \sup_\psi \|{}^k D_{\psi}\|_{\mtr{HS}}^2
    \le
    4C_k\binom{m}{k}^{-1}\binom{m-k}{k}^{-1}\eta^{k},
\end{align}
where we have used $|\Lambda_+|=\binom{m}{k}\binom{m-k}{k}$, which follows
from the definition in Eq.~\eqref{eq:Lambda_+}.

Combining the above, the lower bound on $T$ in
Eq.~\eqref{eq:lower_bound_of_T} can be expressed as
\begin{align}
    T \ge \frac{1}{\purple{192} C_k}
    \binom{m-2k}{\eta-k}^2
    \binom{m}{\eta}^{-2}
    \binom{m}{k}\binom{m-k}{k}
    \frac{\eta^{-k}}{\varepsilon^2}.
\end{align}
To clarify its scaling, we use the two identities
\begin{align}
    \frac{\binom{m-2k}{\eta-k}}{\binom{m}{\eta}}
    =\frac{\binom{\eta}{k}\binom{m-\eta}{k}}{\binom{m}{2k}\binom{2k}{k}},
    \qquad
    \binom{m}{k}\binom{m-k}{k}=\binom{m}{2k}\binom{2k}{k},
\end{align}
which rewrite $T$ in the convenient form
\begin{align}
    T
    \geq 
    \frac{1}{\purple{192} C_k}\,
    \frac{\binom{\eta}{k}^2\binom{m-\eta}{k}^2}
         {\binom{m}{2k}\binom{2k}{k}}\,
    \frac{\eta^{-k}}{\varepsilon^2}.
\end{align}
Here, the binomial coefficients obey
\begin{align}
    \binom{m}{2k}\binom{2k}{k}\le\frac{m^{2k}}{(k!)^2},
\end{align}
Additionally, since $\eta \ge k$, we obtain 
\begin{align}
    \binom{\eta}{k} &= \frac{\eta^k}{k!} \prod^{k-1}_{j=0} \qty(1 - \frac{j}{\eta}) \\
    &\geq \frac{\eta^k}{k!} \prod^{k-1}_{j=0} \qty(1 - \frac{j}{k}) \\
    &= \frac{\eta^k}{k!} \frac{k!}{k^k} = \frac{\eta^k}{k^k}.
    \qquad
\end{align}
Similarly, if we assume that $m - \eta \geq k$, we obtain
\begin{align}
        \binom{m-\eta}{k}\ge\frac{((m-\eta)^k)}{k^k},
    \qquad
\end{align}
Substituting these estimates yields
\begin{align}
    T
    \ge
    \frac{(k!)^2}{\purple{192} C_k\,(k)^{4k}}
    \left(\frac{m-\eta}{m}\right)^{2k}
    \frac{\eta^k}{\varepsilon^2}.
\end{align}
\purple{Since $m=2\eta$, we have $(m-\eta)/m=1/2$.
Substituting this into the preceding bound yields
}
\begin{align}
    T
    \ge 
    \frac{(k!)^2}{\purple{192} C_k\,(2k^2)^{2k}}
    \frac{\eta^k}{\varepsilon^2}  =  \frac{ \Omega_k(\eta^k )}{\varepsilon^2}.
\end{align}

Finally, we translate this bound on $T$ into the sample complexity. We have shown $T=\Omega_k(\eta^k)/\varepsilon^2$ for any no-quantum-memory learning algorithm that solves the two-hypothesis distinguishing task with success probability at least $2/3$. As established above, the element-wise estimation of all $k$-RDM elements to additive error $\varepsilon$ solves this distinguishing task; hence the same lower bound applies to the $k$-RDM estimation problem. Moreover, by the tree representation of Definition~\ref{def:tree_representation}, the depth $T$ is exactly the number of measurement steps, and each step consumes a single copy of $\rho$. Therefore the number of copies required to estimate all $k$-RDM elements within element-wise additive error $\varepsilon$ is also
$\Omega_k(\eta^k)/\varepsilon^2$, which completes the proof of
Theorem~\ref{thm:lower_bound}. \qed

\section{Practical application}
\begin{revision}
This section examines the practical performance of orbital-rotation
shadows. We first compare the query counts required to learn all
entries of the 1-RDM and then evaluate the sample costs of estimating
observables in condensed-matter physics and quantum chemistry.    
\end{revision}

\subsection{1-RDM learning cost}
\label{sec:Sample_cost}

\begin{figure*}[t]
\hspace{-0.75cm}
\includegraphics[width=0.60\linewidth]{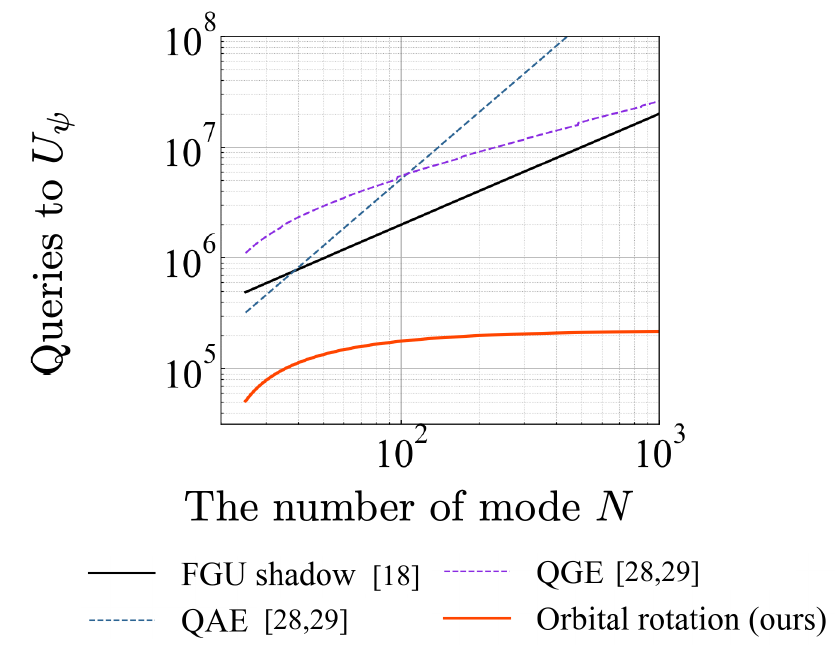}

\caption{\label{Comaprision}
Cost for simultaneously learning all entries of the 1-RDM at fixed particle
number \(\eta=20\) and target accuracy \(\varepsilon=10^{-2}\), in terms of the number of calls to the state preparation unitary $U_{\psi}|0\rangle = |\psi\rangle$ for a target state $|\psi\rangle$.  We compare the
orbital-rotation shadow with the FGU shadow and Heisenberg-limited
estimation strategies like quantum amplitude estimation (QAE) and Quantum gradient estimation (QGE), whose 
cost are evaluated following
Refs.~\cite{Koizumi2026Faster, Koizumi2026Heiseberg}.
}
\label{fig:resource_estimate}
\end{figure*}

\begin{revision}

In this section, we evaluate the query complexity of estimating all entries of the 1-RDM at fixed particle number.  As a representative benchmark, we fix \(\eta=20\) and vary the number of modes \(N\), requiring a mean-squared error of at most \(\varepsilon^2\) for each entry, with \(\varepsilon=10^{-2}\). This setting allows us to examine how particle-number symmetry affects the system-size dependence of the estimation cost. We compare orbital-rotation shadows with FGU shadows and idealized Heisenberg-limited estimation strategies, quantifying their costs in terms of queries to the state-preparation oracles.
We remark that the Heisenberg-limited methods considered here assume coherent access to the state-preparation unitary and its inverse. We grant these methods this stronger access for benchmarking, demonstrating that orbital-rotation shadows still require fewer queries in the parameter regime considered.
\end{revision}

Following Corollary~\ref{cor:entrywise_mse_sample_complexity}, we can derive the upper bound of sufficient copies of the target state to learn all entries.
For the orbital-rotation protocol, each shot produces the full matrix
estimator, so the same single-shot samples can be used to estimate all
\(N^2\) entries.  From the explicit variance formula expressed in Eq.~\eqref{eq:explicit_variance_1rdm_app}, we obtain a
state-independent bound by discarding the negative terms and using
\(\purple{{}^{1}D_{\mkern6mup}^{p}},\purple{{}^{1}D_{\mkern6muq}^{q}}\le 1\):
\begin{align}
    \operatorname{Var}
    \!\left(
        \purple{{}^{1}\widehat D_{\mkern6muq}^{p}}
    \right)
    \le
    \frac{(N+1)(N+1-\eta)}{N(N+2)}(\eta+2).
\end{align}
Note that this off-diagonal bound also dominates the diagonal case.

Figure~\ref{fig:resource_estimate} compares the number of queries to the state preparation unitary $U_\psi$ at \(\varepsilon=10^{-2}\) with the state-of-the-art estimation methods. 
Remark that our comparison also includes idealized
Heisenberg-limited strategies as benchmarks.  These methods assume coherent
access to both a state-preparation circuit and its inverse, so their cost
should be interpreted as coherent queries to the state-preparation oracles,
rather than as single-copy randomized measurement shots~\cite{huggins2022nearly,
Koizumi2026Heiseberg,Koizumi2026Faster}.
The orbital-rotation protocol exhibits essentially $N$-independent scaling, with the dominant dependence appearing only through $\eta$.  By contrast, the FGU-shadow protocol and Heisenberg-limited strategies have a stronger
mode-number dependence, leading to a substantially larger query estimate in the regime shown. 
As a result, our numerical analysis shows that, for estimating one-body fermionic correlations in a system of 100 sites and 20 particles, the proposed method reduces the number of queries to the target state-preparation oracle by a factor of approximately 10 compared with the best existing methods.
Therefore, the orbital-rotation protocol requires fewer costs in the displayed
regime, highlighting the practical advantage of exploiting particle-number
symmetry at the measurement level.

\begin{revision}

\subsection{{Application to condensed-matter physics}}
\label{sec:application_condensed_matter}

As a representative application, we consider the estimation of
physically relevant observables in condensed matter physics.
Specifically, we evaluate the measurement costs of estimating
the {energy density} of the square-lattice Fermi--Hubbard model
and the spin structure factor at $\mathbf{q}=(\pi,\pi)$.
These costs depend not only on the variances of estimators for
individual RDM entries but also on the coefficients with which
these entries contribute to each observable and on the mutual
covariances of their estimators.
Evaluating the full estimator variance for each observable is
therefore important for assessing the practical performance of
classical shadow protocols.

The spinful Fermi--Hubbard model on an $L\times L$ square lattice
has {the} following Hamiltonian:
\begin{equation}
    H_{\mathrm{Hub}}
    =
    -t
    \sum_{\langle p,q\rangle,\sigma}
    \left(
        a_{p\sigma}^{\dagger}a_{q\sigma}
        +
        a_{q\sigma}^{\dagger}a_{p\sigma}
    \right)
    +
    U\sum_p n_{p\uparrow}n_{p\downarrow},
    \label{eq:square-hubbard-hamiltonian}
\end{equation}
where $\langle p,q\rangle$ denotes nearest-neighbor sites,
$\sigma\in\{\uparrow,\downarrow\}$, and
$n_{p\sigma}=a_{p\sigma}^{\dagger}a_{p\sigma}$.
Here, $t$ and $U$ are the hopping and repulsive onsite potential,
respectively, and the model contains $N=2L^2$ fermionic modes.
{
We define the energy density as
$\langle H_{\mathrm{Hub}}\rangle/L^2$.
}
Despite its local form, the model exhibits rich strongly correlated
phenomena and remains difficult to analyze in the doped,
low-temperature regime
\cite{arovas_hubbard_2022}.
Its particle-number conservation, sparse one- and two-body structure,
and scalable system size make it a natural benchmark for
{examining the system-size dependence of the variance of
energy-density estimation}.

In addition to the {energy density}, we consider the static
spin structure factor
\cite{varney2009quantum, mazurenko2016experimental},
which characterizes spatial spin correlations.
The corresponding observable is
\begin{equation}
    S(\mathbf{q})
    =
    \frac{1}{L^2}
    \sum_{p,r}
    e^{i\mathbf{q}\cdot(\mathbf{r}_p-\mathbf{r}_r)}
    \mathbf{S}_p\cdot\mathbf{S}_r,
    \label{eq:spin-structure-factor}
\end{equation}
where $\mathbf{r}_p=(x_p,y_p)$ denotes the position of site $p$,
with the lattice spacing set to unity, and
$\mathbf{S}_p=(S_p^x,S_p^y,S_p^z)$ is the local spin operator.
Setting $\hbar=1$, its components are
\begin{equation}
    S_p^\alpha
    =
    \frac{1}{2}
    \sum_{\sigma,\tau\in\{\uparrow,\downarrow\}}
    a_{p\sigma}^{\dagger}
    (\sigma^\alpha)_{\sigma\tau}
    a_{p\tau},
    \qquad
    \alpha\in\{x,y,z\},
\end{equation}
where $\sigma^\alpha$ denotes the corresponding Pauli matrix.
The sum in Eq.~\eqref{eq:spin-structure-factor} includes the
same-site contributions $p=r$, and the normalization is by
the number of spatial sites, $L^2$.
Here, we focus on $\mathbf{q}=(\pi,\pi)$, for which the phase
factor alternates between the two sublattices.
The expectation value $\langle S(\pi,\pi)\rangle$ therefore
probes staggered antiferromagnetic spin correlations.

We compare three measurement protocols: orbital-rotation shadows,
FGU shadows, and uniform local Pauli shadows.
For local Pauli shadows, each qubit in the Jordan--Wigner encoding
is independently measured in the $X$, $Y$, or $Z$ basis
with probability $1/3$.
We quantify their performance through the single-shot variances
of the complete estimators for {$H_{\mathrm{Hub}}/L^2$} and
$S(\pi,\pi)$, including covariances between their constituent terms.
Here, our comparison focuses on measurement ensembles specified
independently of the target observables.
This allows a common set of snapshots to be reused for estimating
multiple observables by changing only the classical post-processing.
We therefore do not include observable-specific measurement
grouping~\cite{verteletskyi2020measurement}
or an optimized measurement strategy~\cite{hadfield2022measurements}.

The single-shot estimator variance generally depends on the target
state.
Here, we consider fixed-particle-number Slater determinants,
whose correlation functions can be evaluated efficiently using
Wick's theorem~\cite{wick1950evaluation,bravyi2005lagrangian}.
This tractability allows us to assess the measurement cost over
a broad range of system sizes.
To examine the relevance of this benchmark to interacting states,
we compare the estimator variances with those obtained for
numerically computed ground states of small Hubbard systems
in the same particle-number sector.
All single-shot variances reported in
Table~\ref{tab:hubbard-ground-slater} are evaluated directly from
the analytical evaluation of the covariances expressed in Sec.~\ref{sec:verify}, without finite-shot sampling.
Although the energy expectation values differ substantially,
the relative discrepancies between the Slater and ground-state
estimator variances remain below $1.5\%$ for the total energy
and $4.3\%$ for the spin structure factor across the cases
resolved numerically.
{
The relative discrepancies for the energy density are identical
to those for the total energy, since both variances are divided
by the same factor $L^4$.
}
Moreover, these discrepancies decrease with increasing system size
over the range examined.
This trend suggests that the Slater-determinant benchmark may remain
representative of the ground-state measurement cost at larger
system sizes, providing further motivation for its use in the
following scaling analysis.

\begin{table*}[t]
\centering
\small
\setlength{\tabcolsep}{4pt}
\renewcommand{\arraystretch}{1.15}
\caption{%
\justifying
Single-shot estimator variances for the energy density and
$S(\pi,\pi)$ in the ground state ($V_{\mathrm g}$) and a
site-occupation Slater determinant ($V_\Phi$) of the square-lattice
Fermi--Hubbard model.
We use $U/t=4$, open boundaries, $(N_\uparrow,N_\downarrow)=(2,2)$,
and $N=2L^2$ spin orbitals.
The Slater determinant doubly occupies $(0,0)$ and $(1,0)$.
Pauli shadow, FGU shadow, and OR shadow (ours) denote
uniform local Pauli, fermionic Gaussian-unitary, and orbital-rotation
shadows, respectively.
Pauli shadow uses Jordan--Wigner encoding with block-spin ordering.
For {the} spin structure factor, all protocols replace
$3\hat N_e/(2N)$ by $3\eta/(2N)$, with $\eta=4$.
The relative discrepancy
$\delta_{\mathrm{rel}}=|V_{\mathrm g}-V_\Phi|/V_{\mathrm g}$
is computed before rounding and reported as a percentage.
\textit{Unresolved} denotes a difference indistinguishable from
numerical rounding error, without implying exact equality.
}
\label{tab:hubbard-ground-slater}
\begin{tabular*}{\textwidth}{@{\extracolsep{\fill}}clrrrrrr@{}}
\toprule
& & \multicolumn{3}{c}{Energy density}
  & \multicolumn{3}{c}{Spin structure factor $S(\pi,\pi)$} \\
\cmidrule(lr){3-5}
\cmidrule(lr){6-8}
$N$ & Protocol
& $V_{\mathrm g}/t^2$
& $V_{\Phi}/t^2$
& $\delta_{\mathrm{rel}}$ (\%)
& $V_{\mathrm g}$
& $V_{\Phi}$
& $\delta_{\mathrm{rel}}$ (\%) \\
\midrule
18 & Pauli shadow
& $7.617$ & $7.506$ & $1.4599$
& $7.567\times10^{5}$ & $7.567\times10^{5}$
& $1.6289\times10^{-5}$ \\
& FGU shadow
& $52.78$ & $52.75$ & $0.0344$
& $38.78$ & $38.20$ & $1.5009$ \\
& OR shadow (ours)
& $10.51$ & $10.40$ & $1.0967$
& $8.077$ & $7.732$ & $4.2646$ \\
\midrule
32 & Pauli shadow
& $12.25$ & $12.19$ & $0.4979$
& $1.145\times10^{12}$ & $1.145\times10^{12}$
& \textit{unresolved} \\
& FGU shadow
& $89.84$ & $89.84$ & $0.0033$
& $123.8$ & $123.2$ & $0.4860$ \\
& OR shadow (ours)
& $6.777$ & $6.734$ & $0.6092$
& $8.586$ & $8.437$ & $1.7343$ \\
\midrule
50 & Pauli shadow
& $23.86$ & $23.82$ & $0.1373$
& $1.817\times10^{20}$ & $1.817\times10^{20}$
& \textit{unresolved} \\
& FGU shadow
& $137.8$ & $137.8$ & $0.0007$
& $305.3$ & $304.7$ & $0.2034$ \\
& OR shadow (ours)
& $4.587$ & $4.565$ & $0.4836$
& $8.766$ & $8.684$ & $0.9390$ \\
\bottomrule
\end{tabular*}
\end{table*}

\begin{figure*}[t]
    \centering
    \includegraphics[width=0.95\textwidth]{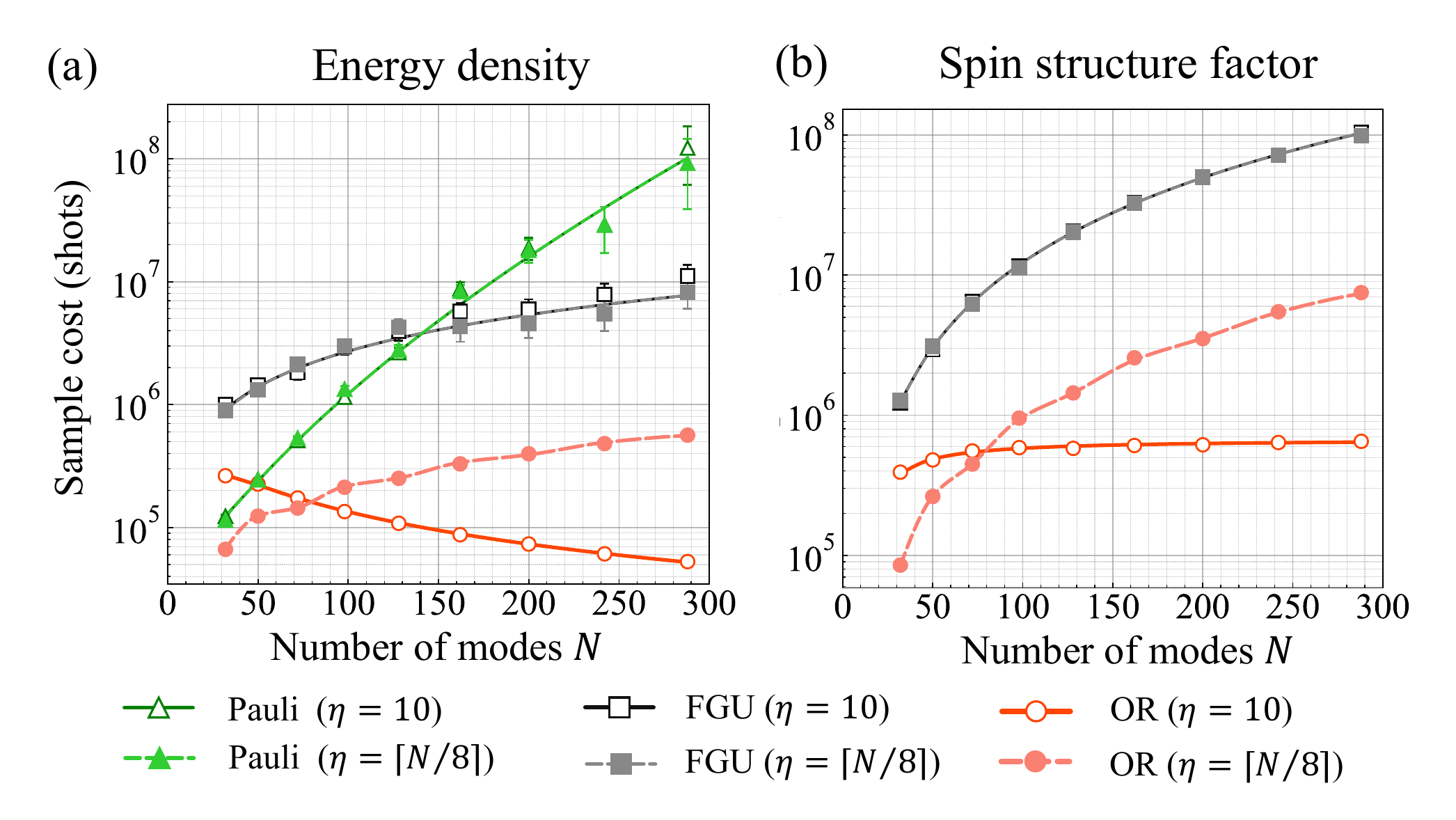}
    \caption{%
    \justifying
    Sample costs for estimating (a) the energy density and (b) the spin structure factor $S(\pi,\pi)$ of the square-lattice Fermi--Hubbard model using Pauli, FGU, and orbital-rotation (OR) shadows, for particle numbers $\eta = 10$ and $\eta = \lceil N/8 \rceil$. In Pauli CS, we consider the Jordan-Wigner encoding. 
    The target accuracies (root-mean-square errors) are $0.01t$ for energy density and
    $\varepsilon=0.01$ for spin structure factor, respectively.
    In both panels, markers show Monte Carlo estimates of the sample
    costs obtained from $10^4$ shots per method and system size,
    with error bars indicating one standard error. In panel (b), the Pauli shadow cost exceeds $10^{16}$ shots throughout the plotted range and is therefore omitted. 
    }
    \label{fig:hubbard-sample-cost}
\end{figure*}


Figure~\ref{fig:hubbard-sample-cost} compares the sample costs for estimating the energy density and spin structure factor at target root-mean-square errors of $\varepsilon_e=0.01t$ and $\varepsilon_S=0.01$, respectively, for $\eta=10$ and $\eta=\lceil N/8\rceil$. We use $U/t=4$ and open boundary conditions, with $N=2L^2$ spin orbitals for $L=4,\ldots,12$. In both cases, the reference state is a site-occupation Slater determinant. For $\eta=10$, we use $(N_\uparrow,N_\downarrow)=(5,5)$, with the sites $(0,0)$, $(1,0)$, $(2,0)$, $(0,1)$, and $(1,1)$ doubly occupied and all other sites empty. This occupation pattern is held fixed as the lattice grows. For $\eta=\lceil N/8\rceil$, we use $(N_\uparrow,N_\downarrow)=(\lceil\eta/2\rceil,\lfloor\eta/2\rfloor)$. For each spin species, the occupied sites are the first $N_\sigma$ sites in lexicographic order of $(\max(x,y),y,x)$, with all remaining sites unoccupied by that species. The occupied region therefore grows with the lattice, maintaining an approximately fixed filling $\eta/N\simeq1/8$, up to integer rounding.

For energy-density estimation
[Fig.~\ref{fig:hubbard-sample-cost}(a)],
the sample cost of orbital-rotation shadows decreases with
system size at fixed particle number, whereas those of FGU
and Pauli shadows increase.
At $N=288$ and $\eta=10$, the sample cost of orbital-rotation
shadows is approximately two orders of magnitude lower than
that of FGU shadows and three orders of magnitude lower than
that of Pauli shadows.
For Pauli shadows, the Jordan--Wigner encoding maps some
nearest-neighbor hopping terms to Pauli strings whose weights
grow with the lattice size, increasing the cost of uniform
local Pauli measurements.
For FGU shadows, the reconstruction factors for fermionic
correlations grow with the number of modes, producing a
variance overhead that persists after normalization by $L^4$.
Orbital-rotation shadows exploit particle-number symmetry
to control RDM-estimator fluctuations; at fixed particle number,
the total-energy estimator variance grows more slowly than
$L^4$, yielding a decreasing energy-density estimation cost.
For $\eta=\lceil N/8\rceil$, the orbital-rotation cost instead
increases with system size, but remains approximately a factor
of $14$ below the FGU cost at $N=288$, demonstrating a substantial
measurement advantage also at approximately fixed filling.

The advantage is also pronounced for the spin structure factor
[Fig.~\ref{fig:hubbard-sample-cost}(b)].
After replacing the particle-number term by its known value,
this observable is a sum of two-body fermionic terms normalized
by the number of lattice sites.
At fixed particle number, the sample cost of orbital-rotation
shadows approaches an approximately constant value over
the range examined, whereas the FGU-shadow cost exhibits
approximately quadratic growth with $N$.
At $N=288$, orbital-rotation shadows reduce the sample cost
by a factor of $161$ relative to FGU shadows.
A substantial reduction persists when the particle number
increases with the lattice size: for $\eta=\lceil N/8\rceil$,
corresponding to approximately fixed filling, the orbital-rotation
cost is approximately $14$ times lower than the FGU cost at $N=288$.
These results demonstrate that the measurement advantage
of orbital-rotation shadows extends to physical observables
whose estimation costs depend on both the variances and
the mutual covariances of their constituent RDM elements.
While this advantage becomes more pronounced with increasing
mode number in the fixed-particle-number regime,
it remains substantial at fixed filling.

\end{revision}

\begin{revision}

\subsection{Application to quantum chemistry}

Energy derivatives are central quantities in quantum chemistry.
In particular, derivatives with respect to nuclear coordinates
determine molecular forces and are therefore required for geometry
optimization, transition-state searches, and ab initio molecular
dynamics~\cite{obrien2022efficient}.
Motivated by these applications,
we focus on first-order nuclear derivatives and compare their
measurement costs, quantified by the single-shot estimator variance,
for orbital-rotation shadows and existing measurement strategies.
In what follows, let
\(
    \mathbf R=(\mathbf R_1,\ldots,\mathbf R_{N_{\mathrm{atom}}})
\)
denote the nuclear geometry, and let
\(R_{A\alpha}=(\mathbf R_A)_\alpha\) be the Cartesian coordinate
of atom \(A\) in direction \(\alpha\in\{x,y,z\}\).

We consider the molecular active-space Hamiltonian
$H_{\mathrm{act}}(\mathbf R)$ in the occupation-number basis
associated with a smoothly varying set of orthonormal molecular
orbitals. For a given atomic basis set, the molecular-orbital
coefficients are obtained from a reference self-consistent-field calculation. In this representation, the Hamiltonian takes
the second-quantized form
\begin{align}
    H_{\mathrm{act}}(\mathbf R)
    &=
    \sum_{p,q}
    h_{p,q}(\mathbf R)\,a_p^\dagger a_q
    +
    \sum_{\vec p,\vec q}
    g_{p_1p_2,q_1q_2}(\mathbf R)\,
    a_{p_1}^\dagger a_{p_2}^\dagger a_{q_2}a_{q_1},
\end{align}
where $p,q\in\{1,\ldots,N\}$ and
$\vec p,\vec q\in\mathcal S_{N,2}$. Here, we ignore the constant factor for simplicity.
The scalar term $E_{\mathrm{const}}(\mathbf R)$ includes
the nuclear-repulsion energy and the inactive-orbital energy.

For a target state $\rho(\mathbf R)$, differentiating
$E(\mathbf R)=\tr[\rho(\mathbf R)H_{\mathrm{act}}(\mathbf R)]$
gives
\begin{align}
    \pdv{E(\mathbf R)}{R_{A\alpha}}
    &=
    \tr\!\left[
        \rho\,\pdv{H_{\mathrm{act}}}{R_{A\alpha}}
    \right]
    +
    \tr\!\left[
        \pdv{\rho}{R_{A\alpha}}H_{\mathrm{act}}
    \right].
    \label{eq:energy-gradient-response}
\end{align}
Here, the second term describes the response of the state
coefficients in the chosen molecular-orbital representation.
It vanishes for a normalized variational ansatz whose
nuclear-coordinate dependence in this representation is entirely
mediated by parameters with respect to which the energy is
stationary \cite{obrien2019calculating,mizukami2020orbital}.
Under this condition, the Hellmann--Feynman relation gives
\begin{align}
    \pdv{E(\mathbf R)}{R_{A\alpha}}
    &=
    \tr\!\left[
        \rho\,\pdv{H_{\mathrm{act}}}{R_{A\alpha}}
    \right].
    \label{eq:energy-gradient-response-2}
\end{align}
For the electronic Hamiltonians considered here, differentiation
preserves the scalar, one-body, and two-body operator structure
\cite{obrien2019calculating,obrien2022efficient}:
\begin{align}
    \pdv{H_{\mathrm{act}}(\mathbf R)}{R_{A\alpha}}
    &=
    \sum_{p,q}
    \pdv{h_{p,q}(\mathbf R)}{R_{A\alpha}}\,
    a_p^\dagger a_q
    +
    \sum_{\vec p,\vec q}
    \pdv{g_{p_1p_2,q_1q_2}(\mathbf R)}{R_{A\alpha}}\,
    a_{p_1}^\dagger a_{p_2}^\dagger a_{q_2}a_{q_1}.
    \label{eq:energy-derivative-operator}
\end{align}
Using the RDM conventions introduced above, the energy derivative
under the same stationarity condition is therefore
\begin{align}
    \pdv{E(\mathbf R)}{R_{A\alpha}}
    &=
    \sum_{p,q}
    \pdv{h_{p,q}(\mathbf R)}{R_{A\alpha}}\,
    {}^{1}D_{\rho;\,q}^{p} +
    \sum_{\vec p,\vec q}
    \pdv{g_{p_1p_2,q_1q_2}(\mathbf R)}{R_{A\alpha}}\,
    {}^{2}D_{\rho;\,\vec q}^{\vec p}.
    \label{eq:energy-derivative-RDM-contraction}
\end{align}
We remark that the coefficient derivatives in
Eq.~\eqref{eq:energy-derivative-operator} include the
nuclear-coordinate dependence of the reference SCF orbitals
\cite{mizukami2020orbital}, together with analytic integral
derivatives, basis-overlap and Pulay terms, and frozen-core
contributions.
Thus, once the derivative coefficients have been determined,
the 1- and 2-RDMs suffice to evaluate the expectation value
of each derivative operator; no higher-order RDMs are required.
The same shadow snapshots can therefore be reused to estimate
all $3N_{\mathrm{atom}}$ Cartesian components by changing only
the classical contraction coefficients.

We considered two molecular systems: $\mathrm{Cr}_2$
and $[\mathrm{Fe}_2\mathrm{S}_2]^{2-}$ cluster.
We selected $\mathrm{Cr}_2$ as a challenging electronic-structure
benchmark because of the interplay of strong static and dynamic
correlation~\cite{larsson2022chromium}.
Our $[\mathrm{Fe}_2\mathrm{S}_2]^{2-}$ model was
motivated by the use of iron--sulfur clusters as chemically relevant
benchmarks for quantum-resource estimation~\cite{georges2025quantum}.
We remark that they serve solely to assess the measurement efficiency of shadow protocols; a quantitative characterization of their chemical properties is beyond the scope of this work.


For $\mathrm{Cr}_2$, we constructed a nonrelativistic
active-space Hamiltonian from restricted Hartree--Fock orbitals
obtained with the cc-pVTZ basis set
\cite{balabanov2005systematically}
for the neutral singlet molecule.
The nuclei were placed at $(0,0,\pm0.84)$\,\AA{},
corresponding to a bond length of $1.68$\,\AA{},
near the experimentally inferred equilibrium distance
\cite{larsson2022chromium}.
Eighteen doubly occupied spatial orbitals were treated as inactive,
leaving $\eta=12$ active electrons.
The choice of the Cr $3d4s$ valence space was motivated by
previous multireference studies
\cite{kurashige2011second}.
Six occupied and six virtual molecular orbitals were selected
using projectors onto the Cr $3d$ and $4s$ atomic orbitals,
and additional unoccupied orbitals were appended in order of
increasing orbital energy.
We constructed active spaces with
$N=32,48,64,80,96,112,128$ spin orbitals by taking the first
$D=N/2$ spatial orbitals from a common ordered set of 64
spatial orbitals, while retaining the same occupied active orbitals.

For the idealized bare $[\mathrm{Fe}_2\mathrm{S}_2]^{2-}$ cluster,
we adopted a model geometry with the Fe atoms at
$(\pm1.35,0,0)$\,\AA{}
and the S atoms at $(0,\pm1.75,0)$\,\AA{}.
Reference orbitals were obtained from a high-spin restricted
open-shell Kohn--Sham (ROKS) calculation with $S=5$, using
the BP86 functional, the cc-pVTZ-DK basis set, and a spin-free
one-electron exact two-component (X2C) relativistic treatment.
These electronic-structure settings were motivated by
the quantum-resource study of Georges et al.
\cite{georges2025quantum}.
The lowest 36 doubly occupied spatial orbitals were treated
as inactive, leaving two doubly occupied and ten singly occupied
orbitals with $\eta=14$ active electrons and spin populations
$(n_\alpha,n_\beta)=(12,2)$.
As for $\mathrm{Cr}_2$, unoccupied orbitals were appended in order
of increasing orbital energy to generate the same sequence
of active-space sizes from a common set of 64 spatial orbitals.

The reference states used for the measurement benchmarks were
fixed-occupation Slater determinants constructed from the corresponding
self-consistent field orbitals.
These states are not assumed to satisfy the stationarity conditions for the correlated active-space energy. Nevertheless, Sec.~\ref{sec:application_condensed_matter} showed that the estimator variances for global observables, including the total energy density and spin structure factor, were close for the Slater determinants and correlated ground states examined. This observation motivates the use of Slater determinants as tractable reference states for the present benchmarks. We therefore focus on the single-shot estimator variance for the expectation value of $\partial H_{\mathrm{act}}/\partial R_{A\alpha}$ in these reference states.

Figures~\ref{fig:cr2-sample-cost}(a) and
\ref{fig:fe2s2-sample-cost}(a) compare the sample costs
for estimating the expectation values of the energy-derivative
operators at a target root-mean-square error of
$\varepsilon=0.01\,\mathrm{Hartree}/\text{\AA}$.
For both systems, the costs of orbital-rotation and FGU
shadows increase as the active space is enlarged,
but the orbital-rotation cost grows substantially more slowly.
The local Pauli costs shown in the insets grow much more
rapidly because of the measurement overhead associated
with high-weight Jordan--Wigner strings.
At $N=128$, orbital-rotation shadows reduce the sample cost
relative to FGU shadows by factors of approximately $50$
for Cr$_2$ and $20$ for $[\mathrm{Fe}_2\mathrm{S}_2]^{2-}$.
Although these reductions demonstrate a substantial advantage,
they remain smaller than the factors $(N/\eta)^2$ suggested
by a simple comparison of the entrywise 2-RDM variance bounds.

To examine this difference, we resolve the sample cost
into one-body, two-body, and covariance contributions
in Figs.~\ref{fig:cr2-sample-cost}(b) and
\ref{fig:fe2s2-sample-cost}(b).
At $N=128$, the one-body variance accounts for approximately
$30\%$ and $27\%$ of the sum of the one-body and two-body
variances for Cr$_2$ and $[\mathrm{Fe}_2\mathrm{S}_2]^{2-}$,
respectively, under orbital-rotation shadows,
compared with approximately $4\%$ and $9\%$ under FGU shadows.
This appreciable one-body contribution helps explain
why the advantage suggested by the 2-RDM variance bounds
is not fully reflected in the total sample cost.
Since the number of possible two-body terms grows faster
than that of one-body terms, we expect the two-body
contribution to become increasingly dominant for
orbital-rotation shadows in larger active spaces,
further enhancing their advantage over FGU shadows.
Nevertheless, the substantial reductions already observed
provide evidence for the utility of orbital-rotation shadows
in quantum-chemistry applications.

\begin{figure*}[p]
    \centering

    \includegraphics[
        width=0.90\textwidth,
        keepaspectratio
    ]{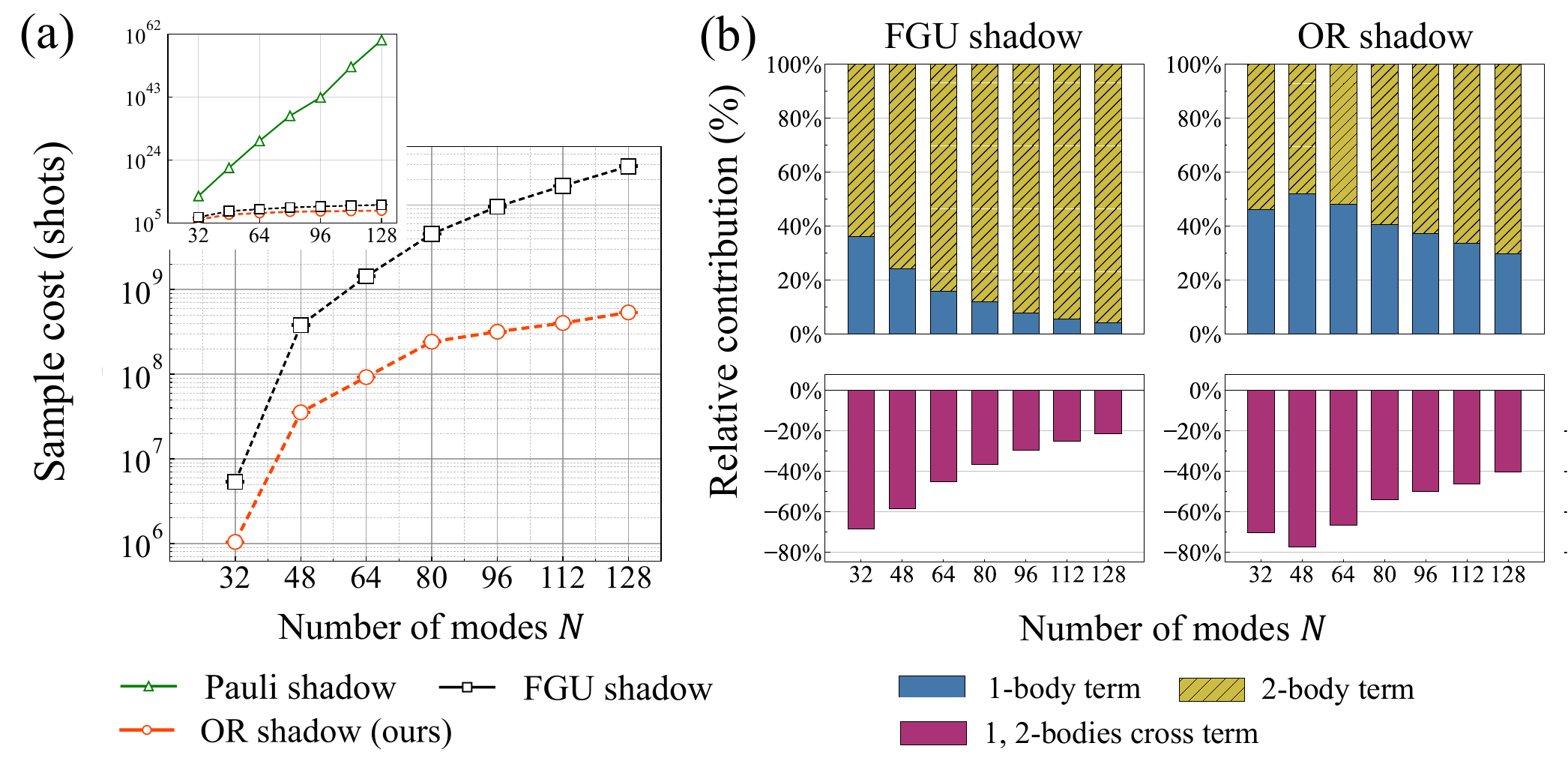}
    \caption{%
    \justifying
    (a) Sample costs for estimating the energy derivative of
    $\mathrm{Cr}_2$ with respect to the displacement of one Cr
    nucleus along the bond axis ($z$ direction), at a target
    root-mean-square error of
    $\varepsilon=0.01\,\text{Hatree/ \AA}$.
    The reference states are fixed-occupation Slater determinants
    from RHF/cc-pVTZ orbitals for the neutral singlet molecule
    at a bond length of $1.68\,\text{\AA}$.
    The active space is enlarged by adding unoccupied orbitals
    at fixed $\eta=12$.
    The main plot compares FGU and OR shadows; the inset also
    includes local Pauli shadows with the Jordan--Wigner encoding.
    FGU and OR estimates use $10^4$ Monte Carlo shots per method
    and system size, with one-standard-error bars;
    Pauli costs are evaluated analytically.
    (b) Relative contributions to the sample cost.
    The upper plots show contributions from estimators of
    one-body correlations $\langle a_p^\dagger a_q\rangle$
    and two-body correlations
    $\langle a_p^\dagger a_q^\dagger a_s a_r\rangle$.
    The lower plots show twice the covariance between the
    corresponding one-body and two-body single-shot estimators.
    All contributions are normalized by the sum of the
    one-body and two-body contributions.
    }
    \label{fig:cr2-sample-cost}

    \vspace{1em}

    \includegraphics[
        width=0.90\textwidth,
        keepaspectratio
    ]{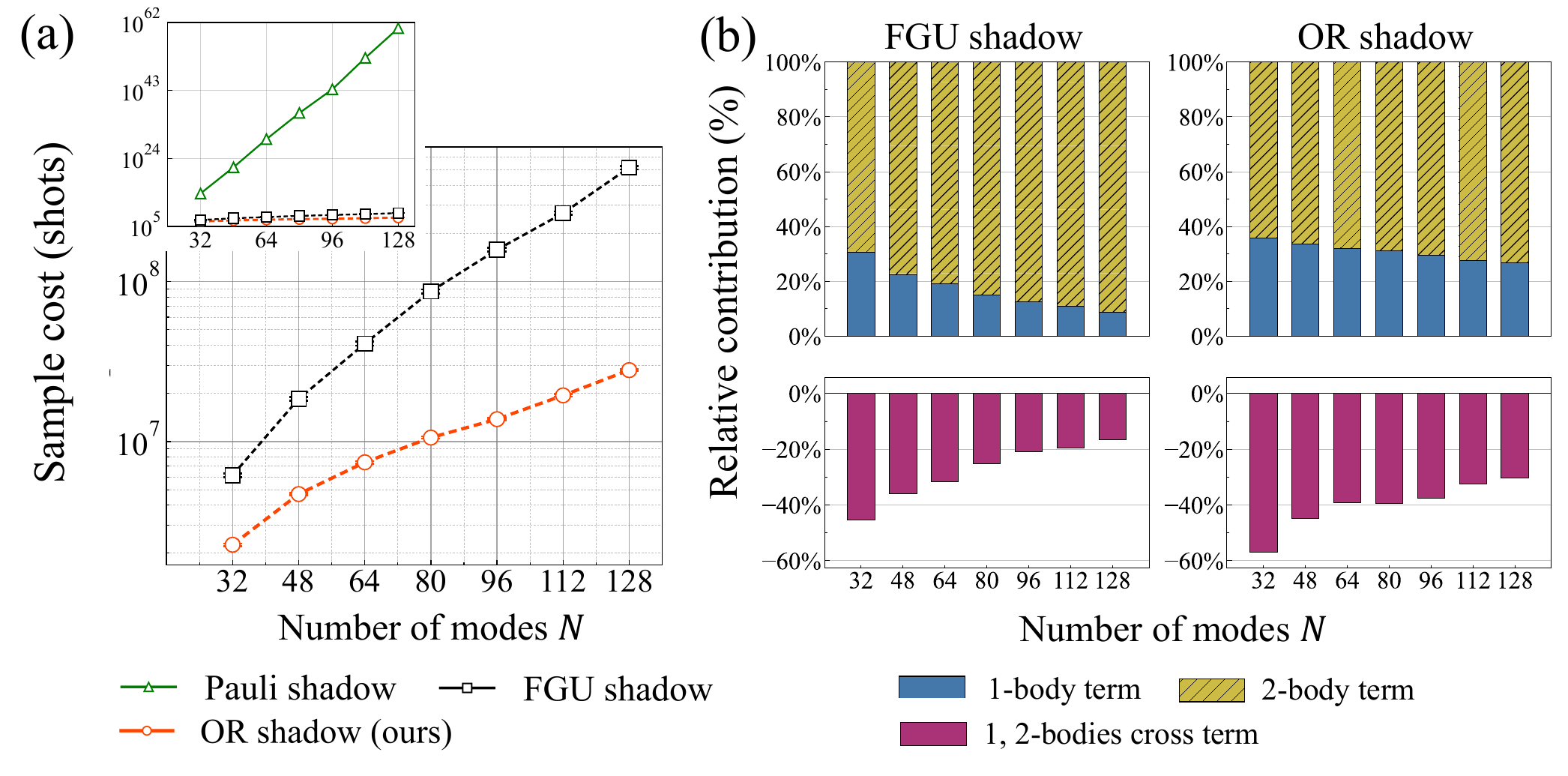}
    \caption{%
    \justifying
(a) Sample costs for estimating the energy derivative of the
$[\mathrm{Fe}_2\mathrm{S}_2]^{2-}$ cluster with respect to
the displacement of one Fe nucleus along the Fe--Fe axis
($x$ direction).
The reference states are fixed-occupation Slater determinants
from high-spin ($S=5$) ROKS/BP86/cc-pVTZ-DK orbitals
with spin-free one-electron X2C.
The Fe and S coordinates are
$(\pm1.35,0,0)\,\text{\AA}$ and
$(0,\pm1.75,0)\,\text{\AA}$, respectively.
The active space is enlarged by adding unoccupied orbitals
at fixed $\eta=14$ and $(n_\alpha,n_\beta)=(12,2)$.
The target accuracy and sampling conventions are the same
as in Fig.~\ref{fig:cr2-sample-cost}(a).
(b) Relative one-body, two-body, and cross-term contributions
to the sample cost, with the same definitions, normalization,
and panel layout as in Fig.~\ref{fig:cr2-sample-cost}(b).
    }
    \label{fig:fe2s2-sample-cost}
\end{figure*}

\end{revision}

\clearpage

\section{Technical lemmas}
\label{sec:technical_lemmas}

\setcounter{secnumdepth}{4}
\renewcommand{\thesubsubsection}{\alph{subsubsection}}
\renewcommand{\theparagraph}{\roman{paragraph}}
\makeatletter
\renewcommand{\p@subsubsection}{\thesection.\thesubsection.}
\renewcommand{\p@paragraph}{\thesection.\thesubsection.\thesubsubsection.}
\makeatother
\renewcommand{\theHsubsubsection}{appendix.subsubsection.\arabic{section}.\arabic{subsection}.\arabic{subsubsection}}
\renewcommand{\theHparagraph}{appendix.paragraph.\arabic{section}.\arabic{subsection}.\arabic{subsubsection}.\arabic{paragraph}}
\addtocontents{toc}{\protect\setcounter{tocdepth}{3}}

\subsection{Geometric Construction of Grassmannian picture}
\label{sec:Geometric_costruction}

\subsubsection{The equivalence between three manifolds}

In this section, we will verify the equivalence among the following three manifolds,
\begin{gather}
    \mtr{U}(N) / \qty( \mtr{U}(\eta) \times  \mtr{U}(N-\eta ) ), \\
    \mtr{Gr}_\eta (\mathbb{C}^N ) = \{ W \subset \mathbb{C}^N \mid \dim_{\mathbb{C} } W = \eta \}, \\
    \mac{P}_\eta = \{ R \in \mtr{End}(\mathbb{C}^N ) \mid R^2 = R, R^\dagger = R , \mtr{rank} R = \eta \}.
\end{gather}
Here, $\mtr{U}(N) \coloneq \{ A \in \mtr{End}(\mathbb{C}^N   ) \mid A^\dagger A = \1_{\mathbb{C}^N } \}.  $ More formally, our goal is to make precise the standard identifications
\begin{align}
    \mtr{Gr}_\eta(\mathbb C^N)
    \simeq
    \mac{P}_\eta
    \simeq
    \frac{\mtr{U}(N)}
    {\mtr{U}(\eta)\times \mtr{U}(N-\eta)}.
\end{align}
Throughout this section, we first recall why the quotient
\(\mtr{U}(N)/(\mtr{U}(\eta)\times \mtr{U}(N-\eta))\) carries a natural
smooth manifold structure.  We then identify each point of the Grassmannian
with the orthogonal projector onto the corresponding subspace.  Finally, by
considering the conjugation orbit of the standard projector \(P_{[\eta]}\)
under \(\mtr{U}(N)\), we prove that these three descriptions are naturally
diffeomorphic.  This homogeneous-space viewpoint also provides a canonical
\(\mtr{U}(N)\)-invariant probability measure on the Grassmannian via the
pushforward of the normalized Haar measure on \(\mtr{U}(N)\).

First, we will overview $\mtr{U}(N)$ as a Lie group.

\begin{definition}[Lie group]
A Lie group is a group \(G\) equipped with a smooth manifold structure
such that the maps
\begin{gather}
    m : G \times G \to G, \qquad m(g,h)=gh,  \\
    \iota : G \to G, \qquad \iota(g)=g^{-1},
\end{gather}
are smooth.
\end{definition}

From Ref.~\cite{lee2003smooth}, we can guarantee that $\mtr{U}(N) / \qty( \mtr{U}(\eta) \times  \mtr{U}(N-\eta ) )$ is a smooth manifold. To this end, let $H:=\mtr{U}(\eta)\times \mtr{U}(N-\eta)$ be embedded into $\mtr{U}(N)$ as the
block-diagonal subgroup
\begin{align*}
H
=
\left\{
\begin{pmatrix}
h_1 & 0\\
0 & h_2
\end{pmatrix}
\;\middle|\;
h_1\in \mtr{U}(\eta),\ h_2\in \mtr{U}(N-\eta)
\right\}.
\end{align*}
With regards this subgroup, the following proposition can be shown

\begin{lemma}
    Let
    \begin{align}
        H
        \coloneq
        \left\{
        \begin{pmatrix}
            h_1 & 0\\
            0 & h_2
        \end{pmatrix}
        \in \mtr{U}(N)
        \;\middle|\;
        h_1\in \mtr{U}(\eta),\ 
        h_2\in \mtr{U}(N-\eta)
        \right\}.
    \end{align}
    Then $H$ is a closed subgroup of $\mtr{U}(N)$.
\end{lemma}

\begin{proof}
    We prove the claim by realizing $H$ as the inverse image of a closed
    set under a continuous map.  For any $A\in \mtr{U}(N)$, write $A$ in
    block form as
    $
        A=
        \begin{pmatrix}
            A_{11} & A_{12}\\
            A_{21} & A_{22}
        \end{pmatrix},
    $ 
    where the block sizes are compatible with the decomposition
    \(\mathbb C^N=\mathbb C^\eta\oplus \mathbb C^{N-\eta}\).  Define
    \begin{align}
        F:\mtr{U}(N)
        &\longrightarrow
        \mathbb C^{\eta\times (N-\eta)}
        \oplus
        \mathbb C^{(N-\eta)\times \eta},\\
        F(A)
        &:=(A_{12},A_{21}).
    \end{align}
    This map is continuous, since it is the restriction to \(\mtr{U}(N)\)
    of the linear projection that extracts the off-diagonal blocks.

    We now observe that $ H=F^{-1}(\{(O,O)\}), $ where \(O\) denotes the zero matrix in each corresponding block space.
    Indeed, if \(A\in H\), then by definition its off-diagonal blocks vanish, i.e. \(F(A)=(O,O)\). Conversely, if \(F(A)=(O,O)\), then $
        A=
        \begin{pmatrix}
            A_{11} & O\\
            O & A_{22}
        \end{pmatrix}.$
    Since \(A\in \mtr{U}(N)\), the condition \(A^\dagger A=\1_{\mathbb{C}^N}\) implies $
        A_{11}^\dagger A_{11}=\1_{\mathbb{C}^\eta },
        A_{22}^\dagger A_{22}=\1_{\mathbb{C}^{N-\eta}}.
    $
    Therefore \(A_{11}\in \mtr{U}(\eta)\) and
    \(A_{22}\in \mtr{U}(N-\eta)\) hold, and it follows \(A\in H\).

    Finally, \(\{(O,O)\}\) is a closed subset of the finite-dimensional
    vector space
    \(\mathbb C^{\eta\times (N-\eta)}
    \oplus
    \mathbb C^{(N-\eta)\times \eta}\). Since the inverse image of a
    closed set under a continuous map is closed, \(H\) is closed in
    \(\mtr{U}(N)\).
\end{proof}

For the Lie group \(\mtr{U}(N)\) and its closed subgroup
\(\mtr{U}(\eta)\times \mtr{U}(N-\eta)\), we use the following standard result without the proof.
\begin{theorem}[Homogeneous space construction theorem, Theorem 21.17 in Ref.~\cite{lee2003smooth}]
    Let \(G\) be a Lie group, and let \(H\) be a closed subgroup of \(G\).
    Then the left coset space \(G/H\) is a topological manifold of dimension
    \(\dim G-\dim H\). Moreover, \(G/H\) admits a unique smooth structure
    such that the quotient map
    \begin{align}
        \pi_H:G\to G/H,
        \qquad
        g \longmapsto gH,
    \end{align}
    is a smooth submersion. With this smooth structure, the natural left
    action of \(G\) on \(G/H\), defined by
    \begin{align}
        g_1\cdot(g_2H)=(g_1g_2)H,
    \end{align}
    makes \(G/H\) into a homogeneous \(G\)-space.
\end{theorem}
    We remark that this fact is
    important because it allows us to treat the complex Grassmannian as a
    smooth quotient manifold,
    \begin{align}
        \mtr{Gr}_\eta (\mathbb{C}^N)
        \simeq
        \mtr{U}(N)/(\mtr{U}(\eta)\times \mtr{U}(N-\eta)).
    \end{align}
    Since \(\mtr{U}(N)\) is compact and \(\mtr{U}(\eta)\times \mtr{U}(N-\eta)\) is a closed subgroup, the
    normalized Haar measure on \(\mtr{U}(N)\) induces a canonical
    \(\mtr{U}(N)\)-invariant probability measure on the quotient space.
    Under the above identification, this measure is precisely the natural
    invariant measure on the Grassmannian.  Thus integrals over
    \(\mtr{Gr}_\eta (\mathbb{C}^N)\) can be rigorously interpreted as integrals over
    the homogeneous space \(\mtr{U}(N)/H\).

    Next, we will define the natural Haar measure on this manifold $\mtr{U}(N) / \qty( \mtr{U}(\eta) \times \mtr{U}(N-\eta) )$.
    Let \( \dd{u} \) denote the normalized Haar measure on the compact Lie group
\(\mtr{U}(N)\).  The quotient map
    \begin{align}
        \label{eq:quotient-map}
        \pi_H:\mtr{U}(N)
        &\longrightarrow
        \mtr{U}(N)/H,
        \qquad
        u \longmapsto uH
    \end{align}
allows us to pushforward this Haar measure to the quotient space.  Namely,
we define a probability measure \(\mu_{\mtr{U}(N)/H}\) on
\(\mtr{U}(N)/H\) by
    \begin{align}
        \mu_{\mtr{U}(N)/H}
        \coloneq
        (\pi_H)_* \dd{u}.
    \end{align}
    This means that, for every Borel subset \(A\subset \mtr{U}(N)/H\), the following relationship holds,
    \begin{align}
        \mu_{\mtr{U}(N)/H}(A)
        =
        \int_{\mtr{U}(N)}
        \1_A(\pi_H(u)) \dd{u}.
    \end{align}
    This measure is naturally \(\mtr{U}(N)\)-invariant.  Indeed, for any
    \(g\in \mtr{U}(N)\), the left action on the quotient is given by
    \(g\cdot (uH)=(gu)H\).  Therefore,
    \begin{align}
        \mu_{\mtr{U}(N)/H}(gA)
        &=
        \int_{\mtr{U}(N)}
        \1_{gA}(\pi_H(u)) \dd{u} \\
        &=
        \int_{\mtr{U}(N)}
        \1_{A}(\pi_H(g^{-1}u)) \dd{u} \\
        &=
        \int_{\mtr{U}(N)}
        \1_{A}(\pi_H(u)) \dd{u} \\
        &=
        \mu_{\mtr{U}(N)/H}(A),
    \end{align}
where the third equality follows from the left invariance of the Haar
measure \( \dd{u} \).  Hence the pushforward of the Haar measure on
\(\mtr{U}(N)\) defines a canonical \(\mtr{U}(N)\)-invariant probability
measure on the homogeneous space \(\mtr{U}(N)/H\).

Next, we introduce the complex Grassmann manifold as follows \cite{milnor1974characteristic}
\begin{align}
    \mtr{Gr}_\eta (\mathbb{C}^N ) \coloneq  \{ W \subset \mathbb{C}^N \mid \dim_{\mathbb{C} } W = \eta \}.
\end{align}
By definition, this is the set of all complex $\eta$-planes through the origin in the complex vector space $\mathbb{C}^{N}$. For instance, when $\eta=1$, this set corresponds to $\mathbb{C} P^{N-1}$. It is known that this set has a natural structure as a smooth manifold. 

An element in this manifold is one subspace $W$ itself, so we might guess that there exists a one-to-one correspondence between $W$ and its orthogonal projector operator $P_W$ onto the subspace $W$. Strictly speaking, this correspondence is the projector representation of the Grassmannian. In this representation, the Grassmannian can be identified with the space of rank-$\eta$ orthogonal projectors, namely 
there is a diffeomorphism
\begin{align}
    \mathrm{Gr}_\eta (V)\simeq \{P\in \mathrm{End}(V)\mid P^2=P,\; P^\dagger=P,\; \mathrm{rank}(P)=\eta\}.
\end{align}
Indeed, if $W\subset V$ is an $\eta$-dimensional subspace, then the orthogonal projector $P_W$ is uniquely determined by $\mathrm{Im}(P_W)=W$, $P_W^2=P_W$, and $P_W^\dagger=P_W$. Conversely, any Hermitian idempotent operator $P$ of rank $\eta$ determines a unique point of the Grassmannian by taking its image, $W=\mathrm{Im}(P)$. Thus, a point of the Grassmannian is equivalently described either as an $\eta$-dimensional subspace $W$ or as the corresponding rank-$\eta$ orthogonal projector $P_W$.

If $\purple{B}\in\mathbb{C}^{N\times \eta}$ is a matrix whose columns form a basis of $W$, then the corresponding projector is given by $P_W=\purple{B}(\purple{B}^\dagger \purple{B})^{-1}\purple{B}^\dagger$. In particular, if the columns of $\purple{B}$ are orthonormal, this formula reduces to $P_W=\purple{B}\purple{B}^\dagger$. This expression is invariant under the change of basis $\purple{B}\mapsto \purple{B}A$ for any invertible matrix $A\in \mtr{End}(\mathbb{C}^\eta ) $, which reflects the fact that the Grassmannian parametrizes subspaces rather than particular choices of bases.

Finally, we show the isomorphism between $\mtr{Gr}_\eta (\mathbb{C}^N)$ and $\mtr{U}(N) / \qty( \mtr{U}(\eta) \times \mtr{U} (N -\eta)  )$. 
\begin{lemma}[Grassmannian as a homogeneous space]
Let
$
    E_{\eta}
    \coloneqq
    \mathrm{span}_{\mathbb C}\{\ket{1},\ldots,\ket{\eta}\}
    \subset \mathbb C^N
$
be the standard $\eta$-dimensional subspace, and let
$
    P_{[\eta]}
    \coloneqq
    \begin{pmatrix}
        I_{\eta} & 0\\
        0 & 0_{N-\eta}
    \end{pmatrix}
$
be the orthogonal projector onto $E_{\eta}$. Then there is a natural diffeomorphism
\begin{align}
    \mtr{Gr}_{\eta}(\mathbb C^N)
    \simeq
    \mathcal P_{\eta}
    \simeq
    \frac{\mtr{U}(N)}
    {\mtr{U}(\eta)\times \mtr{U}(N-\eta)},
\end{align}
where
\begin{align}
    \mathcal P_{\eta}
    \coloneqq
    \left\{
    P\in \mathrm{End}(\mathbb C^N)
    \;\middle|\;
    P^2=P,\;
    P^\dagger=P,\;
    \mathrm{rank}(P)=\eta
    \right\}.
\end{align}
\end{lemma}

\begin{proof}
By the preceding discussion, the Grassmannian can be identified with the
set of rank-$\eta$ orthogonal projectors: the diffeomorphism is given by
    \begin{align}
        \mtr{Gr}_{\eta} (\mathbb{C}^N) &\longrightarrow \mathcal{P}_\eta,
        \qquad
        W \longmapsto P_W,
    \end{align}
    whose inverse is given by
    \begin{align}
        \mathcal{P}_\eta &\longrightarrow \mtr{Gr}_{\eta} (\mathbb{C}^N),
        \qquad
        P \longmapsto \Im (P).
    \end{align}
    We denote the orthogonal projector onto $W$ by $P_W$.

We next show that $\mathcal P_\eta$ is the $\mtr{U}(N)$-orbit of $P_{[\eta]}$.
To this end, we introduce the following mapping $\Phi$,
\begin{align}
    \label{eq:adjoint-P0}
    \Phi:\mtr{U}(N)\longrightarrow \mathcal P_\eta,
    \qquad
    U \longmapsto UP_{[\eta]}U^\dagger,
\end{align}
where $P_{[\eta]} \coloneqq     
    \begin{pmatrix} 
    \1_{\mathbb{C}^\eta} & O \\ O &  O_{_{\mathbb{C}^{N-\eta}} }
    \end{pmatrix}$.
    
We first prove that $\Phi$ is surjective. Let $P\in\mathcal P_\eta$ be
arbitrary, and set $
    W\coloneqq \mathrm{Im}(P).
$
Since $P$ is a rank-$\eta$ orthogonal projector, $W$ is an $\eta$-dimensional
subspace of $\mathbb C^N$, that is,
$
    W\in \mtr{Gr}_\eta(\mathbb C^N).
$
Choose an orthonormal basis
$
    \ket{w_1},\ldots,\ket{w_\eta}
$
of $W$, and extend it to an orthonormal basis
$
    \ket{w_1},\ldots,\ket{w_N}
$
of $\mathbb C^N$. Let $U\in \mtr{U}(N)$ be the unitary matrix defined by
\begin{align}
    \forall i \in \{1,\ldots N\}, \quad U\ket{i}=\ket{w_i}.
\end{align}
This $U$ satisfies $U E_\eta
    =
    \mathrm{span}_{\mathbb C}\{\ket{w_1},\ldots,\ket{w_\eta}\}
    =
    W.
$
Since $P_{[\eta]}$ is the orthogonal projector onto $E_\eta$, the operator
$UP_{[\eta]}U^\dagger$ is the orthogonal projector onto $UE_\eta=W$. This fact follows from
$UP_{[\eta]}U^\dagger=P_W$. On the other hand, $P$ is also the orthogonal
projector onto $W=\mathrm{Im}(P)$. By the uniqueness of the orthogonal
projector onto a fixed subspace, we have $P=UP_{[\eta]}U^\dagger$.
Since $P\in\mathcal P_\eta$ was arbitrary, $\Phi$ is surjective. Equivalently,
\begin{align}
    \mathcal P_\eta
    =
    \left\{
    UP_{[\eta]}U^\dagger
    \;\middle|\;
    U\in \mtr{U}(N)
    \right\}.
\end{align}

It remains to identify the redundancy in this parametrization. The stabilizer
of $P_{[\eta]}$ under the conjugation action of $\mtr{U}(N)$ is
\begin{align}
    \mathrm{Stab}(P_{[\eta]})
    \coloneqq
    \left\{
    U\in \mtr{U}(N)
    \;\middle|\;
    UP_{[\eta]}U^\dagger=P_{[\eta]}
    \right\}.
\end{align}
The condition $UP_{[\eta]}U^\dagger=P_{[\eta]}$ is equivalent to saying that $U$ preserves
the orthogonal decomposition
$
    \mathbb C^N=E_\eta\oplus E_\eta^\perp.
$
Hence $U$ is block diagonal with respect to this decomposition, namely
$
    U=
    \begin{pmatrix}
        A & 0\\
        0 & B
    \end{pmatrix}
$
with $A\in \mtr{U}(\eta)$ and $B\in \mtr{U}(N-\eta)$. Conversely, every
unitary of this form fixes $P_{[\eta]}$. Therefore
$
    \mathrm{Stab}(P_{[\eta]})
    \simeq
    \mtr{U}(\eta)\times \mtr{U}(N-\eta).
$

The map $\Phi$ therefore descends to the quotient space as
\begin{align}
    \label{eq:orbit-map}
    \Phi_{\mtr{Stab}(P_{[\eta]}) }:
    \frac{\mtr{U}(N)}{\mathrm{Stab}(P_{[\eta]})}
    \longrightarrow
    \mathcal P_\eta,
    \qquad
    [U]\longmapsto UP_{[\eta]}U^\dagger ,
\end{align}
where $[U]=U\mathrm{Stab}(P_{[\eta]})$. This map is well-defined because replacing
$U$ by $Uh$ with an arbitrary $h\in\mathrm{Stab}(P_{[\eta]})$ does not change $UP_{[\eta]}U^\dagger$.

Moreover, $\Phi_{\mtr{Stab}(P_{[\eta]}) }$ is injective. Indeed, if
$UP_{[\eta]}U^\dagger=VP_{[\eta]}V^\dagger$, then
$
    (V^\dagger U)P_{[\eta]}(V^\dagger U)^\dagger=P_{[\eta]},
$
and hence $V^\dagger U\in\mathrm{Stab}(P_{[\eta]})$. Therefore $U$ and $V$
represent the same coset in
$
    \mtr{U}(N)/\mathrm{Stab}(P_{[\eta]})
$, which implies that $\Phi_{\mtr{Stab}(P_{[\eta]}) }$ is injective. Since $\Phi$ is surjective,
$\Phi_{\mtr{Stab}(P_{[\eta]}) }$ is also surjective, and hence
\begin{align}
    \mathcal P_\eta
    \simeq
    \frac{\mtr{U}(N)}{\mathrm{Stab}(P_{[\eta]})}
    \simeq
    \frac{\mtr{U}(N)}
    {\mtr{U}(\eta)\times \mtr{U}(N-\eta)}.
\end{align}
$\Phi_{\mtr{Stab}(P_{[\eta]}) }$ is actually a bijective local diffeomorphism, and hence a diffeomorphism.
Combining this with
$
    \mtr{Gr}_{\eta}(\mathbb C^N)\simeq \mathcal P_{\eta},
$
we obtain
\begin{align}
    \mtr{Gr}_{\eta}(\mathbb C^N)
    \simeq
    \frac{\mtr{U}(N)}
    {\mtr{U}(\eta)\times \mtr{U}(N-\eta)}
    .
\end{align}

\end{proof}

\subsubsection{Shadow-Induced Measures on the Grassmannian}

In this section, we formulate Low's number-conserving fermionic
shadow measurement as a probability measure on the Grassmannian of
rank-$\eta$ projectors. We first recall the identification of
$\mathcal P_\eta$ with the homogeneous space $\mtr{U}(N)/H$, and use
this identification to define the canonical invariant probability
measure $\dd R$ on $\mathcal P_\eta$ as the pushforward of the
normalized Haar measure on $\mtr{U}(N)$. We then introduce an equivalent
projector-valued orbit map, $u\mapsto u^\dagger P_{[\eta]}u$, which is
the form naturally appearing in Low's shadow protocol. Finally, we show
that the random measurement outcome $(u,\vec z)$ induces a
state-dependent probability measure $\nu_\rho$ on $\mathcal P_\eta$,
and we compute its density with respect to the Haar-pushforward
reference measure $\dd R$.

Let $H\subset \mtr{U}(N)$ be the stabilizer of $P_{[\eta]}$. Recall the
quotient map
\begin{align}
    \pi_H:\mtr{U}(N)
    \longrightarrow
    \mtr{U}(N)/H,
    \qquad
    \pi_H(u)=uH.
\end{align}
The orbit map
\begin{align}
    \Phi_H:\mtr{U}(N)/H
    \longrightarrow
    \mathcal P_\eta,
    \qquad
    uH \longmapsto uP_{[\eta]}u^\dagger
\end{align}
is well-defined and identifies $\mtr{U}(N)/H$ with $\mathcal P_\eta$.
Thus, if $\dd u$ denotes the normalized Haar measure on $\mtr{U}(N)$,
the canonical invariant probability measure on $\mathcal P_\eta$ is
defined as the pushforward measure
\begin{align}
    \label{eq:pushforward-measure-Peta}
    \dd R
    \equiv
    \mu_{\mathcal P_\eta}
    \coloneqq
    (\Phi_H)_* (\pi_H)_* \dd u
    =
    (\Phi_H \circ \pi_H)_* \dd u.
\end{align}
Equivalently, for every measurable function $f$ on $\mathcal P_\eta$,
\begin{align}
    \int_{\mathcal P_\eta} f(R)\,\dd R
    =
    \int_{\mtr{U}(N)}
    f(uP_{[\eta]}u^\dagger)\,\dd u.
\end{align}

In Low's shadow protocol, it is more convenient to use the equivalent map
\begin{align}
    \widetilde \pi_\eta:\mtr{U}(N)
    \longrightarrow
    \mathcal P_\eta,
    \qquad
    u\longmapsto u^\dagger P_{[\eta]}u .
\end{align}
Since the inversion map $u\mapsto u^\dagger$ preserves the normalized Haar
measure, we have
\begin{align}
    (\widetilde \pi_\eta)_*\dd u
    =
    (\Phi_H \circ \pi_H)_*\dd u
    =
    \dd R .
\end{align}
Hence the notation $\int_{\mathcal P_\eta}(\cdots)\dd R$ is rigorously
understood as integration with respect to the pushforward of the Haar
measure $\dd{u}$ on $\mtr{U}(N)$ through the projector-valued orbit map $\widetilde{\pi}_\eta$.

Under Low's orbital-rotation shadow protocol, a measurement
outcome $(u,\, \vec{z}) \in \mtr{U}(N) \times \mathcal{S}_{N,\eta}$ 
determines the rank-$\eta$ projector
\begin{align}
    R(u,\vec z)
    \coloneqq
    u^\dagger P_{\vec z}u ,
\end{align}
where $P_{\vec{z}} \coloneqq \sum_{i=1}^{\eta} |z_i\rangle\langle z_i|$ is the coordinate projector onto the $\eta$-occupied modes
specified by $\vec z\in\mathcal S_{N,\eta}$. The probability of observing
$\vec z$ conditioned on $u$ is
\begin{align}
    p_\rho(\vec z\mid u)
    =
    \tr\!\left[
        \rho\,
        U_\eta^\dagger(u)
        \bigl( P^{\wedge \eta}_{\vec z}\bigr)
        U_\eta(u)
    \right]
    =
    \tr\!\left[
        \rho\, R(u,\vec z)^{\wedge \eta}
    \right].
\end{align}
Therefore, the shadow measurement induces a probability measure on
$\mathcal P_\eta$ as the pushforward
along the continuous map $R:\mtr{U}(N)\times \mathcal{S}_{N,\eta} \to \mathcal{P}_\eta$
\begin{align}
    \nu_\rho
    \coloneqq
        \sum_{\vec z\in\mathcal S_{N,\eta}}
    R_*
    \left(
        p_\rho(\vec z\mid u)\,\dd u
    \right),
\end{align}
where $R(u,\vec z)=u^\dagger P_{\vec z}u$ and the discrete variable
$\vec z$ is summed over $\mathcal S_{N,\eta}$.

We now compute this pushforward measure explicitly. For each
$\vec z\in\mathcal S_{N,\eta}$, we can choose a permutation unitary
$w_{\vec z}\in\mtr{U}(N)$ such that
$    P_{\vec z}
    =
    w_{\vec z}P_{[\eta]}w_{\vec z}^\dagger
$~\footnote{Such a unitary $w_{\vec{z}}$ always exists because of the surjectivity of the map \eqref{eq:adjoint-P0}.
}.
Then
\begin{align}
    R(u,\vec z)
    =
    u^\dagger P_{\vec z}u
    =
    (w_{\vec z}^\dagger u)^\dagger
    P_{[\eta]}
    (w_{\vec z}^\dagger u)
    =
    \widetilde \pi_\eta(w_{\vec z}^\dagger u).
\end{align}
By the left invariance of the Haar measure, the change of variables
$u'=w_{\vec z}^\dagger u$ leaves $\dd u$ invariant. Hence each term in
the sum over $\vec z$ has the same pushforward to $\mathcal P_\eta$.
Consequently, for every integrable test function $F$ on $\mathcal P_\eta$,
\begin{align}
    \mathbb E_{\rho}[F(R)]
    &=
    \int_{\mtr{U}(N)}\dd u
    \sum_{\vec z\in\mathcal S_{N,\eta}}
    p_\rho(\vec z\mid u)\,
    F(R(u,\vec z))  \\
    &=
    \binom{N}{\eta}
    \int_{\mtr{U}(N)}\dd u\,
    \tr\!\left[
        \rho\,
        \widetilde \pi^{\wedge\eta}_\eta(u)
    \right]
    F(\widetilde \pi_\eta(u))  \\
    &=
    \binom{N}{\eta}
    \int_{\mathcal P_\eta}\dd R\,
    \tr\!\left[
        \rho\ R^{\wedge\eta}
    \right]
    F(R).
\end{align}
Thus the measurement-induced measure on the Grassmannian, identified with
$\mathcal P_\eta$, is
\begin{align}
    \nu_\rho(\dd R)
    =
    \binom{N}{\eta}
    \tr\!\left[
        \rho\, R^{\wedge \eta}
    \right]\dd R.
\end{align}
In particular, by setting $F(R)=1$, we obtain the normalization condition
\begin{align}
    1
    =
    \binom{N}{\eta}
    \int_{\mathcal P_\eta}\dd R\,
    \tr\!\left[
        \rho\, R^{\wedge \eta}
    \right].
\end{align}
Therefore, Low's number-conserving fermionic shadow measurement can be
viewed as a state-dependent probability measure on the Grassmannian, whose
reference measure is the Haar-pushforward measure
$
    \dd R=(\widetilde \pi_\eta)_*\dd u
$.

\subsection{Irreducible decomposition of $\mtr{End}(\mac{H}_\eta )$}
\label{sec:Irrep_endH}

In this section, we prove that every \(\mtr U(N)\)-equivariant linear map from \(\operatorname{End}(\mac H_\eta)\) to \(\operatorname{End}(\mac H_t)\), like $\mac{A}_t$, can be expressed as a linear combination of the contraction--extension maps
\(\mac E_{s,t}\circ \mac C_{\eta,s}\) with representation theory.

For the purposes of this section alone, we use standard notation from
finite-dimensional representation theory, which we briefly recall below.
If \(G=\mtr U(N)\) and \(V,W\) are \(G\)-modules, then
\(\operatorname{Hom}_{G}(V,W)\) denotes the vector space of \(G\)-equivariant
linear maps from \(V\) to \(W\).  Here a \(G\)-module simply means a
finite-dimensional complex vector space equipped with a linear action of
\(G\).  In particular, \(\operatorname{End}(\mac H_r)\) is regarded as a
\(\mtr U(N)\)-module through the conjugation action.
We also use the standard notation \(L_\lambda\) for the irreducible
highest-weight representation of \(\mtr U(N)\) with highest weight
\(\lambda\), and \(\mathbb{S}_\lambda(V)\) for the Schur module associated with a
partition \(\lambda\). Addtionally, $\mtr{GL}_N(\mathbb{C})$ denotes the group of invertible $N\times N$ complex matrices. These notions, as well as the Littlewood--Richardson
rule used below, are standard; see, for example,
Refs.~\cite{fulton1997young,Knapp2023liegroups}.

We first recall the equivariance condition used in the proof.
The spaces \(\operatorname{End}(\mac H_\eta)\) and
\(\operatorname{End}(\mac H_t)\) are equipped with the natural
\(\mtr U(N)\)-actions by conjugation:
\begin{align}
    X
    &\longmapsto
    U_\eta(u)XU_\eta(u)^\dagger, \\
    Y
    &\longmapsto
    U_t(u)YU_t(u)^\dagger .
\end{align}
A linear map
\begin{align}
    \Phi:
    \operatorname{End}(\mac H_\eta)
    \longrightarrow
    \operatorname{End}(\mac H_t)
\end{align}
is said to be \(\mtr U(N)\)-equivariant if it satisfies
\begin{align}
    \Phi\!\left(
        U_\eta(u)XU_\eta(u)^\dagger
    \right)
    =
    U_t(u)\Phi(X)U_t(u)^\dagger
\end{align}
for every \(u\in\mtr U(N)\) and
\(X\in\operatorname{End}(\mac H_\eta)\).

As a preparatory step toward the main goal of this section, we first establish the following lemma.

\begin{lemma}[Irreducible decomposition of $\mtr{End}(\mac{H}_r)$]
    \label{lem:irreducible_dec}
    Let $r$ and $N$ be integers satisfying
    \(0 \leq r\leq N\).  
    The \(\mtr{U}(N)\)-module
    \(\operatorname{End}(\mac H_r)\) decomposes multiplicity-freely into
    \begin{align}
        \operatorname{End}(\mac H_r)
        \simeq
        \bigoplus_{q=0}^{\min(r,N-r)}
        L_{(1^q,\,0^{N-2q},\,(-1)^q)}.
    \end{align}
    Here $L_\lambda$ denotes the irreducible highest-weight representation of
    $\mtr U(N)$ with highest weight $\lambda = (\lambda_1 \ge \lambda_2 \ge \dots \ge \lambda_N)$, and the notation $a^m$ means that
    the entry $a$ is repeated $m$ times.
\end{lemma}

\begin{proof}
    We give a sketch of proof.
    We prove the decomposition at the level of rational
    $\mtr{GL}_N(\mathbb{C})$-modules and then restrict to $\mtr U(N)$.
    Let $V \coloneqq \mathbb{C}^{N}$.
    The conjugation action of $\mtr U(N)$ on $\operatorname{End}(\mac H_r)$ is
    identified with the natural action on
    $
        \mac H_r\otimes \mac H_r^* 
    $, where $\mac H_r^*$ denotes the dual space of $\mac{H}_r$.
    Thus, as $\mtr{GL}_N(\mathbb{C})$-modules,
    \begin{align}
        \operatorname{End}(\mac H_r)
        &\simeq V^{\wedge r} \otimes (V^{\wedge r})^* \\
        &\simeq V^{\wedge r} \otimes V^{\wedge (N-r)} \otimes \det(V)^* .
    \end{align}
    Here
    $
        \det(V):= V^{\wedge N}
    $
    is the one-dimensional determinant representation, on which
    $g\in \mtr{GL}_N(\mathbb{C})$ acts by
    \begin{align}
        g\cdot (v_1\wedge\cdots\wedge v_N)
        =
        \det(g)\, v_1\wedge\cdots\wedge v_N ,
    \end{align}
    and $\det(V)^*$ denotes its dual representation.  The second isomorphism
    comes from the canonical perfect pairing
    \begin{align}
        V^{\wedge r} \otimes V^{\wedge (N-r)}
        &\longrightarrow \det(V), \\
        x\otimes y
        &\longmapsto x\wedge y,
    \end{align}
    which gives
    $
        (V^{\wedge r})^*
        \simeq
        V^{\wedge (N-r)}\otimes \det(V)^* .
    $

    We now use the standard notation $\mathbb{S}_\lambda(V)$ for the Schur module
    associated with a partition $\lambda$
    .  Since
    $
        V^{\wedge r} \simeq \mathbb{S}_{(1^r)}(V),
    $
    the Littlewood-Richardson rule~\cite[Sec.~8.3., Corollary 2]{fulton1997young} gives
    \begin{align}
        V^{\wedge r} \otimes V^{\wedge (N-r)}
        &\simeq
        \mathbb{S}_{(1^r)}(V)\otimes \mathbb{S}_{(1^{N-r})}(V) \\
        &\simeq
        \bigoplus_{q=0}^{\min(r,N-r)}
        \mathbb{S}_{(2^q,1^{N-2q})}(V).
    \end{align}
    This decomposition is multiplicity-free.

    Finally, tensoring by $\det(V)^*$ subtracts $(1^N)$ from the highest weight,
    because $\det(V)$ has highest weight $(1^N)$.  The Schur module
    $\mathbb{S}_{(2^q,1^{N-2q})}(V)$ has highest weight
    $(2^q,1^{N-2q},0^q)$,
    where trailing zeros are written so that the weight has length $N$.
    Therefore
    \[
        (2^q,1^{N-2q},0^q)-(1^N)
        =
        (1^q,0^{N-2q},(-1)^q).
    \]
    Hence, after restricting to $\mtr U(N)$, we obtain
    \[
        \operatorname{End}(\mac H_r)
        \simeq
        \bigoplus_{q=0}^{\min(r,N-r)}
        L_{(1^q,\,0^{N-2q},\,(-1)^q)} .
    \]
    Here $L_\lambda$ denotes the irreducible highest-weight representation of
    $\mtr U(N)$ with highest weight $\lambda$. For more details, see~\cite[Ch.~IX]{Knapp2023liegroups}.
\end{proof}

We will show that any $\mtr{U}(N)$-equivariant linear map can be spanned by $\mac{E}_{s,t}\circ \mac{C}_{\eta,s}$ as follows,

\begin{lemma}
\label{lem:equivariant-map-basis}
Let \(0\leq t\leq \eta \leq N\), and we assume
$
    m:=\min(\eta,N-\eta,t,N-t).
$
Then
\begin{align}
\operatorname{Hom}_{\mtr U(N)}
\bigl(
    \operatorname{End}(\mac H_\eta),
    \operatorname{End}(\mac H_t)
\bigr)
=
\operatorname{span}_{\mathbb C}
\left\{
    \mac E_{s,t}\circ \mac C_{\eta,s}
    :
    0\leq s\leq m
\right\}.    
\end{align}
Moreover, the maps
$
    \{\mac E_{s,t} \circ \mac C_{\eta,s} \}_{s=0}^m,
$ are linearly independent. In particular, the larger family
$
    \{\mac E_{s,t} \circ \mac C_{\eta,s} \}_{s=0}^t,
$
spans the same equivariant Hom-space.
\end{lemma}

\begin{proof}
First, the contraction and extension maps are natural with respect to the
\(\mtr U(N)\)-action. Hence each map
$
    \Phi_s:=\mac E_{s,t}\circ \mac C_{\eta,s}
$
is \(\mtr U(N)\)-equivariant.

We compute the dimension of the equivariant Hom-space, namely $\operatorname{Hom}_{\mtr U(N)}
\bigl(
    \operatorname{End}(\mac H_\eta),
    \operatorname{End}(\mac H_t)
\bigr)$. In what follows, we denote $
    L_q:=L_{(1^q,0^{N-2q},(-1)^q)}
$, which is irreducible representation of $\mtr{U}(N)$.
From Lemma.~\ref{lem:irreducible_dec}, 
\begin{align}
    \operatorname{End}(\mac H_\eta)
    \simeq
    \bigoplus_{q=0}^{\min(\eta,N-\eta)} L_q,
    \qquad
    \operatorname{End}(\mac H_t)
    \simeq
    \bigoplus_{q=0}^{\min(t,N-t)} L_q.    
\end{align}
Therefore, by Schur's lemma,
\begin{align}
    \dim
    \operatorname{Hom}_{\mtr U(N)}
    \bigl(
        \operatorname{End}(\mac H_\eta),
        \operatorname{End}(\mac H_t)
    \bigr) 
    &=
    \dim
    \operatorname{Hom}_{\mtr U(N)}
    \biggl(
        \bigoplus_{q=0}^{\min(\eta,N-\eta)} L_q,
        \bigoplus_{q'=0}^{\min(t,N-t)} L_{q'}
    \biggr) \\
    &=
    \sum_{q=0}^{\min(\eta,N-\eta)}
    \sum_{q'=0}^{\min(t,N-t)}
    \dim
    \operatorname{Hom}_{\mtr U(N)}
    \bigl(
        L_q,L_{q'}
    \bigr) \\
    &=
    \sum_{q=0}^{\min(\eta,N-\eta)}
    \sum_{q'=0}^{\min(t,N-t)}
    \delta_{q,q'} \\
    &=
    \#\Bigl(
        \{q\in\mathbb Z \mid 0\le q\le \min(\eta,N-\eta)\}
        \cap
        \{q\in\mathbb Z \mid 0\le q\le \min(t,N-t)\}
    \Bigr)  \\
    &=
    \#\{q\in\mathbb Z \mid 0\le q\le \min(\eta,N-\eta,t,N-t)\} \\
    &=
    \#\{q\in\mathbb Z \mid 0\le q\le m\} \\
    &=
    m+1.
\end{align}
By the dimension count above, the space of $\mtr U(N)$-equivariant maps has dimension $m+1$.  Since $\{\Phi_s\}_{s=0}^m$ is a family of $m+1$ such maps, it suffices to prove their linear independence.  Thus, it remains to prove that the maps $\Phi_0,\ldots,\Phi_m$ are linearly independent.

For \(0\leq q\leq m\), choose ordered sets
\begin{align}
    A_q=(1,\ldots,q),
    \qquad
    B_q=(N-q+1,\ldots,N),    
\end{align}
with the convention that they are empty for \(q=0\).  Also set
\begin{align}
        C_q=(q+1,\ldots,\eta),
    \qquad
    F_q=(q+1,\ldots,t).
\end{align}
Since \(q\leq N-\eta\) and \(q\leq N-t\), the relevant sets are disjoint.
Define
$
    X_q
    :=
    \purple{\ketbra{A_q\cup C_q}{B_q\cup C_q}}
    \in \operatorname{End}(\mac H_\eta),
$
and define the linear functional
$
    \ell_q(Z)
    :=
    \bra{A_q\cup F_q}Z\ket{B_q\cup F_q},
$ for $Z\in\operatorname{End}(\mac H_t).$

From the creation--annihilation definition of the contraction map,
\begin{align}
    \mac C_{\eta,s}(X_q)=0,
    \qquad \text{for }  s<q.    
\end{align}
Indeed, an \(s\)-particle contraction cannot keep all \(q\) mismatched
bra and ket modes.  On the other hand, for \(s=q\),
\begin{align}
    \mac C_{\eta,q}(X_q)
    =
    \varepsilon_q\ketbra{A_q}{B_q},    
\end{align}
where $ 
    \varepsilon_q\in\{\pm1\}.$
Applying the extension map gives
\begin{align}
    \mac E_{q,t}\bigl(\ketbra{A_q}{B_q}\bigr)
    =
    a^\dagger_{1} \ldots a_q^\dagger a_{N} \ldots a_{N-q+1}\big|_{\mac H_t},    
\end{align}
Since \(F_q\) is disjoint
from \(A_q\cup B_q\), this operator maps
$
    \ket{B_q\cup F_q}
$
to
$
    \pm \ket{A_q\cup F_q}.
$
Hence
\begin{align}
    \ell_q(\Phi_q(X_q))\neq 0.
\end{align}
Based on these observations, we have
\begin{align}
        \ell_q(\Phi_s(X_q))=0 \quad (s<q),
    \qquad
    \ell_q(\Phi_q(X_q))\neq 0.
    \label{eq:tri}
\end{align}

Now suppose that $
    \sum_{s=0}^{m}\alpha_s\Phi_s=0.
$
Apply \(\ell_m\) to the value of this identity on \(X_m\).  Since all
terms with \(s<m\) vanish, we get
\begin{align}
    \alpha_m\,\ell_m(\Phi_m(X_m)) = 0.
\end{align}
Because \(\ell_m(\Phi_m(X_m))\neq0\), it follows that \(\alpha_m=0\).

Proceeding downward, assume that
$
    \alpha_{q+1}=\cdots=\alpha_m=0.
$
Apply \(\ell_q\) to the identity evaluated at \(X_q\).  The terms with
\(s<q\) vanish by Eq.~\eqref{eq:tri}, and the terms with \(s>q\)
vanish by the induction hypothesis.  Thus
\begin{align}
    \alpha_q\,\ell_q(\Phi_q(X_q)) =0.
\end{align}
Again \(\ell_q(\Phi_q(X_q))\neq0\), so \(\alpha_q=0\).  By descending
induction, all \(\alpha_s\) vanish.  Therefore
$
    \{ \Phi_s \}_{s=0}^m
$
are linearly independent.

Since the equivariant Hom-space has dimension \(m+1\), these \(m+1\)
maps form a basis.  Therefore
\begin{align}
\operatorname{Hom}_{\mtr U(N)}
\bigl(
    \operatorname{End}(\mac H_\eta),
    \operatorname{End}(\mac H_t)
\bigr)
=
\operatorname{span}_{\mathbb C}
\left\{
    \mac E_{s,t}\circ \mac C_{\eta,s}
    :
    0\leq s\leq m
\right\}.    
\end{align}
Since this family is contained in the larger family indexed by
\(0\leq s\leq t\), the larger family also spans.  This proves the claim.
\end{proof}

\subsection{Auxiliary identities and bounds}
\label{app:technical_lemma}

\begin{revision}

We will repeatedly use the following consequence of Schur's lemma
to reduce operator-valued averages to scalar operators.

\begin{lemma}[Schur reduction for invariant operator averages]
\label{lem:schur-invariant-average}
Let $\mathcal G$ be a compact group acting measurably on a probability
space $(\Omega,\mu)$, and let
$
    \purple{\pi}\colon \mathcal G\longrightarrow \mtr{U}(\mathcal K)
$
be a finite-dimensional unitary representation on a complex Hilbert
space $\mathcal K$. Suppose that $\mu$ is $\mathcal G$-invariant, namely,
\begin{align}
    \int_{\Omega} h(u\cdot A)\,d\mu(A)
    =
    \int_{\Omega} h(A)\,d\mu(A)
\end{align}
for every $u\in\mathcal G$ and every integrable function $h$ on
$\Omega$.
Additionally, we assume that 
$
    F\colon\Omega\longrightarrow\operatorname{End}(\mathcal K)
$
is an integrable operator-valued function satisfying the covariance
condition
\begin{align}
    F(u\cdot A)
    =
    \purple{\pi}(u)F(A)\purple{\pi}(u)^\dagger
    \label{eq:operator-valued-covariance}
\end{align}
for every $u\in\mathcal G$ and $A\in\Omega$. 

Then
\begin{align}
    \overline F
    &:=
    \int_{\Omega} F(A)\,d\mu(A) = \mathbb{E}[F]
\end{align}
commutes with $\purple{\pi}(u)$ for every $u\in\mathcal G$. In particular, if
$\purple{\pi}$ is irreducible, then
\begin{align}
    \overline F
    =
    c_F\,\1_{\mathcal K},
    \qquad
    c_F
    =
    \frac{1}{\dim\mathcal K}
    \int_{\Omega}
        \operatorname{tr}_{\mathcal K}[F(A)]
    \,d\mu(A).
    \label{eq:schur-invariant-average}
\end{align}
\end{lemma}

\begin{proof}
For every $u\in\mathcal G$, covariance of $F$ and invariance of $\mu$
give
\begin{align}
    \purple{\pi}(u)\overline F\purple{\pi}(u)^\dagger
    &=
    \int_{\Omega}
        \purple{\pi}(u)F(A)\purple{\pi}(u)^\dagger
    \,d\mu(A)
    \\
    &=
    \int_{\Omega}F(u\cdot A)\,d\mu(A)
    \\
    &=
    \int_{\Omega}F(A)\,d\mu(A)
    =
    \overline F.
\end{align}
Thus $\overline F$ belongs to the commutant of $\purple{\pi}(\mathcal G)$.
If $\purple{\pi}$ is irreducible, Schur's lemma implies that
$\overline F=c_F\1_{\mathcal K}$ for some $c_F\in\mathbb C$.
Taking the trace determines $c_F$ and gives
Eq.~\eqref{eq:schur-invariant-average}. 
\end{proof}

\end{revision}

To prove lemma \ref{lem:reference-overlap-moment}, we first state the following lemma,

\begin{lemma}
\label{lem:technical-binomial-sum}
Let \(\vec a\in\mac S_{N,t}\), and set \(q=|\vec a\cap[\eta]|\).
Then
\begin{align}
    \frac{1}{\binom{N+1}{\eta}}
    \sum_{\substack{\vec b\in\mac S_{N,\eta}\\ \vec a\subseteq \vec b}}
    \frac{\eta+1}{\eta+1-|[\eta]\cap \vec b|}
    =
    \frac{1}{\binom{N+1}{t}}
    \sum_{s=0}^{q}
    \frac{\binom{\eta-s}{t-s}}{\binom{t}{s}}
    \binom{q}{s}.
    \label{eq:combination_identity}
\end{align}
\end{lemma}

\begin{proof}
Put \(x=|(\vec b\setminus \vec a)\cap[\eta]|\). Since
\(|\vec a\cap[\eta]|=q\), any \(\vec b\supseteq \vec a\) with this value of \(x\) satisfies $\abs{\vec b\cap[\eta]}= \abs{\vec a \cap [\eta] } +|(\vec b\setminus \vec a)\cap[\eta]|=q+x$.

We now count such configurations \(\vec b\). Since \(\vec b\) has
\(\eta\) elements and already contains the \(t\) elements of \(\vec a\),
we must add \(\eta-t\) elements from \([N]\setminus \vec a\). These
available elements split into the following two disjoint parts:
\begin{align}
    [\eta]\setminus \vec a,
    \qquad
    [N]\setminus([\eta]\cup \vec a).
\end{align}
Their cardinalities are
\begin{align}
    \abs{[\eta]\setminus \vec a}=\eta-q,
    \qquad
    \abs{[N]\setminus([\eta]\cup \vec a)}
    =
    N-\eta-(t-q).    
\end{align}
If \(x\) of the added elements are chosen from \([\eta]\setminus \vec a\),
then the remaining \(\eta-t-x\) elements must be chosen from
\([N]\setminus([\eta]\cup \vec a)\). Hence
\begin{align}
    \sum_{\substack{\vec b\in\mac S_{N,\eta}\\ \vec a\subseteq \vec b}}
    \frac{\eta+1}{\eta+1-\abs{[\eta]\cap \vec b}}
    &=
    \sum_x
    \binom{\eta-q}{x}
    \binom{N-\eta-t+q}{\eta-t-x}
    \frac{\eta+1}{\eta+1-q-x},
\end{align}
where binomial coefficients outside their natural range are understood to
be zero.
We can rewrite this formula as follows, 
\begin{align}
    \sum_{\substack{\vec b\in\mac S_{N,\eta}\\ \vec a\subseteq \vec b}}
    \frac{\eta+1}{\eta+1-|[\eta]\cap \vec b|}
    &=
    \frac{\eta+1}{\eta+1-q}
    \sum_x
    \binom{\eta+1-q}{x}
    \binom{N-\eta-t+q}{\eta-t-x}  \\
    &=
    \frac{\eta+1}{\eta+1-q}
    \binom{N+1-t}{\eta-t}.
\end{align}
The first identity arises from 
\begin{align}
    \binom{\eta-q}{x}
    \frac{\eta+1}{\eta+1-q-x}
    =
    \frac{\eta+1}{\eta+1-q}
    \binom{\eta+1-q}{x},
\end{align}
and the second equality comes from Vandermonde's identity. 
Therefore,
\begin{align}
    \frac{1}{\binom{N+1}{\eta}}
    \sum_{\substack{\vec b\in\mac S_{N,\eta}\\ \vec a\subseteq \vec b}}
    \frac{\eta+1}{\eta+1-|[\eta]\cap \vec b|} 
    &=\frac{\eta+1}{\eta+1-q}
    \binom{N+1-t}{\eta-t} \binom{N+1}{\eta}^{-1}
    \\
    &=
    \frac{1}{\binom{N+1}{t}}
    \binom{\eta}{t}
    \frac{\eta+1}{\eta+1-q}.
    \label{eq:relationship_formula_2}
\end{align}
Here, we use $
    \binom{N+1-t}{\eta-t} \binom{N+1}{\eta}^{-1}
    =
    \binom{\eta}{t}{\binom{N+1}{t}}^{-1}
$.

On the other hand, the RHS of Eq.\eqref{eq:combination_identity} can be expressed as follows.  Since
$
    {\binom{\eta-s}{t-s}}{\binom{t}{s}}^{-1}
    =
    {\binom{\eta}{t}}{\binom{\eta}{s}}^{-1},
$
we have
\begin{align}
    \sum_{s=0}^{q}
    \frac{\binom{\eta-s}{t-s}}{\binom{t}{s}}
    \binom{q}{s}
    &=
    \binom{\eta}{t}
    \sum_{s=0}^{q}
    \frac{\binom{q}{s}}{\binom{\eta}{s}}.
\end{align}
Similarly, $
    {\binom{q}{s}}{\binom{\eta}{s}}^{-1}
    =
    {\binom{\eta-s}{q-s}}{\binom{\eta}{q}}^{-1}
$ holds, 
so we obtain the following identity by the hockey-stick identity,
\begin{align}
    \sum_{s=0}^{q}
    \frac{\binom{q}{s}}{\binom{\eta}{s}}
    &=
    \frac{1}{\binom{\eta}{q}}
    \sum_{s=0}^{q}
    \binom{\eta-s}{q-s}  \\
    &=
    \frac{\binom{\eta+1}{q}}{\binom{\eta}{q}}
    =
    \frac{\eta+1}{\eta+1-q}.
\end{align}
Consequently,
\begin{align}
    \sum_{s=0}^{q}
    \frac{\binom{\eta-s}{t-s}}{\binom{t}{s}}
    \binom{q}{s}
    =
    \binom{\eta}{t}
    \frac{\eta+1}{\eta+1-q}.
    \label{eq:the_relationship_formula}
\end{align}
Substituting Eq.~\eqref{eq:the_relationship_formula} into Eq.~\eqref{eq:relationship_formula_2} gives the claim.
\end{proof}

Based on that, we will provide the proof of lemma.~\ref{lem:reference-overlap-moment},

\noindent
\textit{Proof of Lemma.}~\ref{lem:reference-overlap-moment} \ 
For \(R\in\mac P_\eta\), the principal minors satisfy the marginal identity
\begin{align}
    \det R_{\vec a,\vec a}
    =
    \sum_{\substack{\vec b\in\mac S_{N,\eta}\\ \vec a\subseteq \vec b}}
    \det R_{\vec b,\vec b}.
    \label{eq:minor-marginal-identity}
\end{align}
This follows directly from the contraction identity. Indeed, by the
definition of the contraction map,
\begin{align}
    \bra{\vec a}
    \mac C_{\eta,t}(R^{\wedge\eta})
    \ket{\vec a}
    &= \sum_{\substack{\vec b\in\mac S_{N,\eta} \\ \vec a\subseteq \vec b}}
    \bra{\vec b}R^{\wedge\eta}\ket{\vec b}.
\end{align}
Using \(\mac C_{\eta,t}(R^{\wedge\eta})=R^{\wedge t}\), $\bra{\vec a} R^{\wedge t} \ket{\vec a} = \det R_{\vec a, \vec a}$ and $\bra{\vec b} R^{\wedge \eta} \ket{\vec b} = \det R_{\vec b, \vec b}$, this becomes Eq.~\eqref{eq:minor-marginal-identity}.

Using Eq.~\eqref{eq:minor-marginal-identity}, we obtain
\begin{align}
    \binom{N}{\eta}
    \int_{\mtr{Gr}_{\eta} (\mathbb{C}^N ) }
    \det R_{[\eta],[\eta]}\,
    \det R_{\vec a,\vec a}\,\dd R 
    &=
    \binom{N}{\eta}
    \sum_{\substack{\vec b\in\mac S_{N,\eta}\\ \vec a\subseteq \vec b}}
    \int_{\mtr{Gr}_{\eta} (\mathbb{C}^N ) }
    \det R_{[\eta],[\eta]}\,
    \det R_{\vec b,\vec b}\,\dd R  .
\end{align}
As discussed in Sec.~\ref{sec:Grassmannian_moment}, we can write $R^{\wedge \eta} $ as follows,
\begin{align}
    R^{\wedge \eta}
    =
    U_\eta(u)^\dagger
    \ketbra{[\eta]}{[\eta]}
    U_\eta(u),
    \qquad
    u\sim \mtr U(N),
\end{align}
Therefore,
\begin{align}
    \det R_{[\eta],[\eta]}
    &=
    \bra{[\eta]}
    R^{\wedge \eta}
    \ket{[\eta]} = \bra{[\eta]}
    U_\eta(u)^\dagger
    \ketbra{[\eta]}{[\eta]}
    U_\eta(u)
    \ket{[\eta]},
    \\
    \det R_{\vec b,\vec b}
    &=
    \bra{\vec b}
    R^{\wedge \eta}
    \ket{\vec b} = \bra{\vec b}
    U_\eta(u)^\dagger
    \ketbra{[\eta]}{[\eta]}
    U_\eta(u) \ket{\vec b}.
\end{align}
Hence the Grassmannian integral can be written as the Haar average
\begin{align}
    \int_{\mtr{Gr}_{\eta} (\mathbb{C}^N ) }
    \det R_{[\eta],[\eta]}\,
    \det R_{\vec b,\vec b}\,\dd R 
    &=
    \int_{\mtr U(N)}
    \bra{[\eta]}
    U_\eta(u)^\dagger
    \ketbra{[\eta]}{[\eta]}
    U_\eta(u)
    \ket{[\eta]}
     \\
    &\qquad\qquad \times
    \bra{\vec b}
    U_\eta(u)^\dagger
    \ketbra{[\eta]}{[\eta]}
    U_\eta(u)
    \ket{\vec b}\, \dd{u} .
\end{align}
Equivalently, this is the diagonal matrix element of the two-fold twirling operator $\mac T_{2,\wedge^\eta \mtr U(N)}$ defined by~\cite{low2022classical}, 
\begin{align}
    \mac T_{2,\wedge^\eta \mtr U(N)}
    :=
    \int_{\mtr U(N)}
    \left(
        U_\eta(u)^\dagger
        \ketbra{[\eta]}{[\eta]}
        U_\eta(u)
    \right)^{\otimes 2}
    \dd u .
\end{align}
Indeed,
\begin{align}
    \int_{\mtr{Gr}_{\eta} (\mathbb{C}^N ) }
    \det R_{[\eta],[\eta]}\,
    \det R_{\vec b,\vec b}\,\dd R
    &=
    \bra{[\eta]\otimes \vec b}
    \mac T_{2,\wedge^\eta \mtr U(N)}
    \ket{[\eta]\otimes \vec b}.
    \label{eq:two-fold-twirl-diagonal-entry}
\end{align}
From Theorem 8 in Ref.~\cite{low2022classical},
the evaluation of the diagonal part of this two-fold twirling operator gives
\begin{align}
    \bra{\vec p_1\otimes \vec p_2}
    \mac T_{2,\wedge^\eta \mtr U(N)}
    \ket{\vec p_1\otimes \vec p_2}
    =
    f\!\left(|\vec p_1\cap \vec p_2|\right),
\end{align}
where
\begin{align}
    f(m)
    =
    \frac{1}{
        \binom{N+1}{\eta}
        \binom{N}{\eta}
    }
    \frac{\eta+1}{\eta+1-m}.
\end{align}
Applying this formula with
\(\vec p_1=[\eta]\) and \(\vec p_2=\vec b\), we obtain
\begin{align}
    \int_{\mtr{Gr}_{\eta} (\mathbb{C}^N ) }
    \det R_{[\eta],[\eta]}\,
    \det R_{\vec b,\vec b}\,\dd R 
    &=
    f\!\left(|[\eta]\cap \vec b|\right)
    \\
    &=
    \frac{1}{
        \binom{N+1}{\eta}
        \binom{N}{\eta}
    }
    \frac{\eta+1}{\eta+1-|[\eta]\cap \vec b|}.
\end{align}
Substituting this into the previous expression gives
\begin{align}
    \binom{N}{\eta}
    \int_{\mtr{Gr}_{\eta} (\mathbb{C}^N ) }
    \det R_{[\eta],[\eta]}\,
    \det R_{\vec a,\vec a}\,\dd R 
    &=
    \frac{1}{\binom{N+1}{\eta}}
    \sum_{\substack{\vec b\in\mac S_{N,\eta}\\ \vec a\subseteq \vec b}}
    \frac{\eta+1}{\eta+1-|[\eta]\cap\vec b|}.
\end{align}
Applying Lemma~\ref{lem:technical-binomial-sum}, with
\(q=|\vec a\cap[\eta]|\), we obtain
\begin{align}
    \binom{N}{\eta}
    \int_{\mtr{Gr}_{\eta} (\mathbb{C}^N ) }
    \det R_{[\eta],[\eta]}\,
    \det R_{\vec a,\vec a}\,\dd R 
    =
    \frac{1}{\binom{N+1}{t}}
    \sum_{s=0}^{q}
    \frac{\binom{\eta-s}{t-s}}{\binom{t}{s}}
    \binom{q}{s}.
\end{align}
This proves the claim.
\qed
\\

\begin{lemma}
    \label{lem:trace_occupied_operator}
    Let $\eta , N$ be integers satisfying $\eta \leq N$. We assume that \(R\in \mtr{Gr}_\eta (\mathbb{C}^N )\) and that 
    \(T=\{t_1<\cdots<t_\ell\}\subseteq [N]\) with \(\ell\leq \eta\).
    Then,
    \begin{align}
        \tr_{\mac H_\eta}
        \left[
            \left(\prod_{j\in T} n_j\right) R^{\wedge \eta}
        \right]
        =
        \det R_{T,T},
    \end{align}
    where $n_j \coloneq a_j^\dagger a_j \eval_{\mac{H}_\eta} $ denotes the number operator on mode $j$ and 
    \(R_{T,T}\) denotes the submatrix of \(R\) whose rows and
    columns are indexed by \(T\).
\end{lemma}

\begin{proof}
    For simplicity, we denote $
        \ket{T}
        \coloneqq
        \ket{t_1}\wedge\cdots\wedge\ket{t_\ell}
        \in \mac H_\ell .$
    By the definition of the extension map \(\mac E_{\ell,\eta}\), we have
    \begin{align}
        \mac E_{\ell,\eta}\!\left(\ketbra{T}{T}\right)
        =
        a_{t_1}^\dagger\cdots a_{t_\ell}^\dagger
        a_{t_\ell}\cdots a_{t_1}
        \big|_{\mac H_\eta}.
    \end{align}
    Since the modes in \(T\) are distinct, this operator is exactly the
    projector onto the subspace in which all modes in \(T\) are occupied.
    Therefore, we can express
    $
        \mac E_{\ell,\eta}\!\left(\ketbra{T}{T}\right)
        =
        \prod_{a\in T}n_a .
    $
    From this fact, by using the trace-duality between \(\mac E_{\ell,\eta}\) and
    \(\mac C_{\eta,\ell}\), we obtain
    \begin{align}
        \tr_{\mac H_\eta}
        \left[
            \left(\prod_{a\in T}n_a\right)R^{\wedge\eta}
        \right]
        &=
        \tr_{\mac H_\eta}
        \left[
            R^{\wedge\eta}
            \mac E_{\ell,\eta}\!\left(\ketbra{T}{T}\right)
        \right] \\
        &=
        \tr_{\mac H_\ell}
        \left[
            \purple{\ketbra{T}{T}}
            \mac C_{\eta,\ell}\!\left(R^{\wedge\eta}\right)
        \right].
    \end{align}
    \purple{This becomes}
    \begin{align}
        \tr_{\mac H_\eta}
        \left[
            \left(\prod_{a\in T}n_a\right)R^{\wedge\eta}
        \right]
        =
        \bra{T}
            \mac C_{\eta,\ell}\!\left(R^{\wedge\eta}\right)
        \ket{T}.
    \end{align}

It remains to identify
\(\mac C_{\eta,\ell}(R^{\wedge\eta})\).  By
Lemma~\ref{lem:contraction-exterior-projection}, we have
$
    \mac C_{\eta,\ell}(R^{\wedge\eta})
    =
    R^{\wedge\ell}.
$
We now evaluate the matrix elements of \(R^{\wedge\ell}\). For \(\vec b=(b_1<\cdots<b_\ell)\), we have
\begin{align}
    R^{\wedge \ell}\ket{\vec b}
    &=
    R\ket{b_1}\wedge\cdots\wedge R\ket{b_\ell} \\
    &=
    \left(\sum_{i_1=1}^N R_{i_1 b_1}\ket{i_1}\right)
    \wedge\cdots\wedge
    \left(\sum_{i_\ell=1}^N R_{i_\ell b_\ell}\ket{i_\ell}\right) \\
    &=
    \sum_{i_1,\ldots,i_\ell=1}^N
    \left(\prod_{\beta=1}^{\ell} R_{i_\beta b_\beta}\right)
    \ket{i_1}\wedge\cdots\wedge \ket{i_\ell}.
\end{align}
By antisymmetry, the coefficient of
\(\ket{\vec a}=\ket{a_1}\wedge\cdots\wedge\ket{a_\ell}\) is obtained by
summing over all permutations of \((a_1,\ldots,a_\ell)\). Hence
\begin{align}
    R^{\wedge \ell}\ket{\vec b}
    &=
    \sum_{\vec a\in \mac S_{N,\ell}}
    \left[
        \sum_{\sigma\in \mathfrak S_\ell}
        \operatorname{sgn}(\sigma)
        \prod_{\beta=1}^{\ell}
        R_{a_{\sigma(\beta)} b_\beta}
    \right]
    \ket{\vec a} \\
    &=
    \sum_{\vec a\in \mac S_{N,\ell}}
    \det R_{\vec a,\vec b}\,
    \ket{\vec a}.
\end{align}
Therefore, as an operator on \(\mac H_\ell\),
\begin{align}
    R^{\wedge \ell}
    &=
    \sum_{\vec a,\vec b\in \mac S_{N,\ell}}
    \det R_{\vec a,\vec b}\,
    \ketbra{\vec a}{\vec b}.
\end{align}
In particular, taking \(\vec a=\vec b=T\), we obtain
\begin{align}
    \bra{T}R^{\wedge\ell}\ket{T}
    =
    \det R_{T,T}.
\end{align}
This proves the claim.

\end{proof}


\begin{lemma}
\label{lem:entry-locality}
Let $\vec p, \vec q \in \mac{S}_{N,k}$ and denote
$
\vec s:=\vec p\cup \vec q 
$.
For $0\le t\le k$ and $R \in \mtr{End}(\mathbb{C}^N  ) $, the following identity holds,
\begin{align}
    \bra{\vec p}
\mac{E}_{t, k}(R^{\wedge t})
\ket{\vec q}
=
\sum_{\substack{
\vec  r \subset \vec p\cap \vec q\\
|\vec r|=k-t
}}
\purple{
\varepsilon_{\wedge}(\vec p\setminus\vec r,\vec r)
\varepsilon_{\wedge}(\vec q\setminus\vec r,\vec r)
}
\det R_{\vec p\setminus \vec r,\ \vec q\setminus \vec r},
\end{align}
In particular,
$
\langle \vec p|
\mac{E}_{t,k}(R^{\wedge t})
\ket{\vec q}
$
depends only on the matrix elements $R_{ij}$ with
$i,j\in \vec{s}$.
\end{lemma}

\begin{proof}
Following the proof of lemma.~\ref{lem:trace_occupied_operator}, we will expand $R^{\wedge t}$ in the occupation basis of $\mac{H}_t$,
\begin{align}
R^{\wedge t}
=
\sum_{\vec a,\vec b\in \mac{S}_{N,t}}
\det R_{\vec a,\vec b}
\ketbra{\vec a}{\vec b}.
\end{align}


By applying $\mac{E}_{t, k}$, the term
$|\vec a\rangle\langle \vec b|$ becomes the $t$-body monomial
$a^\dagger_{\vec a}a_{\vec b}$ acting on $\mathcal H_k$.
For the matrix element
$\langle \vec p|a^\dagger_{\vec a}a_{\vec b}|\vec q\rangle$
to be nonzero, the annihilation part must remove $t$ occupied
modes from $\vec q$, and the creation part must produce the
occupied modes of $\vec p$. Hence, the modes that are not acted
on must be the same on the bra $\vec p$ and ket $\vec q$ sides. Therefore, when we denote
$
\vec r\subset \vec p\cap\vec q, |\vec r|=k-t
$,
the allowed modes are expressed as 
\begin{align}
    \vec a=\vec p\setminus \vec r,
\qquad
\vec b=\vec q\setminus \vec r.
\end{align}
Thus
\begin{align}
\left\langle \vec p\middle|
\mac{E}_{t, k}(R^{\wedge t})
\middle|\vec q\right\rangle
=
\sum_{\substack{
\vec r\subset \vec p\cap \vec q\\
|\vec r|=k-t
}}
\purple{
\varepsilon_{\wedge}(\vec p\setminus\vec r,\vec r)
\varepsilon_{\wedge}(\vec q\setminus\vec r,\vec r)
}\,
\det R_{\vec p\setminus \vec r,\vec q\setminus \vec r},
\end{align}
where $\purple{
\varepsilon_{\wedge}(\vec p\setminus\vec r,\vec r)
\varepsilon_{\wedge}(\vec q\setminus\vec r,\vec r)
}$ is the fermionic
ordering sign.
\end{proof}

\begin{revision}

Next, we briefly explain the relationship between random variables $R$ on Grasmann manifold $\mtr{Gr}_{\eta}(\mathbb{C}^N)$ and a Ginibre matrix $G$. To this end, we write $\mtr{St}(N,\eta):=\{L\in\mathbb C^{N\times\eta}:
L^\dagger L=\1_\eta\}$ for the complex Stiefel manifold.
For an integer $m\geq\eta$, $\mtr{CW}_\eta(m,\1_\eta)$ denotes
the complex Wishart distribution, defined as the law of $X^\dagger X$,
where $X\in\mathbb C^{m\times\eta}$ is a standard complex Ginibre matrix.
For further details on complex Wishart distributions and invariant
measures on complex Stiefel manifolds, see
Refs.~\cite[Chap.~3]{forrester2010loggases}
and~\cite[Sec.~2]{forrester2020matrix}, respectively.

\begin{lemma}
\label{lem:gammaS-polar}
Let $N,\eta$ be integers satisfying $1\leq\eta\leq N$, and fix
$\vec s\in\mac S_{N,\eta}$.
Let $G\in\mathbb C^{N\times\eta}$ be distributed according to the
determinant-tilted measure $\gamma_{\vec s}$ defined in
Eq.~\eqref{eq:tilted-ginibre-measure}.
Since $G$ has full column rank almost surely, write its polar
decomposition as
\begin{align}
    G=LT^{1/2},
    \quad 
    L^\dagger L=\1_\eta,
    \quad 
    T=G^\dagger G,
\end{align}
where $L\in\mtr{St}(N,\eta)$ and $T$ is positive definite.
Define the associated rank-$\eta$ orthogonal projector by
\begin{align}
    R:=LL^\dagger.
\end{align}

Under $\gamma_{\vec s}$, the random variables $R$ and $T$ are
independent.  Their marginal distributions are
\begin{align}
    R
    \sim
    \nu_{\vec s}(\dd R), \quad 
    \nu_{\vec s}(\dd R)
    &:=
    \binom{N}{\eta}
    \det(R_{\vec s,\vec s})\,\dd R,
    \label{eq:R-marginal-law}
    \\
    T
    \sim
    \mtr{CW}_\eta(N+1,\1_\eta),
    \label{eq:shifted-wishart}
\end{align}
where $\dd R$ denotes the normalized invariant measure on
$\mtr{Gr}_\eta(\mathbb C^N)$.  In particular, the law of $R$
coincides with the orbital-shadow measurement distribution associated
with the input state $\ketbra{\vec s}$.

Moreover, for every integer $0\leq t\leq\eta$,
\begin{align}
    \mathbb E_{\vec s}\!\left[
        T^{\wedge t}
    \right]
    =
    \fall{N+1}{t}\,
    \1_{\wedge^t\mathbb C^\eta},
    \label{eq:wishart-wedge-moment}
\end{align}
where $\fall{n}{0}:=1$ and
\begin{align}
    \fall{n}{t}
    :=
    n(n-1)\cdots(n-t+1)
\end{align}
for $t\geq1$.
\end{lemma}

\begin{proof}
To determine the joint distribution of the polar factors $L$ and $T$,
we use the change-of-variables formula for the complex rectangular
polar decomposition:
\begin{align}
    \dd G
    =
    \frac{\pi^{N \eta}}{ \pi^{\eta(\eta-1)/ 2} \prod^\eta_{j=1} (N-j)! }
    (\det T)^{N-\eta}\,
    \dd T\,\dd L,
    \label{eq:complex-polar-jacobian}
\end{align}
where $\dd L$ is the normalized invariant measure on the complex Stiefel manifold $\mtr{St}(N,\eta)$ and $\dd T$ is the Lebesgue measure on the cone of positive-definite Hermitian matrices.  This formula follows from the complex case ($\beta=2$) of Ref.~\cite[Proposition~4]{diazgarcia2011wishart}, after normalizing the Stiefel measure. The normalized form used here is stated explicitly in Ref.~\cite[Lemma~3]{wroblewska2025extensions}, with $m=N$ and $r=\eta$.

Combining the expression for the Ginibre measure in
Eq.~\eqref{eq:ginibre-measure} with the change-of-variables formula in
Eq.~\eqref{eq:complex-polar-jacobian}, we express $\gamma_0$ in terms
of $\dd L$ and $\dd T$ as follows:
\begin{align}
    \gamma_0(\dd L,\dd T)
    &\propto
    \dd L\,
    e^{-\tr T}
    (\det T)^{N-\eta}\dd T.
\end{align}
Hence, under $\gamma_0$, the variable $L$ is Haar distributed on
$\mtr{St}(N,\eta)$, the variable $T$ has distribution
$\mtr{CW}_\eta(N,\1_\eta)$, and the two variables are independent.

The determinant appearing in the tilt factorizes as
\begin{align}
    |\det G_{\vec s}|^2
    =
    |\det L_{\vec s}|^2\det T.
    \label{eq:tilt-factorization}
\end{align}
Consequently, the tilted joint density has the product form
\begin{align}
    \gamma_{\vec s}(\dd L,\dd T)
    \propto
    |\det L_{\vec s}|^2\dd L\,
    e^{-\tr T}
    (\det T)^{N-\eta+1}\dd T.
    \label{eq:tilted-product-density}
\end{align}
Thus, $L$ and $T$ remain independent under
$\gamma_{\vec s}$.  The additional factor $\det T$ changes the exponent of $\det T$
from $N-\eta$ to $(N+1)-\eta$.
Hence, the marginal distribution of $T$ is
$\mtr{CW}_\eta(N+1,\1_\eta)$, proving
Eq.~\eqref{eq:shifted-wishart}.

We next determine the angular normalization.  
By the Cauchy--Binet identity,
\begin{align}
    \sum_{\vec u\in\mac S_{N,\eta}}
    |\det L_{\vec u}|^2
    = \det(L^\dagger L)
    = 1.
\end{align}
Since $\dd L$ is invariant under row permutations, all
$\binom{N}{\eta}$ summands have the same expectation. Hence,
\begin{align}
    \int \dd L \!\left[|\det L_{\vec s}|^2\right]
    = \binom{N}{\eta}^{-1}.
    \label{eq:angular-normalization}
\end{align}
The marginal law of $L$ under $\gamma_{\vec s}$ is therefore
\begin{align}
    \binom{N}{\eta}|\det L_{\vec s}|^2\,\dd L.
    \label{eq:tilted-angular-law}
\end{align}
Since $\det(R_{\vec s,\vec s})=|\det L_{\vec s}|^2$
and the map $L\mapsto LL^\dagger$ sends $\dd L$ to $\dd R$,
Eq.~\eqref{eq:R-marginal-law} follows.

Finally, let $1\leq t\leq\eta$.
The probablity measure of $T$ is invariant under unitary conjugation, and
the representation of $\mtr U(\eta)$ on
$\bigwedge^t\mathbb C^\eta$ is irreducible.
Schur's lemma therefore gives
\begin{align}
    \mathbb E_{\vec s}\!\left[T^{\wedge t}\right]
    = c_t\,\1_{\wedge^t\mathbb C^\eta}.
\end{align}
For any $\vec a\in\mac S_{\eta,t}$, the principal submatrix
$T_{\vec a,\vec a}$ has distribution
$\mtr{CW}_t(N+1,\1_t)$.
Evaluating the corresponding diagonal matrix element and using
the standard Wishart determinant moment
\cite[Chap.~3]{forrester2010loggases}, we obtain
\begin{align}
    c_t
    = \mathbb E_{\vec s}\!\left[\det(T_{\vec a,\vec a})\right]
    = \fall{N+1}{t}.
\end{align}
The case $t=0$ is immediate, completing the proof of
Eq.~\eqref{eq:wishart-wedge-moment}.
\end{proof}

\end{revision}

\begin{revision}

\subsection{Gaussian moment bound}
\label{sec:gaussian-moment-bound}

In this section, we aim to show the following important lemma with Wick ordering, Wiener-Ito chaos, and Pieri rule. 
Throughout this section, we restrict to the nonredundant regime
\begin{align}
    1\leq k
    \leq
    \min\{\eta,N-\eta\},
    \label{eq:nonredundant-k-regime}
\end{align}
or equivalently,
\(
k\leq\eta
\)
and
\(
N\geq\eta+k
\).
Indeed, when \(k>N-\eta\), particle--hole duality expresses the
\(k\)-particle correlations in terms of hole correlations of order at
most \(N-\eta\), so they do not constitute a separate estimation
problem. Under this assumption, our goal is to show the following lemma,

\begin{lemma}
    Let $k ,\eta , N$ be integers satisfying $1 \leq k \leq \min(\eta, N-\eta)$ and $\eta < N$. We assume $G \in \mathbb{C}^{N \times \eta}$ be a standard complex Ginibre matrix; namely, its entries satisfy $\mathbb{E}_0 [G_{ia} ] = 0, \ \mathbb{E}_0 [G_{ia} \overline{G}_{jb}] = \delta_{ij} \delta_{ab}$ for any $i, j \in [N]$ and $a,b \in [\eta]$ under the probability measure $\gamma_0$ defined by Eq.~\eqref{eq:ginibre-measure}. Additionally, for $\vec s \in \mac{S}_{N,\eta}$, let $G_{\vec s}$ be the square submatrix construvted by the rows indexed by $\vec s$. 
    
    Then, for every $\vec s$ and every $\vec q, \vec p \in \mac{S}_{N,k}$, the folloiwng inequality holds,
    \begin{align}
        \frac1{\eta!} \mathbb{E}_0 \qty[ \abs{\det G_{\vec s}}^2
        \left|
            \purple{{}^{k}\mathfrak D_{\eta;\vec q}^{\vec p}}
            (GG^\dagger)
        \right|^2 ] \leq \frac{k+1}{k!} \eta^k, \quad (k \geq 2).
    \end{align}
\end{lemma}

We remark that the restriction $k\geq2$ is imposed to obtain the sharper prefactor
$(k+1)/k!$ in the stated bound.

The proof proceeds in five steps. First, we review the
finite-dimensional complex Gaussian chaos and establish the Wick-product,
factorization, and minor-orthogonality identities required in the
subsequent analysis. Second, we express the estimator as a Wick-ordered
exterior-power matrix element and thereby identify its complex
Wiener--It\^o chaos. Third, we decompose the determinant-weighted
estimator into \(k+1\) mutually orthogonal components. Fourth, we identify the exterior Pieri
component supporting each contraction channel and compute the diagonal
matrix elements of the corresponding orthogonal projector. Finally,
unitary invariance and Schur's lemma reduce the second moment of each
channel to these projector diagonals; determining the remaining
normalization and summing over all contraction orders yield the stated
bound.

\subsubsection{Finite-dimensional complex Gaussian chaos}
\label{sec:preliminary_complex_gaussian}

In this section, we review and organize the complex Hermite and Wick
calculus used in the subsequent moment analysis. We first summarize the
tensorized complex Hermite basis, the associated Hermite-sector
decomposition, and Wick ordering. We then show two finite-dimensional
consequences of standard identities for complex multiple Wiener--It\^o
integrals~\cite{ito1952complex,chen2016product,chen2019complex}.
These results provide the analytic tools used in the subsequent
subsections.
We also remark that since all Gaussian variables considered here are the finitely many
matrix entries of \(G\), we work directly with tensorized complex
Hermite polynomials and coordinate derivatives rather than with the
kernel representation \(I_{\alpha,\beta}(f)\) used in Refs.~\cite{ito1952complex,chen2016product,chen2019complex}. 

We begin with one standard circular complex Gaussain variable \(Z\). The standard probability measure is defined by
\begin{align}
    \gamma(\dd z) \coloneq \pi^{-1} e^{-\abs{z}^2} \dd^2 z.
\end{align}
Under this probability distribution, \(\purple{\mathbb{E}_0}[Z] = \purple{\mathbb{E}_0}[Z^2] = \purple{\mathbb{E}_0}[\overline{Z}^{\,2}] = 0\) and \(\purple{\mathbb{E}_0}[\abs{Z}^2] = 1\). Following Refs.~\cite{ito1952complex,ismail2015complex}, for non-negative integers $m,n$, we introduce the complex Hermite polynomials \(H_{m,n} : \mathbb{C} \to \mathbb{C}\), defined by
\begin{align}
    H_{m,n}(z,\overline{z})
    \coloneq
    \sum_{j=0}^{\min(m,n)}
    (-1)^j j!
    \binom{m}{j}
    \binom{n}{j}
    z^{m-j}\overline{z}^{\,n-j}.
\end{align}
For instance, \(H_{1,0} = z, H_{0,1} = \overline{z}, H_{1,1} = \abs{z}^2 - 1, H_{2,1} = z^2 \overline{z} - 2z\). As stated in Refs~\cite{ito1952complex,ismail2015complex}, Their generating function is given by
\begin{align}
    e^{\overline{\lambda}z+\lambda\overline{z}-\abs{\lambda}^2}
    =
    \sum_{m,n\geq 0}
    \frac{\lambda^n\overline{\lambda}^{\,m}}{n!m!}
    H_{m,n}(z,\overline{z}).
\end{align}
Based on this function, taking the expectation of this generating function times the complex conjugate of a second immediately yields
\begin{align}
    \purple{\mathbb{E}_0}
    \left[
        H_{m,n}(Z,\overline{Z})
        \overline{
            H_{m',n'}(Z,\overline{Z})
        }
    \right]
    =
    m!n!\,
    \delta_{m,m'}\delta_{n,n'}.
\end{align}


Next, we extend these complex Hermite polynomials to the multivariate
setting~\cite{ito1952complex,chen2019complex}. In what follows, we regard
the \(N\eta\) entries of \(G\in\mathbb{C}^{N\times\eta}\) as independent standard
complex Gaussian variables: using row-major ordering, we can write
\begin{align}
    \mathbf Z
    \coloneq
    \bigl(
        G_{1,1},G_{1,2},\ldots,G_{1,\eta},
        G_{2,1},\ldots,G_{N,\eta}
    \bigr)
    =
    \bigl(Z_1,Z_2,\ldots,Z_{N\eta}\bigr),
\end{align}
where \(Z_1,\ldots,Z_{N\eta}\) are independent and identically distributed
according to \(\gamma(\dd z_i)\).
For multi-indices
\(
    \mathbf m,\mathbf n\in\mathbb N_0^{N\eta},
\)
we define the tensorized complex Hermite polynomial
\begin{align}
    \mathbf H_{\mathbf m,\mathbf n}
    (\mathbf Z,\overline{\mathbf Z})
    \coloneq
    \prod_{j=1}^{N\eta}
    H_{m_j,n_j}(Z_j,\overline Z_j).
    \label{eq:tensorized-complex-Hermite}
\end{align}
By the independence of \(Z_1,\ldots,Z_{N\eta}\) and the one-variable
orthogonality relation, these polynomials satisfy
\begin{align}
    \purple{\mathbb{E}_0}\!\left[
        \mathbf H_{\mathbf m,\mathbf n}
        (\mathbf Z,\overline{\mathbf Z})
        \overline{
            \mathbf H_{\mathbf m',\mathbf n'}
            (\mathbf Z,\overline{\mathbf Z})
        }
    \right]
    =
    \mathbf m!\,\mathbf n!\,
    \delta_{\mathbf m,\mathbf m'}
    \delta_{\mathbf n,\mathbf n'},
    \label{eq:tensorized-Hermite-orthogonality}
\end{align}
where
\(
    \mathbf m!\coloneq\prod_{j=1}^{N\eta}m_j!
\)
and similarly for \(\mathbf n!\).
This is the finite-dimensional tensor-product Hermite basis described in
Refs.~\cite[Theorems~13.1 and~13.2]{ito1952complex}
and~\cite[Proposition~2.3]{chen2019complex}.
For nonnegative integers \(\alpha\) and \(\beta\), we then define the
Hermite subspace of bidegree \((\alpha,\beta)\) by
\begin{align}
    \mathcal W_{\alpha,\beta}
    \coloneq
    \operatorname{span}
    \left\{
        \mathbf H_{\mathbf m,\mathbf n}
        \,\middle|\,
        \abs{\mathbf m}=\alpha,\ 
        \abs{\mathbf n}=\beta
    \right\},
    \label{eq:finite-dimensional-complex-chaos}
\end{align}
where \(\abs{\mathbf m} \coloneq \sum_{j=1}^{N\eta}m_j\).
This subspace is commonly referred to as
the \((\alpha,\beta)\)-th Wiener--It\^o chaos.

Next, we introduce Wick ordering to rewrite the target quantities
\(\det G_{\vec s}\) and
\(\purple{{}^{k}\mathfrak D_{\vec q}^{\vec p}}(GG^\dagger)\)
in the orthogonal basis of tensorized complex Hermite polynomials.
Let \(P(\mathbf Z,\overline{\mathbf Z})\) be a polynomial such that every monomial
appearing in \(P\) has total degree \(\alpha\) in
\(Z_1,\ldots,Z_{N\eta}\) and total degree \(\beta\) in
\(\overline Z_1,\ldots,\overline Z_{N\eta}\): more explicitly
\begin{align}
    P(\mathbf Z,\overline{\mathbf Z})
    &=
    \sum_{\substack{
        \boldsymbol\mu,\boldsymbol\nu\in\mathbb N_0^{N\eta}\\
        |\boldsymbol\mu|=\alpha,\;
        |\boldsymbol\nu|=\beta
    }}
    p_{\boldsymbol\mu,\boldsymbol\nu}\,
    \mathbf Z^{\boldsymbol\mu}
    \overline{\mathbf Z}^{\boldsymbol\nu},
    \label{eq:bihomogeneous-polynomial}
\end{align}
where we denote
\(
    \mathbf Z^{\boldsymbol\mu}
    :=
    \prod_{j=1}^{N\eta}Z_j^{\mu_j},
    \
    \overline{\mathbf Z}^{\boldsymbol\nu}
    :=
    \prod_{j=1}^{N\eta}\overline Z_j^{\,\nu_j}.
\)
Here, the Wick ordering of \(P\), denoted by \( :\!P\!:\), is defined as
its orthogonal projection onto the corresponding Hermite sector
\(\mathcal W_{\alpha,\beta}\).
Equivalently, Wick ordering introduces the lower-degree counterterms
required to cancel all internal Gaussian contractions.  

To compute this projection explicitly, we regard
\(Z_j\) and \(\overline Z_j\) as algebraically independent formal
variables. Then Wick ordering is given by
\begin{align}
    :\!P(\mathbf Z,\overline{\mathbf Z})\!:
    &=
    \exp\!\left(
        -\sum_{j=1}^{N\eta}
        \partial_{Z_j}\partial_{\overline Z_j}
    \right)
    P(\mathbf Z,\overline{\mathbf Z})
    \\
    &=
    \sum_{r=0}^{\min(\alpha,\beta)}
    \frac{(-1)^r}{r!}
    \left(
        \sum_{j=1}^{N\eta}
        \partial_{Z_j}\partial_{\overline Z_j}
    \right)^r
    P(\mathbf Z,\overline{\mathbf Z}).
    \label{eq:wick-ordering-differential}
\end{align}
Here, the coordinate derivatives are characterized by
\begin{align}
    \partial_{Z_j}Z_k
    =
    \partial_{\overline Z_j}\overline Z_k
    =
    \delta_{j,k},
    \
    \partial_{Z_j}\overline Z_k
    =
    \partial_{\overline Z_j}Z_k
    =
    0,
\end{align}
and are extended to arbitrary polynomials by linearity and the
Leibniz rule.
The series terminates because each application of the differential
operator lowers the degrees in \(Z\) and \(\overline Z\) by one.

To verify Eq.~\eqref{eq:wick-ordering-differential}, recall that, for
one complex variable and with the present normalization,
\begin{align}
    \exp\!\left(
        -\partial_z\partial_{\overline z}
    \right)
    z^m\overline z^{\,n}
    =
    H_{m,n}(z,\overline z).
    \label{eq:one-variable-Hermite-operational}
\end{align}
Since the differential operators associated with distinct coordinates
commute, this identity tensorizes directly. Namely, for
\(\boldsymbol\mu,\boldsymbol\nu\in\mathbb N_0^{N\eta}\),
\begin{align}
    \exp\!\left(
        -\sum_{j=1}^{N\eta}
        \partial_{Z_j}\partial_{\overline Z_j}
    \right)
    \mathbf Z^{\boldsymbol\mu}
    \overline{\mathbf Z}^{\boldsymbol\nu}
    &=
    \prod_{j=1}^{N\eta}
    \exp\!\left(
        -\partial_{Z_j}\partial_{\overline Z_j}
    \right)
    Z_j^{\mu_j}\overline Z_j^{\,\nu_j}
\\
    &=
    \prod_{j=1}^{N\eta}
    H_{\mu_j,\nu_j}(Z_j,\overline Z_j)
    =
    \mathbf H_{\boldsymbol\mu,\boldsymbol\nu}
    (\mathbf Z,\overline{\mathbf Z}).
    \label{eq:multivariate-Hermite-operational}
\end{align}
The difference between
\(\mathbf Z^{\boldsymbol\mu}\overline{\mathbf Z}^{\boldsymbol\nu}\) and
\(\mathbf H_{\boldsymbol\mu,\boldsymbol\nu}
(\mathbf Z,\overline{\mathbf Z})\) consists only of terms
belonging to the lower Hermite sectors
$ \mathcal W_{
        |\boldsymbol\mu|-r,\,
        |\boldsymbol\nu|-r
    }$ for $ r\geq 1$ .
Consequently,
\(\mathbf H_{\boldsymbol\mu,\boldsymbol\nu}
(\mathbf Z,\overline{\mathbf Z})\) is precisely the
orthogonal projection of
\(\mathbf Z^{\boldsymbol\mu}\overline{\mathbf Z}^{\boldsymbol\nu}\) onto
\(\mathcal W_{|\boldsymbol\mu|,|\boldsymbol\nu|}\).
Equation~\eqref{eq:wick-ordering-differential} then follows by
linearity.

For a monomial, this operation admits the following expansion
\begin{align}
    &:\!
        Z_{i_1}\cdots Z_{i_\alpha}
        \overline Z_{j_1}\cdots\overline Z_{j_\beta}:
     =
    \sum_{\pi\in\operatorname{PM}(\alpha,\beta)}
    (-1)^{\abs{\pi}}
    \prod_{(r,s)\in\pi}
        \delta_{i_r,j_s}
    \prod_{r\notin L(\pi)}Z_{i_r}
    \prod_{s\notin R(\pi)}\overline Z_{j_s},
    \label{eq:wick-ordering-pairing}
\end{align}
where \(\operatorname{PM}(\alpha,\beta)\) denotes the set of partial
matchings between the \(\alpha\) nonconjugate and \(\beta\) conjugate
factors.  Here \(L(\pi)\) and \(R(\pi)\) denote the sets of matched
positions on the two sides, respectively.  Each matched pair contributes
the Gaussian contraction
$
    \mathbb{E}_0[Z_i\overline Z_j]=\delta_{i,j}.
$
For example,
\begin{align}
    : Z_i \overline{Z_j} : = Z_i \overline{Z_j} - \delta_{i,j}
\end{align}
\begin{align}
    :\!Z_iZ_k\overline Z_j\!:
    &=
    Z_iZ_k\overline Z_j
    -
    \wick{\c1{Z_i}Z_k\c1{\overline Z_j}}
    -
    \wick{Z_i\c1{Z_k}\c1{\overline Z_j}}
    \\
    &=
    Z_iZ_k\overline Z_j
    -\delta_{i,j}Z_k
    -\delta_{k,j}Z_i.
    \label{eq:wick-order-three}
\end{align}
\begin{align}
    :\!Z_iZ_k\overline Z_j\overline Z_\ell\!:
    &=
    Z_iZ_k\overline Z_j\overline Z_\ell
    -
    \wick{\c1{Z_i}Z_k\c1{\overline Z_j}\overline Z_\ell}
    -
    \wick{\c1{Z_i}Z_k\overline Z_j\c1{\overline Z_\ell}}
    -
    \wick{Z_i\c1{Z_k}\c1{\overline Z_j}\overline Z_\ell}
    -
    \wick{Z_i\c1{Z_k}\overline Z_j\c1{\overline Z_\ell}}
    +
    \wick{\c1{Z_i}\c2{Z_k}
          \c1{\overline Z_j}\c2{\overline Z_\ell}}
    +
    \wick{\c1{Z_i}\c2{Z_k}
          \c2{\overline Z_j}\c1{\overline Z_\ell}}
    \nonumber\\
    &=
    Z_iZ_k\overline Z_j\overline Z_\ell
    -\delta_{i,j}Z_k\overline Z_\ell
    -\delta_{i,\ell}Z_k\overline Z_j
    -\delta_{k,j}Z_i\overline Z_\ell
    -\delta_{k,\ell}Z_i\overline Z_j
    +\delta_{i,j}\delta_{k,\ell}
    +\delta_{i,\ell}\delta_{k,j}.
    \label{eq:wick-order-four}
\end{align}


This Wick ordering formulation assembles standard facts that are commonly
presented in different frameworks.  The interpretation of Wick
ordering as an orthogonal projection follows from the complex
Wiener--It\^o decomposition; see, for example,
Ref.~\cite{ito1952complex}.  The one-variable operational identity
\eqref{eq:one-variable-Hermite-operational} is a standard
representation of the complex Hermite polynomials; see
Ref.~\cite[Prop.~2.3]{liu2020complex}.
Equation~\eqref{eq:multivariate-Hermite-operational} is its direct
multivariate tensorization, included here to make the correspondence
with Wick ordering explicit.

Next, we will show the relevant lemmas to bound the target qunatites. The following two lemmas are finite-dimensional coordinate
consequences of standard identities for complex multiple
Wiener--It\^o integrals.  The first follows by specializing the
complex product formula
\cite[Theorem~3.2, Eq.~(3.21)]{chen2016product};
see also \cite[Theorem~2.13, Eq.~(2.14)]{chen2019complex}.
The second follows from the complex Wiener--It\^o isometry
\cite[Theorem~7, properties (I.4)--(I.6)]{ito1952complex}.

\begin{lemma}
\label{lem:one-sided-Wick-product}
Let $\eta, N$ be integers satisfying $1 \leq \eta < N$ and $a,b,c$ be positive integers, and suppose that
\begin{align}
    F_a(\mathbf Z)
    =
    \sum_{\substack{\boldsymbol\mu\in\mathbb N_0^{N\eta}\\|\boldsymbol\mu|=a}}
    f_{\boldsymbol\mu} \mathbf Z^{\boldsymbol\mu}, 
    \quad 
    Q_{b,c}(\mathbf Z,\overline{\mathbf Z})
    =
    \sum_{\substack{\boldsymbol\rho,\boldsymbol\sigma\in\mathbb N_0^{N\eta}\\
                    |\boldsymbol\rho|=b,\ |\boldsymbol\sigma|=c}}
    q_{\boldsymbol\rho,\boldsymbol\sigma}
    \mathbf Z^{\boldsymbol\rho}\overline{\mathbf Z}^{\boldsymbol\sigma}.
\end{align}

Then, the following identity satisfies
\begin{align}
    \bigl(:\!F_a\!:\bigr)
    \bigl(:\!Q_{b,c}\!:\bigr)
    &=
    \sum_{r=0}^{\min(a,c)}
    \frac{1}{r!}
    \sum_{j_1,\ldots,j_r=1}^{N\eta }
    :\!
        \left(
            \partial_{Z_{j_1}}\cdots
            \partial_{Z_{j_r}}F_a
        \right)
        \left(
            \partial_{\overline Z_{j_1}}\cdots
            \partial_{\overline Z_{j_r}}Q_{b,c}
        \right)
    \!: .
    \label{eq:one-sided-Wick-product}
\end{align}
For each fixed \(r\), the corresponding summand belongs to
$
    \mathcal W_{a+b-r,c-r}.
$
\end{lemma}
\begin{proof}
Since all partial derivatives commute, one has
\begin{align}
    \exp\left(
        \sum_{j=1}^{N\eta}
        \partial_{Z_j}\partial_{\overline Z_j}
    \right)
    \exp\left(
        -\sum_{j=1}^{N\eta}
        \partial_{Z_j}\partial_{\overline Z_j}
    \right)
    Q_{b,c} 
    &= \qty[\sum_{s \geq 0} \frac1{s!} \qty(\sum_{j=1}^{N\eta}
        \partial_{Z_j}\partial_{\overline Z_j})^s ]
        \qty[\sum_{s' \geq 0} \frac{(-1)^{s'}}{(s')!} \qty(\sum_{j=1}^{N\eta}
        \partial_{Z_j}\partial_{\overline Z_j})^{s'} ]Q_{b,c}
    \\
    &= 
    \sum_{s,s' \geq 0} \frac{(-1)^{s'}}{s!(s')!} \qty(\sum_{j=1}^{N\eta}\partial_{Z_j}\partial_{\overline Z_j})^{s+s'} Q_{b,c}.
\end{align}
If we take $t=s+s'$, this identity can be expressed as 
\begin{align}
    \sum_{s,s' \geq 0} \frac{(-1)^{s'}}{s!(s')!} \qty(\sum_{j=1}^{N\eta}\partial_{Z_j}\partial_{\overline Z_j})^{s+s'} Q_{b,c}
    &=
    \sum_{t\geq 0}
    \frac{1}{t!}
    \left[
        \sum_{s'=0}^{t}
        (-1)^{s'} \binom{t}{s'}
    \right]
    \left(
        \sum_{j=1}^{N\eta}
        \partial_{Z_j}\partial_{\overline Z_j}
    \right)^t
    Q_{b,c} \\
    &\quad=
    Q_{b,c}.
    \label{eq:Wick-exponential-cancellation}
\end{align}
In the last identity, we use $\sum^{t}_{s'=0} (-1)^{s'} \binom{t}{s'} = \delta_{t,0} $.

Since \(F_a\) contains no \(\overline Z_j\), one has
$
    :\!F_a\!:
    =
    F_a.
$
Repeated application of the Leibniz rule gives
\begin{align}
    &\exp\left(
        \sum_{j=1}^{N\eta}
        \partial_{Z_j}\partial_{\overline Z_j}
    \right)
    \left[
        F_a
        \exp\left(
            -\sum_{j=1}^{N\eta}
            \partial_{Z_j}\partial_{\overline Z_j}
        \right)
        Q_{b,c}
    \right]
    \nonumber\\
    &\qquad=
    \sum_{r=0}^{\min(a,c)}
    \frac{1}{r!}
    \sum_{j_1,\ldots,j_r=1}^{N\eta}
    \left(
        \partial_{Z_{j_1}}\cdots
        \partial_{Z_{j_r}}F_a
    \right)
    \left(
        \partial_{\overline Z_{j_1}}\cdots
        \partial_{\overline Z_{j_r}}Q_{b,c}
    \right).
    \label{eq:inverse-Wick-Leibniz}
\end{align}
Indeed, for every
polynomial \(P\),
\begin{align}
    \left(
        \sum_{j=1}^{N\eta}
        \partial_{Z_j}\partial_{\overline Z_j}
    \right)^n
    (F_aP)
    &=
    \sum_{r=0}^{n}
    \binom{n}{r}
    \sum_{j_1,\ldots,j_r=1}^{N\eta}
    \left(
        \partial_{Z_{j_1}}\cdots
        \partial_{Z_{j_r}}F_a
    \right)
    \partial_{\overline Z_{j_1}}\cdots
    \partial_{\overline Z_{j_r}}
    \left(
        \sum_{j=1}^{N\eta}
        \partial_{Z_j}\partial_{\overline Z_j}
    \right)^{n-r}
    P.
    \label{eq:iterated-Wick-Leibniz}
\end{align}
Consequently,
\begin{align}
    \exp\left(
        \sum_{j=1}^{N\eta}
        \partial_{Z_j}\partial_{\overline Z_j}
    \right)
    (F_aP)
    &=
    \sum_{r=0}^{a}
    \frac{1}{r!}
    \sum_{j_1,\ldots,j_r=1}^{N\eta}
    \left(
        \partial_{Z_{j_1}}\cdots
        \partial_{Z_{j_r}}F_a
    \right)
    \exp\left(
        \sum_{j=1}^{N\eta}
        \partial_{Z_j}\partial_{\overline Z_j}
    \right)
    \left(
        \partial_{\overline Z_{j_1}}\cdots
        \partial_{\overline Z_{j_r}}P
    \right).
    \label{eq:exponential-Wick-Leibniz}
\end{align}
Substituting
$
    P
    =
    \exp\left(
        -\sum_{j=1}^{N\eta}
        \partial_{Z_j}\partial_{\overline Z_j}
    \right)Q_{b,c},
$
and using the commutativity of all partial derivatives, we obtain
\begin{align}
    \exp\left(
        \sum_{j=1}^{N\eta}
        \partial_{Z_j}\partial_{\overline Z_j}
    \right)
    \left(
        \partial_{\overline Z_{j_1}}\cdots
        \partial_{\overline Z_{j_r}}P
    \right)=
    \partial_{\overline Z_{j_1}}\cdots
    \partial_{\overline Z_{j_r}}
    Q_{b,c}.
\end{align}
Since the derivatives of \(F_a\) and \(Q_{b,c}\) vanish for
\(r>a\) and \(r>c\), respectively, Eq.~\eqref{eq:inverse-Wick-Leibniz}
follows.

Applying
\(
\exp\bigl(
-\sum_{j=1}^{N\eta}
\partial_{Z_j}\partial_{\overline Z_j}
\bigr)
\)
to both sides of Eq.~\eqref{eq:inverse-Wick-Leibniz} yields
\begin{align}
    F_a:\!Q_{b,c}\!:
    &=
    \sum_{r=0}^{\min(a,c)}
    \frac{1}{r!}
    \sum_{j_1,\ldots,j_r=1}^{N\eta}
    :\!
        \left(
            \partial_{Z_{j_1}}\cdots
            \partial_{Z_{j_r}}F_a
        \right)
        \left(
            \partial_{\overline Z_{j_1}}\cdots
            \partial_{\overline Z_{j_r}}Q_{b,c}
        \right)
    \!: .
\end{align}
Since \(F_a=:\!F_a\!:\), this proves
Eq.~\eqref{eq:one-sided-Wick-product}.

For fixed \(r\), the polynomial inside the Wick ordering has degree
\(a+b-r\) in \(\mathbf Z\) and degree \(c-r\) in \(\overline{\mathbf Z}\).
Consequently, its Wick ordering belongs to
\(\mathcal W_{a+b-r,c-r}\).
\end{proof}

\begin{lemma}
\label{lem:two-block-Wick-factorization}
Let \(\eta,N\) be integers satisfying \(1\leq\eta<N\), and let
\(\ell,m\in\mathbb N_0\). Suppose that
\begin{align}
    F_\ell(\mathbf Z)
    =
    \sum_{\substack{
        \boldsymbol\mu\in\mathbb N_0^{N\eta}\\
        |\boldsymbol\mu|=\ell
    }}
    f_{\boldsymbol\mu}\,
    \mathbf Z^{\boldsymbol\mu},
    & \
    F'_\ell(\mathbf Z)
    =
    \sum_{\substack{
        \boldsymbol\mu\in\mathbb N_0^{N\eta}\\
        |\boldsymbol\mu|=\ell
    }}
    f'_{\boldsymbol\mu}\,
    \mathbf Z^{\boldsymbol\mu},
    \\
    H_m(\mathbf Z)
    =
    \sum_{\substack{
        \boldsymbol\nu\in\mathbb N_0^{N\eta}\\
        |\boldsymbol\nu|=m
    }}
    h_{\boldsymbol\nu}\,
    \mathbf Z^{\boldsymbol\nu},
    & \
    H'_m(\mathbf Z)
    =
    \sum_{\substack{
        \boldsymbol\nu\in\mathbb N_0^{N\eta}\\
        |\boldsymbol\nu|=m
    }}
    h'_{\boldsymbol\nu}\,
    \mathbf Z^{\boldsymbol\nu}.
\end{align}
Then
\begin{align}
    \mathbb E_0
    \left[
        :\!F_\ell\overline{H_m}\!:\,
        \overline{
            :\!F'_\ell\overline{H'_m}\!:
        }
    \right]
    &=
    \mathbb E_0
    \left[
        F_\ell\overline{F'_\ell}
    \right]
    \mathbb E_0
    \left[
        H'_m\overline{H_m}
    \right].
    \label{eq:two-block-Wick-factorization}
\end{align}
\end{lemma}

\begin{proof}
By the linearity of Wick ordering and
Eq.~\eqref{eq:multivariate-Hermite-operational},
\begin{align}
    :\!F_\ell\overline{H_m}\!:
    &=
    \sum_{\substack{
        |\boldsymbol\mu|=\ell\\
        |\boldsymbol\nu|=m
    }}
    f_{\boldsymbol\mu}
    \overline{h_{\boldsymbol\nu}}\,
    \mathbf H_{\boldsymbol\mu,\boldsymbol\nu}
    (\mathbf Z,\overline{\mathbf Z}).
    \label{eq:two-block-Hermite-expansion}
\end{align}
Similarly,
\begin{align}
    :\!F'_\ell\overline{H'_m}\!:
    &=
    \sum_{\substack{
        |\boldsymbol\mu'|=\ell\\
        |\boldsymbol\nu'|=m
    }}
    f'_{\boldsymbol\mu'}
    \overline{h'_{\boldsymbol\nu'}}\,
    \mathbf H_{\boldsymbol\mu',\boldsymbol\nu'}
    (\mathbf Z,\overline{\mathbf Z}).
\end{align}
Using the tensorized complex Hermite orthogonality relation expressed in Eq.~\eqref{eq:tensorized-Hermite-orthogonality},
we obtain
\begin{align}
    \purple{\mathbb E_0}
    \left[
        :\!F_\ell\overline{H_m}\!:\,
        \overline{
            :\!F'_\ell\overline{H'_m}\!:
        }
    \right]
    &= \sum_{\substack{
        |\boldsymbol\mu'|=\ell\\
        |\boldsymbol\nu'|=m
    }} \sum_{\substack{
        |\boldsymbol\mu|=\ell\\
        |\boldsymbol\nu|=m
    }}
    f_{\boldsymbol\mu}
    \overline{h_{\boldsymbol\nu}}\,
\overline{f'_{\boldsymbol\mu'}
    \overline{h'_{\boldsymbol\nu'}}} \,
    \purple{\mathbb{E}_0} \qty[ \mathbf H_{\boldsymbol\mu,\boldsymbol\nu}(\mathbf Z,\overline{\mathbf Z}) 
        \overline{\mathbf H}_{\boldsymbol\mu',\boldsymbol\nu'}
    (\mathbf Z,\overline{\mathbf Z})]
     \\
    &=
    \sum_{\substack{
        |\boldsymbol\mu|=\ell\\
        |\boldsymbol\nu|=m
    }}
    \boldsymbol\mu!\,\boldsymbol\nu!\,
    f_{\boldsymbol\mu}
    \overline{f'_{\boldsymbol\mu}}\,
    h'_{\boldsymbol\nu}
    \overline{h_{\boldsymbol\nu}}
    \\
    &=
    \left(
        \sum_{|\boldsymbol\mu|=\ell}
        \boldsymbol\mu!\,
        f_{\boldsymbol\mu}
        \overline{f'_{\boldsymbol\mu}}
    \right)
    \left(
        \sum_{|\boldsymbol\nu|=m}
        \boldsymbol\nu!\,
        h'_{\boldsymbol\nu}
        \overline{h_{\boldsymbol\nu}}
    \right).
    \label{eq:two-block-factorized-sum}
\end{align}
On the other hand, Gaussian monomial orthogonality gives
\begin{align}
    \purple{\mathbb E_0}
    \left[
        F_\ell\overline{F'_\ell}
    \right]
    &=
    \sum_{|\boldsymbol\mu|=\ell}
    \boldsymbol\mu!\,
    f_{\boldsymbol\mu}
    \overline{f'_{\boldsymbol\mu}},
    \\
    \purple{\mathbb E_0}
    \left[
        H'_m\overline{H_m}
    \right]
    &=
    \sum_{|\boldsymbol\nu|=m}
    \boldsymbol\nu!\,
    h'_{\boldsymbol\nu}
    \overline{h_{\boldsymbol\nu}}.
\end{align}
Substitution into
Eq.~\eqref{eq:two-block-factorized-sum}
proves Eq.~\eqref{eq:two-block-Wick-factorization}.
\end{proof}

The following orthogonality relation for Ginibre minors will also be
used independently in the covariance calculation below.

\begin{lemma}
\label{lem:exterior-minor-orthogonality}
Let \(0\leq m\leq\eta<N\), and let
\(G\in\mathbb C^{N\times\eta}\) be a standard complex Ginibre matrix.

Then, for every
\(
\vec a,\vec b\in\mac S_{N,m}
\)
and
\(
\vec c,\vec d\in\mac S_{\eta,m},
\)
one has
\begin{align}
    \purple{\mathbb E_0}
    \left[
        \det G_{\vec a,\vec c}\,
        \overline{\det G_{\vec b,\vec d}}
    \right]
    &=
    m!\,
    \delta_{\vec a,\vec b}\,
    \delta_{\vec c,\vec d}.
    \label{eq:Ginibre-minor-orthogonality}
\end{align}
\end{lemma}

\begin{proof}
Expanding both determinants and applying the complex Wick formula gives
\begin{align}
    \mathbb E_0
    \left[
        \det G_{\vec a,\vec c}\,
        \overline{\det G_{\vec b,\vec d}}
    \right]
    =
    \sum_{\sigma,\tau\in\mathfrak S_m}
    \sgn(\sigma)\sgn(\tau)
    \sum_{\pi\in\mathfrak S_m}
    \prod_{i=1}^m
        \delta_{a_i,b_{\pi(i)}}\,
        \delta_{c_{\sigma(i)},d_{\tau(\pi(i))}}.
    \label{eq:Ginibre-minor-Wick-expansion}
\end{align}
Since the tuples \(\vec a\) and \(\vec b\) are written in increasing
order, the row Kronecker deltas vanish unless
\(
\vec a=\vec b
\);
when \(\vec a=\vec b\), they force
\(
\pi=\operatorname{id}.
\)
The remaining column Kronecker deltas then vanish unless
\(
\vec c=\vec d
\);
when \(\vec c=\vec d\), they force
\(
\sigma=\tau.
\)
Consequently,
\begin{align}
    \mathbb E_0
    \left[
        \det G_{\vec a,\vec c}\,
        \overline{\det G_{\vec b,\vec d}}
    \right]
    &=
    \delta_{\vec a,\vec b}\,
    \delta_{\vec c,\vec d}
    \sum_{\sigma\in\mathfrak S_m}
        \sgn(\sigma)^2
    \\
    &=
    m!\,
    \delta_{\vec a,\vec b}\,
    \delta_{\vec c,\vec d},
\end{align}
which proves
Eq.~\eqref{eq:Ginibre-minor-orthogonality}.

\end{proof}


\subsubsection{Wick-ordered representation of the estimator}

In this subsection, we show that
\(\purple{{}^{k}\mathfrak D_{\eta;\vec q}^{\vec p}}(GG^\dagger)\)
is \(1/k!\) times the Wick-ordered exterior-power matrix element
\(\bra{\vec q}(GG^\dagger)^{\wedge k}\ket{\vec p}\).
This identity places the estimator in the \((k,k)\)-th complex
Wiener--It\^o chaos \(\mathcal W_{k,k}\) and provides the starting point
for the decomposition in the next subsection.

\begin{lemma}
\label{lem:wick-ordered-minor-representation}
Let \(1\leq k\leq\eta<N\), and let
\(G\in\mathbb C^{N\times\eta}\) be a standard complex Ginibre matrix.
Then, for every
\(\vec q,\vec p\in\mathcal S_{N,k}\),
\begin{align}
    :\!
    \bra{\vec q}
        (GG^\dagger)^{\wedge k}
    \ket{\vec p}
    \!:
    &=
    \sum_{t=0}^{k}
    (-1)^{k+t}
    \frac{(\eta-t)!}{(\eta-k)!}
    \bra{\vec q}
        \mac E_{t,k}
        \left(
            (GG^\dagger)^{\wedge t}
        \right)
    \ket{\vec p}
    \nonumber\\
    &=
    k!\,
    \purple{{}^{k}\mathfrak D_{\eta;\vec q}^{\vec p}}
    (GG^\dagger),
    \label{eq:wick-ordered-minor-representation}
\end{align}
where \(\mathfrak D_{\eta}^{(k)}\) is defined in
Eq.~\eqref{eq:dilute-estimator}.
\end{lemma}

Before proving this lemma, we provide the following lemma,

\begin{lemma}[Higher derivatives of matrix minors]
\label{lem:higher-derivatives-matrix-minors}
Let \(r,k,\eta,N\) be integers satisfying
$ 1 \leq k \leq \eta < N, \
    0 \leq r \leq k,
$
and let
\(
\vec q\in\mac S_{N,k}
\)
and
\(
\vec a\in\mac S_{\eta,k}.
\)

Then, for $\vec m\in\mac S_{N,r},
    \vec\alpha  \in    \mac S_{\eta,r}$,
and
\(
\pi\in\mathfrak S_r,
\)
we have
\begin{align}
    \prod_{\ell=1}^{r}
    \partial_{G_{m_\ell,\alpha_{\pi(\ell)}}}
    \det G_{\vec q,\vec a}
    =
    \begin{cases}
        \purple{
            \varepsilon_{\wedge}
            (\vec m,\vec q\setminus\vec m)
            \varepsilon_{\wedge}
            (\vec\alpha,\vec a\setminus\vec\alpha)
            \sgn(\pi)
        }\,
        \det G_{\vec q\setminus\vec m,\,
                   \vec a\setminus\vec\alpha},
        &
        \vec m\subseteq\vec q,\quad
        \vec\alpha\subseteq\vec a,
        \\[2mm]
        0,
        &
        \text{otherwise},
    \end{cases}
    \label{eq:higher-derivatives-matrix-minors}
\end{align}

Furthermore, if
\(
\vec p,\vec q\in\mac S_{N,k},
\)
\(
\vec m\subseteq\vec q\cap\vec p,
\)
and
\(
\vec\alpha\subseteq\vec a,
\)
then
\begin{align}
    \purple{
        \varepsilon_{\wedge}
        (\vec q\setminus\vec m,\vec m)
        \varepsilon_{\wedge}
        (\vec p\setminus\vec m,\vec m)
    }
    =
    \bra{\vec q}
        a_{\vec q\setminus\vec m}^{\dagger}
        a_{\vec p\setminus\vec m}
    \ket{\vec p}.
    \label{eq:cofactor-sign-fermionic-sign}
\end{align}
\end{lemma}

\begin{proof}
The case \(r=0\) is immediate. If
\(
\vec m\not\subseteq\vec q
\)
or
\(
\vec\alpha\not\subseteq\vec a
\),
then at least one of the differentiation variables does not occur in
\(\det G_{\vec q,\vec a}\). The derivative therefore vanishes, proving
the zero case in
Eq.~\eqref{eq:higher-derivatives-matrix-minors}.
Hence, it remains to consider
\(
\vec m\subseteq\vec q
\)
and
\(
\vec\alpha\subseteq\vec a
\).

By the definition of the fermionic shuffle sign,
\begin{align}
    \det G_{\vec q,\vec a}
    &=
    \purple{
        \varepsilon_{\wedge}
        (\vec m,\vec q\setminus\vec m)
        \varepsilon_{\wedge}
        (\vec\alpha,\vec a\setminus\vec\alpha)
    }
    \det
    G_{
        (\vec m,\vec q\setminus\vec m),\,
        (\vec\alpha,\vec a\setminus\vec\alpha)
    }.
\end{align}
For disjoint increasing tuples \(\vec x\) and \(\vec y\), we use
\((\vec x,\vec y)\) to denote their ordered concatenation. In
particular, the first \(r\) rows and columns of
$
    G_{
        (\vec m,\vec q\setminus\vec m),\,
        (\vec\alpha,\vec a\setminus\vec\alpha)
    }
$
are indexed by \(\vec m\) and \(\vec\alpha\), respectively.

The Leibniz expansion gives
\begin{align}
    \det G_{
        (\vec m,\vec q\setminus\vec m),\,
        (\vec\alpha,\vec a\setminus\vec\alpha)
    }
    &=
    \sum_{\sigma\in\mathfrak S_k}
    \sgn(\sigma)
    \prod_{j=1}^{k}
    G_{(\vec m,\vec q\setminus\vec m)_j,(\vec\alpha,\vec a\setminus\vec\alpha)_{\sigma(j)}}.
\end{align}
After applying
\(
\prod_{\ell=1}^{r}
\partial_{G_{m_\ell,\alpha_{\pi(\ell)}}}
\),
a term indexed by \(\sigma\) survives only if
\begin{align}
    \sigma(\ell)
    &=
    \pi(\ell),
    \qquad
    1\leq\ell\leq r.
\end{align}
{Since \(\pi\) is a permutation of \([r]\), this condition implies
\(\sigma([r])=[r]\)}. Hence, the restriction of \(\sigma\) to the
remaining \(k-r\) positions is a permutation
\(\tau\in\mathfrak S_{k-r}\), and every surviving permutation has the
unique block form
\begin{align}
    \sigma
    &=
    \pi\oplus\tau.
\end{align}
Here, \(\pi\oplus\tau\) acts by \(\pi\) on the first \(r\) positions
and by \(\tau\) on the remaining \(k-r\) positions. Its sign factors as
$
    \sgn(\pi\oplus\tau)=
    \sgn(\pi)\sgn(\tau).
$
Therefore,
\begin{align}
    \prod_{\ell=1}^{r}
    \partial_{G_{m_\ell,\alpha_{\pi(\ell)}}}
    \det
    G_{
        (\vec m,\vec q\setminus\vec m),\,
        (\vec\alpha,\vec a\setminus\vec\alpha)
    }
    &=
    \sum_{\tau\in\mathfrak S_{k-r}}
    \sgn(\pi)\sgn(\tau)
    \prod_{u=1}^{k-r}
    G_{
        (\vec q\setminus\vec m)_u,\,
        (\vec a\setminus\vec\alpha)_{\tau(u)}
    }
    \\
    &=
    \sgn(\pi)
    \det G_{\vec q\setminus\vec m,\,
               \vec a\setminus\vec\alpha}.
\end{align}
Combining the preceding two identities proves
Eq.~\eqref{eq:higher-derivatives-matrix-minors}.

To prove the last identity, observe that
\begin{align}
    a_{\vec p\setminus\vec m}\ket{\vec p}
    &=
    \purple{
        \varepsilon_{\wedge}
        (\vec p\setminus\vec m,\vec m)
    }
    \ket{\vec m},
    \\
    a_{\vec q\setminus\vec m}^{\dagger}\ket{\vec m}
    &=
    \purple{
        \varepsilon_{\wedge}
        (\vec q\setminus\vec m,\vec m)
    }
    \ket{\vec q}.
\end{align}
Therefore,
\begin{align}
    \bra{\vec q}
        a_{\vec q\setminus\vec m}^{\dagger}
        a_{\vec p\setminus\vec m}
    \ket{\vec p}
    &=
    \purple{
        \varepsilon_{\wedge}
        (\vec q\setminus\vec m,\vec m)
        \varepsilon_{\wedge}
        (\vec p\setminus\vec m,\vec m)
    },
\end{align}
which proves
Eq.~\eqref{eq:cofactor-sign-fermionic-sign}.
\end{proof}

\begin{lemma}[Iterated contraction identity for matrix minors of possibly different sizes]
\label{lem:iterated-contraction-matrix-minors}
Let \(1\leq \ell,k\leq\eta<N\), \(0\leq r\leq\min\{\ell,k\}\),
and \(G\in\mathbb C^{N\times\eta}\). Then, for every
$
    \vec q\in\mac S_{N,\ell},
    \vec a\in\mac S_{\eta,\ell},
    \vec p\in\mac S_{N,k},
    \vec b\in\mac S_{\eta,k},
$
we have
\begin{align}
    &\frac{1}{r!}
    \left(
        \sum_{x=1}^{N}
        \sum_{\alpha=1}^{\eta}
        \partial_{G_{x,\alpha}}
        \partial_{\overline G_{x,\alpha}}
    \right)^{r}
    \left(
        \det G_{\vec q,\vec a}\,
        \overline{\det G_{\vec p,\vec b}}
    \right)
    \nonumber\\
    &\qquad=
    r!
    \sum_{\substack{
        \vec x\in\mac S_{N,r}\\
        \vec x\subseteq\vec q\cap\vec p
    }}
    \sum_{\substack{
        \vec\alpha\in\mac S_{\eta,r}\\
        \vec\alpha\subseteq\vec a\cap\vec b
    }}
    \purple{
        \varepsilon_{\wedge}
            (\vec x,\vec q\setminus\vec x)\,
        \varepsilon_{\wedge}
            (\vec\alpha,\vec a\setminus\vec\alpha)\,
        \varepsilon_{\wedge}
            (\vec x,\vec p\setminus\vec x)\,
        \varepsilon_{\wedge}
            (\vec\alpha,\vec b\setminus\vec\alpha)
    }
    \det G_{\vec q\setminus\vec x,\,
               \vec a\setminus\vec\alpha}\,
    \overline{
        \det G_{\vec p\setminus\vec x,\,
                   \vec b\setminus\vec\alpha}
    }.
    \label{eq:iterated-contraction-matrix-minors}
\end{align}

In particular, suppose that the two minors have the same size,
\(\ell=k\), and the same column-index tuple,
\(\vec b=\vec a\). Then
\begin{align}
    &\frac{1}{r!}
    \left(
        \sum_{x=1}^{N}
        \sum_{\alpha=1}^{\eta}
        \partial_{G_{x,\alpha}}
        \partial_{\overline G_{x,\alpha}}
    \right)^{r}
    \left(
        \det G_{\vec q,\vec a}\,
        \overline{\det G_{\vec p,\vec a}}
    \right)
    \nonumber\\
    &\qquad=
    r!
    \sum_{\substack{
        \vec x\in\mac S_{N,r}\\
        \vec x\subseteq\vec q\cap\vec p
    }}
    \bra{\vec q}
        a_{\vec q\setminus\vec x}^{\dagger}
        a_{\vec p\setminus\vec x}
    \ket{\vec p}
    \sum_{\substack{
        \vec\alpha\in\mac S_{\eta,r}\\
        \vec\alpha\subseteq\vec a
    }}
    \det G_{\vec q\setminus\vec x,\,
               \vec a\setminus\vec\alpha}\,
    \overline{
        \det G_{\vec p\setminus\vec x,\,
                   \vec a\setminus\vec\alpha}
    }.
    \label{eq:iterated-contraction-equal-size-minors}
\end{align}
For \(r=0\), both identities are understood with the usual
empty-tuple convention.
\end{lemma}

\begin{proof}
Since the first determinant is holomorphic in the entries of \(G\),
whereas the second is antiholomorphic, expanding the \(r\)-th power
gives
\begin{align}
    &\frac{1}{r!}
    \sum_{x_1,\ldots,x_r=1}^{N}
    \sum_{\alpha_1,\ldots,\alpha_r=1}^{\eta}
    \left(
        \partial_{G_{x_1,\alpha_1}}\cdots
        \partial_{G_{x_r,\alpha_r}}
        \det G_{\vec q,\vec a}
    \right)
    \left(
        \partial_{\overline G_{x_1,\alpha_1}}\cdots
        \partial_{\overline G_{x_r,\alpha_r}}
        \overline{\det G_{\vec p,\vec b}}
    \right).
    \label{eq:expanded-iterated-contraction}
\end{align}
A term is nonzero only if the row indices \(x_1,\ldots,x_r\)
are pairwise distinct and belong to \(\vec q\cap\vec p\), and the
column indices \(\alpha_1,\ldots,\alpha_r\) are pairwise distinct
and belong to \(\vec a\cap\vec b\). Since the differential operators
commute, the prefactor \(1/r!\) removes the \(r!\) possible orderings
of every fixed set of \(r\) contracted pairs. For fixed increasing
tuples \(\vec x\) and \(\vec\alpha\), these sets of pairs are indexed
by \(\pi\in\mathfrak S_r\), where
\(x_\nu\leftrightarrow\alpha_{\pi(\nu)}\).
Lemma~\ref{lem:higher-derivatives-matrix-minors} therefore yields
\begin{align}
    &\frac{1}{r!}
    \left(
        \sum_{x=1}^{N}
        \sum_{\alpha=1}^{\eta}
        \partial_{G_{x,\alpha}}
        \partial_{\overline G_{x,\alpha}}
    \right)^{r}
    \left(
        \det G_{\vec q,\vec a}\,
        \overline{\det G_{\vec p,\vec b}}
    \right)
    \nonumber\\
    &=
    \sum_{\substack{
        \vec x\in\mac S_{N,r}\\
        \vec x\subseteq\vec q\cap\vec p
    }}
    \sum_{\substack{
        \vec\alpha\in\mac S_{\eta,r}\\
        \vec\alpha\subseteq\vec a\cap\vec b
    }}
    \sum_{\pi\in\mathfrak S_r}
    \purple{
        \operatorname{sgn}(\pi)^2\,
        \varepsilon_{\wedge}
            (\vec x,\vec q\setminus\vec x)\,
        \varepsilon_{\wedge}
            (\vec\alpha,\vec a\setminus\vec\alpha)\,
        \varepsilon_{\wedge}
            (\vec x,\vec p\setminus\vec x)\,
        \varepsilon_{\wedge}
            (\vec\alpha,\vec b\setminus\vec\alpha)
    }
    \det G_{\vec q\setminus\vec x,\,
               \vec a\setminus\vec\alpha}\,
    \overline{
        \det G_{\vec p\setminus\vec x,\,
                   \vec b\setminus\vec\alpha}
    } \\
    &= 
    \sum_{\substack{
        \vec x\in\mac S_{N,r}\\
        \vec x\subseteq\vec q\cap\vec p
    }}
    \sum_{\substack{
        \vec\alpha\in\mac S_{\eta,r}\\
        \vec\alpha\subseteq\vec a\cap\vec b
    }}
    \sum_{\pi\in\mathfrak S_r}
    \purple{
        \varepsilon_{\wedge}
            (\vec x,\vec q\setminus\vec x)\,
        \varepsilon_{\wedge}
            (\vec\alpha,\vec a\setminus\vec\alpha)\,
        \varepsilon_{\wedge}
            (\vec x,\vec p\setminus\vec x)\,
        \varepsilon_{\wedge}
            (\vec\alpha,\vec b\setminus\vec\alpha)
    }
    \det G_{\vec q\setminus\vec x,\,
               \vec a\setminus\vec\alpha}\,
    \overline{
        \det G_{\vec p\setminus\vec x,\,
                   \vec b\setminus\vec\alpha}
    }
    .
    \label{eq:iterated-contraction-sign-expansion}
\end{align}
Here, the sign product is independent of \(\pi\). Summing over
\(\pi\in\mathfrak S_r\) produces the factor \(r!\), proving
Eq.~\eqref{eq:iterated-contraction-matrix-minors}.

Finally, suppose that \(\ell=k\) and \(\vec b=\vec a\). In this
case $\varepsilon_{\wedge} (\vec\alpha,\vec a\setminus\vec\alpha) = \varepsilon_{\wedge}   (\vec\alpha,\vec b\setminus\vec\alpha)$. By the fermionic sign identity in
Lemma~\ref{lem:higher-derivatives-matrix-minors},
\begin{align}
    \purple{
        \varepsilon_{\wedge}
            (\vec x,\vec q\setminus\vec x)\,
        \varepsilon_{\wedge}
            (\vec\alpha,\vec a\setminus\vec\alpha)\,
        \varepsilon_{\wedge}
            (\vec x,\vec p\setminus\vec x)\,
        \varepsilon_{\wedge}
            (\vec\alpha,\vec a\setminus\vec\alpha)
    }
    &=
    \purple{
        \varepsilon_{\wedge}
            (\vec x,\vec q\setminus\vec x)\,
        \varepsilon_{\wedge}
            (\vec x,\vec p\setminus\vec x)
    }
    \nonumber\\
    &=
    \bra{\vec q}
        a_{\vec q\setminus\vec x}^{\dagger}
        a_{\vec p\setminus\vec x}
    \ket{\vec p}.
\end{align}
This expression proves
Eq.~\eqref{eq:iterated-contraction-equal-size-minors}.
\end{proof}

\begin{proof}
We first apply the Cauchy--Binet formula to express the matrix element
as a sum of products of matrix minors. We then use the differential
representation of Wick ordering and
Lemma~\ref{lem:higher-derivatives-matrix-minors} to evaluate each
contraction term in terms of complementary minors.

Following Lemma~\ref{lem:entry-locality}, we obtain
\begin{align}
    \bra{\vec q}
        (GG^\dagger)^{\wedge k}
    \ket{\vec p}
    &=
    \det (GG^\dagger)_{\vec q,\vec p}
    \nonumber\\
    &=
    \sum_{\vec a\in\mac S_{\eta,k}}
    \det G_{\vec q,\vec a}\,
    \overline{\det G_{\vec p,\vec a}}.
\end{align}
By the differential representation of Wick ordering,
\begin{align}
    :\!
    \bra{\vec q}
        (GG^\dagger)^{\wedge k}
    \ket{\vec p}
    \!:
    &=
    \sum_{\vec a\in\mac S_{\eta,k}}
    \sum_{r=0}^{k}
    \frac{(-1)^r}{r!}
    \left(
        \sum_{x=1}^{N}
        \sum_{\alpha=1}^{\eta}
        \partial_{G_{x,\alpha}}
        \partial_{\overline{G}_{x,\alpha}}
    \right)^{r}
    \left(
        \det G_{\vec q,\vec a}\,
        \overline{\det G_{\vec p,\vec a}}
    \right).
\end{align}
Applying
Lemma~\ref{lem:iterated-contraction-matrix-minors}, we obtain
\begin{align}
    &\frac{1}{r!}
    \left(
        \sum_{x=1}^{N}
        \sum_{\alpha=1}^{\eta}
        \partial_{G_{x,\alpha}}
        \partial_{\overline G_{x,\alpha}}
    \right)^{r}
    \left(
        \det G_{\vec q,\vec a}\,
        \overline{\det G_{\vec p,\vec a}}
    \right)
    \nonumber\\
    &\qquad=
    r!
    \sum_{\substack{
        \vec x\in\mac S_{N,r}\\
        \vec x\subseteq\vec q\cap\vec p
    }}
    \bra{\vec q}
        a_{\vec q\setminus\vec x}^{\dagger}
        a_{\vec p\setminus\vec x}
    \ket{\vec p}
    \sum_{\substack{
        \vec\alpha\in\mac S_{\eta,r}\\
        \vec\alpha\subseteq\vec a
    }}
    \det G_{\vec q\setminus\vec x,\,
               \vec a\setminus\vec\alpha}\,
    \overline{
        \det G_{\vec p\setminus\vec x,\,
                   \vec a\setminus\vec\alpha}
    }.
\end{align}

Set \(t:=k-r\). For a fixed uncontracted column tuple
\(\vec\beta\in\mac S_{\eta,t}\), there are
\(\binom{\eta-t}{r}\) contracted column tuples
\(\vec\alpha\in\mac S_{\eta,r}\) satisfying
\(\vec\beta\cap\vec\alpha=\varnothing\), and each pair
\((\vec\beta,\vec\alpha)\) uniquely determines
\(\vec a=\vec\beta\cup\vec\alpha\). Therefore, for each fixed
\(\vec x\in\mac S_{N,r}\) with
\(\vec x\subseteq\vec q\cap\vec p\),
\begin{align}
    \sum_{\vec a\in\mac S_{\eta,k}}
    \sum_{\substack{
        \vec\alpha\in\mac S_{\eta,r}\\
        \vec\alpha\subseteq\vec a
    }}
    \det G_{\vec q\setminus\vec x,\,
               \vec a\setminus\vec\alpha}\,
    \overline{
        \det G_{\vec p\setminus\vec x,\,
                   \vec a\setminus\vec\alpha}
    }
    &=
    \binom{\eta-t}{r}
    \sum_{\vec\beta\in\mac S_{\eta,t}}
    \det G_{\vec q\setminus\vec x,\vec\beta}\,
    \overline{
        \det G_{\vec p\setminus\vec x,\vec\beta}
    }
    \\
    &=
    \binom{\eta-t}{r}
    \bra{\vec q\setminus\vec x}
        (GG^\dagger)^{\wedge t}
    \ket{\vec p\setminus\vec x},
\end{align}
where the last equality follows from the Cauchy--Binet formula.

Consequently, the part containing exactly \(r\) contractions is
\begin{align}
    W_{\vec q,\vec p}^{(r)}
    &=
    (-1)^r
    r!\binom{\eta-t}{r}
    \sum_{\substack{
        \vec x\in\mac S_{N,r}\\
        \vec x\subseteq\vec q\cap\vec p
    }}
    \bra{\vec q}
        a_{\vec q\setminus\vec x}^{\dagger}
        a_{\vec p\setminus\vec x}
    \ket{\vec p}
    \bra{\vec q\setminus\vec x}
        (GG^\dagger)^{\wedge t}
    \ket{\vec p\setminus\vec x}
    \nonumber\\
    &=
    (-1)^r
    \frac{(\eta-t)!}{(\eta-k)!}
    \sum_{\substack{
        \vec x\in\mac S_{N,r}\\
        \vec x\subseteq\vec q\cap\vec p
    }}
    \bra{\vec q}
        a_{\vec q\setminus\vec x}^{\dagger}
        a_{\vec p\setminus\vec x}
    \ket{\vec p}
    \bra{\vec q\setminus\vec x}
        (GG^\dagger)^{\wedge t}
    \ket{\vec p\setminus\vec x}.
    \label{eq:fixed-r-contraction}
\end{align}
Here, we used
\(
r!\binom{\eta-t}{r}
=
(\eta-t)!/(\eta-k)!
\).

Taking the \((\vec q,\vec p)\) matrix element of the extension map
gives
\begin{align}
    \bra{\vec q}
        \mac E_{t,k}(X)
    \ket{\vec p}
    &=
    \sum_{\substack{
        \vec x\in\mac S_{N,k-t}\\
        \vec x\subseteq\vec q\cap\vec p
    }}
    \bra{\vec q}
        a_{\vec q\setminus\vec x}^{\dagger}
        a_{\vec p\setminus\vec x}
    \ket{\vec p}
    \bra{\vec q\setminus\vec x}
        X
    \ket{\vec p\setminus\vec x}.
    \label{eq:extension-entry-contracted}
\end{align}
Since \(r=k-t\), Eq.~\eqref{eq:fixed-r-contraction} becomes
\begin{align}
    W_{\vec q,\vec p}^{(r)}
    &=
    (-1)^{k-t}
    \frac{(\eta-t)!}{(\eta-k)!}
    \bra{\vec q}
        \mac E_{t,k}
        \left(
            (GG^\dagger)^{\wedge t}
        \right)
    \ket{\vec p}
    \nonumber\\
    &=
    (-1)^{k+t}
    \frac{(\eta-t)!}{(\eta-k)!}
    \bra{\vec q}
        \mac E_{t,k}
        \left(
            (GG^\dagger)^{\wedge t}
        \right)
    \ket{\vec p},
    \label{eq:fixed-r-extension}
\end{align}
where we used
\(
(-1)^{k-t}=(-1)^{k+t}
\).

Summing over \(r=0,\ldots,k\), or equivalently over
\(t=k-r\), yields
\begin{align}
    :\!
    \bra{\vec q}
        (GG^\dagger)^{\wedge k}
    \ket{\vec p}
    \!:
    &=
    \sum_{t=0}^{k}
    (-1)^{k+t}
    \frac{(\eta-t)!}{(\eta-k)!}
    \bra{\vec q}
        \mac E_{t,k}
        \left(
            (GG^\dagger)^{\wedge t}
        \right)
    \ket{\vec p}.
    \label{eq:wick-ordered-entry}
\end{align}
Comparing this identity with the definition of
\(\mathfrak D_{\eta}^{(k)}\), we conclude that
\begin{align}
    :\!
    \bra{\vec q}
        (GG^\dagger)^{\wedge k}
    \ket{\vec p}
    \!:
    &=
    k!\,
    \purple{{}^{k}\mathfrak D_{\eta;\vec q}^{\vec p}}
    (GG^\dagger).
\end{align}
This proves the claim.
\end{proof}

We remark that $\bra{\vec q} (GG^\dagger)^{\wedge k} \ket{\vec p}$ is a Gaussian polynomial of bidegree $(k,k)$. Recalling that wick ordering is a projector of the corresponding Wiener-It\^o chao, $\mathfrak{D}^{(k)}_{\vec q,\vec p}$ belongs to $(k,k)$-th Wiener-It\^o chao.

\subsubsection{Decomposition of the determinant-weighted estimator}
\label{sec:decomposition_of_determinant}

In this section, we decompose the determinant-weighted Gaussian
polynomial
\(
    \det G_{\vec s}\,
    \purple{{}^{k}\mathfrak D_{\vec q}^{\vec p}}(GG^\dagger)
\)
into \(k+1\) components indexed by the number \(r\) of
cross-contractions and show that the \(r\)-th component belongs to 
\(
    \mac W_{\eta+k-r,k-r}
\);
consequently, orthogonality of the distinct chaoses reduces the target
second moment to the sum of the squared \(L^2(\gamma_0)\) norms of the
individual contraction components.

We introduce
\begin{align}
    \Xi_{\vec s;\vec q,\vec p}
    &:=
    \frac{1}{k!}\det G_{\vec s}
    :\!
    \bra{\vec q}
        (GG^\dagger)^{\wedge k}
    \ket{\vec p}
    \!: .
    \label{eq:Xi-definition}
\end{align}
We remark that the determinant \(\det G_{\vec s}\) contains \(\eta\) Gaussian factors and no conjugate-Gaussian factor, and hence has bidegree
\((\eta,0)\).  On the other hand, the Wick-ordered factor belongs to
the \((k,k)\)-th complex Gaussian chaos.  Therefore, the only possible
contractions pair a Gaussian factor in \(\det G_{\vec s}\) with a
conjugate-Gaussian factor in the Wick-ordered polynomial.

Lemma.~\ref{lem:one-sided-Wick-product} provides the explicit decomposition of $\Xi_{\vec s; \vec q,\vec p}$. By setting $F_\eta = \det G_{\vec s}$ and $Q_{k,k} = \bra{\vec q} (GG^\dagger)^{\wedge k} \ket{\vec p}$ and using the fact $\min(k,\eta) = k$, the product formula then
gives
\begin{align}
    \Xi_{\vec s; \vec q,\vec p} 
    &= \frac1{k!}  \qty(:\!\det G_{\vec s }     \!:)  
    \qty(:\!
    \bra{\vec q}
        (GG^\dagger)^{\wedge k}
    \ket{\vec p}
    \!: )
    \\
    &= \frac{1}{k!} \sum^{k}_{r=0} \frac1{r!} 
        \sum_{x_1, \ldots x_r =1}^{N}
        \sum_{\alpha_1, \ldots, \alpha_r=1}^{\eta}
        :\qty(\partial_{G_{x_1,\alpha_1}} \cdots \partial_{G_{x_r,\alpha_r}} \det G_{\vec s})
        \qty(\partial_{\overline{G}_{x_1,\alpha_1}}  \cdots \partial_{\overline{G}_{x_r,\alpha_r}} \bra{\vec q} \purple{(GG^\dagger)^{\wedge k}} \ket{\vec p} ):
\end{align}
We will focus on the fixed $r$-th term expressed by 
\begin{align}
    \Xi^{(r)}_{\vec s;\vec q,\vec p} 
        &\coloneq  \frac1{k!r!} 
        \sum_{x_1, \ldots x_r =1}^{N}
        \sum_{\alpha_1, \ldots, \alpha_r=1}^{\eta}
        :\qty(\partial_{G_{x_1,\alpha_1}} \cdots \partial_{G_{x_r,\alpha_r}} \det G_{\vec s})
        \qty(\partial_{\overline{G}_{x_1,\alpha_1}}  \cdots \partial_{\overline{G}_{x_r,\alpha_r}} \bra{\vec q} \purple{(GG^\dagger)^{\wedge k}} \ket{\vec p} ): 
        \label{eq:Xi_defintion}
        \\
        &=  \frac1{k!r!} 
        \sum_{x_1, \ldots x_r =1}^{N}
        \sum_{\alpha_1, \ldots, \alpha_r=1}^{\eta}
        :\qty(\partial_{G_{x_1,\alpha_1}} \cdots \partial_{G_{x_r,\alpha_r}} \det G_{\vec s})
        \qty(\partial_{\overline{G}_{x_1,\alpha_1}}  \cdots \partial_{\overline{G}_{x_r,\alpha_r}} \sum_{\vec a\in\mac S_{\eta,k}}
    \det G_{\vec q,\vec a}\,
    \overline{\det G_{\vec p,\vec a}} ): \\
    &=  \frac1{k!r!} 
    \sum_{\vec a\in\mac S_{\eta,k}}
        \sum_{x_1, \ldots x_r =1}^{N}
        \sum_{\alpha_1, \ldots, \alpha_r=1}^{\eta}
        :\qty(\partial_{G_{x_1,\alpha_1}} \cdots \partial_{G_{x_r,\alpha_r}}  \det G_{\vec s})
        \qty(\partial_{\overline{G}_{x_1,\alpha_1}}  \cdots \partial_{\overline{G}_{x_r,\alpha_r}} 
    \overline{\det G_{\vec p,\vec a}} )     \det G_{\vec q,\vec a}\, :  
    \label{eq:Xi-contraction-decomposition}
\end{align}
Applying
Lemma~\ref{lem:iterated-contraction-matrix-minors}
to the pair
\(
\det G_{\vec s}
\)
and
\(
\overline{\det G_{\vec p,\vec a}}
\),
we find that a nonzero \(r\)-fold contraction must select
\(
\vec x\subseteq\vec s\cap\vec p
\)
and
\(
\vec\alpha\subseteq\vec a
\).
Since the lemma already accounts for all orderings of the contracted
coordinates, we obtain
\begin{align}
    \Xi_{\vec s;\vec q,\vec p}^{(r)}
    &=
    \frac{r!}{k!}
    \sum_{\vec a\in\mac S_{\eta,k}}
    \sum_{\substack{
        \vec x\in\mac S_{N,r}\\
        \vec x\subseteq\vec s\cap\vec p
    }}
    \sum_{\substack{
        \vec\alpha\in\mac S_{\eta,r}\\
        \vec\alpha\subseteq\vec a
    }}
    \varepsilon_{\wedge}
    (\vec x,\vec s\setminus\vec x)\,
    \varepsilon_{\wedge}
    (\vec\alpha,[\eta]\setminus\vec\alpha)\,
    \varepsilon_{\wedge}
    (\vec x,\vec p\setminus\vec x)\,
    \varepsilon_{\wedge}
    (\vec\alpha,\vec a\setminus\vec\alpha) \\
    &\quad \times
    :\!
    \det G_{\vec s\setminus\vec x,\,
               [\eta]\setminus\vec\alpha}\,
    \det G_{\vec q,\vec a}\,
    \overline{
        \det G_{\vec p\setminus\vec x,\,
                   \vec a\setminus\vec\alpha}
    }
    \!: .
    \label{eq:Xi-r-minor-form}
\end{align}
Each polynomial appearing inside the Wick ordering is homogeneous in the Gaussian entries and their complex conjugates, with bidegree
\begin{align}
    (\eta-r,0)+(k,0)+(0,k-r)
    =
    (\eta+k-r,k-r).
\end{align}
Equivalently, each \(r\)-fold contraction lowers both components of the original bidegree \((\eta+k,k)\) by \(r\). Its Wick ordering therefore belongs to the \((\eta+k-r,k-r)\)th complex Wiener--It\^o chaos, which leads to 
$
    \Xi^{(r)}_{\vec s;\vec q,\vec p}
    \in
    \mac W_{\eta+k-r,k-r}.
    \label{eq:Xi-contraction-bidegree}
$

Recalling the orthogonality of Wiener It\^o chaos, we can derive
\begin{align}
    \purple{\mathbb{E}_0}[\purple{\Xi_{\vec s;\vec q,\vec p}^{(r)}} \overline{\purple{\Xi_{\vec s;\vec q,\vec p}^{(r')}}} ] = 0, 
\end{align}
for $r \neq r'$. By defintion expressed in Eq.~\eqref{eq:Xi_defintion}, we can express  $\Xi_{\vec s; \vec q, \vec p} = \sum^{k}_{r=0}\purple{\Xi_{\vec s;\vec q,\vec p}^{(r)}} $ and we obtain
\begin{align}
    \purple{\mathbb E_0}\!\left[
        \left|\det G_{\vec s} \
            \purple{{}^{k}\mathfrak D_{\vec q}^{\vec p}}(GG^\dagger)
        \right|^2
    \right] 
    &= 
    \purple{\mathbb E_0}\!\left[
            \purple{\sum^k_{r=0} \Xi_{\vec s;\vec q,\vec p}^{(r)} \sum^k_{r'=0} \overline{\Xi_{\vec s;\vec q,\vec p}^{(r')}}}
    \right]
    \\
    &=
    \sum_{r=0}^{k}
    \purple{\mathbb E_0}\!\left[
        \left|
            \purple{\Xi_{\vec s;\vec q,\vec p}^{(r)}}
        \right|^2
    \right].
\end{align}

From the above discussion, we can reduce to the evaluation of the following quantity
\begin{align}
    \frac{1}{\eta!} \mathbb{E}_0 \qty[\abs{ \purple{\Xi_{\vec s;\vec q,\vec p}^{(r)}}}^2 ]. 
\end{align}
For each
\(
\vec x\subseteq\vec s\cap\vec p
\)
with \(|\vec x|=r\), let
\(
\Xi_{\vec s;\vec q,\vec p;\vec x}^{(r)}
\)
denote the subsum of
Eq.~\eqref{eq:Xi-r-minor-form} . Then
\begin{align}
    \Xi_{\vec s;\vec q,\vec p}^{(r)}
    &=
    \sum_{\substack{
        \vec x\in\mac S_{N,r}\\
        \vec x\subseteq\vec s\cap\vec p
    }} \qty(\frac{r!}{k!}\,
    \sum_{\vec d\in\mac S_{\eta,m}}
    \sum_{\substack{
        \vec\alpha\subseteq[\eta]\setminus\vec d\\
        |\vec\alpha|=r
    }}
    \purple{
        \varepsilon_{\wedge}
        (\vec x,\vec s\setminus\vec x)
        \varepsilon_{\wedge}
        (\vec\alpha,[\eta]\setminus\vec\alpha)
        \varepsilon_{\wedge}
        (\vec x,\vec p\setminus\vec x)
        \varepsilon_{\wedge}
        (\vec\alpha,\vec d)
    }
    :\!
    \det G_{\vec s\setminus\vec x,\,
               [\eta]\setminus\vec\alpha}\,
    \det G_{\vec q,\,
               \vec\alpha\sqcup\vec d}\,
    \overline{
        \det G_{\vec p\setminus\vec x,\vec d}
    }
    \!:
    ) \\
    &= \sum_{\substack{
        \vec x\in\mac S_{N,r}\\
        \vec x\subseteq\vec s\cap\vec p
    }}
    \Xi_{\vec s;\vec q,\vec p;\vec x}^{(r)}.
    \label{eq:Xi-r-fixed-c-decomposition}
\end{align}

For simplicitly, 
we set \(m:=k-r\) and $\sigma_{\vec x} =\varepsilon_{\wedge}
    (\vec s\setminus\vec x,\vec x)
    \varepsilon_{\wedge}
    (\vec x,\vec p\setminus\vec x)$. Reindexing the column summation by
\(
\vec d=\vec a\setminus\purple{\vec\alpha}
\), we obtain
\begin{align}
    \Xi_{\vec s;\vec q,\vec p;\vec x}^{(r)}
    &=
    \frac{r!}{k!}\,
    \sigma_{\vec x}
    \sum_{\vec d\in\mac S_{\eta,m}}
    \sum_{\substack{
        \vec\alpha\subseteq[\eta]\setminus\vec d\\
        |\vec\alpha|=r
    }}
    \purple{
        \varepsilon_{\wedge}
        ([\eta]\setminus\vec\alpha,\vec\alpha)
        \varepsilon_{\wedge}
        (\vec\alpha,\vec d)
    }
    :\!
    \det G_{\vec s\setminus\vec x,\,
               [\eta]\setminus\vec\alpha}\,
    \det G_{\vec q,\,
               \vec\alpha\sqcup\vec d}\,
    \overline{
        \det G_{\vec p\setminus\vec x,\vec d}
    }
    \!:
    \\
    &=
    \frac{r!}{k!}\,
    \sigma_{\vec x}
    \sum_{\vec d\in\mac S_{\eta,m}}
    :\!
    \left(
        \bra{\vec s\setminus\vec x}
        \otimes
        \bra{\vec q}
    \right)
    \left(
        G^{\wedge(\eta-r)}
        \otimes
        G^{\wedge k}
    \right)
    \left[
        \sum_{\substack{
            \vec\alpha\subseteq[\eta]\setminus\vec d\\
            |\vec\alpha|=r
        }}
        \purple{
            \varepsilon_{\wedge}
            ([\eta]\setminus\vec\alpha,\vec\alpha)
            \varepsilon_{\wedge}
            (\vec\alpha,\vec d)
        }\,
        \ket{[\eta]\setminus\vec\alpha}
        \otimes
        \ket{\vec\alpha\sqcup\vec d}
    \right]
    \nonumber\\
    &\hspace{42mm}\times
    \overline{
        \bra{\vec p\setminus\vec x}
        G^{\wedge m}
        \ket{\vec d}
    }
    \!:
    \\
    &=
    \frac{r!}{k!}\,
    \sigma_{\vec x}
    \sum_{\vec d\in\mac S_{\eta,m}}
    :\!
    \left(
        \bra{\vec s\setminus\vec x}
        \otimes
        \bra{\vec q}
    \right)
    \left(
        G^{\wedge(\eta-r)}
        \otimes
        G^{\wedge k}
    \right)
    \iota_r
    \left(
        \ket{[\eta]}
        \otimes
        \ket{\vec d}
    \right)
    \overline{
        \bra{\vec p\setminus\vec x}
        G^{\wedge m}
        \ket{\vec d}
    }
    \!: .
    \label{eq:Xi-fixed-x-iota-representation}
\end{align}
In the last identity, we introduced the following linear map
$
    \iota_r
    \colon
    \bigwedge^{\eta}\mathbb C^\eta
    \otimes
    \bigwedge^{m}\mathbb C^\eta
    \rightarrow
    \bigwedge^{\eta-r}\mathbb C^\eta
    \otimes
    \bigwedge^{k}\mathbb C^\eta
$:
for every
\(
\vec d\in\mac S_{\eta,m}
\),
its action on the occupation-number basis is defined by
\begin{align}
    \iota_r
    \left(
        \ket{[\eta]}
        \otimes
        \ket{\vec d}
    \right)
    &:=
    \sum_{\substack{
        \vec\alpha\subseteq[\eta]\setminus\vec d\\
        |\vec\alpha|=r
    }}
    \varepsilon_{\wedge}
    ([\eta]\setminus\vec\alpha,\vec\alpha)\,
    \varepsilon_{\wedge}
    (\vec\alpha,\vec d)\,
    \ket{[\eta]\setminus\vec\alpha}
    \otimes
    \ket{\vec\alpha\sqcup\vec d}.
    \label{eq:def-iota-r}
\end{align}
Since
\(
\bigwedge^\eta\mathbb C^\eta
\)
is one-dimensional and
\(
\{\ket{\vec d}:\vec d\in\mac S_{\eta,m}\}
\)
is a basis of
\(
\bigwedge^m\mathbb C^\eta
\),
this prescription uniquely determines \(\iota_r\)
~\footnote{
Since
$
    \left\{
        \ket{[\eta]}\otimes\ket{\vec d}
        :
        \vec d\in\mac S_{\eta,m}
    \right\}
$
is a basis of
\(
\bigwedge^\eta\mathbb C^\eta
\otimes
\bigwedge^m\mathbb C^\eta
\),
Eq.~\eqref{eq:def-iota-r} uniquely extends by linearity to the map
\(\iota_r\). Therefore, \(\iota_r\) is well-defined.
}. 
We remark that this linear map plays an important role in the
representation-theoretic analysis below. Its intrinsic construction as
an exterior Pieri inclusion, together with its relevant properties, will
be discussed later.

Consequently, Eq.~\eqref{eq:Xi-fixed-x-iota-representation} may be
rewritten as
\begin{align}
    \Xi_{\vec s;\vec q,\vec p;\vec x}^{(r)}
    &=
    \frac{r!}{k!}\,
    \sigma_{\vec x}
    \sum_{\vec d\in\mac S_{\eta,m}}
    :\!
        \left(
            \bra{\vec s\setminus\vec x}
            \otimes
            \bra{\vec q}
        \right)
        \left(
            G^{\wedge(\eta-r)}
            \otimes
            G^{\wedge k}
        \right)
        \iota_r
        \left(
            \ket{[\eta]}
            \otimes
            \ket{\vec d}
        \right)
    \nonumber\\
    &\hspace{42mm}\times
    \overline{
        \bra{\vec p\setminus\vec x}
        G^{\wedge m}
        \ket{\vec d}
    }
    \!: .
    \label{eq:Xi-fixed-c-projected-representation}
\end{align}
Recalling lemma.~\ref{lem:two-block-Wick-factorization} and setting $F,H$ as follows
\begin{align}
    F(\vec x, \vec d) &=   (\bra{\vec s\setminus\vec x}
            \otimes
            \bra{\vec q} ) 
           \qty( G^{\wedge(\eta-r)}
            \otimes
            G^{\wedge k}) 
            \iota_r
        \left(
            \ket{[\eta]}
            \otimes
            \ket{\vec d}
        \right) \\
    H(\vec x, \vec d) &= \bra{\vec p\setminus\vec x}
        G^{\wedge m}
        \ket{\vec d}.
\end{align}
we obtain
\begin{align}
    \purple{\mathbb{E}_0} \qty[\Xi_{\vec s;\vec q,\vec p;\vec x}^{(r)} \overline{\Xi_{\vec s;\vec q,\vec p;\vec y}^{(r)}} ]
    &= \frac{(r!)^2}{(k!)^2}
    \purple{
        \sigma_{\vec x}\sigma_{\vec y}
    }
    \sum_{\vec d \in \mac{S}_{\eta, m}} \sum_{\vec c \in \mac{S}_{\eta, m}} \purple{\mathbb{E}_0}[:F(\vec x, \vec d)\overline{H(\vec x, \vec d)}: : \overline{F(\vec y, \vec c) \overline{H(\vec y, \vec c)}} :] \\
    &= \frac{(r!)^2}{(k!)^2}
    \purple{
        \sigma_{\vec x}\sigma_{\vec y}
    }
    \sum_{\vec d \in \mac{S}_{\eta, m}} \sum_{\vec c \in \mac{S}_{\eta, m}}
    \purple{\mathbb{E}_0}[F(\vec x, \vec d) \overline{F(\vec y, \vec c)} ] \purple{\mathbb{E}_0}[H(\vec y, \vec c) \overline{H(\vec x, \vec d)}]
\end{align}
From lemma \ref{lem:exterior-minor-orthogonality}, we obtain
\begin{align}
    \purple{\mathbb{E}_0}[H(\vec y, \vec c) \overline{H(\vec x, \vec d)}] &= \purple{\mathbb{E}_0}[\det G_{\vec p \setminus \vec y, \vec c} \overline{\det G_{\vec p \setminus \vec x, \vec d}} 
    ] \\
    &= m! \delta_{\vec p \setminus \vec y, \vec p \setminus \vec x} \delta_{\vec c, \vec d} = m! \delta_{\vec y, \vec x} \delta_{\vec c, \vec d}. 
\end{align}
From this fact, we can rewrite
\begin{align}
    \purple{\mathbb{E}_0} \qty[\Xi_{\vec s;\vec q,\vec p;\vec x}^{(r)} \overline{\Xi_{\vec s;\vec q,\vec p;\vec y}^{(r)}} ] 
    &=  \frac{(r!)^2}{(k!)^2}
    \purple{
        \sigma_{\vec x}\sigma_{\vec y}
    }
    \sum_{\vec d \in \mac{S}_{\eta, m}} \sum_{\vec c \in \mac{S}_{\eta, m}}
    \purple{\mathbb{E}_0}[F(\vec x, \vec d) \overline{F(\vec y, \vec c)} ] \cdot m! \delta_{\vec y, \vec x} \delta_{\vec c, \vec d} \\
    &= \frac{(r!)^2 m!}{(k!)^2} \sum_{\vec d \in \mac{S}_{\eta,m}} \purple{\mathbb{E}_0} \qty[\abs{F(\vec x, \vec d \ )}^2] \delta_{\vec x, \vec y}.
    \label{eq:target_quantity}
\end{align}
Here, we use $\abs{\sigma_{\vec c}}^2 =1$.

\subsubsection{Exterior Pieri components and projector diagonals}
\label{sec:exterior-pieri-components}

We first state the exterior Pieri rule and compute the occupation-basis
diagonal of each associated projector. We then use the decomposition to
characterize an exterior Pieri inclusion up to a nonzero scalar and
establish its behavior under linear maps. Finally, we show that the map
$\iota_r$ defined in Eq.~\eqref{eq:def-iota-r} is a particular
normalization of this inclusion. This identifies the irreducible
component relevant to the contraction decomposition and justifies the
insertion of its projector in the subsequent covariance calculation.

Throughout this subsection, we set $\mathbb{S}_\lambda(V)=\{0\}$ when
$\ell(\lambda)>\dim V$ and $\bigwedge^jV=\{0\}$ when
$j<0$ or $j>\dim V$. For a nonnegative integer $n$, we use
$\binom nj=0$ when $j<0$ or $j>n$.

\begin{lemma}[Exterior Pieri rule]
\label{lem:exterior-pieri-rule}
Let $V$ be an $n$-dimensional complex Hermitian vector space, and let
$a,b$ be integers satisfying $0\leq a,b\leq n$.

Then there is a multiplicity-free decomposition
\begin{align}
    \bigwedge^aV\otimes\bigwedge^bV
    &\simeq
    \bigoplus_{d=\max\{0,a+b-n\}}^{\min\{a,b\}}
    \mathbb{S}_{(2^d,1^{a+b-2d})}(V).
    \label{eq:exterior-pieri-decomposition}
\end{align}
We denote the orthogonal projector onto the summand indexed by $d$ by
$\Pi_d^{a,b;V}$, omitting $V$ when it is clear from the context.
The inner product is the one induced by the Hermitian structure on $V$.
\end{lemma}

\begin{proof}
Since $\bigwedge^aV\simeq\mathbb{S}_{(1^a)}(V)$ and
$\bigwedge^bV\simeq\mathbb{S}_{(1^b)}(V)$, the Littlewood--Richardson rule
\cite[Sec.~8.3, Corollary~2]{fulton1997young} gives the summands
$\mathbb{S}_{(2^d,1^{a+b-2d})}(V)$ for $0\leq d\leq\min\{a,b\}$,
each with multiplicity one. The partition $(2^d,1^{a+b-2d})$ has
$a+b-d$ parts, so its Schur module is nonzero precisely when
$d\geq a+b-n$. This gives Eq.~\eqref{eq:exterior-pieri-decomposition}.
\end{proof}

We now compute the occupation-basis diagonal of an exterior Pieri
projector. This calculation uses only the exterior Pieri rule and the
unitary symmetries of the projector.

\begin{lemma}[Diagonal of the relevant exterior Pieri projector]
\label{lem:exterior-pieri-projector-diagonal}
Let $V=\mathbb C^N$, and let $a,b,d$ be integers satisfying
$0\leq a,b\leq N$ and
$\max\{0,a+b-N\}\leq d\leq\min\{a,b\}$.
Let $\vec u\in\mathcal S_{N,a}$ and
$\vec t\in\mathcal S_{N,b}$, and set $h:=|\vec u\cap\vec t|$.

Then
\begin{align}
    \left\|
        \Pi_d^{a,b}(|\vec u\rangle\otimes|\vec t\rangle)
    \right\|^2
    &=p_{a,b,d}(h),
    \label{eq:exterior-pieri-projector-diagonal}
    \\
    p_{a,b,d}(h)
    &:=
    \frac{
        \binom{a+b-2h}{d-h}-\binom{a+b-2h}{d-h-1}
    }{\binom{a+b-2h}{b-h}}.
    \label{eq:exterior-pieri-diagonal-function}
\end{align}
The denominator may equivalently be written as
$\binom{a+b-2h}{a-h}$.
\end{lemma}

\begin{proof}
Let
\begin{align}
    \purple{\mathbb T_N}&:=\{\operatorname{diag}(z_1,\ldots,z_N)\mid |z_i|=1\}
    \subseteq\mathrm U(N),
    &
    \purple{\pi}(D)&:=D^{\wedge a}\otimes D^{\wedge b}
    \quad(D\in \purple{\mathbb T_N}).
    \label{eq:exterior-pieri-torus-action}
\end{align}
Consider the subspace
\begin{align}
    K&:=\operatorname{span}\left\{
        |\vec q\rangle\otimes|\vec p\rangle\mid
        \begin{array}{l}
            \vec q\in\mathcal S_{N,a},\quad
            \vec p\in\mathcal S_{N,b},\\
            \vec q\cap\vec p=\vec u\cap\vec t,\quad
            \vec q\cup\vec p=\vec u\cup\vec t
        \end{array}
    \right\}.
    \label{eq:exterior-pieri-weight-basis}
\end{align}
To specify such a basis tensor, it suffices to choose the $b-h$ indices
in $\vec p\setminus\vec q$ from the $a+b-2h$ indices in
$(\vec u\cup\vec t)\setminus(\vec u\cap\vec t)$. Therefore,
\begin{align}
    \dim K&=\binom{a+b-2h}{b-h}.
    \label{eq:exterior-pieri-weight-space-dimension}
\end{align}

For $D=\operatorname{diag}(z_1,\ldots,z_N)\in \purple{\mathbb T_N}$, an occupation-basis
tensor satisfies
\begin{align}
    \purple{\pi}(D)(|\vec q\rangle\otimes|\vec p\rangle)
    &={}
    \left(\prod_{i\in\vec q\cap\vec p}z_i^2\right)
    \left(\prod_{i\in(\vec q\cup\vec p)\setminus(\vec q\cap\vec p)}z_i\right)
    |\vec q\rangle\otimes|\vec p\rangle.
    \label{eq:exterior-pieri-basis-torus-action}
\end{align}
Thus every basis tensor in $K$ has eigenvalue
\begin{align}
    \chi(D)&:=\prod_{i\in\vec u}z_i\prod_{i\in\vec t}z_i.
    \label{eq:exterior-pieri-weight-character}
\end{align}
Conversely, if an occupation-basis tensor has eigenvalue $\chi(D)$ for
every $D\in \purple{\mathbb T_N}$, varying the phases independently shows that each index
occurs with the same multiplicity as in
$|\vec u\rangle\otimes|\vec t\rangle$.
Hence its intersection and union agree with $\vec u\cap\vec t$ and
$\vec u\cup\vec t$, respectively. Since occupation-basis tensors
simultaneously diagonalize the action of $\purple{\mathbb T_N}$, it follows that
\begin{align}
    K&=\left\{
        v\in\bigwedge^aV\otimes\bigwedge^bV\mid
        \purple{\pi}(D)v=\chi(D)v\text{ for every }D\in \purple{\mathbb T_N}
    \right\}.
    \label{eq:exterior-pieri-weight-space}
\end{align}

The Pieri summand $\mathbb{S}_{(2^d,1^{a+b-2d})}(V)$ is
$\mathrm U(N)$-invariant. Because the induced representation
$U^{\wedge a}\otimes U^{\wedge b}$ is unitary, the orthogonal
complement of this summand is also invariant. Hence its orthogonal
projector commutes with the induced action of every
$U\in\mathrm U(N)$. In particular, for $v\in K$ and $D\in \purple{\mathbb T_N}$,
\begin{align}
    \purple{\pi}(D)\Pi_d^{a,b}v
    &=\Pi_d^{a,b}\purple{\pi}(D)v
    =\chi(D)\Pi_d^{a,b}v.
    \label{eq:exterior-pieri-weight-space-preservation}
\end{align}
Thus $\Pi_d^{a,b}(K)\subseteq K$, and
$\Pi_d^{a,b}|_K$ is the orthogonal projector onto
$K\cap\operatorname{im}\Pi_d^{a,b}$.

Let $\mathcal B$ be the standard orthonormal basis of $K$.
Permutation matrices preserving both $\vec u\cap\vec t$ and
$\vec u\cup\vec t$ setwise act transitively, up to signs, on
$\mathcal B$. Since $\Pi_d^{a,b}$ commutes with their induced actions,
its diagonal entries on $\mathcal B$ have a common value $c$.
Orthogonality of the projector gives
$\|\Pi_d^{a,b}|e\rangle\|^2=\langle e|\Pi_d^{a,b}|e\rangle=c$
for every $|e\rangle\in\mathcal B$, while
$\tr(\Pi_d^{a,b}|_K)=(\dim K)c$.
Since $|\vec u\rangle\otimes|\vec t\rangle\in\mathcal B$, we obtain
\begin{align}
    \left\|
        \Pi_d^{a,b}(|\vec u\rangle\otimes|\vec t\rangle)
    \right\|^2
    &=\frac{\tr(\Pi_d^{a,b}|_K)}{\dim K}
    =\frac{\dim(K\cap\operatorname{im}\Pi_d^{a,b})}{\dim K}.
    \label{eq:exterior-pieri-diagonal-dimension-ratio}
\end{align}
It remains to compute the numerator.

The exterior Pieri rule gives
\begin{align}
    \bigwedge^{a+b-d}V\otimes\bigwedge^dV
    &\simeq\bigoplus_{j=0}^{d}\mathbb{S}_{(2^j,1^{a+b-2j})}(V).
    \label{eq:exterior-pieri-cumulative-decomposition}
\end{align}
For $d\geq1$, it also gives
\begin{align}
    \bigwedge^{a+b-d+1}V\otimes\bigwedge^{d-1}V
    &\simeq\bigoplus_{j=0}^{d-1}\mathbb{S}_{(2^j,1^{a+b-2j})}(V),
    \label{eq:exterior-pieri-previous-decomposition}
\end{align}
where exterior powers and Schur modules that exceed the dimension
bounds are understood to be zero. The right-hand sides differ by exactly the summand
$\mathbb{S}_{(2^d,1^{a+b-2d})}(V)$.

Both decompositions respect the action of $\purple{\mathbb T_N}$.
For $0\leq j\leq d$, let $m_j$ be the dimension of the subspace of
$\mathbb{S}_{(2^j,1^{a+b-2j})}(V)$ on which every $D\in \purple{\mathbb T_N}$ acts by
$\chi(D)$. By Eq.~\eqref{eq:exterior-pieri-weight-space} and the
identification
$\operatorname{im}\Pi_d^{a,b}\simeq\mathbb{S}_{(2^d,1^{a+b-2d})}(V)$,
\begin{align}
    m_d&=\dim(K\cap\operatorname{im}\Pi_d^{a,b}).
    \label{eq:exterior-pieri-numerator-multiplicity}
\end{align}

To determine $m_d$, we apply the same common-eigenspace construction
as for $K$ to the exterior tensor products in the two Pieri
decompositions expressed in Eqs.~\eqref{eq:exterior-pieri-cumulative-decomposition} and \eqref{eq:exterior-pieri-previous-decomposition}. We keep the character $\chi$ fixed and change
only the exterior degrees. The multiplicity of $\chi$ in each Schur
module is independent of its realization in an exterior tensor
product. Thus, for $d\geq1$, the dimensions of the corresponding
common eigenspaces are $\sum_{j=0}^{d}m_j$ and
$\sum_{j=0}^{d-1}m_j$, respectively, and their difference gives $m_d$.

Set $\vec I:=\vec u\cap\vec t$ and
$\vec X:=(\vec u\cup\vec t)\setminus\vec I$, so that
$|\vec I|=h$ and $|\vec X|=a+b-2h$.
For $0\leq r\leq d$, consider the common eigenspace consisting of
vectors $v\in\bigwedge^{a+b-r}V\otimes\bigwedge^rV$ satisfying
\begin{align*}
    \left(D^{\wedge(a+b-r)}\otimes D^{\wedge r}\right)v
    &=\chi(D)v
    \qquad\text{for every }D\in \purple{\mathbb T_N}.
\end{align*}
As in the construction of $K$, an occupation-basis tensor satisfies
this condition precisely when each index in $\vec I$ occurs in both
tensor factors, each index in $\vec X$ occurs in exactly one factor,
and no other index occurs.
The second factor now contains $r$ indices, of which $h$ are already
fixed by $\vec I$. Its remaining indices therefore form a subset
$\vec B\subseteq\vec X$ of size $r-h$, and the first factor contains
$\vec I\cup(\vec X\setminus\vec B)$.
Consequently, the common eigenspace has the orthonormal basis
\begin{align}
    \left\{
        |\vec I\cup(\vec X\setminus\vec B)\rangle
        \otimes|\vec I\cup\vec B\rangle\mid
        \vec B\subseteq\vec X,\ |\vec B|=r-h
    \right\}.
    \label{eq:exterior-pieri-multiplicity-basis}
\end{align}
All indices in each ket are arranged in increasing order.
Each choice of $\vec B$ gives one distinct basis tensor, so the
dimension of this common eigenspace is
$\binom{a+b-2h}{r-h}$.
If $r<h$, the eigenspace is $\{0\}$, consistently with our convention
for binomial coefficients.

If $d=0$, Eq.~\eqref{eq:exterior-pieri-cumulative-decomposition}
directly gives $m_0=\binom{a+b-2h}{-h}$, which agrees with the claimed
formula because $\binom{a+b-2h}{-h-1}=0$.
For $d\geq1$, applying the count above with $r=d$ and $r=d-1$ to the
two Pieri decompositions yields
\begin{align}
    \sum_{j=0}^{d}m_j
    &=\binom{a+b-2h}{d-h},
    \label{eq:exterior-pieri-cumulative-multiplicity}
    \\
    \sum_{j=0}^{d-1}m_j
    &=\binom{a+b-2h}{d-h-1}.
    \label{eq:exterior-pieri-previous-multiplicity}
\end{align}
Subtracting these identities gives
\begin{align}
    m_d
    &=\binom{a+b-2h}{d-h}-\binom{a+b-2h}{d-h-1}.
    \label{eq:exterior-pieri-weight-multiplicity}
\end{align}
As noted above, this identity also holds when $d=0$.
Finally, substituting
Eqs.~\eqref{eq:exterior-pieri-weight-space-dimension},
\eqref{eq:exterior-pieri-numerator-multiplicity}, and
\eqref{eq:exterior-pieri-weight-multiplicity} into
Eq.~\eqref{eq:exterior-pieri-diagonal-dimension-ratio} proves the claim.
\end{proof}

The exterior Pieri rule determines the inclusion relevant to the
contraction decomposition through the following uniqueness property.

\begin{lemma}[Exterior Pieri inclusion]
\label{lem:exterior-pieri-inclusion}
Let $\eta,k,r$ be integers satisfying $\eta\geq1$ and
$0\leq r\leq k\leq\eta$, and set $m:=k-r$.

Then, there exists a nonzero $\mtr{GL}_{\eta}(\mathbb{C})$-equivariant map
\begin{align}
    \iota_r^{\mathbb C^\eta}\colon
    \bigwedge^\eta\mathbb C^\eta\otimes\bigwedge^m\mathbb C^\eta
    &\longrightarrow
    \bigwedge^{\eta-r}\mathbb C^\eta\otimes\bigwedge^k\mathbb C^\eta.
    \label{eq:exterior-pieri-inclusion}
\end{align}
Such a map is unique up to multiplication by a nonzero scalar.
Moreover, it is injective and satisfies
\begin{align}
    \operatorname{im}\iota_r^{\mathbb C^\eta}
    &=\operatorname{im}\Pi_m^{\eta-r,k;\mathbb C^\eta}
    \simeq\mathbb{S}_{(2^m,1^{\eta-m})}(\mathbb C^\eta).
    \label{eq:exterior-pieri-abstract-inclusion-image}
\end{align}
\end{lemma}

\begin{proof}
We compare the irreducible decompositions of the domain and target.
The exterior Pieri rule gives
\begin{align}
    \bigwedge^\eta\mathbb C^\eta\otimes\bigwedge^m\mathbb C^\eta
    &\simeq\mathbb{S}_{(2^m,1^{\eta-m})}(\mathbb C^\eta),
    \label{eq:exterior-pieri-inclusion-domain}
    \\
    \bigwedge^{\eta-r}\mathbb C^\eta\otimes\bigwedge^k\mathbb C^\eta
    &\simeq
    \bigoplus_{q=m}^{\min\{\eta-r,k\}}
    \mathbb{S}_{(2^q,1^{\eta+m-2q})}(\mathbb C^\eta).
    \label{eq:exterior-pieri-inclusion-target}
\end{align}
The domain is irreducible, and the same irreducible representation
occurs exactly once in the target, at $q=m$. Schur's lemma therefore
gives
\begin{align}
    \dim\operatorname{Hom}_{\mtr{GL}_{\eta}(\mathbb{C})}\!\left(
        \bigwedge^\eta\mathbb C^\eta\otimes\bigwedge^m\mathbb C^\eta,
        \bigwedge^{\eta-r}\mathbb C^\eta\otimes\bigwedge^k\mathbb C^\eta
    \right)
    &=1.
    \label{eq:exterior-pieri-intertwiner-dimension}
\end{align}
Choose a nonzero map in this one-dimensional space and denote it by
$\iota_r^{\mathbb C^\eta}$. Its kernel is an invariant subspace of the
irreducible domain. Since the map is nonzero, its kernel is zero.
Its image is therefore isomorphic to the domain and must coincide with
the unique summand indexed by $q=m$ in
Eq.~\eqref{eq:exterior-pieri-inclusion-target}.
This proves Eq.~\eqref{eq:exterior-pieri-abstract-inclusion-image},
injectivity, and uniqueness up to a nonzero scalar.
\end{proof}

We next show that every such inclusion map $\iota_r^{\mathbb{C}^\eta }$ has the required compatibility
with linear maps between the underlying vector spaces.
This property does not depend on its scalar normalization.

\begin{corollary}
\label{cor:exterior-pieri-naturality}
Let $\eta,k,r,m$ and $\iota_r^{\mathbb C^\eta}$ be as in
Lemma~\ref{lem:exterior-pieri-inclusion}, let $N\geq\eta$ be an integer,
and let $G\colon\mathbb C^\eta\to\mathbb C^N$ be linear.

Then
\begin{align}
    \left(G^{\wedge(\eta-r)}\otimes G^{\wedge k}\right)
    \operatorname{im}\iota_r^{\mathbb C^\eta}
    &\subseteq\operatorname{im}\Pi_m^{\eta-r,k;\mathbb C^N}.
    \label{eq:exterior-pieri-abstract-image-under-G}
\end{align}
Consequently,
\begin{align}
    \Pi_m^{\eta-r,k;\mathbb C^N}
    \left(G^{\wedge(\eta-r)}\otimes G^{\wedge k}\right)
    \iota_r^{\mathbb C^\eta}
    &=\left(G^{\wedge(\eta-r)}\otimes G^{\wedge k}\right)
    \iota_r^{\mathbb C^\eta}.
    \label{eq:exterior-pieri-abstract-naturality}
\end{align}
\end{corollary}

\begin{proof}
We first show that a linear map between the underlying vector spaces
induces a map between the corresponding exterior Pieri summands.
Applying this fact to $G$ and using
Lemma~\ref{lem:exterior-pieri-inclusion} gives the claimed image
inclusion. The projector identity then follows because the projector
acts as the identity on its image.

Set $a:=\eta-r$, $b:=k$, and
$\lambda:=(2^m,1^{\eta-m})=(2^m,1^{a+b-2m})$.
We use the fact that linear maps between the underlying vector spaces
induce maps between the corresponding Pieri summands
\cite[Sec.~8.3]{fulton1997young}. More precisely, for every linear map
$A\colon V\to W$,
\begin{align}
    \left(A^{\wedge a}\otimes A^{\wedge b}\right)
    \left(\mathbb{S}_\lambda(V)\right)
    &\subseteq\mathbb{S}_\lambda(W),
    \label{eq:exterior-pieri-functoriality}
\end{align}
where the Schur modules are regarded as the corresponding summands in
the exterior Pieri decompositions.


By Lemma~\ref{lem:exterior-pieri-inclusion},
$\operatorname{im}\iota_r^{\mathbb C^\eta}
=\mathbb{S}_\lambda(\mathbb C^\eta)$ under this identification.
Applying Eq.~\eqref{eq:exterior-pieri-functoriality} to $A=G$ proves
Eq.~\eqref{eq:exterior-pieri-abstract-image-under-G}.
The projector acts as the identity on its image, which gives
Eq.~\eqref{eq:exterior-pieri-abstract-naturality}.
\end{proof}

We finally relate the representation-theoretic inclusion $\iota_r^{\mathbb{C}^\eta }$ to the map
$\iota_r$ defined in Eq.~\eqref{eq:def-iota-r}.
By the uniqueness statement in Lemma~\ref{lem:exterior-pieri-inclusion},
it suffices to verify that $\iota_r$ is nonzero and
$\mtr{GL}_{\eta}(\mathbb{C})$-equivariant. Its defining occupation-basis formula
then fixes a normalization of the exterior Pieri inclusion.

\begin{lemma}[Identification of the contraction map]
\label{lem:contraction-pieri-identification}
Let $\eta,k,r,m$ and $\iota_r^{\mathbb C^\eta}$ be as in
Lemma~\ref{lem:exterior-pieri-inclusion}, and let $\iota_r$ be the map
defined in Eq.~\eqref{eq:def-iota-r}.
Let $N\geq\eta$ be an integer and let
$G\colon\mathbb C^\eta\to\mathbb C^N$ be linear.

Then $\iota_r$ is nonzero and $\mtr{GL}_{\eta}(\mathbb{C})$-equivariant, and is
therefore a nonzero scalar multiple of $\iota_r^{\mathbb C^\eta}$.
In particular, it is injective and satisfies
\begin{align}
    \operatorname{im}\iota_r
    &=\operatorname{im}\iota_r^{\mathbb C^\eta}
    =\operatorname{im}\Pi_m^{\eta-r,k;\mathbb C^\eta}.
    \label{eq:exterior-pieri-inclusion-image}
\end{align}
Moreover,
\begin{align}
    \left(G^{\wedge(\eta-r)}\otimes G^{\wedge k}\right)
    \operatorname{im}\iota_r
    &\subseteq\operatorname{im}\Pi_m^{\eta-r,k;\mathbb C^N},
    \label{eq:exterior-pieri-image-under-G}
    \\
    \Pi_m^{\eta-r,k;\mathbb C^N}
    \left(G^{\wedge(\eta-r)}\otimes G^{\wedge k}\right)\iota_r
    &=\left(G^{\wedge(\eta-r)}\otimes G^{\wedge k}\right)\iota_r.
    \label{eq:exterior-pieri-inclusion-naturality}
\end{align}
\end{lemma}

\begin{proof}
We first prove equivariance directly from the occupation-basis formula
for $\iota_r$. To compare applying a linear transformation before and
after $\iota_r$, we express the same formula for arbitrary exterior
factors. For $v_1,\ldots,v_\eta\in\mathbb C^\eta$ and
$w\in\bigwedge^m\mathbb C^\eta$, we have
\begin{align}
    &\iota_r\bigl((v_1\wedge\cdots\wedge v_\eta)\otimes w\bigr)
    \notag\\
    &\qquad=
    \sum_{\vec\alpha\in\mathcal S_{\eta,r}}
    \varepsilon_\wedge([\eta]\setminus\vec\alpha,\vec\alpha)
    \left(\bigwedge_{i\in[\eta]\setminus\vec\alpha}v_i\right)
    \otimes
    \left(\left(\bigwedge_{i\in\vec\alpha}v_i\right)\wedge w\right).
    \label{eq:exterior-pieri-inclusion-construction}
\end{align}
Each indexed exterior product is taken in increasing index order, and
an empty exterior product equals $1$.
To verify this identity, note that both sides are alternating and
multilinear in $v_1,\ldots,v_\eta$ and linear in $w$.
Since $\bigwedge^\eta\mathbb C^\eta$ is one-dimensional, it suffices to
check it for $v_i=e_i$ and $w=|\vec d\rangle$, with
$\vec d\in\mathcal S_{\eta,m}$.
For these inputs, terms with $\vec\alpha\cap\vec d\neq\varnothing$
vanish. In the remaining terms, exterior multiplication gives
\begin{align}
    \left(\bigwedge_{i\in\vec\alpha}e_i\right)\wedge|\vec d\rangle
    &=\varepsilon_\wedge(\vec\alpha,\vec d)
      |\vec\alpha\sqcup\vec d\rangle.
    \label{eq:exterior-pieri-inclusion-basis-sign}
\end{align}
Thus the right-hand side of
Eq.~\eqref{eq:exterior-pieri-inclusion-construction} reproduces exactly
Eq.~\eqref{eq:def-iota-r}.

Substituting $Av_i$ for $v_i$ and $A^{\wedge m}w$ for $w$ in
Eq.~\eqref{eq:exterior-pieri-inclusion-construction} applies
$A^{\wedge(\eta-r)}\otimes A^{\wedge k}$ to each term on its
right-hand side, while leaving the shuffle signs unchanged.
Hence, for every $A\in\mtr{GL}_{\eta}(\mathbb{C})$,
\begin{align}
    \iota_r\left(A^{\wedge\eta}\otimes A^{\wedge m}\right)
    &=\left(A^{\wedge(\eta-r)}\otimes A^{\wedge k}\right)\iota_r.
    \label{eq:exterior-pieri-inclusion-equivariance}
\end{align}
Furthermore, the defining formula shows that
$\iota_r(|[\eta]\rangle\otimes|\vec d\rangle)$ is a sum of
$\binom{\eta-m}{r}>0$ distinct orthonormal basis tensors with
coefficients $\pm1$. Thus $\iota_r$ is nonzero.

By Eq.~\eqref{eq:exterior-pieri-intertwiner-dimension}, the nonzero
equivariant maps $\iota_r$ and $\iota_r^{\mathbb C^\eta}$ differ by a
nonzero scalar. Lemma~\ref{lem:exterior-pieri-inclusion} therefore gives
injectivity and Eq.~\eqref{eq:exterior-pieri-inclusion-image}.
Finally, Corollary~\ref{cor:exterior-pieri-naturality} applies to every
nonzero scalar multiple of $\iota_r^{\mathbb C^\eta}$, so it gives
Eqs.~\eqref{eq:exterior-pieri-image-under-G} and
\eqref{eq:exterior-pieri-inclusion-naturality} for $\iota_r$.
\end{proof}

\subsubsection{Pieri reduction and the second-moment bound}
\label{sec:Pieri-reduction-second-moment}

We now convert the contraction decomposition obtained in
Sec.~\ref{sec:decomposition_of_determinant} into a uniform bound on the
determinant-weighted second moment. For each contraction order \(r\),
unitary invariance and Schur's lemma reduce the channel covariance to a
scalar multiple of diagonal matrix elements of a single irreducible
exterior Pieri projector, while a reference configuration determines
the remaining scalar. Summing the mutually orthogonal contraction
components then yields the desired
\((k+1)\eta^k/k!\) bound.

\begin{lemma}
\label{lem:Xi-r-Pieri-reduction}
Let
\(
0\leq r\leq k\leq\eta
\)
and suppose that
\(
N\geq\eta+k.
\). Additionally, we
set
\(
m:=k-r
\),
and let
\(
\vec s\in\mac S_{N,\eta}
\)
and
\(
\vec q,\vec p\in\mac S_{N,k}.
\)
For every
\(
\vec c\subseteq\vec s\cap\vec p
\)
with
\(
|\vec c|=r
\),
define
\begin{align}
    h_{\vec c}
    &:=
    \left|
        \vec q\cap(\vec s\setminus\vec c)
    \right|
    =
    |\vec q\cap\vec s|
    -
    |\vec q\cap\vec c|.
    \label{eq:h-c-definition}
\end{align}
For the relevant Pieri component, introduce the shorthand
\begin{align}
    p_r(h)
    &:=
    p_{\eta-r,k,k-r}(h)
    \nonumber\\
    &=
    \frac{\eta-k+r+1}{\eta-h+1}
    \frac{\binom{k-h}{r}}{\binom{\eta-h}{r}},
    \label{eq:p-r-definition}
\end{align}
where
\(
p_{\eta-r,k,k-r}
\)
is the diagonal function from
Lemma~\ref{lem:exterior-pieri-projector-diagonal}.

Then, there exists a scalar
\(
\alpha_{\eta,k,r}\geq0
\) 
such that
\begin{align}
    \frac{1}{\eta!}
    \mathbb E_0
    \left[
        \left|
            \Xi_{\vec s;\vec q,\vec p}^{(r)}
        \right|^2
    \right]
    &=
    \alpha_{\eta,k,r}
    \sum_{\substack{
        \vec c\subseteq\vec s\cap\vec p\\
        |\vec c|=r
    }}
    p_r(h_{\vec c}).
    \label{eq:Xi-r-Pieri-diagonal}
\end{align}
Here and below,
\(
\binom{n}{j}=0
\)
when
\(
j<0
\)
or
\(
j>n
\).
\end{lemma}
\begin{proof}
We first express the second moment of the $r$-contraction component
in terms of a positive semidefinite Hermitian form on the relevant
exterior Pieri component. Unitary invariance and Schur's lemma then
show that this form is a scalar multiple of the inner product.
Finally, we evaluate the resulting squared projection norms using
the diagonal formula for the exterior Pieri projector.

Recall that the determinant-weighted second moment is expressed in
terms of the mutually orthogonal contraction channels as
\begin{align}
    \frac{1}{\eta!}
    \mathbb E_0
    \left[
        \left|\Xi_{\vec s;\vec q,\vec p}\right|^2
    \right]
    &=
    \sum_{j=0}^k
    \frac{1}{\eta!}
    \mathbb E_0
    \left[
        \left|\Xi_{\vec s;\vec q,\vec p}^{(j)}\right|^2
    \right].
\end{align}
We therefore fix a contraction order $r$ and evaluate its contribution.
Since
\begin{align}
    \Xi_{\vec s;\vec q,\vec p}^{(r)}
    &=
    \sum_{\substack{
        \vec c\subseteq\vec s\cap\vec p\\
        |\vec c|=r
    }}
    \Xi_{\vec s;\vec q,\vec p;\vec c}^{(r)},
\end{align}
this amounts to evaluating the covariances of the terms indexed by
the contraction sets $\vec c$.

For $\vec d\in\mac S_{\eta,m}$, write
\begin{align}
    \Phi_{\vec d}^{(N)}(G)
    &:=
    \left(
        G^{\wedge(\eta-r)}
        \otimes
        G^{\wedge k}
    \right)
    \iota_r
    \left(
        \ket{[\eta]}
        \otimes
        \ket{\vec d}
    \right).
    \label{eq:Phi-d-definition}
\end{align}
By the naturality of the exterior Pieri inclusion,
$\Phi_{\vec d}^{(N)}(G)\in\operatorname{im}\Pi_m^{\eta-r,k}$.
The covariances of the fixed-contraction terms can be expressed
through the following Hermitian form on this Pieri component: namely,
for $\ket{z},\ket{z'}\in\operatorname{im}\Pi_m^{\eta-r,k}$, we can define
\begin{align}
    \mathsf B_N(\ket{z}, \ket{z'} )
    &:=
    \frac{(r!)^2m!}{(k!)^2\eta!}
    \sum_{\vec d\in\mac S_{\eta,m}}
    \mathbb E_0
    \left[
        \braket{z}{\Phi_{\vec d}^{(N)}(G)}
        \braket{\Phi_{\vec d}^{(N)}(G)}{z'}
    \right].
    \label{eq:Xi-invariant-Hermitian-form}
\end{align}
This form is positive semidefinite. Using this Hermitian form, Eq.~\eqref{eq:target_quantity} can be rewritten by
\begin{align}
    \frac{1}{\eta!}
    \mathbb E_0
    \left[
        \Xi_{\vec s;\vec q,\vec p;\vec c}^{(r)}
        \overline{
            \Xi_{\vec s;\vec q,\vec p;\vec c'}^{(r)}
        }
    \right]
    &=
    \delta_{\vec c,\vec c'}\,
    \mathsf B_N
    \left(
        \Pi_m^{\eta-r,k}
        \left(
            \ket{\vec s\setminus\vec c}
            \otimes
            \ket{\vec q}
        \right),
        \Pi_m^{\eta-r,k}
        \left(
            \ket{\vec s\setminus\vec c'}
            \otimes
            \ket{\vec q}
        \right)
    \right).
    \label{eq:Xi-fixed-c-norm-quadratic-form}
\end{align}
Thus, distinct contraction sets contribute orthogonally, and
\begin{align}
    \frac{1}{\eta!}
    \mathbb E_0
    \left[
        \left|\Xi_{\vec s;\vec q,\vec p}^{(r)}\right|^2
    \right]
    &=
    \sum_{\substack{
        \vec c\subseteq\vec s\cap\vec p\\
        |\vec c|=r
    }}
    \mathsf B_N
    \left(
        \Pi_m^{\eta-r,k}
        \left(
            \ket{\vec s\setminus\vec c}
            \otimes
            \ket{\vec q}
        \right),
        \Pi_m^{\eta-r,k}
        \left(
            \ket{\vec s\setminus\vec c}
            \otimes
            \ket{\vec q}
        \right)
    \right).
    \label{eq:Xi-r-second-moment-Hermitian-form}
\end{align}

It remains to determine $\mathsf B_N$ using unitary invariance.
Let
\begin{align}
    \purple{\pi_N}(U)
    &:=
    \left.
    \left(
        U^{\wedge(\eta-r)}
        \otimes
        U^{\wedge k}
    \right)
    \right|_{\operatorname{im}\Pi_m^{\eta-r,k}},
    \qquad U\in \mtr{U}(N).
\end{align}
The identity
\begin{align}
    \Phi_{\vec d}^{(N)}(U^\dagger G)
    &=
    \purple{\pi_N}(U)^\dagger\Phi_{\vec d}^{(N)}(G),
\end{align}
together with the fact that $U^\dagger G$ has the same distribution
as $G$, gives
\begin{align}
    \mathsf B_N
    \left(
        \purple{\pi_N}(U) \ket{z},
        \purple{\pi_N}(U) \ket{z'}
    \right)
    &=
    \mathsf B_N(\ket{z}, \ket{z'} ).
    \label{eq:Xi-Hermitian-form-invariance}
\end{align}
Equivalently, the positive operator $\mathsf M_N$ determined by
$\mathsf B_N(z,z')=\bra{z}\mathsf M_N\ket{z'}$
commutes with $\purple{\pi_N}(U)$ for every $U\in \mtr{U}(N)$.
Since $
    \operatorname{im}\Pi_m^{\eta-r,k}
    \simeq
    \mathbb{S}_{(2^m,1^{\eta-m})}(\mathbb C^N)
$
is an irreducible $\mtr{U}(N)$-representation, Schur's lemma gives $\mathsf{M}_N \propto \1_{\mathbb{S}_{(2^m,1^{\eta-m})}(\mathbb C^N)}$, which leads to
\begin{align}
    \mathsf B_N(\ket{z}, \ket{z'} )
    &=
    \alpha_{N;\eta,k,r}\braket{z}{z'}
    \label{eq:Xi-Hermitian-operator-Schur}
\end{align}
for some $\alpha_{\eta,k,r}\geq0$.
Substituting this identity into
Eq.~\eqref{eq:Xi-r-second-moment-Hermitian-form} yields
\begin{align}
    \frac{1}{\eta!}
    \mathbb E_0
    \left[
        \left|\Xi_{\vec s;\vec q,\vec p}^{(r)}\right|^2
    \right]
    &=
    \alpha_{\eta,k,r}
    \sum_{\substack{
        \vec c\subseteq\vec s\cap\vec p\\
        |\vec c|=r
    }}
    \left\|
        \Pi_m^{\eta-r,k}
        \left(
            \ket{\vec s\setminus\vec c}
            \otimes
            \ket{\vec q}
        \right)
    \right\|^2,
\end{align}
which proves the existence of a sclar $\alpha_{\eta,k,r}$.

Finally, apply
Lemma~\ref{lem:exterior-pieri-projector-diagonal}
with $(a,b,d,h)=(\eta-r,k,k-r,h_{\vec c})$.
By Eq.~\eqref{eq:p-r-definition},
\begin{align}
    \left\|
        \Pi_m^{\eta-r,k}
        \left(
            \ket{\vec s\setminus\vec c}
            \otimes
            \ket{\vec q}
        \right)
    \right\|^2
    &=p_r(h_{\vec c}),
\end{align}
which gives Eq.~\eqref{eq:Xi-r-Pieri-diagonal}.
\end{proof}

\begin{lemma}
\label{lem:fixed-contraction-normalization}
Under the same assumption as in Lemma.~\ref{lem:Xi-r-Pieri-reduction}, the normalization is expressed as,
\begin{align}
    \alpha_{\eta,k,r}
    &=
    \frac{\eta+1}
         {(\eta-k+r+1)\binom{k}{r}}
    \binom{\eta}{k}.
    \label{eq:alpha-explicit}
\end{align}
In particular,
\(
\alpha_{\eta,k,r}>0
\).
Consequently, the following identity holds,
\begin{align}
    \frac{1}{\eta!}
    \mathbb E_0
    \left[
        \left|
            \Xi_{\vec s;\vec q,\vec p}^{(r)}
        \right|^2
    \right]
    &=
    \frac{\eta+1}
         {(\eta-k+r+1)\binom{k}{r}}
    \binom{\eta}{k}
    \sum_{\substack{
        \vec c\subseteq\vec s\cap\vec p\\
        |\vec c|=r
    }}
    p_r(h_{\vec c}).
    \label{eq:Xi-r-explicit-norm}
\end{align}
\end{lemma}

\begin{proof}
We determine $\alpha_{\eta,k,r}$ by evaluating the same second
moment in two ways for a reference configuration satisfying
$\vec q\cap\vec s=\varnothing$ and $\vec p\subseteq\vec s$.
The Pieri reduction expresses this moment in terms of
$\alpha_{\eta,k,r}$, while independence of the disjoint Gaussian
row blocks and orthogonality of minors allow us to evaluate it
directly. Comparing the two expressions yields the stated formula
and establishes its independence of $N$.

Since $0 \leq k \leq \min(\eta, N-\eta)$, we can choose a reference configuration satisfying
$
    \vec q\cap\vec s
    =
    \varnothing,\
    \vec p
    \subseteq
    \vec s.
    \label{eq:reference-configuration}
$
Then, 
\(
h_{\vec c}=0
\)
for every
\(
\vec c\subseteq\vec p
\)
with
\(
|\vec c|=r
\).
Hence Eq.~\eqref{eq:Xi-r-Pieri-diagonal} gives
\begin{align}
    \frac{1}{\eta!}
    \mathbb E_0
    \left[
        \left|
            \Xi_{\vec s;\vec q,\vec p}^{(r)}
        \right|^2
    \right]
    &=
    \binom{k}{r}\,
    \alpha_{\eta,k,r}\,
    p_r(0).
    \label{eq:reference-Schur-evaluation}
\end{align}

We now evaluate the same quantity directly using
Eq.~\eqref{eq:target_quantity}.
By Eqs.~\eqref{eq:Xi-fixed-x-iota-representation} and \eqref{eq:def-iota-r}, expanding the exterior
Pieri inclusion gives
\begin{align}
    F(\vec x,\vec d)
    &=  (\bra{\vec s\setminus\vec x}
            \otimes
            \bra{\vec q} ) 
        \Pi_m^{\eta-r,k}
           \qty( G^{\wedge(\eta-r)}
            \otimes
            G^{\wedge k}) 
            \iota_r
        \left(
            \ket{[\eta]}
            \otimes
            \ket{\vec d}
        \right)
        \\
    &=
    \left(
        \bra{\vec s\setminus\vec x}
        \otimes
        \bra{\vec q}
    \right)
    \left(
        G^{\wedge(\eta-r)}
        \otimes
        G^{\wedge k}
    \right)
    \left[
        \sum_{\substack{
            \vec\alpha\subseteq[\eta]\setminus\vec d\\
            |\vec\alpha|=r
        }}
        \purple{
            \varepsilon_{\wedge}
            ([\eta]\setminus\vec\alpha,\vec\alpha)
            \varepsilon_{\wedge}
            (\vec\alpha,\vec d)
        }\,
        \ket{[\eta]\setminus\vec\alpha}
        \otimes
        \ket{\vec\alpha\sqcup\vec d}
    \right] \\
    &=
    \sum_{\substack{
        \vec\alpha\subseteq[\eta]\setminus\vec d\\
        |\vec\alpha|=r
    }}
    \tau_{\vec\alpha,\vec d}\,
    \det G_{\vec s\setminus\vec x,\,
              [\eta]\setminus\vec\alpha}\,
    \det G_{\vec q,\,
              \vec\alpha\sqcup\vec d},
    \label{eq:reference-F-expansion}
\end{align}
where
$
    \tau_{\vec\alpha,\vec d}
    :=
    \varepsilon_{\wedge}
    ([\eta]\setminus\vec\alpha,\vec\alpha)\,
    \varepsilon_{\wedge}
    (\vec\alpha,\vec d).
$
Fix
\(
\vec x\subseteq\vec p
\)
with
\(
|\vec x|=r
\)
and
\(
\vec d\in\mac S_{\eta,m}.
\)
Because
\(
\vec s\setminus\vec x\subseteq\vec s
\)
and
\(
\vec q\cap\vec s=\varnothing,
\)
the Gaussian entries belonging to the row blocks
\(
\vec s\setminus\vec x
\)
and
\(
\vec q
\)
are independent. Therefore, for any admissible
\(
\vec\alpha,\vec\beta
\),
\begin{align}
    &\purple{\mathbb E_0}
    \left[
        \det G_{\vec s\setminus\vec x,\,
                  [\eta]\setminus\vec\alpha}\,
        \det G_{\vec q,\,
                  \vec\alpha\sqcup\vec d}
        \,
        \overline{
            \det G_{\vec s\setminus\vec x,\,
                      [\eta]\setminus\vec\beta}\,
            \det G_{\vec q,\,
                      \vec\beta\sqcup\vec d}
        }
    \right] \\
    &=
    \purple{\mathbb E_0}
    \left[
        \det G_{\vec s\setminus\vec x,\,
                  [\eta]\setminus\vec\alpha}\,
        \overline{
            \det G_{\vec s\setminus\vec x,\,
                      [\eta]\setminus\vec\beta}
        }
    \right]
    \purple{\mathbb E_0}
    \left[
        \det G_{\vec q,\,
                  \vec\alpha\sqcup\vec d}\,
        \overline{
            \det G_{\vec q,\,
                      \vec\beta\sqcup\vec d}
        }
    \right]
    \\
    &=
    (\eta-r)!\,k!\,
    \delta_{[\eta]\setminus\vec\alpha,\,
            [\eta]\setminus\vec\beta}\,
    \delta_{\vec\alpha\sqcup\vec d,\,
            \vec\beta\sqcup\vec d}
    \\
    &=
    (\eta-r)!\,k!\,
    \delta_{\vec\alpha,\vec\beta}.
    \label{eq:reference-F-cross-term}
\end{align}
Here, Lemma~\ref{lem:exterior-minor-orthogonality} was applied
separately to the minors of degrees
\(
\eta-r
\)
and
\(
k
\).
Moreover,
\(
[\eta]\setminus\vec\alpha
=
[\eta]\setminus\vec\beta
\)
is equivalent to
\(
\vec\alpha=\vec\beta
\);
equivalently, since both
\(
\vec\alpha
\)
and
\(
\vec\beta
\)
are disjoint from
\(
\vec d
\),
the equality
\(
\vec\alpha\sqcup\vec d
=
\vec\beta\sqcup\vec d
\)
also implies
\(
\vec\alpha=\vec\beta
\).

Expanding the squared norm in
Eq.~\eqref{eq:reference-F-expansion} and using
Eq.~\eqref{eq:reference-F-cross-term}, we thus obtain
\begin{align}
    \purple{\mathbb E_0}
    \left[
        \left|
            F(\vec x,\vec d)
        \right|^2
    \right]
    &=
    (\eta-r)!\,k!
    \sum_{\substack{
        \vec\alpha\subseteq[\eta]\setminus\vec d\\
        |\vec\alpha|=r
    }}
    \left|
        \tau_{\vec\alpha,\vec d}
    \right|^2
    =
    (\eta-r)!\,k! \binom{\eta-m}{r}.
    \label{eq:reference-F-norm}
\end{align}
Indeed,
\(
|[\eta]\setminus\vec d|=\eta-m
\),
so there are exactly
\(
\binom{\eta-m}{r}
\)
admissible choices of
\(
\vec\alpha
\).

Summing Eq.~\eqref{eq:target_quantity} over the contracted row sets
and using
\(
|\vec p|=k
\)
therefore yields
\begin{align}
    \frac{1}{\eta!}
    \purple{\mathbb E_0}
    \left[
        \left|
            \Xi_{\vec s;\vec q,\vec p}^{(r)}
        \right|^2
    \right]
    &=
    \frac{1}{\eta!}
    \sum_{\substack{
        \vec x\in\mac S_{N,r}\\
        \vec x\subseteq\vec s\cap\vec p
    }}
    \purple{\mathbb E_0}
    \left[
        \left|
            \Xi_{\vec s;\vec q,\vec p;\vec x}^{(r)}
        \right|^2
    \right]
    &&
    \qty(
        \because\
        \text{Eqs.~\eqref{eq:Xi-r-fixed-c-decomposition}
        and~\eqref{eq:target_quantity}}
    )
    \\
    &=
    \frac{1}{\eta!}
    \sum_{\substack{
        \vec x\in\mac S_{N,r}\\
        \vec x\subseteq\vec s\cap\vec p
    }}
    \frac{(r!)^2m!}{(k!)^2}
    \sum_{\vec d\in\mac S_{\eta,m}}
    \purple{\mathbb E_0}
    \left[
        \left|
            F(\vec x,\vec d)
        \right|^2
    \right]
    &&
    \qty(
        \because\
        \text{Eq.~\eqref{eq:target_quantity}}
    )
    \\
    &=
    \frac{1}{\eta!}
    \sum_{\substack{
        \vec x\in\mac S_{N,r}\\
        \vec x\subseteq\vec s\cap\vec p
    }}
    \frac{(r!)^2m!}{(k!)^2}
    \sum_{\vec d\in\mac S_{\eta,m}}
    \binom{\eta-m}{r}
    (\eta-r)!\,k!
    &&
    \qty(
        \because\
        \text{Eq.~\eqref{eq:reference-F-norm}}
    )
    \\
    &=
    \frac{1}{\eta!}
    \binom{k}{r}
    \frac{(r!)^2m!}{(k!)^2}
    \binom{\eta}{m}
    \binom{\eta-m}{r}
    (\eta-r)!\,k!
    &&
    \qty(
        \because\
        \vec p\subseteq\vec s,\ 
        |\vec p|=k,\ 
        |\mac S_{\eta,m}|=\binom{\eta}{m}
    )
    \\
    &=
    \binom{\eta-r}{m}
    =
    \binom{\eta-r}{k-r}.
    \label{eq:reference-direct-evaluation}
\end{align}

Comparing
Eqs.~\eqref{eq:reference-Schur-evaluation} and
\eqref{eq:reference-direct-evaluation}, we obtain
\begin{align}
    \alpha_{\eta,k,r}
    &=
    \frac{
        \binom{\eta-r}{k-r}
    }{
        \binom{k}{r}\,p_r(0)
    }.
\end{align}
Substituting \(p_r(0)\) from
Eq.~\eqref{eq:p-r-definition} and simplifying gives
\begin{align}
    \alpha_{\eta,k,r}
    &=
    \frac{\eta+1}
         {(\eta-k+r+1)\binom{k}{r}}
    \binom{\eta}{k}.
\end{align}
The asserted positivity and independence of \(N\) and of the index
sets follow immediately.
\end{proof}

\begin{proof}[
Proof of Lemma~\ref{lem:gaussian-moment-bound}
]
As shown in Sec.~\ref{sec:decomposition_of_determinant}, we have
\begin{align}
    \det G_{\vec s}\,
    \purple{{}^{k}\mathfrak D_{\vec q}^{\vec p}}
    (GG^\dagger) = \sum_{r=0}^{k}
    \Xi_{\vec s;\vec q,\vec p}^{(r)},
\end{align}
and $\Xi_{\vec s;\vec q,\vec p}^{(r)} \in \mac{W}_{\eta + k -r, k -r}$. Hence the orthogonality of
the complex Wiener--It\^o chaoses gives
\begin{align}
    \frac{1}{\eta!}
    \mathbb E_0
    \left[
        \left|
            \det G_{\vec s}
        \right|^2
        \left|
            \purple{{}^{k}\mathfrak D_{\vec q}^{\vec p}}
            (GG^\dagger)
        \right|^2
    \right]
    =
    \sum_{r=0}^{k}
    \frac{1}{\eta!}
    \mathbb E_0
    \left[
        \left|
            \Xi_{\vec s;\vec q,\vec p}^{(r)}
        \right|^2
    \right].
    \label{eq:sum-over-contraction-levels}
\end{align}

For each admissible contracted set
\(
\vec c\subseteq\vec s\cap\vec p
\)
with
\(
|\vec c|=r,
\)
set
$
    h_{\vec c}
    :=
    \left|
        \vec q\cap(\vec s\setminus\vec c)
    \right|.
$
By Lemmas~\ref{lem:Xi-r-Pieri-reduction} and
\ref{lem:fixed-contraction-normalization},
\begin{align}
    \frac{1}{\eta!}\purple{\mathbb E_0}
    \left[
        \left|
            \Xi_{\vec s;\vec q,\vec p}^{(r)}
        \right|^2
    \right]
    =
    \frac{\eta+1}
         {(\eta-k+r+1)\binom{k}{r}}
    \binom{\eta}{k}
    \sum_{\substack{
        \vec c\subseteq\vec s\cap\vec p\\
        |\vec c|=r
    }}
    p_r(h_{\vec c}).
    \label{eq:r-channel-explicit-second-moment}
\end{align}
By definition,
\(
p_r(h_{\vec c})
\)
is the squared norm of the orthogonal projection of the unit vector
\(
\ket{\vec s\setminus\vec c}\otimes\ket{\vec q}
\)
onto the relevant exterior Pieri component. Therefore,
\begin{align}
    0
    \leq
    p_r(h_{\vec c})
    \leq
    1.
    \label{eq:p-r-unit-bound}
\end{align}
Moreover,
\begin{align}
    \#\left\{
        \vec c\subseteq\vec s\cap\vec p:
        |\vec c|=r
    \right\}
    &=
    \binom{|\vec s\cap\vec p|}{r}
    \leq
    \binom{k}{r}.
\end{align}
It follows from
Eq.~\eqref{eq:r-channel-explicit-second-moment} that
\begin{align}
    \frac{1}{\eta!}
    \mathbb E_0
    \left[
        \left|
            \Xi_{\vec s;\vec q,\vec p}^{(r)}
        \right|^2
    \right]
    &\leq
    \frac{\eta+1}{\eta-k+r+1}
    \binom{\eta}{k}.
    \label{eq:r-channel-upper-bound}
\end{align}
Summing Eq.~\eqref{eq:r-channel-upper-bound} over \(r\) and using
Eq.~\eqref{eq:sum-over-contraction-levels}, we obtain
\begin{align}
    &\frac{1}{\eta!}
    \mathbb E_0
    \left[
        \left|
            \det G_{\vec s}
        \right|^2
        \left|
            \purple{{}^{k}\mathfrak D_{\vec q}^{\vec p}}
            (GG^\dagger)
        \right|^2
    \right]
    \nonumber\\
    &\qquad\leq
    \binom{\eta}{k}
    (\eta+1)
    \sum_{r=0}^{k}
    \frac{1}{\eta-k+r+1}.
    \label{eq:pre-harmonic-bound}
\end{align}

We bound the right-hand side of
Eq.~\eqref{eq:pre-harmonic-bound} directly.
For $k\geq3$ and $\eta\geq k$, cancellation of each denominator
against the corresponding factor gives
\begin{align}
    \binom{\eta}{k}(\eta+1)
    \sum_{r=0}^{k}
    \frac{1}{\eta-k+r+1}
    &=
    \frac{1}{k!}
    \sum_{j=-1}^{k-1}
    \prod_{\substack{\ell=-1\\ \ell\neq j}}^{k-1}
    (\eta-\ell)
    \nonumber\\
    &\leq
    \frac{k+1}{k!}\eta^k.
    \label{eq:direct-product-bound}
\end{align}
To justify the inequality, consider each product separately.
If the factor $\eta+1$ is absent, all $k$ factors are at most $\eta$.
Otherwise, since $k\geq3$, at least one factor $\eta-\ell$
with $\ell\geq1$ remains. Pairing these two factors yields
\begin{align}
    (\eta+1)(\eta-\ell)
    \leq
    (\eta+1)(\eta-1)
    \leq
    \eta^2,
\end{align}
while each of the remaining $k-2$ factors is at most $\eta$.
Thus, every product is bounded by $\eta^k$.

For $k=2$, direct calculation gives
\begin{align}
    \binom{\eta}{2}(\eta+1)
    \left(
        \frac{1}{\eta-1}
        +\frac{1}{\eta}
        +\frac{1}{\eta+1}
    \right)
    =
    \frac{3\eta^2-1}{2}
    \leq
    \frac{3}{2}\eta^2.
\end{align}
Consequently, Eq.~\eqref{eq:pre-harmonic-bound} implies
\begin{align}
    \frac{1}{\eta!}
    \mathbb E_0
    \left[
        \left|\det G_{\vec s}\right|^2
        \left|
            {}^{k}\mathfrak D_{\vec q}^{\vec p}
            (GG^\dagger)
        \right|^2
    \right]
    \leq
    \frac{k+1}{k!}\eta^k,
    \qquad k\geq2.
\end{align}
\end{proof}
   
\end{revision}

\subsection{Weingarten Integral over Grassmannian}
\label{appendix:explicit_weingarten}

\subsubsection{Partition, Young Diagram, Irreducible character}
\label{sec:prelim}

In this section, we introduce the mathematical tools based on Young diagrams and representation theory. Our purpose here is to briefly introduce the concepts that underpin the Weingarten evaluation carried out in Section~\ref{sec:weingarten}: integer
partitions and their Young diagrams, the content and hook length of a box, the
irreducible characters of the symmetric group, and the principal specialisation of
the Schur functions. We collect here only what is subsequently needed, and refer the
reader to the standard references (e.g. Refs.~\cite{fulton1997young, georgi2000lie, ceccherini2010representation}) for proofs and a fuller account.

A partition $\lambda$ of a non-negative integer $n$, written $\lambda\vdash n$, is defined as a weakly
decreasing sequence $\lambda=(\lambda_1\ge\lambda_2\ge\cdots\ge\lambda_\ell>0)$ of
positive integers with $\sum_i\lambda_i=n$.
We identify $\lambda$ with its Young diagram as the
left-justified array of $n$ boxes whose $i$-th row contains $\lambda_i$ boxes. For instance, $\lambda =(2,1)$ can be represented as a Young diagram given by
\begin{align}
    \lambda = \ydiagram{2,1} \ .
\end{align}
Additionally, we
denote by $(i,j)\in\lambda$ the box in row $i$ and column $j$. For instance, $(1,1) \in \lambda$. 

To each box $(i,j)\in\lambda$ we attach two integers as the content: The content of the box is defined by
\begin{equation}
    c(i,j)\coloneqq j-i,
\end{equation}
which increases by one along each row and decreases by one along each column. For
$\lambda=(3,1)$, labelling each box with its content, we can express it as follows,
\begin{equation}
    \begin{ytableau} 0 & 1 & 2 \\ \text{-}1 \end{ytableau}\,.
\end{equation}

We also recall the parametrization of conjugacy classes in
$\mathfrak S_n$. For a permutation $\sigma\in\mathfrak S_n$, its cycle
type is the partition $\purple{\kappa}\vdash n$ formed by the lengths of the
disjoint cycles of $\sigma$, with each fixed point contributing a part
of size $1$. For example,
$
\sigma=(1,2,4)(3,5)\in\mathfrak S_5
$
has cycle type $\purple{\kappa}=(3,2)\vdash 5$. This cycle type characterizes the
conjugacy class of $\sigma$ in $\mathfrak S_n$: two permutations in
$\mathfrak S_n$ are conjugate if and only if they have the same cycle
type.

Characters are basis-independent invariants of representations. For
finite groups, irreducible characters distinguish irreducible
representations. For each $\lambda\vdash n$, let $V_\lambda$ be the irreducible
representation space labeled by $\lambda$, and let
\(\pi^\lambda(\sigma)\) denote the linear operator representing
\(\sigma\in\mathfrak S_n\) on \(V_\lambda\). Its character is defined by
\begin{align}
    \chi^\lambda(\sigma)
    \coloneqq
    \tr_{V_\lambda}\!\left[\pi^\lambda(\sigma)\right].
\end{align}
In particular,
$\chi^\lambda(\mtr{id})$ is the dimension of the irreducible
representation $V_\lambda$.

Since characters are constant on conjugacy classes, the value
$\chi^\lambda(\sigma)$ depends only on the cycle type of $\sigma$.
Accordingly, for a partition $\purple{\kappa}\vdash n$, we define
\begin{align}
\chi^\lambda(\purple{\kappa})
\coloneqq
\chi^\lambda(\sigma),
\end{align}
where $\sigma\in\mathfrak S_n$ is any permutation of cycle type $\purple{\kappa}$..
This definition is independent of the choice of \(\sigma\).
Thus, in the notation \(\chi^\lambda(\purple{\kappa})\), the partition \(\lambda\)
labels the irreducible character, whereas \(\purple{\kappa}\) labels the conjugacy
class on which the character is evaluated.

As a simple illustration, consider \(n=2\). The partitions of \(2\) are
\((2)\) and \((1,1)\). Hence, \(\mathfrak S_2\) has two isomorphism
classes of irreducible complex representations, namely the trivial
representation and the sign representation, labeled by \((2)\) and
\((1,1)\), respectively. Both representations are one-dimensional:
\[
    \chi^{(2)}(\mtr{id})
    =
    \chi^{(1,1)}(\mtr{id})
    =
    1.
\]
The conjugacy classes of \(\mathfrak S_2\) are likewise indexed by the
cycle types \((1,1)\) and \((2)\), corresponding to the identity
permutation and the transposition, respectively. The corresponding
character table is
\begin{equation}
    \renewcommand{\arraystretch}{1.2}
    \begin{array}{c|cc}
        \lambda\backslash\purple{\kappa}
            & (1,1) & (2) \\\hline
        (2)     & 1 & 1 \\
        (1,1)   & 1 & -1
    \end{array}
    \label{eq:table_character}
\end{equation}
The column indexed by \(\purple{\kappa}=(1,1)\) corresponds to the identity
conjugacy class and therefore gives the dimensions
\(\chi^\lambda(\mtr{id})\). The row indexed by \(\lambda=(2)\) is the
trivial character, whereas the row indexed by \(\lambda=(1,1)\) is the
sign character.

More generally, for \(d\geq 1\) and partitions
\(\lambda,\purple{\kappa}\vdash d\), the character value
\(\chi^\lambda(\purple{\kappa})\) can be computed recursively using the
Murnaghan--Nakayama rule; see, for example,
Ref.~\cite{stanley1999enumerative}. We use this rule below when needed.


\subsubsection{Core techniques to calculate Weingarten integral}
\label{sec:weingarten}

In this section, we will explain how
Weingarten integral are evaluated.
The variances of orbital-rotation shadow protocol are determined by the Weingarten integral expressed by
\(\int_{\mtr{Gr}_{\eta} (\mathbb{C}^N ) } R_{i_1j_1}\cdots R_{i_dj_d}\dd R\). To this end, two results suffice: a convolution formula
(Theorem~\ref{thm:Grassmannian_convolution}), which reduces an arbitrary monomial
moment to the Weingarten function \(\Wg^{P_\eta}_{N,\eta}(\sigma)\) labelled by a
permutation \(\sigma\in\mathfrak S_d\); and a character formula, which expresses
\(\Wg^{P_\eta}_{N,\eta}(\sigma)\) through the irreducible characters of the symmetric
group and the contents of Young diagrams. We state both results and then illustrate their use on the simplest nontrivial moment.

First, we show important results from Ref.~\cite{coulter2025integration}.
\begin{theorem}[Grassmannian convolution formula, Theorem 2.3 in Ref.~\cite{coulter2025integration}]
Let $d, \eta, N$ be non-negative integers satisfying $\eta \leq N$ and $d \leq N$, and we assume
\(\
    \mac{P}_\eta
    :=
    \left\{
        R\in \operatorname{End}(\mathbb C^N)
        \mid
        R^2=R,\ R^\dagger=R,\ \operatorname{tr}R=\eta
    \right\}
\)
be equipped with the normalized \( \mtr{U}(N)\)-invariant probability measure
\(\dd R\).  For a non-negative integer \(d \) and \(\sigma\in \mathfrak{S}_d\), we define the Weingarten function on $\mac{P}_\eta $ as 
\[
    \operatorname{Wg}^{P_\eta}_{N,\eta}(\sigma)
    :=
    \int_{\mac{P}_\eta}
    R_{1,\sigma(1)}
    R_{2,\sigma(2)}
    \cdots
    R_{d,\sigma(d)}
    \dd R .
\]
Then, for arbitrary indices
$
    i_1,\ldots,i_d,\ j_1,\ldots,j_d\in [N],
$
one has
\begin{align}
    \int_{\mac{P}_\eta}
    R_{i_1,j_1}
    R_{i_2,j_2}
    \cdots
    R_{i_d,j_d}
    \dd R
    =
    \sum_{\sigma\in \mathfrak{S}_d}
    \left(
        \prod_{a=1}^d
        \delta_{i_{\sigma(a)},j_a}
    \right)
    \operatorname{Wg}^{P_\eta}_{N,\eta}(\sigma).    
\end{align}
\label{thm:Grassmannian_convolution}
\end{theorem}

\begin{lemma}[Character formula, Corollary 2.15 in Ref.~\cite{coulter2025integration}]
    \label{lem:character_formula}
    Let $d, \eta, N$ be non-negative integers satisfying $\eta \leq N$ and $d \leq N$. For $\sigma \in \mathfrak{S}_d $, the following equality holds, 
    \begin{align}
        \Wg_{N,\eta}^{\mac{P}_\eta }(\sigma)  = \frac{1}{d!} \sum_{\lambda \vdash d } \chi^\lambda (\mtr{id} ) \chi^\lambda (\sigma ) \prod_{ (i,j) \in \lambda  } \frac{\eta + c(i,j) }{N + c(i,j)  },
    \end{align}
    where $\chi^\lambda$ denotes the irreducible character of $\mathfrak S_d$ associated with the partition $\lambda\vdash d$, $\mtr{id}\in\mathfrak S_d$ is the
    identity permutation, and $c(i,j)=j-i$ is the content of the box $(i,j)\in\lambda$.
\end{lemma}
The background on partitions, Young diagrams, and irreducible
characters of the symmetric group is briefly reviewed in
Section~\ref{sec:prelim}. Together,
Theorem~\ref{thm:Grassmannian_convolution} and
Lemma~\ref{lem:character_formula} reduce arbitrary moments of the matrix
entries of \(R\) to explicit finite sums over permutations
\(\sigma\in\mathfrak S_d\) and partitions \(\lambda\vdash d\).

As the simplest example, we compute \(\int_{\mtr{Gr}_{\eta} (\mathbb{C}^N ) } R_{1,2}R_{2,1}\dd R\) by
applying Theorem~\ref{thm:Grassmannian_convolution} and the character formula with
\(d=2\). Reading off the indices, the two factors \(R_{12}R_{21}\) correspond to
\((i_1,j_1)=(1,2)\) and \((i_2,j_2)=(2,1)\). The symmetric group is
\(\mathfrak S_2=\{\mtr{id},\tau\}\) with \(\tau=(1\,2)\), and the Kronecker factor
\(\prod_{a=1}^{2}\delta_{i_{\sigma(a)},\,j_a}\) selects a single permutation:
\begin{align}
    \sigma=\mtr{id}:&\quad
        \delta_{i_1,j_1}\delta_{i_2,j_2}=\delta_{1,2}\,\delta_{2,1}=0,\\
    \sigma=\tau:&\quad
        \delta_{i_2,j_1}\delta_{i_1,j_2}=\delta_{2,2}\,\delta_{1,1}=1 .
\end{align}
Hence Theorem~\ref{thm:Grassmannian_convolution} gives
\begin{equation}
    \int_{\mtr{Gr}_{\eta} (\mathbb{C}^N ) } R_{1,2}R_{2,1}\dd R
    =\Wg^{P_\eta}_{N,\eta}(\tau).
\end{equation}
To evaluate \(\Wg^{P_\eta}_{N,\eta}(\tau)\) we use the character formula with \(d=2\).
There are two partitions of \(2\). Labelling each box with its content, we have
\begin{equation}
    \lambda=(2):\quad \begin{ytableau} 0 & 1 \end{ytableau}
    \qquad\text{and}\qquad
    \lambda=(1,1):\quad \begin{ytableau} 0 \\ \text{-}1 \end{ytableau}\,,
\end{equation}
so that the product \(\prod_{(i,j)\in\lambda}\frac{\eta+c(i,j)}{N+c(i,j)}\) becomes
\begin{gather}
    \prod_{(i,j)\in(2)}\frac{\eta+c(i,j)}{N+c(i,j)}
        =\frac{\eta}{N}\cdot\frac{\eta+1}{N+1},\\
    \prod_{(i,j)\in(1,1)}\frac{\eta+c(i,j)}{N+c(i,j)}
        =\frac{\eta}{N}\cdot\frac{\eta-1}{N-1}.
\end{gather}
As Eq.~\eqref{eq:table_character} shows, character values are
\(\chi^{(2)}(\mtr{id})=\chi^{(1,1)}(\mtr{id})=1\), \(\chi^{(2)}(\tau)=1\), and
\(\chi^{(1,1)}(\tau)=-1\). From these facts, we obtain
\begin{align}
    \Wg^{P_\eta}_{N,\eta}(\tau)
    &=\frac{1}{2!}\Bigg[
        \chi^{(2)}(\mtr{id})\,\chi^{(2)}(\tau)\,
            \frac{\eta}{N}\cdot\frac{\eta+1}{N+1}
        +\chi^{(1,1)}(\mtr{id})\,\chi^{(1,1)}(\tau)\,
            \frac{\eta}{N}\cdot\frac{\eta-1}{N-1}
      \Bigg]\notag\\
    &=\frac{1}{2}\left[
        \frac{\eta(\eta+1)}{N(N+1)}-\frac{\eta(\eta-1)}{N(N-1)}
      \right] \\
    &=\frac{\eta\left(N-\eta\right)}{N\left(N-1\right)\left(N+1\right)} .
\end{align}
Finally, we obtain
\begin{align}
    \int_{\mtr{Gr}_{\eta} (\mathbb{C}^N ) } R_{1,2}R_{2,1}\dd R
    =\frac{\eta\left(N-\eta\right)}{N\left(N-1\right)\left(N+1\right)}.
\end{align}
Remark that this reproduces the entry for the transposition \(\sigma=(1\,2)\) in the table of the
Weingarten function given in Ref.~\cite{coulter2025integration}.

The calculation process described above determines, in principle, every coordinate
moment. While each moment is explicit, its evaluation becomes laborious as the degree
\(d\) grows: the convolution formula ranges over the symmetric group \(\mathfrak S_d\),
and each Weingarten value \(\Wg^{P_\eta}_{N,\eta}(\sigma)\) is itself a sum over the
partitions \(\lambda\vdash d\) with characters \(\chi^\lambda(\sigma)\) obtained from
the Murnaghan--Nakayama rule. Already the \(k=2\) estimator variances require moments of
degree up to \(2k+\abs{L}=8\), whose evaluation as exact rational functions of \(N\)
and \(\eta\) is lengthy to carry out reliably by hand. We therefore evaluate the
higher-degree moments symbolically by computer.

\subsubsection{Weingarten integral evaluation}

First, we evaluate Weingarten integral over Grassmannian required to determine the coefficients \(A_T\) appearing in the decomposition of
\( {}^1B^1_{\ 2}\) in Sec.~\ref{sec:1-rdm_variance}. 
From Theorem \ref{thm:Grassmannian_convolution} and Lemma \ref{lem:character_formula}, we obtain the following identities
\begin{align}
    \int R_{1,2}R_{2,1}\,\mathrm{d}R &= \frac{\eta \left(N - \eta\right)}{N \left(N - 1\right) \left(N + 1\right)}, \\
    \int R_{1,2}R_{2,1}R_{1,1}\,\mathrm{d}R &= \frac{\eta \left(N - \eta\right) \left(\eta + 1\right)}{N \left(N - 1\right) \left(N + 1\right) \left(N + 2\right)}, \\
    \int R_{1,2}R_{2,1}R_{2,2}\,\mathrm{d}R &= \frac{\eta \left(N - \eta\right) \left(\eta + 1\right)}{N \left(N - 1\right) \left(N + 1\right) \left(N + 2\right)}, \\
    \int R_{1,2}R_{2,1}R_{1,1}R_{2,2}\,\mathrm{d}R &= \frac{\eta \left(N - \eta\right) \left(\eta + 1\right) \left(N \eta + N + \eta - 1\right)}{N^{2} \left(N - 1\right) \left(N + 1\right) \left(N + 2\right) \left(N + 3\right)}, \\
    \int (R_{1,2}R_{2,1})^2\,\mathrm{d}R &= \frac{2 \eta \left(N - \eta\right) \left(\eta + 1\right) \left(N - \eta + 1\right)}{N^{2} \left(N - 1\right) \left(N + 1\right) \left(N + 2\right) \left(N + 3\right)}.
\end{align}
From this result, we obtain
\begin{align}
    \int R_{1,2}R_{2,1}R_{1,1}\,\dd R
    =
    \int R_{1,2}R_{2,1}R_{2,2}\,\dd R
    &=
    \frac{\eta(N-\eta)(\eta+1)}
    {N(N-1)(N+1)(N+2)}, \\
    \int R_{1,2}R_{2,1}
    \left(R_{1,1}R_{2,2}-R_{1,2}R_{2,1}\right)\,\dd R
    &=
    \frac{\eta(N-\eta)(\eta-1)(\eta+1)}
    {N^2(N-1)(N+1)(N+2)}.
\end{align}
Therefore,
\begin{align}
    A_{\emptyset}
    &=
    \frac{\binom{N}{\eta}}{\binom{N-2}{\eta}}
    \int
    R_{1,2}R_{2,1}
    \left(
        1-R_{1,1}-R_{2,2}+R_{1,1}R_{2,2}-R_{1,2}R_{2,1}
    \right)
    \dd R  \\
    &=
    \frac{\eta(N-\eta+1)}
    {N(N+1)(N+2)}, \\
    A_{\{1\}}
    &=
    \frac{\binom{N}{\eta}}{\binom{N-2}{\eta-1}}
    \int
    R_{1,2}R_{2,1}
    \left(
        R_{1,1}-R_{1,1}R_{2,2}+R_{1,2}R_{2,1}
    \right)
    \dd R  \\
    &=
    \frac{(\eta+1)(N-\eta+1)}
    {N(N+1)(N+2)}, \\
    A_{\{2\}}
    &=
    \frac{\binom{N}{\eta}}{\binom{N-2}{\eta-1}}
    \int
    R_{1,2}R_{2,1}
    \left(
        R_{2,2}-R_{1,1}R_{2,2}+R_{1,2}R_{2,1}
    \right)
    \dd R  \\
    &=
    \frac{(\eta+1)(N-\eta+1)}
    {N(N+1)(N+2)}, \\
    A_{\{1,2\}}
    &=
    \frac{\binom{N}{\eta}}{\binom{N-2}{\eta-2}}
    \int
    R_{1,2}R_{2,1}
    \left(
        R_{1,1}R_{2,2}-R_{1,2}R_{2,1}
    \right)
    \dd R  \\
    &=
    \frac{(N-\eta)(\eta+1)}
    {N(N+1)(N+2)}.
\end{align}

Next, we evaluate Weingarten integral over Grassmannian required to determine the coefficients \(A_T\) appearing in the decomposition of
\( {}^1B^1_{\ 1}\) in Sec.~\ref{sec:1-rdm_variance}. 
From Theorem \ref{thm:Grassmannian_convolution} and Lemma \ref{lem:character_formula}, we obtain the following identities
\begin{align}
    \int_{\mtr{Gr}_\eta(\mathbb C^N)}
    R_{1,1}^2 \dd R
    &=
    \frac{\eta(\eta+1)}{N(N+1)}, \\
    \int_{\mtr{Gr}_\eta(\mathbb C^N)}
    R_{1,1}^3 \dd R
    &=
    \frac{\eta(\eta+1)(\eta+2)}
    {N(N+1)(N+2)}.
\end{align}
Therefore,
\begin{align}
    \int_{\mtr{Gr}_\eta(\mathbb C^N)}
    R_{1,1}^2(1-R_{1,1}) \dd R
    &=
    \frac{\eta(\eta+1)(N-\eta)}
    {N(N+1)(N+2)}.
\end{align}
Substituting these identities gives
\begin{align}
    A_\emptyset
    &=
    \frac{\eta(\eta+1)}{(N+1)(N+2)}, \\
    A_{\{1\}}
    &=
    \frac{(\eta+1)(\eta+2)}{(N+1)(N+2)}.
\end{align}

\begin{revision}

\section{Analytical evaluation of the covariance}
\label{sec:verify}

In this section, we describe the method for evaluating the moment
operators, treating $k$ as a fixed constant and assuming
$
  2k \le \eta \le N-2k ,
  \label{eq:app_practical_range}
$
which covers the practically relevant cases $k=1,2$. This range
guarantees that every block $\Psi_W$ is nondegenerate on $\mac H_\eta$,
since $\ell_++\ell_0\le 2k$ never exceeds either $\eta$ or $N-\eta$, so
that the normalizations appearing below are invertible.

\subsubsection{Classification of the moment operators via charge and incidence profiles}

In this section, we overview the method to efficiently evaluate the factors
$\purple{{}^{k}B_{\vec q;\,\vec s}^{\vec p;\,\vec r}}$ that govern the variance and
covariance of the orbital shadow protocol. Following the definition of the
main text, these factors are defined as
\begin{align}
  \purple{{}^{k}B_{\vec q;\,\vec s}^{\vec p;\,\vec r}}
  \coloneq
  \binom{N}{\eta}
  \int_{\mtr{Gr}_\eta(\mathbb{C}^N)}
  \purple{{}^{k}\widehat D_{\mkern6mu\vec q}^{\vec p}(R)}\,
  \purple{\overline{{}^{k}\widehat D_{\mkern6mu\vec s}^{\vec r}(R)}}\,
  R^{\wedge\eta} \dd{R} .
  \label{eq:app_B_def}
\end{align}
However, evaluating all possible
$\purple{{}^{k}B_{\vec q;\,\vec s}^{\vec p;\,\vec r}}$ separately is inefficient, since
the number of index quadruples $\purple{(\vec p,\vec q;\vec r,\vec s)}$ scales as $N^{4k}$. To organize this
family, we introduce two complementary profiles: the charge profile,
which determines the form of the linear decomposition of
$\purple{{}^{k}B_{\vec q;\,\vec s}^{\vec p;\,\vec r}}$, and the incidence profile, which determines the associated coefficient values. 

Let $L\coloneq\vec q\cup\vec p\cup\vec r\cup\vec s$ and consider the
subgroup $G=\mathrm U(1)^L\times\mathrm U(L^c)\subset\mathrm U(N)$.
The moment operator is invariant under the $\mathrm U(L^c)$ action and
transforms covariantly under the local phase rotations $\mathrm U(1)^L$,
with the phase determined by its charge profile below. After ordering the
modes so that those in $L$ come first and writing $\ell=\abs{L}$, the
elements of $G$ have the block form
\begin{align}
  G = \{ V=\mtr{diag}(e^{i\theta_1},\dots,e^{i\theta_\ell})\oplus u,
  \mid \theta_x\in\mathbb R,\quad u\in\mathrm U(N-\ell) \}.
\end{align}
Here, each $V\in G$ acts by independent phases on the modes in $L$
and by an arbitrary unitary rotation on the complement $L^c$. The local
phase rotations impose a charge-selection rule, while Schur's lemma for
the $\mathrm U(L^c)$ action determines the allowed blocks, yielding the
decomposition derived below.

For each mode $x\in L$, we will define the \textit{charge} as follows,
\begin{align}
  \zeta_x
  \coloneq
  \purple{
    \purple{\mathbf 1_{\{x\in\vec q\}}}-\purple{\mathbf 1_{\{x\in\vec p\}}}
  +\purple{\mathbf 1_{\{x\in\vec r\}}}-\purple{\mathbf 1_{\{x\in\vec s\}}}
  }.
  \label{eq:app_charge_def}
\end{align}
Equivalently, $\zeta_x$ records
\purple{the occurrences of mode $x$ in $\purple{\vec q}$ and $\purple{\vec r}$ with positive
sign and those in $\purple{\vec p}$ and $\purple{\vec s}$ with negative sign in the phase of
the kernel in Eq.~\eqref{eq:app_B_def}.} Indeed, under the adjoint action of
$V=\mtr{diag}(e^{i\theta_1},\dots,e^{i\theta_N})$, each entry transforms as
$R_{ab}\mapsto e^{i(\theta_a-\theta_b)}R_{ab}$, and the kernel acquires
exactly the phase
\begin{align}
  \purple{{}^{k}\widehat D_{\mkern6mu\vec q}^{\vec p}(VRV^\dagger)}\,
  \purple{\overline{{}^{k}\widehat D_{\mkern6mu\vec s}^{\vec r}(VRV^\dagger)}}
  =
  e^{i\sum_{x\in L}\zeta_x\theta_x}\,
  \purple{{}^{k}\widehat D_{\mkern6mu\vec q}^{\vec p}(R)}\,
  \purple{\overline{{}^{k}\widehat D_{\mkern6mu\vec s}^{\vec r}(R)}} .
  \label{eq:app_kernel_phase}
\end{align}
As shown this, $\zeta_x$ is the integer exponent with which the kernel
transforms under the $\mathrm U(1)$ phase rotation of mode $x$; we
therefore call it the charge of mode $x$, in analogy with the
total particle number, which is the charge associated with the global phase
rotation $\theta_x\equiv\theta$.
We refer to the collection $(\zeta_x)_{x\in L}$ as the \textit{charge profile}
of the index quadruple $\purple{(\vec p,\vec q;\vec r,\vec s)}$. According to the
sign of the charge, the local support $L$ decomposes into the three subsets
\begin{align}
  L_\pm \coloneq \{x\in L \mid \zeta_x=\pm1\},
  \qquad
  L_0 \coloneq \{x\in L \mid \zeta_x=0\},
  \label{eq:app_charge_sets}
\end{align}
and we write $\ell_+ \coloneq\abs{L_+}$ and $\ell_0\coloneq\abs{L_0}$ for their
cardinalities.
Since each tuple $\vec q ,\vec p, \vec r, \vec s$ contributes $k$ indices, the charges sum to
$\sum_{x\in L}\zeta_x=k-k+k-k=0$; in particular $\abs{L_+}=\abs{L_-}=\ell_+$
whenever all charges satisfy $\abs{\zeta_x}\le1$, and then
$\abs{L}=2\ell_++\ell_0$.
For such charge profiles, the four index tuples contain $4k$ index
occurrences in total. Each of the $2\ell_+$ charged modes occurs at least
once, while each of the $\ell_0$ neutral modes occurs at least twice.
Consequently, the following inequalities satisfy
\begin{align}
  2\ell_++2\ell_0\le4k,\qquad \ell_++\ell_0\le2k.
  \label{eq:app_charge_support_bound}
\end{align}

The charge profile decides the decomposition of
$\purple{{}^{k}B_{\vec q;\,\vec s}^{\vec p;\,\vec r}}$ as follows. Repeating the
substitution $R\mapsto V^\dagger RV$ in Eq.~\eqref{eq:app_B_def} with the
Haar invariance of the measure, Eq.~\eqref{eq:app_kernel_phase} gives
\begin{align}
  \purple{{}^{k}B_{\vec q;\,\vec s}^{\vec p;\,\vec r}}
  &=\binom{N}{\eta}\int \,
  \purple{{}^{k}\widehat D_{\mkern6mu\vec q}^{\vec p}(VRV^\dagger)}\,
  \purple{\overline{{}^{k}\widehat D_{\mkern6mu\vec s}^{\vec r}(VRV^\dagger)}} (VRV^\dagger)^{\wedge\eta}\dd{ (VRV^\dagger) }\\
  &=\binom{N}{\eta}e^{i\sum_{x\in L}\zeta_x\theta_x}\,
   V^{\wedge\eta}\,
   \qty[ \int \,
  \purple{{}^{k}\widehat D_{\mkern6mu\vec q}^{\vec p}(R)}\,
  \purple{\overline{{}^{k}\widehat D_{\mkern6mu\vec s}^{\vec r}(R)}} R^{\wedge\eta}\dd{ (VRV^\dagger) } ]
   (V^{\wedge\eta})^\dagger  \\
    &=\binom{N}{\eta}e^{i\sum_{x\in L}\zeta_x\theta_x}\,
   V^{\wedge\eta}\,
   \qty[ \int \,
  \purple{{}^{k}\widehat D_{\mkern6mu\vec q}^{\vec p}(R)}\,
  \purple{\overline{{}^{k}\widehat D_{\mkern6mu\vec s}^{\vec r}(R)}} R^{\wedge\eta}\dd{ R } ]
   (V^{\wedge\eta})^\dagger  \\
  &=e^{i\sum_{x\in L}\zeta_x\theta_x}\,
   V^{\wedge\eta}\,\purple{{}^{k}B_{\vec q;\,\vec s}^{\vec p;\,\vec r}}\,
   (V^{\wedge\eta})^\dagger ,
\end{align}
Equivalently, the moment operator satisfies
\begin{align}
  V^{\wedge\eta}\,{}^{k}B_{\vec q;\,\vec s}^{\vec p;\,\vec r}
  (V^{\wedge\eta})^\dagger
  &=e^{-i\sum_{x\in L}\zeta_x\theta_x}\,
  {}^{k}B_{\vec q;\,\vec s}^{\vec p;\,\vec r}.
  \label{eq:app_B_covariance}
\end{align}
The same transformation law holds for every $V\in G$, since the
kernel depends only on matrix entries of $R$ with indices in $L$ and is
therefore unchanged by rotations on $L^c$. In particular, the moment
operator is invariant under $\mathrm U(L^c)$; in the zero-charge case it
is invariant under the full group $G$.
Taking matrix elements of $\purple{{}^{k}B_{\vec q;\,\vec s}^{\vec p;\,\vec r}} \in \mtr{End}(\mac{H}_\eta )$
between occupation configurations $\vec a,\vec b\in\mac S_{N,\eta}$ and
using $V^{\wedge\eta}\ket{\vec a}=e^{i\sum_{x\in\vec a}\theta_x}\ket{\vec a}$,
this identity becomes
\begin{align}
  \bra{\vec a}\purple{{}^{k}B_{\vec q;\,\vec s}^{\vec p;\,\vec r}}\ket{\vec b}
  =
  e^{i\sum_{x}\qty(\zeta_x+\purple{\mathbf 1_{\{x\in\vec a\}}}-\purple{\mathbf 1_{\{x\in\vec b\}}})\theta_x}\,
  \bra{\vec a}\purple{{}^{k}B_{\vec q;\,\vec s}^{\vec p;\,\vec r}}\ket{\vec b} .
  \label{eq:app_matrix_element_phase}
\end{align}
Averaging over the angles, the matrix element can be nonzero only if
\begin{align}
  \purple{\mathbf 1_{\{x\in\vec b\}}}-\purple{\mathbf 1_{\{x\in\vec a\}}}=\zeta_x
  \qquad\text{for every mode }x :
  \label{eq:app_selection_rule}
\end{align}
\purple{The ket configuration carries the positive charges and the bra
configuration the negative ones.} We consider the following three cases.

First, if $\abs{\zeta_x}=2$ for some mode $x$, no pair of occupation
configurations can match the charge profile, since occupation numbers shift
by at most one per mode. Hence
\begin{align}
  \purple{{}^{k}B_{\vec q;\,\vec s}^{\vec p;\,\vec r}}=0
  \qquad\text{if } \abs{\zeta_x}=2 \text{ for some } x\in L .
  \label{eq:app_impossible_charge}
\end{align}
This identity implies that $\mtr{Cov}_\rho [\purple{{}^{k}\widehat D_{\mkern6mu\vec q}^{\vec p}}, \purple{{}^{k}\widehat D_{\mkern6mu\vec s}^{\vec r}}  ] = 0 - \purple{{}^{k}D_{\rho;\,\vec q}^{\vec p}} \purple{{}^{k}D_{\rho;\,\vec r}^{\vec s}} $ for such a $\purple{(\vec p,\vec q;\vec r,\vec s)}$.

Second, if the charge profile vanishes identically, $\zeta_x=0$ for all
$x\in L$, as in the variance case $\purple{(\vec r,\vec s)=(\vec p,\vec q)}$, the
selection rule \eqref{eq:app_selection_rule} forces $\vec a=\vec b$. 
Note that when $\purple{{}^{k}B_{\vec q;\,\vec s}^{\vec p;\,\vec r}} = 0 $, we call this tuples $\purple{(\vec p,\vec q;\vec r,\vec s)}$ belongs to class-I.

Third, in the general case $\abs{\zeta_x}\le1$ with $\ell_+\ge1$, the selection
rule \eqref{eq:app_selection_rule} requires each mode in $L_+$ to be
annihilated, each mode in $L_-$ to be created, and each neutral mode in
$L_0$ to be dressed by a number operator $n_x$ or a hole projector
$\1-n_x$. Labelling each admissible block by the mode set $W\subseteq L_0$, we define
\begin{align}
  \Psi_W
  \coloneq
  \qty(\prod_{x\in L_-} a_x^\dagger)
  \qty(\prod_{x\in L_+} a_x)\,
  \prod_{c\in W} n_c
  \prod_{d\in L_0\setminus W} (\1-n_d)
  \;\Bigg|_{\mac H_\eta},
  \label{eq:app_Psi_def}
\end{align}
where the creation (annihilation) operators are ordered by increasing
(decreasing) mode index. Schur's lemma on the complement $L^c$ then reduces
the decomposition to a single sum over the shared occupations $W$,
\begin{align}
  \purple{{}^{k}B_{\vec q;\,\vec s}^{\vec p;\,\vec r}}
  =
  \sum_{\substack{W\subseteq L_0\\ \ell_++\abs{W}\le\eta}}
  A_W\,\Psi_W ,
  \label{eq:app_B_charge_expansion}
\end{align}
which contains the zero-charge case as the specialization
$L_\pm=\emptyset$: there $\Psi_W=\prod_{c\in W}n_c
\prod_{d\in L_0\setminus W}(\1-n_d)\,\big|_{\mac H_\eta}$ is the occupation
projector $\Pi_W$. 

Futtheremore, this decomposition motivates the following classification. We say that
two index quadruples $\purple{(\vec p,\vec q;\vec r,\vec s)}$ belong to the same
structural class $\mac{C}$ if they have the same triple
$
(\ell_+,\ell_0).
$
For a fixed class, the expansion
\eqref{eq:app_B_charge_expansion} has the same block structure: it
contains $2^{\ell_0}$ possible coefficients $A_W$, indexed by
$W\subseteq L_0$, and the corresponding RDM elements have orders ranging
from $\ell_+$ to $\ell_+ + \ell_0$, subject to the particle-number
constraint $\ell_+ + \abs{W}\le\eta$.

Next, we introduce the concept of the incidence profile.
Within a fixed structural class $\mac{C}$, the operator
$\purple{{}^{k}B_{\vec q;\,\vec s}^{\vec p;\,\vec r}}$ admits a decomposition as a
linear combination of the operators $\Psi_W$, and the admissible subsets
$W\subseteq L_0$ are determined by $\mac{C}$. However, even when two
index quadruples $(\vec p,\vec q;\vec r,\vec s)$ belong to the same
structural class, the corresponding coefficients $A_W$ may differ.
Motivated by this observation, we introduce the incidence profile. For
each mode $x\in L$, we define its incidence vector by
\begin{align}
    \alpha(x)
    \coloneq
    \left(
        \purple{\mathbf 1_{\{x\in\vec p\}}}, \
        \purple{\mathbf 1_{\{x\in\vec q\}}},\
        \purple{\mathbf 1_{\{x\in\vec r\}}}, \
        \purple{\mathbf 1_{\{x\in\vec s\}}}
    \right)
    \in
    \{0,1\}^4\setminus\{(0,0,0,0)\}.
\end{align}
For each incidence vector $\alpha$, let $n_\alpha$ denote the number of
modes $x\in L$ satisfying $\alpha(x)=\alpha$. We refer to the collection
$\boldsymbol n=(n_\alpha)_\alpha$ as the incidence profile of the index
quadruple $\purple{(\vec p,\vec q;\vec r,\vec s)}$. Together, the structural and
incidence profiles completely characterize the decomposition of
$\purple{{}^{k}B_{\vec q;\,\vec s}^{\vec p;\,\vec r}}$.
For example, the two $k=1$ index quadruples $(1,1;2,2)$ and
$(1,2;1,2)$ have the same structural profile
$(\ell_+,\ell_0)=(0,2)$ but different incidence profiles. Indeed, the
former has $n_{1100}=n_{0011}=1$, whereas the latter has
$n_{1010}=n_{0101}=1$, with all other $n_\alpha$ equal to zero. Thus,
the incidence profile distinguishes Classes~III and~III', which cannot
be distinguished by $(\ell_+,\ell_0)$ alone.

Remark that this classification allows us to organize the many index quadruples into
finitely many structural types. In particular, instead of treating each
$\purple{{}^{k}B_{\vec q;\,\vec s}^{\vec p;\,\vec r}}$ separately, the coefficients in
the linear expansion can be evaluated class by class, according to the
charge profile induced by the quadruple
$(\vec p,\vec q;\vec r,\vec s)$. For example, when $k=1$, the $N^4$
possible index quadruples are reduced to only eight classes as shown in Table \ref{tab:k1_cov_classes}.

Additionally, to treat cross-covariances between estimators of different
orders, we introduce the mixed-order moment operator
\begin{align}
  \purple{{}^{k,k'}B_{\vec q;\,\vec s}^{\vec p;\,\vec r}}
  \coloneq
  \binom{N}{\eta}
  \int_{\mtr{Gr}_\eta(\mathbb{C}^N)}
  \purple{{}^{k}\widehat D_{\mkern6mu\vec q}^{\vec p}(R)}\,
  \purple{\overline{{}^{k'}\widehat D_{\mkern6mu\vec s}^{\vec r}(R)}}\,
  R^{\wedge\eta}\dd{R},
  \label{eq:app_B_mixed_def}
\end{align}
where $\vec q,\vec p\in\mac S_{N,k}$ and
$\vec s,\vec r\in\mac S_{N,k'}$. The equal-order operator in
Eq.~\eqref{eq:app_B_def} is recovered as
\[
  \purple{{}^{k}B_{\vec q;\,\vec s}^{\vec p;\,\vec r}}
  =
  \purple{{}^{k,k}B_{\vec q;\,\vec s}^{\vec p;\,\vec r}}.
\]
The preceding charge-profile argument and block decomposition apply
equally to this mixed-order operator, since
\[
  \sum_{x\in L}\zeta_x
  =
  k-k+k'-k'
  =
  0.
\]
In particular, $\purple{{}^{1,2}B_{\vec q;\,\vec s}^{\vec p;\,\vec r}}$ governs the
cross-covariance between the single-shot $1$- and $2$-RDM estimators:
\begin{align}
  \mtr{Cov}\!\left[
    \purple{{}^{1}\widehat D_{\mkern6mu\vec q}^{\vec p}},
    \purple{{}^{2}\widehat D_{\mkern6mu\vec s}^{\vec r}}
  \right]
  =
  \tr_{\mac H_\eta}\!\left[
    \rho \purple{{}^{1,2}B_{\vec q;\,\vec s}^{\vec p;\,\vec r}}
  \right]
  -
  \purple{{}^{1}D_{\rho;\,\vec q}^{\vec p}}
  \purple{{}^{2}D_{\rho;\,\vec r}^{\vec s}}.
\end{align}
The same block decomposition and classification are obtained by repeating
the preceding argument.

The remaining task is the evaluation of these
coefficients $\{A_W\} $, to which we now turn.

\subsubsection{Evaluation of coefficients $A_W$}
In this section we organize the evaluation of coeffcients $A_W$ into three steps.
First, we derive a projection formula for $A_W$.
Next, we combine these coefficients into the covariance coefficients $\Gamma_K$.
Finally, we evaluate the required Grassmannian moments using Weingarten calculus.

As Step~1, we derive a projection formula for the sector
coefficients $A_W$.
The blocks $\Psi_W$ are mutually orthogonal with respect to the
Hilbert--Schmidt inner product on $\mtr{End}(\mac H_\eta)$. Indeed,
$\Psi_W$ maps each configuration $\ket{L_+\cup W\cup S}$ with
$S\subseteq L^c$, $\abs{S}=\eta-\ell_+-\abs{W}$, to
$\pm\ket{L_-\cup W\cup S}$ and annihilates every other configuration, so
\begin{align}
  \tr_{\mac H_\eta}\!\left[\Psi_W^\dagger\,\Psi_{W'}\right]
  =
  \delta_{W,W'}\,
  \binom{N-\abs{L}}{\eta-\ell_+-\abs{W}},
  \label{eq:app_block_orthogonality}
\end{align}
where $\Psi_W^\dagger$ is the block of Eq.~\eqref{eq:app_Psi_def} with
$L_+$ and $L_-$ interchanged, and the squared fermionic signs leave one
unit contribution per spectator configuration $S$. Pairing
Eq.~\eqref{eq:app_B_charge_expansion} with $\Psi_W^\dagger$ therefore
isolates a single coefficient,
\begin{align}
  A_W
  =
  \binom{N-\abs{L}}{\eta-\ell_+-\abs{W}}^{-1}
  \tr_{\mac H_\eta}\!\left[
    \Psi_W^\dagger\,
    \purple{{}^{k}B_{\vec q;\,\vec s}^{\vec p;\,\vec r}}
  \right],
  \label{eq:app_AW_projection}
\end{align}
and inserting the definition \eqref{eq:app_B_def} gives the projection
formula
\begin{align}
  A_W
  =
  \binom{N}{\eta}
  \binom{N-\abs{L}}{\eta-\ell_+-\abs{W}}^{-1}
  \int_{\mtr{Gr}_\eta(\mathbb{C}^N)}
  \purple{{}^{k}\widehat D_{\mkern6mu\vec q}^{\vec p}(R)}\,
  \purple{\overline{{}^{k}\widehat D_{\mkern6mu\vec s}^{\vec r}(R)}}\,
  \tr_{\mac H_\eta}\!\left[
    \Psi_W^\dagger\,R^{\wedge\eta}
  \right]
  \dd{R} .
  \label{eq:app_AW_integral}
\end{align}
The trace factor is an explicit finite sum of $\eta\times\eta$ minors of
$R$,
\begin{align}
  \tr_{\mac H_\eta}\!\left[\Psi_W^\dagger\,R^{\wedge\eta}\right]
  =
  \sum_{\substack{S\subseteq L^c\\ \abs{S}=\eta-\ell_+-\abs{W}}}
  \purple{
    \varepsilon_{\wedge}(L_-,W\cup S)\,
    \varepsilon_{\wedge}(L_+,W\cup S)\,
    \det R_{\,L_-\cup W\cup S,\;L_+\cup W\cup S}
  }\,,
  \label{eq:app_trace_factor}
\end{align}
where $R_{A,B}$ denotes the submatrix of $R$ with rows $A$ and columns $B$
listed in increasing order, and $\purple{\varepsilon_{\wedge}(L_-,W\cup S)\,
\varepsilon_{\wedge}(L_+,W\cup S)}\in\{\pm1\}$ is the fermionic sign
produced by the normal ordering in Eq.~\eqref{eq:app_Psi_def}.

Equivalently, the projection can be realized as a single matrix element.
Because $\purple{{}^{k}B_{\vec q;\,\vec s}^{\vec p;\,\vec r}}$ commutes with the
permutations of the modes in $L^c$, the sign-corrected matrix element below is independent
of the choice of the spectator configuration $S_0\subseteq L^c$ with
$\abs{S_0}=\eta-\ell_+-\abs{W}$, and
\begin{align}
  A_W
  =
  \purple{\varepsilon_{\wedge}(L_-,W\cup S_0)\,
  \varepsilon_{\wedge}(L_+,W\cup S_0)}
  \bra{L_-\cup W\cup S_0}
  \purple{{}^{k}B_{\vec q;\,\vec s}^{\vec p;\,\vec r}}
  \ket{L_+\cup W\cup S_0} .
  \label{eq:app_AW_matrix_element}
\end{align}
The trace form \eqref{eq:app_AW_integral} is recovered by averaging
Eq.~\eqref{eq:app_AW_matrix_element} over the
$\binom{N-\abs{L}}{\eta-\ell_+-\abs{W}}$ admissible spectator configurations; in
the symbolic implementation we use the trace form after eliminating
the spectator indices.

In summary, Step~1 reduces each sector coefficient $A_W$ to a scalar
Grassmannian integral whose integrand is a polynomial in the entries of
$R$: the kernel
$\purple{{}^{k}\widehat D_{\mkern6mu\vec q}^{\vec p}(R)}
 \purple{\overline{{}^{k}\widehat D_{\mkern6mu\vec s}^{\vec r}(R)}}$
contributes degree at most $2k$, and the trace factor
\eqref{eq:app_trace_factor} contributes the minors. The evaluation of the
resulting coordinate moments is deferred to Step~3.

As Step~2, we express the covariance coefficients $\Gamma_K$
in terms of the sector coefficients $A_W$. Taking
the expectation of Eq.~\eqref{eq:app_B_charge_expansion} in $\rho$ requires
$\tr_{\mac H_\eta}[\rho\,\Psi_W]$. Expanding the hole projectors in
Eq.~\eqref{eq:app_Psi_def} as
$\prod_{d\in L_0\setminus W}(\1-n_d)
 =\sum_{U\subseteq L_0\setminus W}(-1)^{\abs{U}}\prod_{d\in U}n_d$
turns $\Psi_W$ into a signed sum of bare transition operators, whose
expectations are local RDM elements:
\begin{align}
  \tr_{\mac H_\eta}\!\left[\rho\,\Psi_W\right]
  =
  \sum_{U\subseteq L_0\setminus W}
  (-1)^{\abs{U}}\,
  \purple{\varepsilon_{\wedge}(L_-,W\cup U)\,
  \varepsilon_{\wedge}(L_+,W\cup U)}\,
  \purple{{}^{\ell_++\abs{W}+\abs{U}}D_{\rho;L_+\cup W\cup U}^{L_-\cup W\cup U}} ,
  \label{eq:app_Psi_expectation}
\end{align}
with all index tuples ordered increasingly and
\purple{the fermionic merge signs displayed explicitly}. Substituting
Eq.~\eqref{eq:app_Psi_expectation} into
$\tr[\rho\,\purple{{}^{k}B_{\vec q;\,\vec s}^{\vec p;\,\vec r}}]
 =\sum_W A_W\tr[\rho\,\Psi_W]$ and collecting the terms with $K=W\cup U$
gives
\begin{align}
  \tr_{\mac H_\eta}\!\left[\rho\,\purple{{}^{k}B_{\vec q;\,\vec s}^{\vec p;\,\vec r}}\right]
  &=
  \sum_{\substack{K\subseteq L_0\\ \ell_++\abs{K}\le\eta}}
  \purple{\varepsilon_{\wedge}(L_-,K)\,
  \varepsilon_{\wedge}(L_+,K)}\,\Gamma_K\,
  \purple{{}^{\ell_++\abs{K}}D_{\rho;L_+\cup K}^{L_-\cup K}} ,
  \\
  \Gamma_K
  &\coloneq
  \sum_{W\subseteq K}
  (-1)^{\abs{K}-\abs{W}}\,
  A_W ,
  \label{eq:app_Gamma_def}
\end{align}
so that the covariance reads
\begin{align}
  \operatorname{Cov}_\rho\!\left(
    \purple{{}^{k}\widehat D_{\mkern6mu\vec q}^{\vec p}},\,
    \purple{{}^{k}\widehat D_{\mkern6mu\vec s}^{\vec r}}
  \right)
  =
  \sum_{\substack{K\subseteq L_0\\ \ell_++\abs{K}\le\eta}}
  \purple{\varepsilon_{\wedge}(L_-,K)\,
  \varepsilon_{\wedge}(L_+,K)} \Gamma_K\,
  \purple{{}^{\ell_++\abs{K}}D_{\rho;L_+\cup K}^{L_-\cup K}}
  -
  \purple{{}^{k}D_{\rho;\,\vec q}^{\vec p}}\,
  \purple{{}^{k}D_{\rho;\,\vec r}^{\vec s}} .
  \label{eq:app_cov_final}
\end{align}
The collected coefficients $\Gamma_K$, rather than the $A_W$ themselves,
are what multiply the physically accessible RDM elements; they are the
quantities tabulated by the symbolic routine of Step~3.

As Step~3, we evaluate the Grassmannian moments.
After Steps~1 and~2, every coefficient is a linear combination of monomial
moments
\begin{align}
  \int_{\mtr{Gr}_\eta(\mathbb{C}^N)}
  R_{i_1j_1}\cdots R_{i_dj_d}
  \dd{R} ,
  \label{eq:app_monomial_moment}
\end{align}
obtained by expanding the kernel and the minors of
Eq.~\eqref{eq:app_trace_factor} into monomials. Identifying
$\mtr{Gr}_\eta(\mathbb{C}^N)$ with the projector ensemble $\mac P_\eta$,
these moments are evaluated in closed form by the Weingarten calculus of
Appendix~\ref{appendix:explicit_weingarten}: the convolution formula
(Theorem~\ref{thm:Grassmannian_convolution}) reduces
Eq.~\eqref{eq:app_monomial_moment} to the Weingarten function
$\Wg^{P_\eta}_{N,\eta}(\sigma)$ labelled by $\sigma\in\mathfrak S_d$, and
the character formula expresses the latter through the irreducible
characters of $\mathfrak S_d$ and the contents of Young diagrams, yielding
rational functions of $N$ and $\eta$.

Two remarks keep the computation finite and small. First,
after simplifying the trace factor \eqref{eq:app_trace_factor},
the required moments have degree $d\le4k\le N$
under the input conditions $2k \leq \eta \leq N -2k$. Second, since
$\Wg^{P_\eta}_{N,\eta}(\sigma)$ depends only on the cycle type of
$\sigma$, the finitely many values with $d\le4k$ can be tabulated once
and reused for every quadruple $\purple{(\vec p,\vec q;\vec r,\vec s)}$.
For mixed orders $(k,k')$, the corresponding bound is
$d\le2(k+k')\le N$ under $k+k'\le\eta\le N-(k+k')$.

\subsubsection{Pseudocode of the analytical evaluation routine}
\label{subsec:app_pseudocode}

This subsection summarizes the implementation of Steps~1--3 in pseudocode.
The routine is organized in three layers, presented bottom-up.
Algorithms~\ref{alg:monomial_moment} and \ref{alg:weingarten_value}
evaluate a single monomial moment by the Weingarten calculus of
Appendix~\ref{appendix:explicit_weingarten} (Step~3);
Algorithm~\ref{alg:sector_coefficient} assembles a sector coefficient
$A_W$ from these moments (Step~1); and
Algorithm~\ref{alg:gamma_coefficient} collects the sector coefficients
into the covariance coefficients $\Gamma_K$ (Step~2). Throughout, $N$ and
$\eta$ are treated as symbolic variables, and every output is a rational
function of $N$ and $\eta$. The numerical implementation of the symbolic procedure in
Algorithms~\ref{alg:monomial_moment}--\ref{alg:gamma_coefficient} is available in Ref.~\cite{koizumi_or_shadow_covariances}.

Algorithm~\ref{alg:monomial_moment} implements the convolution formula of
Theorem~\ref{thm:Grassmannian_convolution}: the loop selects the
permutations surviving the Kronecker factor, each contributing its
Weingarten value. The latter is supplied by
Algorithm~\ref{alg:weingarten_value}, which implements the character
formula of Lemma~\ref{lem:character_formula} and is memoized in the cycle
type denoted by $\mtr{type}(\sigma)$, on which alone it depends.

\begin{algorithm}[H]
  \caption{Symbolic evaluation of the monomial moments
  $\int_{\mac P_\eta} R_{i_1j_1}\cdots R_{i_dj_d}\dd{R}$ via
  Theorem~\ref{thm:Grassmannian_convolution}.}
  \label{alg:monomial_moment}
\begin{algorithmic}[1]
\Statex \textnormal{\textbf{Input}: degree $d$ and index sequences
  $(i_1,\dots,i_d),(j_1,\dots,j_d)\in[N]^d$ with $d\le N$; $N$ and $\eta$ are kept
  symbolic.}
\Statex \textnormal{\textbf{Output}: the moment
  $\int_{\mac P_\eta} R_{i_1j_1}\cdots R_{i_dj_d}\dd{R}$ as a rational
  function of $N$ and $\eta$.}
\Function{\texttt{MonomialMoment}}{$(i_1,\dots,i_d),(j_1,\dots,j_d)$}
  \If{the multisets $\{i_a\}$ and $\{j_a\}$ differ}
      \State \textbf{return} $0$
      \Comment{no permutation can match all Kronecker deltas}
  \EndIf
  \State $\mathcal M \gets 0$
  \For{$\sigma \in \mathfrak S_d$}
      \If{$i_{\sigma(a)} = j_a$ for all $a\in[d]$}
          \State $\mathcal M \gets \mathcal M
            + \texttt{Wg}\bigl(\mtr{type}(\sigma)\bigr)$
            \Comment{Algorithm~\ref{alg:weingarten_value}}
      \EndIf
  \EndFor
  \State \textbf{return} $\mathcal M$
\EndFunction
\end{algorithmic}
\end{algorithm}

\begin{algorithm}[H]
  \caption{The Weingarten function $\Wg^{P_\eta}_{N,\eta}$ via the
  character formula of Lemma~\ref{lem:character_formula}.}
  \label{alg:weingarten_value}
\begin{algorithmic}[1]
\Statex \textnormal{\textbf{Input}: a cycle type $\purple{\kappa}\vdash d$ with $d\le N$; $N$ and $\eta$ are
  kept symbolic.}
\Statex \textnormal{\textbf{Output}: $\Wg^{P_\eta}_{N,\eta}(\sigma)$ for any
  $\sigma\in\mathfrak S_d$ of cycle type $\purple{\kappa}$, as a rational function of
  $N$ and $\eta$. Values are memoized in $\purple{\kappa}$.}
\Function{\texttt{Wg}}{$\purple{\kappa}$}
  \State $w \gets 0$
  \For{$\lambda \vdash d$}
      \State $f^\lambda \gets \chi^\lambda(\mtr{id})$,\quad
        $\chi^\lambda_{\purple{\kappa}} \gets \chi^\lambda(\purple{\kappa})$
        \Comment{Murnaghan--Nakayama rule}
      \State $w \gets w + f^\lambda\,\chi^\lambda_{\purple{\kappa}}\, \prod_{(i,j)\in\lambda}\dfrac{\eta+c(i,j)}{N+c(i,j)}$
  \EndFor
  \State \textbf{return} $w/d!$
\EndFunction
\end{algorithmic}
\end{algorithm}

Algorithm~\ref{alg:sector_coefficient} assembles the sector coefficient
$A_W$ of Eq.~\eqref{eq:app_AW_integral}. It builds the kernel and the
spectator-reduced trace factor as polynomials in the entries of $R$
supported on $L$, expands their product into monomials, and integrates
termwise with Algorithm~\ref{alg:monomial_moment}.

\begin{algorithm}[H]
  \caption{Symbolic evaluation of the sector coefficients $A_W$ of
  Eq.~\eqref{eq:app_AW_integral}.}
  \label{alg:sector_coefficient}
\begin{algorithmic}[1]
\Statex \textnormal{\textbf{Input}: an index quadruple
  $(\vec p,\vec q;\vec r,\vec s)$, $\vec p,\vec q,\vec r,\vec s\in\mac
  S_{N,k}$, and a label $W\subseteq L_0$; $N$ and $\eta$ are kept
  symbolic.}
\Statex \textnormal{\textbf{Output}: $A_W$ as a rational function of $N$ and
  $\eta$.}
\Function{\texttt{SectorCoefficient}}{$(\vec p,\vec q;\vec r,\vec s),W$}
  \State $L \gets \vec q\cup\vec p\cup\vec s\cup\vec r$
  \State compute $\zeta_x$ for $x\in L$ by Eq.~\eqref{eq:app_charge_def};
    set $L_\pm$, $L_0$, $\ell_+\gets\abs{L_+}$
  \If{$\abs{\zeta_x}=2$ for some $x\in L$}
      \State \textbf{return} $0$
      \Comment{Eq.~\eqref{eq:app_impossible_charge}}
  \EndIf
  \State $f(R) \gets
    \purple{{}^{k}\widehat D_{\mkern6mu\vec q}^{\vec p}(R)}\,
    \purple{\overline{{}^{k}\widehat D_{\mkern6mu\vec s}^{\vec r}(R)}}$
    \Comment{kernel; polynomial of degree $\le 2k$ supported on $L$}
  \State $t_W(R) \gets \tr_{\mac H_\eta}\!\bigl[\Psi_W^\dagger
    R^{\wedge\eta}\bigr]$ 
    \Comment{polynomial of degree {$\le\ell_++\abs{L_0}\le2k$} supported on $L$}
  \State expand $f(R)\,t_W(R)
    = \sum_{t} c_t\, R_{i^t_1 j^t_1}\cdots R_{i^t_{d_t} j^t_{d_t}}$
    \Comment{{$d_t\le4k\le N$}; $c_t$ rational in $N,\eta$}
  \State $A \gets \sum_{t} c_t\,
    \texttt{MonomialMoment}\bigl((i^t_a)_a,(j^t_a)_a\bigr)$
    \Comment{Algorithm~\ref{alg:monomial_moment}}
  \State \textbf{return}
    $\binom{N}{\eta}\binom{N-\abs{L}}{\eta-\ell_+-\abs{W}}^{-1} A$,
    simplified as a rational function of $N,\eta$
\EndFunction
\end{algorithmic}
\end{algorithm}

Finally, Algorithm~\ref{alg:gamma_coefficient} collects the sector
coefficients into the covariance coefficients $\Gamma_K$ of
Eq.~\eqref{eq:app_Gamma_def}, which completes the evaluation of the
covariance \eqref{eq:app_cov_final} for a given index quadruple.

\begin{algorithm}[H]
  \caption{Symbolic evaluation of the covariance coefficients $\Gamma_K$
  of Eq.~\eqref{eq:app_Gamma_def}.}
  \label{alg:gamma_coefficient}
\begin{algorithmic}[1]
\Statex \textnormal{\textbf{Input}: an index quadruple
  $(\vec p,\vec q;\vec r,\vec s)$, $\vec p,\vec q,\vec r,\vec s\in\mac
  S_{N,k}$; $N$ and $\eta$ are kept symbolic.}
\Statex \textnormal{\textbf{Output}: the coefficients
  $\{\Gamma_K\}_{K\subseteq L_0}$ of Eq.~\eqref{eq:app_cov_final} as
  rational functions of $N$ and $\eta$.}
\Function{\texttt{CovarianceCoefficients}}{$(\vec p,\vec q;\vec r,\vec s)$}
  \State $L \gets \vec q\cup\vec p\cup\vec s\cup\vec r$;\quad
    compute $\zeta_x$, $L_\pm$, $L_0$, $\ell_+$ as in
    Algorithm~\ref{alg:sector_coefficient}
  \If{$\abs{\zeta_x}=2$ for some $x\in L$}
      \State \textbf{return} $\emptyset$
      \Comment{$\operatorname{Cov}_\rho
        =-\purple{{}^{k}D_{\rho;\,\vec q}^{\vec p}}\purple{{}^{k}D_{\rho;\,\vec r}^{\vec s}}$,
        Eq.~\eqref{eq:app_impossible_charge}}
  \EndIf
  \For{$W \subseteq L_0$}
      \State $A_W \gets
        \texttt{SectorCoefficient}\bigl((\vec p,\vec q;\vec r,\vec s),
        W\bigr)$
        \Comment{Algorithm~\ref{alg:sector_coefficient}}
  \EndFor
  \For{$K \subseteq L_0$}
      \State $\Gamma_K \gets
        \sum_{W\subseteq K}(-1)^{\abs{K}-\abs{W}}\,A_W$
        \Comment{Eq.~\eqref{eq:app_Gamma_def}}
  \EndFor
  \State \textbf{return} $\{\Gamma_K\}_{K\subseteq L_0}$
\EndFunction
\end{algorithmic}
\end{algorithm}



\subsection{Enumeration of the covariance classes}
\label{subsec:cov_class_enumeration}

Having reduced every moment operator to the block form
\eqref{eq:app_B_charge_expansion}, whose structure is fixed by the
structural profile $(\ell_+,\ell_0)$ and refined by the incidence profile
$\boldsymbol n=(n_\alpha)_\alpha$, we enumerate the resulting classes for
the covariance of the $k$-RDM estimator with $k=1,2$ and for the $(1,2)$
cross-covariance between the $1$- and $2$-RDM estimators. Each class is
shown through a representative index quadruple, written in the order
$\purple{(\vec p,\vec q;\vec r,\vec s)}$, from which the rest follow
by mode relabeling and the Hermitian symmetries. Within each structural
class, the incidence variants are ordered lexicographically according to
this convention and labeled by subscripts $a,b,\dots$; Roman numerals label
the structural profiles.

The $k=1$ classes are listed in Table~\ref{tab:k1_cov_classes} ($\ell_++\ell_0\le 2$;
\purple{seven} admissible classes and the impossible-charge class~$\mathrm{I}$),
the $k=2$ classes in Table~\ref{tab:k2_cov_classes} ($\ell_++\ell_0\le 4$;
$36$ admissible patterns and $22$ impossible-charge patterns), and the
$(1,2)$ cross-covariance classes in Table~\ref{tab:k12_cov_classes}
($18$ admissible patterns and $6$ impossible-charge patterns).

\begin{table}[H]
  \centering
  \renewcommand{\arraystretch}{1.2}
    \caption{
  \justifying{
  Classification of the $1$-RDM covariance blocks by the
  structural profile $(\ell_+,\ell_0)$ and the incidence profile
  $\boldsymbol n=(n_\alpha)_\alpha$. For $k=1$ every listed component
  equals one and all omitted components vanish. Class~$\mathrm{I}$
  collects the impossible-charge cases, for which
  $\purple{{}^{1}B_{j;\,\ell}^{i;\,k}}=0$. The column ``\#$\Gamma_K$'' counts the coefficients
  $2^{\ell_0}$, and ``RDM orders'' lists the orders
  $\ell_+,\dots,\ell_++\ell_0$ appearing in Eq.~\eqref{eq:app_cov_final}
  before the truncation $\ell_++\abs{K}\le\eta$. The variance patterns are
  the special cases $\purple{(k,\ell)=(i,j)}$: the diagonal-entry variance is
  Class~$\mathrm{II}$ and the off-diagonal variance is
  Class~$\mathrm{III}_{b}$.}
  }
  \begin{tabular}{ccc l cc l}
    \toprule
    Class & $\ell_+$ & $\ell_0$ & Nonzero $n_\alpha$
    & \#$\Gamma_K$ & RDM orders &
    Representative $\purple{(i,j;k,\ell)}$ \\
    \midrule
    $\mathrm{I}$
    & --- & ---
    & $\purple{n_{1001},\,n_{0110}}$
    & $0$ & ---
    & $\purple{(1,2;2,1)}$ \\
    \midrule
    $\mathrm{II}$
    & $0$ & $1$
    & $\purple{n_{1111}}$
    & $2$ & $0$--$1$
    & $\purple{(1,1;1,1)}$ \\
    \midrule
    $\mathrm{III}_{a}$
    & \multirow{2}{*}{$0$} & \multirow{2}{*}{$2$}
    & $\purple{n_{1100},\,n_{0011}}$
    & \multirow{2}{*}{$4$} & \multirow{2}{*}{$0$--$2$}
    & $\purple{(1,1;2,2)}$ \\
    $\mathrm{III}_{b}$
    & & &
    $\purple{n_{1010},\,n_{0101}}$
    & & &
    $\purple{(1,2;1,2)}$ \\
    \midrule
    $\mathrm{IV}$
    & $1$ & $0$
    & $\purple{n_{1011},\,n_{0100}}$
    & $1$ & $1$
    & $\purple{(1,2;1,1)}$ \\
    \midrule
    $\purple{\mathrm{V}_{a}}$
    & \multirow{2}{*}{$1$} & \multirow{2}{*}{$1$}
    & $\purple{n_{1100},\,n_{0010},\,n_{0001}}$
    & \multirow{2}{*}{$2$} & \multirow{2}{*}{$1$--$2$}
    & $\purple{(1,1;2,3)}$ \\
    $\purple{\mathrm{V}_{b}}$
    & & &
    $\purple{n_{1010},\,n_{0100},\,n_{0001}}$
    & & &
    $\purple{(1,2;1,3)}$ \\
    \midrule
    $\mathrm{VI}$
    & $2$ & $0$
    & $\purple{n_{1000},\,n_{0100},\,n_{0010},\,n_{0001}}$
    & $1$ & $2$
    & $\purple{(1,2;3,4)}$ \\
    \bottomrule
  \end{tabular}
  \label{tab:k1_cov_classes}
\end{table}

\begin{table}[H]
  \centering
  \renewcommand{\arraystretch}{1.3}
  \setlength{\tabcolsep}{4pt}
    \caption{
    \justifying{Classification of the $2$-RDM covariance blocks by the
  structural profile $(\ell_+,\ell_0)$ and the incidence profile
  $\boldsymbol n=(n_\alpha)_\alpha$. The column ``Nonzero $n_\alpha$''
  lists only the nonzero components for the displayed representative;
  their multiplicities can be read from the representative. Subscripted
  variants share the structural profile of their parent but have distinct
  incidence profiles and coefficient values. Enumeration up to relabeling
  and the Hermitian symmetries yields the $36$ admissible patterns
  listed, together with $22$ impossible-charge patterns collected into
  class~$\mathrm{I}$, for which $\purple{{}^{2}B_{\vec q;\,\vec s}^{\vec p;\,\vec r}}=0$.
  The variance patterns are the special cases
  $\purple{(\vec r,\vec s)=(\vec p,\vec q)}$: diagonal (Class~$\mathrm{II}$), one
  shared mode (Class~$\mathrm{III}_{b}$), and disjoint
  (Class~$\mathrm{IV}_{c}$).}}
  \begin{tabular}{ccc l cc l}
    \hline
    Class & $\ell_+$ & $\ell_0$ & Nonzero $n_\alpha$
    & \#$\Gamma_K$ & RDM orders &
    Representative $\purple{(\vec p,\vec q;\vec r,\vec s)}$ \\
    \hline
    $\mathrm{I}$ & --- & ---
      & $\purple{n_{1000},n_{0110},n_{0100},n_{0010},n_{0001}}$
      & $0$ & --- & $\purple{(12,34;35,67)}$ \\
    \hline
    $\mathrm{II}$ & $0$ & $2$
      & $\purple{n_{1111}}$
      & $4$ & $0$--$2$ & $\purple{(12,12;12,12)}$ \\
    \hline
    $\mathrm{III}_{a}$
      & \multirow{2}{*}{$0$}
      & \multirow{2}{*}{$3$}
      & $\purple{n_{1111},n_{1100},n_{0011}}$
      & \multirow{2}{*}{$8$}
      & \multirow{2}{*}{$0$--$3$}
      & $\purple{(12,12;13,13)}$ \\
    $\mathrm{III}_{b}$
      & & &
      $\purple{n_{1111},n_{1010},n_{0101}}$
      & & & $\purple{(12,13;12,13)}$ \\
    \hline
    $\mathrm{IV}_{a}$
      & \multirow{3}{*}{$0$}
      & \multirow{3}{*}{$4$}
      & $\purple{n_{1100},n_{0011}}$
      & \multirow{3}{*}{$16$}
      & \multirow{3}{*}{$0$--$4$}
      & $\purple{(12,12;34,34)}$ \\
    $\mathrm{IV}_{b}$
      & & &
      $\purple{n_{1100},n_{1010},n_{0101},n_{0011}}$
      & & & $\purple{(12,13;24,34)}$ \\
    $\mathrm{IV}_{c}$
      & & &
      $\purple{n_{1010},n_{0101}}$
      & & & $\purple{(12,34;12,34)}$ \\
    \hline
    $\mathrm{V}_{a}$
      & \multirow{3}{*}{$1$}
      & \multirow{3}{*}{$1$}
      & $\purple{n_{1111},n_{1110},n_{0001}}$
      & \multirow{3}{*}{$2$}
      & \multirow{3}{*}{$1$--$2$}
      & $\purple{(12,12;12,13)}$ \\
    $\mathrm{V}_{b}$
      & & &
      $\purple{n_{1110},n_{1101},n_{0011}}$
      & & & $\purple{(12,12;13,23)}$ \\
    $\mathrm{V}_{c}$
      & & &
      $\purple{n_{1110},n_{1011},n_{0101}}$
      & & & $\purple{(12,13;12,23)}$ \\
    \hline
    $\mathrm{VI}_{a}$
      & \multirow{5}{*}{$1$}
      & \multirow{5}{*}{$2$}
      & $\purple{n_{1111},n_{1100},n_{0010},n_{0001}}$
      & \multirow{5}{*}{$4$}
      & \multirow{5}{*}{$1$--$3$}
      & $\purple{(12,12;13,14)}$ \\
    $\mathrm{VI}_{b}$
      & & &
      $\purple{n_{1110},n_{1100},n_{0011},n_{0001}}$
      & & & $\purple{(12,12;13,34)}$ \\
    $\mathrm{VI}_{c}$
      & & &
      $\purple{n_{1111},n_{1010},n_{0100},n_{0001}}$
      & & & $\purple{(12,13;12,14)}$ \\
    $\mathrm{VI}_{d}$
      & & &
      $\purple{n_{1110},n_{1010},n_{0101},n_{0001}}$
      & & & $\purple{(12,13;12,34)}$ \\
    $\mathrm{VI}_{e}$
      & & &
      $\purple{n_{1110},n_{1000},n_{0101},n_{0011}}$
      & & & $\purple{(12,13;14,34)}$ \\
    \hline
    $\purple{\mathrm{VII}_{a}}$
      & \multirow{4}{*}{$1$}
      & \multirow{4}{*}{$3$}
      & $\purple{n_{1100},n_{0011},n_{0010},n_{0001}}$
      & \multirow{4}{*}{$8$}
      & \multirow{4}{*}{$1$--$4$}
      & $\purple{(12,12;34,35)}$ \\
    $\purple{\mathrm{VII}_{b}}$
      & & &
      $\purple{n_{1100},n_{1010},n_{0101},n_{0010},n_{0001}}$
      & & & $\purple{(12,13;24,35)}$ \\
    $\purple{\mathrm{VII}_{c}}$
      & & &
      $\purple{n_{1100},n_{1010},n_{0100},n_{0011},n_{0001}}$
      & & & $\purple{(12,13;24,45)}$ \\
    $\mathrm{VII}_{d}$
      & & &
      $\purple{n_{1010},n_{0101},n_{0100},n_{0001}}$
      & & & $\purple{(12,34;12,35)}$ \\
    \hline
    $\purple{\mathrm{VIII}_{a}}$
      & \multirow{4}{*}{$2$}
      & \multirow{4}{*}{$0$}
      & $\purple{n_{1110},n_{0001}}$
      & \multirow{4}{*}{$1$}
      & \multirow{4}{*}{$2$}
      & $\purple{(12,12;12,34)}$ \\
    $\purple{\mathrm{VIII}_{b}}$
      & & &
      $\purple{n_{1110},n_{1011},n_{0100},n_{0001}}$
      & & & $\purple{(12,13;12,24)}$ \\
    $\purple{\mathrm{VIII}_{c}}$
      & & &
      $\purple{n_{1110},n_{1000},n_{0111},n_{0001}}$
      & & & $\purple{(12,13;13,34)}$ \\
    $\purple{\mathrm{VIII}_{d}}$
      & & &
      $\purple{n_{1011},n_{1000},n_{0111},n_{0100}}$
      & & & $\purple{(12,34;13,13)}$ \\
    \hline
    $\purple{\mathrm{IX}_{a}}$
      & \multirow{5}{*}{$2$}
      & \multirow{5}{*}{$1$}
      & $\purple{n_{1110},n_{1100},n_{0010},n_{0001}}$
      & \multirow{5}{*}{$2$}
      & \multirow{5}{*}{$2$--$3$}
      & $\purple{(12,12;13,45)}$ \\
    $\purple{\mathrm{IX}_{b}}$
      & & &
      $\purple{n_{1110},n_{1010},n_{0100},n_{0001}}$
      & & & $\purple{(12,13;12,45)}$ \\
    $\purple{\mathrm{IX}_{c}}$
      & & &
      $\purple{n_{1111},n_{1000},n_{0100},n_{0010},n_{0001}}$
      & & & $\purple{(12,13;14,15)}$ \\
    $\mathrm{IX}_{d}$
      & & &
      $\purple{n_{1110},n_{1000},n_{0101},n_{0010},n_{0001}}$
      & & & $\purple{(12,13;14,35)}$ \\
    $\mathrm{IX}_{e}$
      & & &
      $\purple{n_{1110},n_{1000},n_{0100},n_{0011},n_{0001}}$
      & & & $\purple{(12,13;14,45)}$ \\
    \hline
    $\purple{\mathrm{X}_{a}}$
      & \multirow{5}{*}{$2$}
      & \multirow{5}{*}{$2$}
      & $\purple{n_{1100},n_{0010},n_{0001}}$
      & \multirow{5}{*}{$4$}
      & \multirow{5}{*}{$2$--$4$}
      & $\purple{(12,12;34,56)}$ \\
    $\purple{\mathrm{X}_{b}}$
      & & &
      $\purple{n_{1100},n_{1010},n_{0100},n_{0010},n_{0001}}$
      & & & $\purple{(12,13;24,56)}$ \\
    $\purple{\mathrm{X}_{c}}$
      & & &
      $\purple{n_{1100},n_{1000},n_{0100},n_{0011},n_{0010},n_{0001}}$
      & & & $\purple{(12,13;45,46)}$ \\
    $\mathrm{X}_{d}$
      & & &
      $\purple{n_{1010},n_{0100},n_{0001}}$
      & & & $\purple{(12,34;12,56)}$ \\
    $\mathrm{X}_{e}$
      & & &
      $\purple{n_{1010},n_{1000},n_{0101},n_{0100},n_{0010},n_{0001}}$
      & & & $\purple{(12,34;15,36)}$ \\
    \hline
    $\mathrm{XI}$ & $3$ & $0$
      & $\purple{n_{1011},n_{1000},n_{0100},n_{0010},n_{0001}}$
      & $1$ & $3$ & $\purple{(12,34;15,16)}$ \\
    \hline
    $\purple{\mathrm{XII}_{a}}$
      & \multirow{2}{*}{$3$}
      & \multirow{2}{*}{$1$}
      & $\purple{n_{1010},n_{1000},n_{0100},n_{0010},n_{0001}}$
      & \multirow{2}{*}{$2$}
      & \multirow{2}{*}{$3$--$4$}
      & $\purple{(12,34;15,67)}$ \\
    $\purple{\mathrm{XII}_{b}}$
      & & &
      $\purple{n_{1000},n_{0100},n_{0010},n_{0011},n_{0001}}$
      & & & $\purple{(12,34;56,67)}$ \\
    \hline
    $\mathrm{XIII}$ & $4$ & $0$
      & $\purple{n_{1000},n_{0100},n_{0010},n_{0001}}$
      & $1$ & $4$ & $\purple{(12,34;56,78)}$ \\
    \hline
  \end{tabular}
  \label{tab:k2_cov_classes}
\end{table}

\begin{table}[H]
  \centering
  \renewcommand{\arraystretch}{1.3}
  \setlength{\tabcolsep}{4pt}
  \caption{
  \justifying{Class-independent dictionary expressing the covariance
coefficients $\Gamma_K$ in terms of the sector coefficients $A_W$ through
Eq.~\eqref{eq:app_Gamma_def}, for $K=\{c_1,\dots,c_{\abs{K}}\}\subseteq
L_0$. The entry occurs whenever $\abs{K}\le\ell_0$; the $A_W$ are those of
the class at hand. Class~$\mathrm{I}$ carries no coefficient.}}
  \begin{tabular}{ccc l cc l}
    \hline
    Class & $\ell_+$ & $\ell_0$ & Nonzero $n_\alpha$
    & \#$\Gamma_K$ & RDM orders &
    Representative $\purple{(i,j;\vec r,\vec s)}$ \\
    \hline
    $\mathrm{I}$ & --- & ---
      & $\purple{n_{1011},n_{0110},n_{0001}}$
      & $0$ & --- & $\purple{(1,2;12,13)}$ \\
    \hline
    $\mathrm{II}$ & $0$ & $2$
      & $\purple{n_{1111},n_{0011}}$
      & $4$ & $0$--$2$ & $\purple{(1,1;12,12)}$ \\
    \hline
    $\mathrm{III}_{a}$
      & \multirow{2}{*}{$0$}
      & \multirow{2}{*}{$3$}
      & $\purple{n_{1100},n_{0011}}$
      & \multirow{2}{*}{$8$}
      & \multirow{2}{*}{$0$--$3$}
      & $\purple{(1,1;23,23)}$ \\
    $\mathrm{III}_{b}$
      & & &
      $\purple{n_{1010},n_{0101},n_{0011}}$
      & & & $\purple{(1,2;13,23)}$ \\
    \hline
    $\mathrm{IV}$ & $1$ & $0$
      & $\purple{n_{1011},n_{0111}}$
      & $1$ & $1$ & $\purple{(1,2;12,12)}$ \\
    \hline
    $\mathrm{V}_{a}$
      & \multirow{4}{*}{$1$}
      & \multirow{4}{*}{$1$}
      & $\purple{n_{1111},n_{0010},n_{0001}}$
      & \multirow{4}{*}{$2$}
      & \multirow{4}{*}{$1$--$2$}
      & $\purple{(1,1;12,13)}$ \\
    $\mathrm{V}_{b}$
      & & &
      $\purple{n_{1110},n_{0011},n_{0001}}$
      & & & $\purple{(1,1;12,23)}$ \\
    $\mathrm{V}_{c}$
      & & &
      $\purple{n_{1010},n_{0111},n_{0001}}$
      & & & $\purple{(1,2;12,23)}$ \\
    $\mathrm{V}_{d}$
      & & &
      $\purple{n_{1011},n_{0100},n_{0011}}$
      & & & $\purple{(1,2;13,13)}$ \\
    \hline
    $\mathrm{VI}_{a}$
      & \multirow{4}{*}{$1$}
      & \multirow{4}{*}{$2$}
      & $\purple{n_{1100},n_{0011},n_{0010},n_{0001}}$
      & \multirow{4}{*}{$4$}
      & \multirow{4}{*}{$1$--$3$}
      & $\purple{(1,1;23,24)}$ \\
    $\mathrm{VI}_{b}$
      & & &
      $\purple{n_{1010},n_{0101},n_{0010},n_{0001}}$
      & & & $\purple{(1,2;13,24)}$ \\
    $\mathrm{VI}_{c}$
      & & &
      $\purple{n_{1010},n_{0100},n_{0011},n_{0001}}$
      & & & $\purple{(1,2;13,34)}$ \\
    $\mathrm{VI}_{d}$
      & & &
      $\purple{n_{1000},n_{0100},n_{0011}}$
      & & & $\purple{(1,2;34,34)}$ \\
    \hline
    $\mathrm{VII}_{a}$
      & \multirow{2}{*}{$2$}
      & \multirow{2}{*}{$0$}
      & $\purple{n_{1110},n_{0010},n_{0001}}$
      & \multirow{2}{*}{$1$}
      & \multirow{2}{*}{$2$}
      & $\purple{(1,1;12,34)}$ \\
    $\mathrm{VII}_{b}$
      & & &
      $\purple{n_{1011},n_{0100},n_{0010},n_{0001}}$
      & & & $\purple{(1,2;13,14)}$ \\
    \hline
    $\mathrm{VIII}_{a}$
      & \multirow{3}{*}{$2$}
      & \multirow{3}{*}{$1$}
      & $\purple{n_{1100},n_{0010},n_{0001}}$
      & \multirow{3}{*}{$2$}
      & \multirow{3}{*}{$2$--$3$}
      & $\purple{(1,1;23,45)}$ \\
    $\mathrm{VIII}_{b}$
      & & &
      $\purple{n_{1010},n_{0100},n_{0010},n_{0001}}$
      & & & $\purple{(1,2;13,45)}$ \\
    $\mathrm{VIII}_{c}$
      & & &
      $\purple{n_{1000},n_{0100},n_{0011},n_{0010},n_{0001}}$
      & & & $\purple{(1,2;34,35)}$ \\
    \hline
    $\mathrm{IX}$ & $3$ & $0$
      & $\purple{n_{1000},n_{0100},n_{0010},n_{0001}}$
      & $1$ & $3$ & $\purple{(1,2;34,56)}$ \\
    \hline
  \end{tabular}%
  \label{tab:k12_cov_classes}
\end{table}

\begin{table}[t]
\centering
\renewcommand{\arraystretch}{1.5}
  \caption{
  \justifying{Classification of the $(1,2)$ cross-covariance blocks, in the
  format of Table~\ref{tab:k2_cov_classes}. Enumeration up to relabeling
  and the conjugation symmetry
  $\purple{(i,j;\vec r,\vec s)\to(j,i;\vec s,\vec r)}$ yields the $18$ admissible
  patterns listed, together with $6$ impossible-charge patterns collected
  into class~$\mathrm{I}$, for which $\operatorname{Cov}_\rho
  =-\purple{{}^{1}D_{\rho;\,j}^{i}}\purple{{}^{2}D_{\rho;\,\vec r}^{\vec s}}$. No variance pattern
  exists in this family.}}
\begin{tabular}{c l c}
\hline
$\abs{K}$ & $\Gamma_K$ & profiles \\
\hline
$0$ & $A_\emptyset$ & all \\
$1$ & $A_{\{c_1\}}-A_\emptyset$ & $\ell_0\ge1$ \\
$2$ & $A_{\{c_1,c_2\}}
      -A_{\{c_1\}}-A_{\{c_2\}}
      +A_\emptyset$ & $\ell_0\ge2$ \\
$3$ & $A_{\{c_1,c_2,c_3\}}
      -\sum_{i<j}A_{\{c_i,c_j\}}
      +\sum_{i}A_{\{c_i\}}
      -A_\emptyset$ & $\ell_0\ge3$ \\
$4$ & $A_{\{c_1,\dots,c_4\}}
      -\sum_{i<j<k}A_{\{c_i,c_j,c_k\}}
      +\sum_{i<j}A_{\{c_i,c_j\}}
      -\sum_{i}A_{\{c_i\}}
      +A_\emptyset$ & $\ell_0=4$ \\
\hline
\end{tabular}
\label{tab:gamma_dictionary}
\end{table}

\subsection{Numerical validation by Monte Carlo sampling}
\label{subsec:numerical_validation}

In this section, we verify the correctness of the analytical evaluation developed in Sec.~\ref{sec:verify} by comparing its covariance predictions with the direct simulation of orbital-rotation shadows.
The preceding subsections reduce the evaluation of estimator covariances
to the classification of moment operators by charge and incidence
profiles, the computation of sector coefficients through Grassmannian
moments, and the contraction of the resulting covariance coefficients
with RDM elements. We now test this procedure against a direct
Monte Carlo simulation of the orbital-rotation shadow measurement process,
using representative entries of the $1$- and $2$-RDM estimators.

For system sizes with $N=8,9$, we draw a Haar-random pure state in
$\mathcal H_4$ and keep it fixed throughout the simulation.
We select $43$ estimator entries and compute their
$43\times43$ unbiased sample covariance matrix from $10^6$
independent measurement shots.
This matrix includes \red{the $1$-RDM, $2$-RDM, and cross-covariance sectors}.

Independently of the measurement samples, we evaluate the theoretical
covariance for all $1849$ ordered estimator pairs.
For each pair, the theoretical covariance calculation routine in
Sec.~\ref{subsec:app_pseudocode} is applied to the actual index tuples
to obtain the sector coefficients and their collected covariance
coefficients. We then use Eq.~\eqref{eq:app_cov_final}, together with the exact RDMs of the same state\red{,} to obtain the predicted covariance.
The comparison between the theoretical and sampled covariance matrices therefore provides a combined check of the implementation of the charge selection rules, the evaluation of Grassmannian moments, the collection of sector coefficients, the fermionic merge signs, and the final RDM contraction.

Figure~\ref{fig:app_covariance_validation} displays the absolute value of each entry in the theoretical and sampled covariance matrices, together with the absolute value of the difference between each pair of corresponding complex entries.
The theoretical and sampled matrices exhibit the same block structure.
For $(N,\eta)=(8,4)$, the maximum and root-mean-square discrepancies
over all $1849$ entries are $1.22\times10^{-2}$ and
$3.19\times10^{-3}$, respectively.
Repeating the test for $(N,\eta)=(9,4)$ gives corresponding
discrepancies of $1.05\times10^{-2}$ and $3.50\times10^{-3}$.
These results are consistent with sampling fluctuations and
provide a numerical consistency check of the analytical covariance evaluation.

\begin{figure}[t]
\centering
\includegraphics[width=0.90\textwidth]{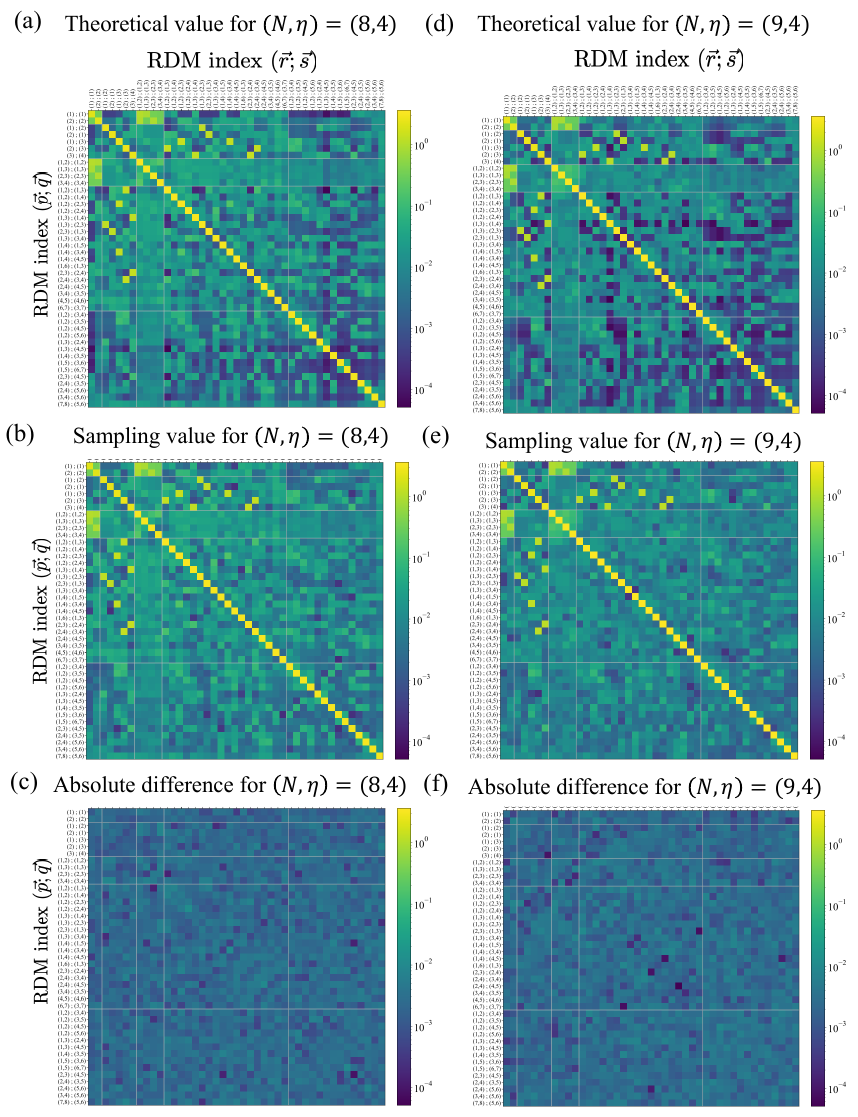}
\caption{%
\justifying{
Numerical validation of the analytical covariance
evaluation for $(N,\eta)=(8,4)$ \red{and $(9,4)$} using $10^6$ independent
\red{measurement shots for each system}.
Panels (a)--(c) show the results for $(N,\eta)=(8,4)$,
and panels (d)--(f) show those for $(N,\eta)=(9,4)$.
Panels (a) and (d) display the absolute value of each entry
in the theoretical covariance matrix, \red{computed using the procedure described} in Sec.~\ref{sec:verify}.
Panels (b) and (e) display the absolute value of each entry
in the unbiased sample covariance matrix obtained from the
measurement shots.
Panels (c) and (f) display the absolute differences between
the corresponding complex entries of the theoretical and
sample covariance matrices.
\red{For each system size, both} matrices refer to the same fixed Haar-random pure state.
The rows and columns enumerate the same $43$ representative
$1$- and $2$-RDM estimators, with white lines separating the
estimator blocks.}
}
\label{fig:app_covariance_validation}
\end{figure}

\end{revision}

\end{document}